\documentclass[reqno]{amsart}
\usepackage[utf8]{inputenc}
\usepackage{amsthm,amsfonts,amstext,amssymb,mathrsfs,amsmath,latexsym,mathtools}
\usepackage{cite}
\usepackage{graphicx}
\usepackage{hyperref}
\usepackage{comment}
\usepackage{enumerate}
\usepackage{caption}
\usepackage{overpic}
\usepackage{enumitem}
\usepackage[super]{nth}
\usepackage[dvips]{color}

\newtheorem{thm}{Theorem}[section]
\newtheorem{prop}[thm]{Proposition}
\newtheorem{lem}[thm]{Lemma}
\newtheorem{cor}[thm]{Corollary}

\theoremstyle{definition}
\newtheorem{definition}[thm]{Definition}

\numberwithin{equation}{section}

\theoremstyle{remark}
\newtheorem{remark}{Remark}[section]

\DeclareMathOperator{\ai}{Ai}
\DeclareMathOperator{\re}{Re}
\DeclareMathOperator{\im}{Im}
\DeclareMathOperator{\disc}{Discr}
\DeclareMathOperator{\const}{const}
\DeclareMathOperator{\result}{Resultant}
\DeclareMathOperator{\trace}{Tr}
\DeclareMathOperator{\area}{Area}
\DeclareMathOperator{\res}{Res}
\DeclareMathOperator{\supp}{supp}

\newcommand{\Boh}{\mathcal{O}}

\newcommand{\R}{{\mathbb R}}
\newcommand{\C}{{\mathbb C}}

\newcommand{\N}{{\mathbb N}}

\newcommand{\D}{\mathbb D}

\newcommand{\ga}{\gamma}
\newcommand{\Ga}{\Gamma}

\newcommand{\la}{\lambda}

\newcommand{\de}{\delta}

\newcommand\restr[2]{{
  \left.\kern-\nulldelimiterspace 
  #1 
  \vphantom{\big|} 
  \right|_{#2} 
  }}

 \newcommand\pder[2]{{
  \frac{\partial #1}{\partial #2} 
  }}

\title{The mother body phase transition in the normal matrix model}

\author[P.~Bleher]{Pavel M. Bleher}

\address[PB]{Department of Mathematical Sciences, Indiana University-Purdue University Indianapolis, 402 N. Blackford St., Indianapolis, IN 46202, USA.}

\email{pbleher@iupui.edu}

\author[G.~Silva]{Guilherme L.~F.~Silva}

\address[GS]{KU Leuven, Department of Mathematics, Celestijnenlaan 200B bus 2400, B-3001 Leuven, Belgium}

\email{guilherme.silva@wis.kuleuven.be}

\date{}

 

\begin{document}

\begin{abstract}
The normal matrix model with algebraic potential has gained a lot of attention recently, partially in virtue of its connection to several other topics as 
quadrature domains, inverse potential problems and the Laplacian growth.

In this present paper we consider the normal matrix model with cubic plus linear potential. In order to regularize the model, we follow Elbau \& Felder and 
introduce a cut-off. In the large size limit, the eigenvalues of the model accumulate uniformly within a certain domain $\Omega$ that we determine explicitly by finding the 
rational parametrization of its boundary. 

We also study in details the mother body problem associated to $\Omega$. It turns out that the mother body measure $\mu_*$ displays a novel phase transition that we call the {\it 
mother body phase transition}: although $\partial \Omega$ evolves analytically, the mother body measure undergoes a ``one-cut  to three-cut'' phase 
transition. 

To construct the mother body measure, we define a quadratic differential $\varpi$ on the associated spectral curve, and embed $\mu_*$ into its critical graph. Using 
deformation techniques for quadratic differentials, we are able to get precise information on $\mu_*$. In particular, this allows us to determine the phase diagram for 
the mother body phase transition explicitly.

Following previous works of Bleher \& Kuijlaars and Kuijlaars \& López, we consider multiple orthogonal polynomials associated with the normal matrix model. Applying the 
Deift-Zhou nonlinear steepest descent method to the associated Riemann-Hilbert problem, we obtain strong asymptotic formulas for these polynomials. Due to the presence of the 
linear term in the potential, there are no rotational symmetries in the model. This makes the construction of the associated $g$-functions significantly more involved, and the 
critical graph of $\varpi$ becomes the key technical tool in this analysis as well. 

\end{abstract}

\maketitle
\tableofcontents

\section{Introduction}

We are interested in the eigenvalues of the normal matrix model given by the probability distribution
\begin{equation}\label{normal_matrix_ensemble}
d\pi_n(M)=\frac{1}{\widetilde Z_n}e^{-n\trace \mathcal V(M)}dM,
\end{equation}
where $M$ is an $n\times n$ normal matrix and $\mathcal V$ is a given function of $M$. Its induced joint probability 
distribution on the eigenvalues $\lambda=(\lambda_1,\hdots,\lambda_n)\in \C^n$ is given explicitly by
\begin{equation}\label{normal_matrix_ensemble_eigenvalues}
d\pi_n(\la)=\frac{1}{Z_n}\prod_{j<k}|\lambda_j-\lambda_k|^2 e^{-n\sum_{j=1}^n \mathcal V(\lambda_j)} d\lambda,
\end{equation}
where $d\lambda$ is the Lebesgue measure on $\C^n$ and $Z_n$ is the corresponding partition function 
\cite{chau_zaboronsky_normal_matrix_model,elbau_felder_density_eigenvalues_normal_matrices}.
 
This model has been studied in the literature for different choices of the potential $\mathcal V$ and under various perspectives
\cite{ameur_hedenmalm_makarov_fluctuations,ameur_hedenmalm_makarov_ward_identities,hedenmalm_makarov_coulomb_gas,chau_zaboronsky_normal_matrix_model,lee_teodorescu_wiegmann,
marchetti_pereira_veneziani,marchetti_pereira_veneziani_2,roman_riser_thesis}. Of particular interest is the choice 
\begin{equation}\label{entire_potential}
\mathcal V(z)=\frac{1}{t_0}(|z|^2-V(z)-\overline{V(z)}), \quad z\in \C, \quad V(z)=\sum_{k=1}^{d+1} \frac{t_k}{k}z^k,\quad t_{d+1}\neq 0.
\end{equation}

As formally observed by Kostov, Krichever, Mineev-Weinstein, Wiegmann and Zabrodin \cite{kostov_et_al_tau_function}, in this situation the eigenvalues should accumulate on a 
domain $\Omega=\Omega(t_0,t_1,\hdots,t_d)$, whose boundary $\partial \Omega$ evolves in time $t_0>0$ according to the Laplacian growth model with given harmonic moments
\begin{equation}\label{harmonic_moments_area_integral}
\area(\Omega)=\pi t_0,\qquad -\frac{1}{\pi}\iint_{\C\setminus \Omega}\frac{dA(z)}{z^k}= t_k,\quad k=1,2,3,\hdots,
\end{equation}
where we set $t_{j}=0$, for $j\geq d+2$, and $dA$ is the Lebesgue measure on $\C$.

However, the model \eqref{normal_matrix_ensemble} for $\mathcal V$ given by \eqref{entire_potential} is in general purely formal. If $d\geq 2$, the density in 
\eqref{normal_matrix_ensemble} is not integrable, hence the normal matrix model is ill-defined. To overcome this essential issue, Elbau and Felder 
\cite{elbau_felder_density_eigenvalues_normal_matrices} proposed to consider a {\it cut off} model. Instead of integrating 
\eqref{normal_matrix_ensemble} over the whole set of normal matrices, they consider \eqref{normal_matrix_ensemble} as a distribution over normal matrices whose eigenvalues 
are constrained to lie within a fixed bounded domain $D\subset \C$. In this setup, the model becomes well-defined and the eigenvalue density 
\eqref{normal_matrix_ensemble_eigenvalues} can be rewritten as
\begin{equation}\label{normal_matrix_ensemble_eigenvalues_cut_off}
d\pi_n(\la)=\frac{1}{Z_n}\prod_{j<k}|\lambda_j-\lambda_k|^2 \prod_{j=1}^n \chi_D(\lambda_j) e^{-n \mathcal V(\lambda_j)} d\lambda,
\end{equation}
where $\chi_D$ is the characteristic function of $D$. Let
\begin{equation*}
d\mu_\lambda(z)=\frac{1}{n}\, \sum_{j=1}^n \de(z-\lambda_j)dA(z)
\end{equation*}
be a probability atomic measure on $D$ with atoms at the points $\la_j$ of an eigenvalue configuration $\la$ and
\begin{equation}
H(\mu)=\iint_{x\not=z} \log\frac{1}{|z-x|}d\mu(x)d\mu(z)+\int \mathcal V(z)d\mu(z)
\end{equation}
the Coulomb gas Hamiltonian, where $\mu$ is an arbitrary probability measure on  $D$, giving a distribution of
the Coulomb gas particles. Then formula \eqref{normal_matrix_ensemble_eigenvalues_cut_off} can be written as 
\begin{equation*}
d\pi_n(\la)=\frac{1}{Z_n}e^{-n^2 H(\mu_\lambda)} d\lambda,
\end{equation*}
The factor $n^2$ in the exponent suggests that, as $n\to\infty$, the measure $d\pi_n(\la)$ concentrates in a
shrinking neighborhood of the equilibrium measure $\mu_0$, which is the probability measure minimizing 
\begin{equation}\label{energy_functional}
\iint \log\frac{1}{|z-x|}d\mu(x)d\mu(z)+\int \mathcal V(z)d\mu(z)
\end{equation}
over all probability measures supported on $D$. This concentration phenomenon has been proved rigorously for instance in 
\cite{elbau_felder_density_eigenvalues_normal_matrices,hedenmalm_makarov_coulomb_gas} under different assumptions.

Under the additional requirements that
\begin{enumerate}
  \item the boundary $\partial D$ of the cut off is sufficiently smooth,
  \item the potential $\mathcal V$ has exactly one minimum 
in the cut off $D$, and 
  \item the time parameter $t_0$ is \textit{sufficiently small}, 
\end{enumerate}
Elbau and Felder proved that the (unique) probability measure on $D$ minimizing \eqref{energy_functional} has the form
\begin{equation}\label{limiting_measure_eigenvalues}
d\mu_0(z)=\frac{1}{\pi t_0}\chi_{\Omega}(z)dA(z),
\end{equation}
where $\chi_{\Omega}$ is the characteristic function of a simply connected domain $\Omega=\Omega(t_0,V)$ contained in $D$, whose boundary $\partial \Omega$ is a 
{\it polynomial curve} of degree $d$: there exists a rational function of the form
\begin{equation}\label{parametrization_polynomial_curve_general}
h(w)=rw + a_0 + \frac{a_1}{w}+\cdots + \frac{a_d}{w^d}, \quad w\in \C,\quad r>0, \ a_d\neq 0,
\end{equation}
which is injective on the boundary of the unit disc $\D$ and such that $\partial \Omega = h(\partial \D)$.
Moreover, $h$ gives a conformal map of $\C\setminus \D$ onto  $\C\setminus\Omega$. Furthermore, still for sufficiently small time $t_0$, Elbau and Felder proved rigorously the 
connection of $\Omega$ with the Laplacian growth as in \eqref{harmonic_moments_area_integral}.

Concerning local statistics, Elbau \cite{elbau_thesis} pointed out that the eigenvalue distribution \eqref{normal_matrix_ensemble_eigenvalues_cut_off} is a bona fide 
determinantal point process with kernel
$$
K_n(z,w)=e^{-\frac{n}{2}(\mathcal V(z)+\mathcal V(w))}\sum_{j=0}^n \frac{q_{j,n}(z)\overline{q_{j,n}(w)}}{h_{n,j}},
$$
where $q_{j,n}(z)=z^j+\hdots$ are {\it planar orthogonal polynomials in the external field $\mathcal V$} (or simply POP's)
\begin{equation}\label{planar_orthogonality_conditions}
\iint_D q_{j,n}(z)\overline{q_{k,n}(z)}e^{-n\mathcal V(z)}dA(z)=h_{n,j}\delta_{j,k}, \quad j,k\in \N.
\end{equation}

A natural question is to understand the behavior of the polynomials $(q_{n,n})$, and in particular of their zeros, as $n\to \infty$. Elbau showed that any weak limit $\mu_*$ of 
the zero counting measures
\begin{equation*}
d\mu_n(z)=\frac{1}{n}\, \sum_{q_{n,n}(w)=0} \de(z-w)dA(z)
\end{equation*}
should be supported in $D$ and satisfy the mother body property for $\mu_0$, namely
\begin{equation}\label{mother body_problem_general}
\int \log|s-z|d\mu_0(s)=\int \log|s-z|d\mu_*(s),\quad z\in \C\setminus D.
\end{equation}

The goal of the present paper is to study in details the cubic plus linear model
\begin{equation}\label{generic_cubic_potential}
\mathcal V(z)=\frac{1}{t_0}(|z|^2-2\re V(z)),\qquad V(z)=\frac{z^3}{3}+t_1 z,\quad -\frac{3}{4}< t_1 < \frac{1}{4}\ ,
\end{equation}
for values of $t_0$ up to a critical time $t_{0,crit}=t_{0,crit}(t_1)$. The restriction on $t_1$ above comes from the corresponding potential $\mathcal V$: as a simple analysis 
shows, if either $t_1\leq -3/4$ or $t_1\geq 1/4$, then the potential $\mathcal V$ has no local minimum. Consequently, for any choice of cut off $D$, the corresponding eigenvalues 
should accumulate on the boundary of $D$ as $n\to \infty$, so that the limiting shape of eigenvalues $\Omega$ intersects $\partial D$ and is very sensitive to the precise choice 
of 
$D$.

For $t_1$ as in \eqref{generic_cubic_potential}, we are able to determine precisely the underlying phase diagram in the $(t_0,t_1)$-plane, as is shown in 
Figure~\ref{phase_diagram_full}. Given $t_1$ as above and for all positive values of $t_0$ up to the critical time $t_{0,crit}=t_{0,crit}(t_1)$, we find the parametrization $h$ 
in 
\eqref{parametrization_polynomial_curve_general} of the corresponding polynomial curve $\partial \Omega$, and consequently we obtain the associated limiting eigenvalue 
distribution \eqref{limiting_measure_eigenvalues} explicitly. When $t_0\to t_{0,crit}$, the boundary of $\partial \Omega$ creates either one (if $t_1>0$) or two (if $t_1<0$) 
cusps: we call this phenomenon the {\it cusp phase transition}. We are able to compute the critical curve $(t_{0,crit},t_1)$ explicitly.

\begin{figure}[t]
 \centering
 \begin{overpic}[scale=1]
{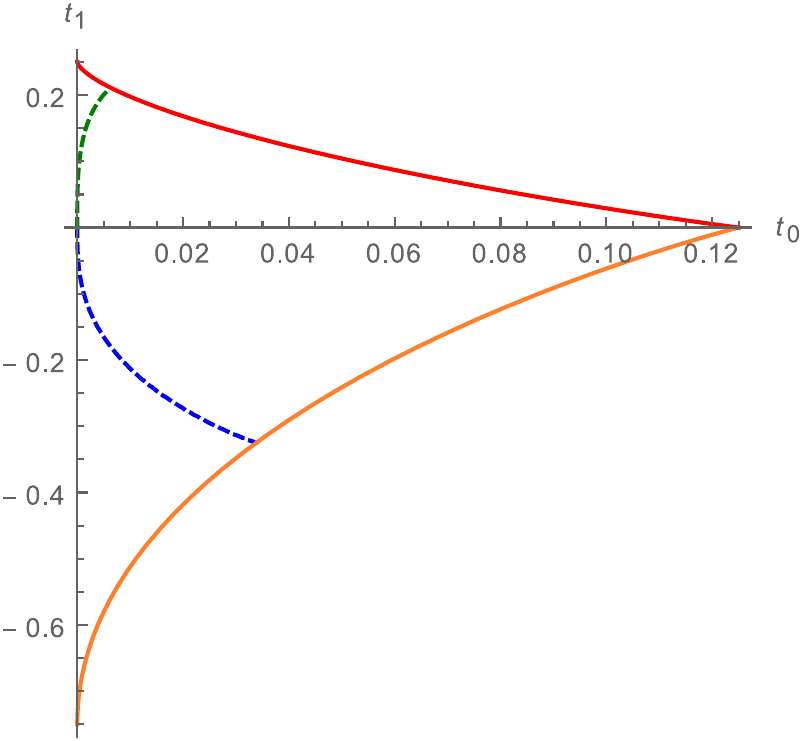}
 \end{overpic}
 \caption{Phase diagram on the $(t_0,t_1)$-plane: the solid curves are the pairs $(t_{0,crit},t_1)$, which correspond to the cusp phase transition. The dashed curves are the 
points of the form $(\tilde t_{0,crit},t_1)$ and they correspond to the mother body phase transition. We should remark that the dashed curves on the 
upper and lower half planes are \emph{not analytic continuation of each other}: the curve on the upper half plane 
is algebraic, whereas the one on the lower half plane is transcendental.}\label{phase_diagram_full}
\end{figure}

For all values of $(t_0,t_1)$ in our phase diagram, we also study the mother body equation \eqref{mother body_problem_general} in detail. We construct a measure $\mu_*$ satisfying 
\eqref{mother body_problem_general} and whose support consists of a finite union of analytic arcs. Furthermore, we find another critical time $\tilde t_{0,crit}=\tilde 
t_{0,crit}(t_1)<t_{0,crit}$ such that if $t_0<\tilde t_{0,crit}$, then $\supp\mu_*$ consists of one analytic arc, whereas for $t_0>\tilde t_{0,crit}$ the set $\supp\mu_*$ consists 
of three analytic arcs meeting at a common point. We call this transition the {\it mother body phase transition}.

The critical value $\tilde t_{0,crit}$ is depicted in Figure~\ref{phase_diagram_full}. We emphasize that the boundary of $\partial \Omega$ depends analytically on the parameters 
$(t_0,t_1)$ in the phase diagram. So what our results show is that the measure $\mu_*$ solving the mother body problem \eqref{mother body_problem_general} displays a 
phase 
transition that is not felt by $\partial \Omega$. We refer to Figure~\ref{phase_diagram_complete} for a visualization of this transition. To our knowledge, this is the first time 
such a phenomenon is described.

\begin{figure}[t]
 \centering
 \begin{overpic}[scale=1]
{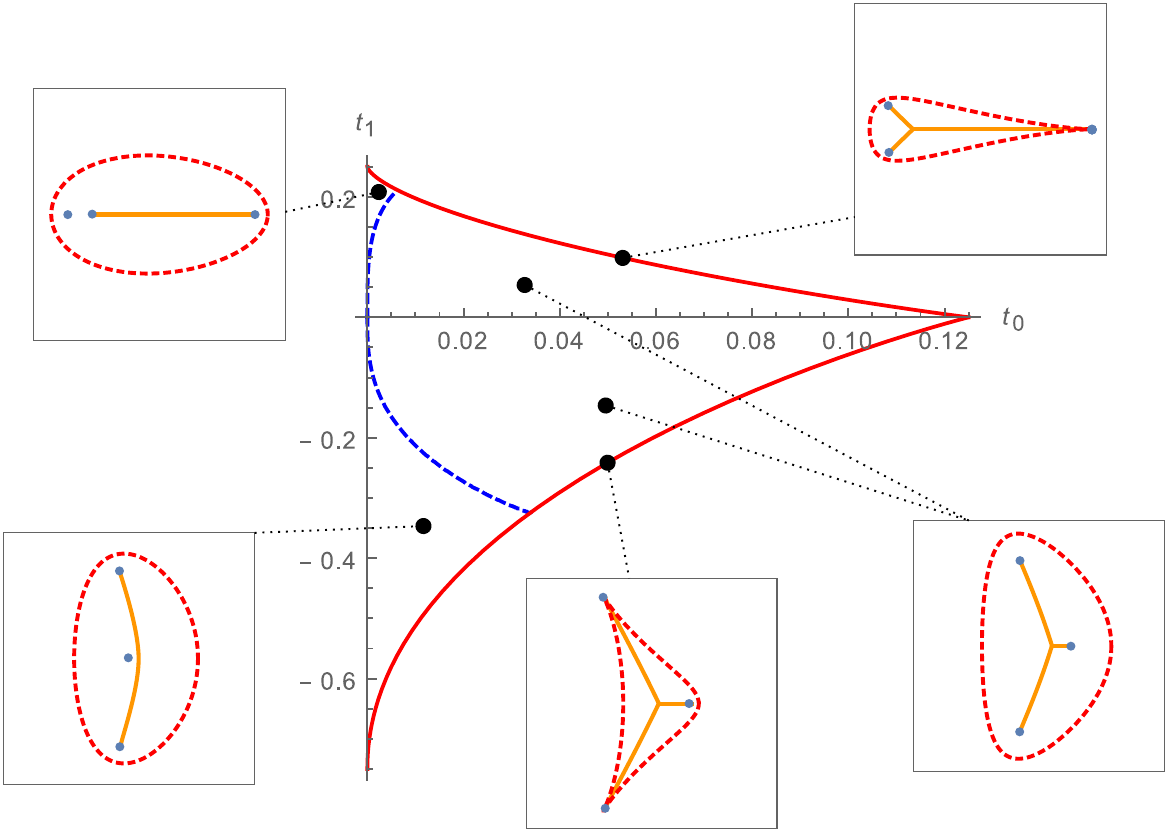}
 \end{overpic}
 \caption{The phase diagram and the various topological configurations of $\partial \Omega$ (dashed lines in each frame) and $\supp\mu_*$ ( 
solid lines in each frame). The dots inside each frame are the branch points of the associated spectral curve.}\label{phase_diagram_complete}
\end{figure}

The construction of the measure $\mu_*$ is, in our opinion, our main technical contribution. We construct a quadratic differential on the three-sheeted Riemann surface (a.k.a. 
spectral curve) associated with the model, and lift the problem \eqref{mother body_problem_general} to the trajectories of this quadratic differential. Following the 
recent work of Martínez-Finkelshtein and the second author \cite{martinez_silva},  we use deformation techniques to describe the critical graph of this quadratic differential.  
When we project some of these trajectories back to the complex plane, we recover the measure $\mu_*$. This 
critical graph displays some phase transitions; these transitions are determined by the critical value $\tilde t_{0,crit}$ and correspond to the phase transitions of 
$\supp\mu_*$.

We follow previous works of Bleher and Kuijlaars \cite{bleher_kuijlaars_normal_matrix_model} and Kujlaars and López \cite{kuijlaars_lopez_normal_matrix_model} and introduce a new 
sequence of polynomials $(P_{n,n})$ which has, at the heuristic level, the same asymptotic behavior as the sequence $(q_{n,n})$ in 
\eqref{planar_orthogonality_conditions}. We characterize this sequence 
$(P_{n,n})$ in terms of multiple orthogonality of non-hermitian type, and using the Deift-Zhou steepest descent method we obtain strong asymptotic formulas for these polynomials. 
As one of the consequences, we prove that the sequence of zero counting measures for $(P_{n,n})$ converges weakly to the measure $\mu_*$.

The case $t_1=0$ was studied before by Bleher and Kuijlaars \cite{bleher_kuijlaars_normal_matrix_model}, and it plays a substantial role here as well. Many auxiliary results 
require separate proofs depending whether $t_1<0$ or $t_1>0$, and a complete analysis would make the already lengthy paper much longer. So our main focus is in the case 
$t_1>0$, whose main results are stated in Section~\ref{section_statement_of_results}. The corresponding main results for $t_1<0$ are stated and discussed in Section 
\ref{section_negative_t}, but their proofs are analogous and not provided.

\section{Statement of main results}\label{section_statement_of_results}

\subsection{Phase diagram of the cubic model}\label{section_phase_diagram}

For the choice of potential \eqref{entire_potential} with $V$ as in \eqref{generic_cubic_potential}, the rational function $h$ in 
\eqref{parametrization_polynomial_curve_general} assumes the form
$$
h(w)=r w +a_0 +\frac{a_1}{w}+\frac{a_2}{w^2}.
$$

According to Elbau and Felder \cite[page~442]{elbau_felder_density_eigenvalues_normal_matrices}, the coefficients of $h$ should be related to the normal matrix model 
\eqref{normal_matrix_ensemble} with cubic potential \eqref{generic_cubic_potential} through the system of equations
\begin{equation}\label{system_elbau_felder}
\left\{
\begin{aligned}
& \frac{a_2}{r^2} =1, \\
& \frac{a_1}{r}-\frac{2a_0a_2}{r^2} = 0,\\
& a_0-\frac{a_0a_1}{r}-\frac{a_2(2a_1r-a_0^2)}{r^2} = t_1,\\
& r^2-a_1^2-2a_2^2  = t_0.
\end{aligned}
\right.
\end{equation}
This system of equations is obtained by computing the (expected) exterior harmonic moments of $\partial \Omega$ in terms of the rational function $h$. Solving in terms of $r$, 
it gives us
\begin{equation}\label{values_a1_a2}
a_1=2 r a_0,\quad a_2 =r^2, \quad a_0=\frac{1-4r^2+\delta \sqrt{(1-4r^2)^2-t_1}}{2}.
\end{equation}
where $\delta =\pm 1$. When $t_1\to 0$, we expect $a_0=0$ \cite{bleher_kuijlaars_normal_matrix_model}; thus $\delta=-1$, and $h$ reduces to
\begin{equation}\label{rational_parametrization}
h(w)=rw + a_0 + \frac{2a_0 r}{w}+\frac{r^2}{w^2},
\end{equation}
where
\begin{equation}\label{definition_a_0}
a_0 = a_0(t_0,t_1) =  \frac{1-4r^2-\sqrt{(1-4r^2)^2-4t_1}}{2},
\end{equation}
and $r=r(t_0,t_1)$ is to be determined. Using the values \eqref{values_a1_a2} in the last equation in \eqref{system_elbau_felder}, after a lengthy calculation we conclude that $r$ 
should be a root of the polynomial
\begin{multline}\label{definition_polynomial_p}
p(x) =  128 x^{10}-124 x^8+ (64 t_0 -16 t_1 + 36)x^6 \\ 
         + \left(16 t_1^2+8 t_1-28 t_0 -3\right)x^4 +  t_0(2-8 t_1)  x^2 + t_0^2
\end{multline}
When $t_1\to 0$, we compare again with the results by Kuijlaars and the first author \cite{bleher_kuijlaars_normal_matrix_model} to get that $r$ should be the smallest positive 
root of $p$. An analysis of the discriminant of $p$ then leads us to consider the domain $\mathcal F$ on the $(t_0,t_1)$-plane, bounded by the segments
\begin{align*}
(t_0,0), & \quad 0\leq t_0 \leq 1/8, \\
(0,t_1), & \quad 0\leq t_1\leq 1/4,
\end{align*}
and the critical curve $\Gamma_c$, parametrized by
\begin{equation}\label{definition_critical_curve_phase_transition}
\Gamma_c: \quad t_0=-6s^4+4s^3,\quad t_1=4s^3-3s^2+1/4, \quad 0 \leq s \leq 1/2. 
\end{equation}

\begin{prop}\label{proposition_definition_r}
For $(t_0,t_1)\in \mathcal F$, the polynomial $p$ in \eqref{definition_polynomial_p} has a smallest positive root $r=r(t_0,t_1)$, which is simple.
\end{prop}

Proposition \ref{proposition_definition_r} is proved in Section \ref{analysis_polynomial_p}. Theorem \ref{theorem_monotonicity_r} in Section \ref{analysis_polynomial_p} gives an 
important refinement of Proposition \ref{proposition_definition_r}.

The curve $\Gamma_c$ corresponds to the cusp phase transition for $t_1>0$. A plot of the region $\mathcal F$ and the critical curve $\Gamma_c$ are displayed in 
Figure~\ref{phase_diagram}. 

\begin{figure}[t]
 \centering
 \begin{overpic}[scale=1]
{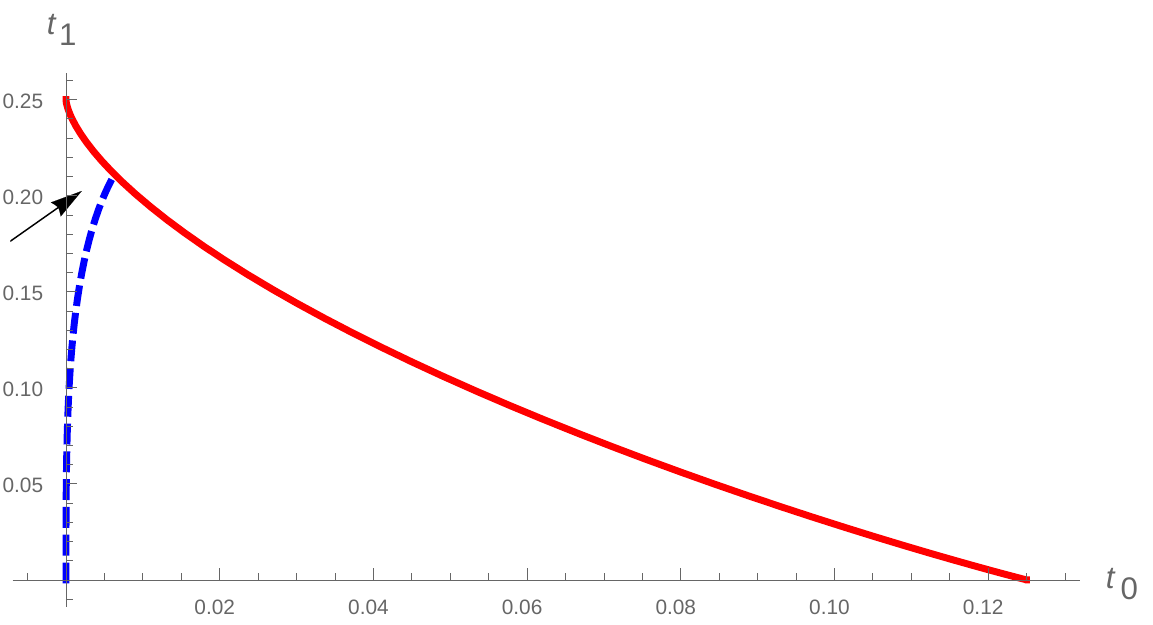}
\put(40,23){$\Gamma_c$}
\put(8,25){$\gamma_c$}
\put(25,15){$\mathcal F_1$}
\put(-3,32){$\mathcal F_2$}
 \end{overpic}
 \caption{Phase diagram on the $(t_0,t_1)$-plane: The solid curve is $\Gamma_c$ given in \eqref{definition_critical_curve_phase_transition} and the dashed curve is $\gamma_c$ 
given in \eqref{critical_curve_limiting_zero_distribution}. The region determined by $\Gamma_c$, $\gamma_c$ and the $t_0$-axis is the three-cut region 
$\mathcal F_1$, whereas the region between $\Gamma_c$, $\gamma_c$ and the $t_1$-axis is the one-cut region $\mathcal F_2$. Axes are scaled differently.}\label{phase_diagram}
\end{figure}
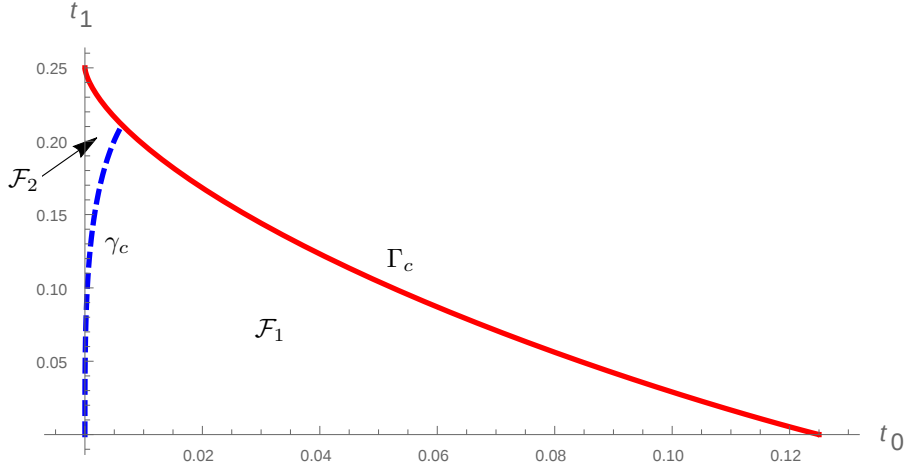

Since the function $r=r(t_0,t_1)$ is a simple root of the polynomial $p$, it is analytic with respect to both variables $t_0,t_1$, as long as $(t_0,t_1)\in \mathcal F$, and it is 
continuous up to the boundary of $\mathcal F$. The function $r$ plays a fundamental role in the rest of the paper.

\subsection{The limiting boundary of eigenvalues as a polynomial curve}

As it was heuristically explained in Section~\ref{section_phase_diagram}, the rational function $h$ in \eqref{rational_parametrization}, with coefficients $a_0$ and $r$ as in 
\eqref{definition_a_0} and Proposition~\ref{proposition_definition_r}, should give the parametrization of the polynomial curve $\partial \Omega$ for the potential 
\eqref{generic_cubic_potential}. This is rigorously established by our next result.

\begin{thm}\label{theorem_rational_parametrization_polynomial_curve}
For $(t_0,t_1)\in \mathcal F$ and $r,a_0$ given respectively by Proposition \ref{proposition_definition_r} and equation \eqref{definition_a_0}, the rational function $h$ is 
injective on $\overline \C\setminus \D$. The image $h(\partial \D)$ is an analytic curve whose interior is a simply connected domain $\Omega$ with area given by
$$
\area(\Omega)=\pi t_0.
$$

Moreover, the exterior harmonic moments of $\Omega$ with respect to any point $\zeta \in\Omega$ are given by
\begin{equation}\label{harmonic_moments}
\frac{1}{2\pi i}\int_{\partial \Omega} \frac{\overline z}{(z-\zeta)^k} dz=
\begin{cases}
t_1+\zeta^2, & k=1, \\
2\zeta, & k=2,\\
1,   & k=3, \\
0,   & k\geq 4.
\end{cases}
\end{equation}
\end{thm}

Theorem \ref{theorem_rational_parametrization_polynomial_curve} is proved in Section~\ref{section_analysis_rational_parametrization}.

The equality $\partial \Omega=h(\partial \D)$ with $h$ rational and injective on $\partial \D$ says that $\partial \Omega$ is a {\it polynomial curve}, in the sense of Elbau and 
Felder \cite{elbau_felder_density_eigenvalues_normal_matrices}. We refer to Figure~\ref{chain_polynomial_curves} for examples.

\begin{figure}[t]
\begin{overpic}[scale=1]{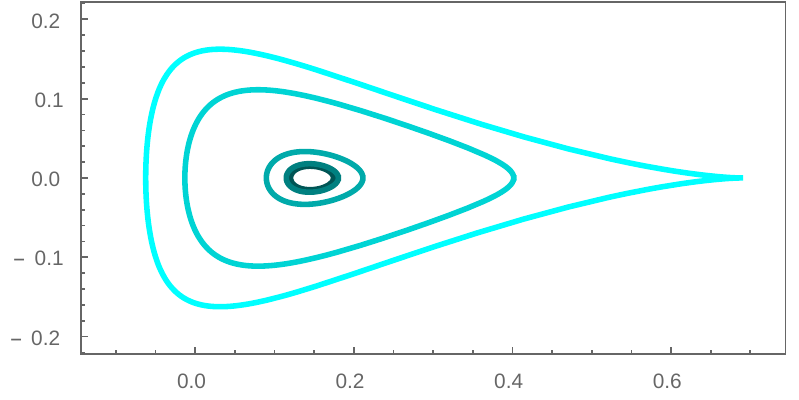}
\end{overpic}
\caption{The boundary $\partial \Omega$ corresponding to $t_1=\frac{1}{8}$ and $t_0=\frac{1}{2500}, \frac{1}{1700}, \frac{1}{500}, \frac{1}{50}, \frac{5}{128}$ (from 
dark to bright color, respectively). For the pair $(t_0,t_1)=(\frac{5}{128},\frac{1}{8})$, which belongs to the critical curve $\Gamma_c$, a cusp is created at the boundary. 
Numerical 
output.}\label{chain_polynomial_curves}
\end{figure}

If $0\in \Omega$ and $k\geq 3$, then the integrals in \eqref{harmonic_moments_area_integral} are convergent for $\zeta=0$ and Green's Theorem applied to $\C\setminus\Omega$ gives 
us the formula
$$
-\frac{1}{\pi}\iint_{\C\setminus\Omega} \frac{dA(z)}{z^k}=\frac{1}{2\pi i}\int_{\partial \Omega} \frac{\overline z}{z^k}\; dz,
$$
leading to the connection previously mentioned in \eqref{harmonic_moments_area_integral}.

For a measure $\nu$ on the complex plane, we denote by
\begin{equation}\label{definition_log_potential}
 U^{\nu}(z)=\int \log\frac{1}{|s-z|}d\nu(s),\quad z\in \C,
\end{equation}
its logarithmic potential, which is harmonic in $\C\setminus\supp\nu$ and superharmonic in $\C$ \cite{Saff_book}.

Define the measure $\mu_0$ by the formula
\begin{equation}\label{definition_planar_measure}
d\mu_0(z)=\frac{1}{\pi t_0}\chi_\Omega(z)dA(z).
\end{equation}

\begin{thm}\label{thm_equilibrium_measure}
 Suppose $(t_0,t_1)\in \mathcal F$. Consider the potential $\mathcal V$ in \eqref{entire_potential} for the cubic polynomial $V$ in \eqref{generic_cubic_potential}. There exist 
an open neighborhood $D$ of $\overline{\Omega}$ and a constant $l$ such that
 \begin{align}
 2U^{\mu_0}(z)+\mathcal V(z)  & = l,\quad z\in \overline{\Omega}, \label{variational_equality_planar_measure} \\
 2U^{\mu_0}(z)+\mathcal V(z) & > l,\quad z\in D\setminus \overline{\Omega}.\label{variational_inequality_planar_measure}
 \end{align}
\end{thm}

As an immediate consequence, we recover the connection with the normal matrix model.

\begin{thm}\label{theorem_density_eigenvalues}
Suppose $(t_0,t_1)\in \mathcal F$, and $\mathcal V$ and $V$ are as in Theorem~\ref{thm_equilibrium_measure}. Suppose in 
addition that a given domain $D$ contains $\overline{\Omega}$ and satisfies the conclusions of Theorem~\ref{thm_equilibrium_measure}. Then the measure $\mu_0$ in 
\eqref{definition_planar_measure} is the limiting eigenvalue distribution of the normal matrix model with cubic potential \eqref{generic_cubic_potential} and cut off $D$.
\end{thm}

Theorems \ref{thm_equilibrium_measure} and \ref{theorem_density_eigenvalues} are proved in Section \ref{section_proof_s_property}. The evolution of the domain $\Omega$ in time 
$t_0>0$ is displayed in Figure~\ref{chain_polynomial_curves}.

\subsection{Spectral curve}

The pairs of points
$$
(\xi,z)=(h(w^{-1}),h(w)),\quad w\in \overline \C
$$
are expected to be solutions of an algebraic equation of the form
$$
F(\xi,z)=0,
$$
where $F$ is a symmetric polynomial in $\xi$ and $z$, with $\deg_z F=\deg_\xi F=3$. This equation is known in random matrix terminology as the {\it spectral curve} or {\it master 
loop equation}. Using equations \eqref{system_elbau_felder}--\eqref{definition_a_0}, after a 
lengthy calculation we arrive at 
\begin{equation}\label{spectral_curve}
F(\xi,z):=\xi^3+z^3-z^2\xi^2-t_1(\xi^2+z^2)-(1+t_0)z\xi+(B+t_1)(\xi+z)+A=0,
\end{equation}
where
\begin{equation}\label{equation_B}
 B = 4 a^3_0 r^2+4 a^2_0 r^4+4 a_0 r^4-a_0 r^2
\end{equation}
and
\begin{multline}\label{equation_A}
 A = -\left(a_0^4 \left(1-4 r^2\right)-2 a_0^3 \left(1-2 r^2\right)^2 \right. \\
      \left. +a_0^2 \left(-4 r^6+6 r^4-3 r^2+1\right)+r^2 \left(r^2-1\right)^3\right),
\end{multline}
$r=r(t_0,t_1)$ is as in Proposition \ref{proposition_definition_r} and $a_0$ is given in \eqref{definition_a_0}. 

For each $z$, there are three solutions $\xi_1,\xi_2,\xi_3$ to \eqref{spectral_curve}, labeled according to the expansions
\begin{align}\label{asymptotics_xi}
& \xi_1(z)=z^2+t_1+\frac{t_0}{z}+\mathcal O(z^{-2}), & \nonumber \\
& \xi_2(z)=-z^{1/2}+\frac{t_1}{2z^{1/2}}-\frac{t_0}{2z}+\mathcal O(z^{-3/2}), & \mbox{ as } z\to \infty, \\
& \xi_3(z)=z^{1/2}-\frac{t_1}{2z^{1/2}}-\frac{t_0}{2z}+\mathcal O(z^{-3/2}),\nonumber
\end{align}
where the square root is considered with a branch cut along the negative axis, and the branch is chosen such that $z^{1/2}>0$ along the positive axis.
In particular, the solution $\xi_1$ is meromorphic at $z=\infty$, whereas $\xi_2,\xi_3$ are branched at $z=\infty$.

A map $w\mapsto (\psi(w),\phi(w))$, $\psi,\phi$ rational, is a {\it rational parametrization} of \eqref{spectral_curve} if
$$
F(\psi(w),\phi(w))=0,\quad w\in \overline \C.
$$

A rational parametrization as above is called {\it proper} if every but a finite number of pairs $(\xi,z)$ satisfying \eqref{spectral_curve} is generated by
$(\xi,z)=(\psi(w),\phi(w))$ for exactly one choice $w\in \overline \C$.

\begin{thm}\label{theorem_schwarz_function}
Suppose that $(t_0,t_1)\in \mathcal F$. The map 
\begin{equation}\label{rational_function_bijection}
w\mapsto (\xi,z)=(h(w^{-1}),h(w)),\quad w \in \overline \C,
\end{equation}
where $h$ is given in \eqref{rational_parametrization}, is a proper rational parametrization of the algebraic equation \eqref{spectral_curve}.

Moreover, the function $\xi_1$ is the Schwarz function of $\partial \Omega$. That is, there exists a simply connected domain $E\subset \overline \C$, containing $\partial \Omega$ 
and the point $\infty$, 
and such that $\xi_1$ admits a meromorphic continuation to $E$, with pole only at $\infty$, and satisfies
\begin{equation}\label{equation_schwarz_function}
\xi_1(z)=\overline z,\quad z\in \partial \Omega.
\end{equation}
\end{thm}

Theorem~\ref{theorem_schwarz_function} is proved in Section~\ref{section_analysis_rational_parametrization}.

When $(t_0,t_1)\in \Gamma_c$, the function $\xi_1$ becomes branched in $\partial \Omega$ (see Section~\ref{section_behavior_critical_time} below).

In particular, the existence of the Schwarz function implies that $\overline \C\setminus \overline\Omega$ is a quadrature domain, we refer the reader to 
\cite{aharonov_shapiro,lee_makarov} for more details.

\subsection{Phase transition of the spectral curve}

The curve
\begin{equation}\label{critical_curve_limiting_zero_distribution}
\gamma_c: \quad t_0=-2s^{12}+s^{6}-9s^{10},\quad t_1=\frac{3}{2}s^{2}-\frac{9}{4}s^{4}-6s^{8},\quad 0\leq s \leq \frac{1}{2}
\end{equation}
splits the parameter region $\mathcal F$ into two parts $\mathcal F_1$, $\mathcal F_2$. The first part, $\mathcal F_1$, consists of points $(t_0,t_1)$ that lie to the right of 
$\gamma_c$, whereas the second part, $\mathcal F_2$, consists of points that lie to the left of $\gamma_c$, see Figure~\ref{phase_diagram}. For reasons that will become apparent 
later, we call $\mathcal F_1$ the {\it three-cut} region and $\mathcal F_2$ the {\it one-cut} region.

\begin{thm}\label{theorem_singular_points_spectral_curve}
For $(t_0,t_1)\in \mathcal F\setminus \gamma_c$, the spectral curve \eqref{spectral_curve} has three branch points $z_0,z_1,z_2$ of order two, and no other branch 
points. These points are located as follows.
\begin{enumerate}[label=(\roman*)]
\item For $(t_0,t_1)\in \mathcal F_1$,  
$$
\im z_1<0,\quad \overline{z}_2=z_1,\quad z_0>0,\quad z_0,z_1,z_2\in \Omega. 
$$

\item For $(t_0,t_1)\in \mathcal F_2$,
$$
z_2<z_1<z_0,\quad z_0>0,\quad z_0,z_1 \in \Omega.
$$

\item For $(t_0,t_1)\in \gamma_c$, \eqref{spectral_curve} has a branch point $z_0>0$ of order two and a branch point $z_1=z_2\in \R$ of order three, with $z_1<z_0$. Furthermore, 
$z_0,z_1\in \Omega$.
\end{enumerate}

Moreover, $\infty$ is always a branch point of order two of \eqref{spectral_curve}.

Finally, for $(t_0,t_1)\in\mathcal F$, \eqref{spectral_curve} has three critical points $\hat z_0,\hat z_1,\hat z_2 \in \C\setminus \overline \Omega$, satisfying
$$
\hat z_0>z_0, \quad \hat z_1,\hat z_2\in \C\setminus \R,\quad \im \hat z_1<0,\quad \overline{\hat z}_2=z_1.
$$
\end{thm}

The proof of Theorem~\ref{theorem_singular_points_spectral_curve} is provided in Section~\ref{section_topology_spectral_curve}.
In Theorem \ref{theorem_singular_points_spectral_curve}, by a critical point $\hat z_j$ we mean that it satisfies
$$
\frac{\partial F}{\partial \xi}(\xi_k,\hat z_j)=\frac{\partial F}{\partial z}(\xi_k,\hat z_j)=0,
$$
for some choice of $\xi_k=\xi_k(\hat z_j)$ for which the pair $(\xi_k,\hat z_j)$ satisfies \eqref{spectral_curve}.

Theorem~\ref{theorem_singular_points_spectral_curve} can be summarized in the following manner. The critical curve $\gamma_c$ determines two different regimes for the spectral 
curve: for pairs $(t_0,t_1)$ to the left of $\gamma_c$ the spectral curve has three real branch points, whereas for $(t_0,t_1)$ to the right of $\gamma_c$ the spectral curve has 
one real and two non real branch points. At $\gamma_c$, these non real branch points coalesce. We refer the reader to Figure~\ref{figure_critical_points} for a depiction of the 
branch points and critical points.

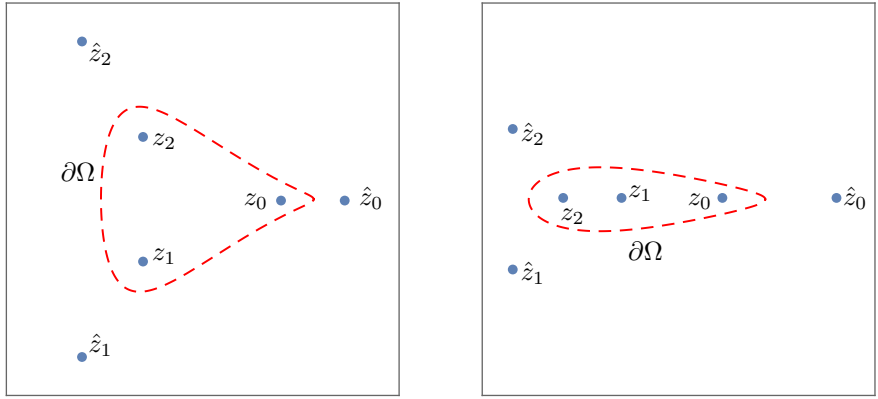
\begin{figure}[t]
\begin{minipage}[c]{0.5\textwidth}
\centering
  \begin{overpic}[scale=1]
  {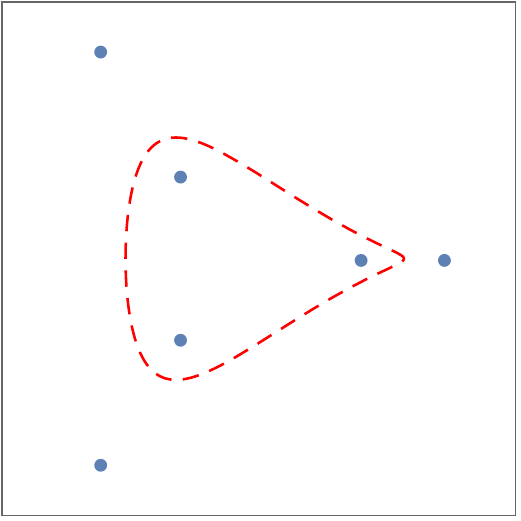}
\put(21,85){$\hat z_2$}
\put(89,49){$\hat z_0$}
\put(21,12){$\hat z_1$}
\put(37,64){$z_2$}
\put(60,49){$z_0$}
\put(37,35){$z_1$}
\put(14,55){$\partial\Omega$}
 \end{overpic}
\end{minipage}%
\begin{minipage}[c]{0.5\textwidth}
\centering
\begin{overpic}[scale=1]
{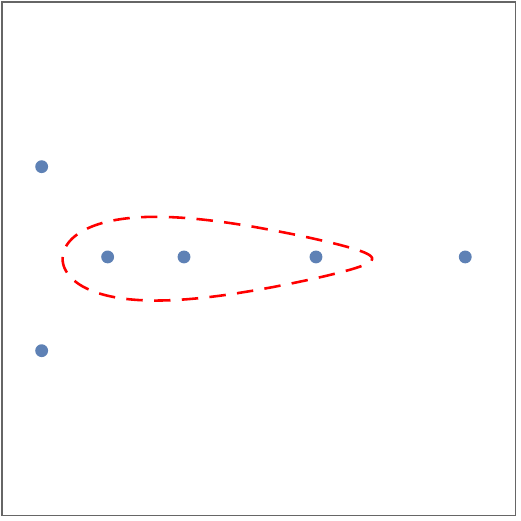}
\put(10,65){$\hat z_2$}
\put(91,49){$\hat z_0$}
\put(10,30){$\hat z_1$}
\put(20,45){$z_2$}
\put(52,49){$z_0$}
\put(37,51){$z_1$}
\put(37,35){$\partial\Omega$}
 \end{overpic}
 \end{minipage}
 \caption{The boundary of $\Omega$ and the branch points and critical points in the three-cut (left-hand panel) and one-cut (right-hand panel) cases.}\label{figure_critical_points}
 \end{figure}

\subsection{The parameters $(r,a_0)$ as a change of variables}

The functions $r=r(t_0,t_1)$ and $a_0=a_0(t_0,t_1)$, given by Proposition \ref{proposition_definition_r} and equation \eqref{definition_a_0}, respectively, can be seen as a 
change of variables. It turns out that we can express the inverse change of coordinates 
$(r,a_0)\mapsto (t_0,t_1)$ explicitly. 

\begin{prop}\label{proposition_change_of_coordinates}
 Suppose $(t_0,t_1)\in \overline{\mathcal F}$. The functions $r$ and $a_0$ satisfy the nonlinear system
 \begin{align}
  & 2r^4-r^2(1-4a_0^2)=-t_0 \label{system_change_coordinates_a}\\
  & a_0^2-(1-4r^2)a_0=-t_1 \label{system_change_coordinates_b}  
 \end{align}
\end{prop}

The equations \eqref{system_change_coordinates_a}--\eqref{system_change_coordinates_b} are nothing but the last two equations in \eqref{system_elbau_felder}, taking into account 
the values of $a_1$ and $a_2$ in \eqref{values_a1_a2}. 

As a consequence of Proposition~\ref{proposition_change_of_coordinates}, we can compute our phase diagram in the $(r,a_0)$-plane, as it is established by the next Theorem and it 
is shown in Figure~\ref{figure_phase_diagram_a_r_plane}.

\begin{figure}[t]
 \centering
 \includegraphics[scale=1]{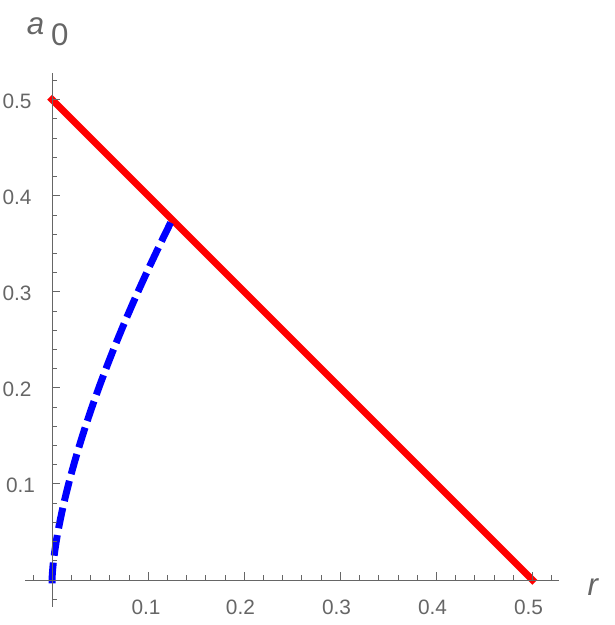}
 \caption{Image of the phase diagram in Figure~\ref{phase_diagram} through the change of variables $(t_0,t_1)\mapsto (r,a_0)$. The solid curve is the image of $\Gamma_c$ and the 
dashed curve is the image of $\gamma_c$. The region determined by the solid and dashed curves and the $r$-axis corresponds to $\mathcal F_1$, whereas the region between the 
solid and dashed curves and the $a_0$-axis corresponds to $\mathcal F_2$.}\label{figure_phase_diagram_a_r_plane}
\end{figure}

\begin{thm}\label{theorem_change_coordinates_phase_diagram}
On the $(r,a_0)$-plane, the curve $\Gamma_c$ assumes the form
\begin{equation}\label{cusp_critical_curve_positive_t1}
r=s,\quad a_0=\frac{1-2s}{2},\qquad 0<s<\frac{1}{2},
\end{equation}
and the curve $\gamma_c$ assumes the form
\begin{equation}\label{critical_curve_limiting_zero_distribution_a_r_plane}
r=s^3, \quad a_0=\frac{3}{2}s^2,\qquad 0<s<\frac{1}{2}.
\end{equation}
\end{thm}

The proofs of Proposition~\ref{proposition_change_of_coordinates} and Theorem~\ref{theorem_change_coordinates_phase_diagram} are given in 
Section~\ref{section_analysis_rational_parametrization}.

\subsection{The mother body problem}

We now focus our attention to the mother body problem \eqref{mother body_problem_general}. In what follows, given a set $E\subset \C$ we denote 
\begin{equation}\label{notation_conjugate_set}
E^*=\{z\in \C \; \mid \;\overline z\in E \} \quad \mbox{and} \quad \C_\pm=\{ z\in \C \; \mid \; \pm \im z >0 \}.
\end{equation}
Recall also that, according to Theorem~\ref{theorem_singular_points_spectral_curve}, $z_0,z_1$ and $z_2$ are the branch points of the spectral curve \eqref{spectral_curve}.

\begin{thm}\label{theorem_limiting_support_zeros}
There exists a contour $\Sigma_*$ with
\begin{equation}\label{inclusion_star_Omega}
 \Sigma_*\subset \Omega,
\end{equation}
and for which the solution $\xi_1$ in \eqref{asymptotics_xi} admits an analytic continuation to $\C\setminus \Sigma_*$ that satisfies 
\begin{equation}\label{s_property}
(\xi_{1+}(s)-\xi_{1-}(s))ds \in i\R,\quad s\in \Sigma_*,
\end{equation}
where $ds$ is a tangent vector to $\Sigma_*$ at the point $s$.
The contour $\Sigma_*$ is symmetric with respect to the real axis
$$
(\Sigma_*)^*=\Sigma_*
$$
and has the following geometric properties.
\begin{enumerate}[label=(\roman*)]
\item (Three-cut case) For $(t_0,t_1)\in \mathcal F_1$, the contour $\Sigma_*$ can be decomposed into
$$
\Sigma_*=\Sigma_{*,0}\cup \Sigma_{*,1}\cup \Sigma_{*,2},
$$
where each $\Sigma_{*,j}$ is a smooth oriented contour from a common point $z_*\in (-\infty,z_0)$ to the branch point $z_j$ and
$$
\Sigma_{*,0}=[z_*,z_0],\quad (\Sigma_{*,2})^*=\Sigma_{*,1}\subset \overline\C_-.
$$

\item (One-cut case) For $(t_0,t_1)\in \mathcal F_2 $, the contour $\Sigma_*$ is given by
$$
\Sigma_*=[z_1,z_0].
$$
\end{enumerate}

Moreover, the measure
\begin{equation}\label{limiting_measure_zeros}
d\mu_*(z)=\frac{1}{2\pi i t_0}(\xi_{1-}(z)-\xi_{1+}(z))dz,\quad z\in \Sigma_*
\end{equation}
is a probability measure on $\Sigma_*$. 
\end{thm}

The phase diagram displayed in Figure~\ref{phase_diagram_complete} shows several configurations of $\supp\mu_*$, and we refer to Figure~\ref{figure_support_mu_star} for 
more detailed numerical evaluations of $\supp\mu_*$, displaying the evolution of $\supp\mu_*$ in time $t_0$ while $t_1$ is kept fixed. 

\begin{figure}[t]
\begin{minipage}[c]{0.5\textwidth}
\centering
  \begin{overpic}[scale=1]
  {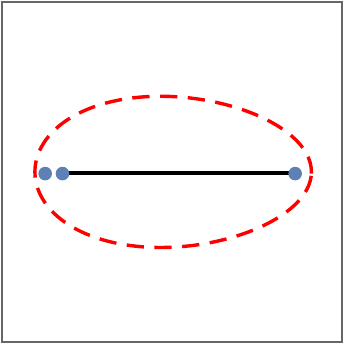}
 \end{overpic}
 \caption*{$(t_0,t_1)=(\frac{1}{2500},\frac{1}{8})\in\mathcal F_2$}
\end{minipage}%
\begin{minipage}[c]{0.5\textwidth}
\centering
 \begin{overpic}[scale=1]
  {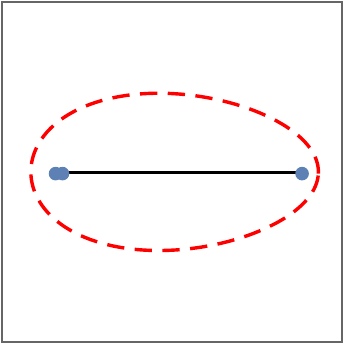}
 \end{overpic}
 \caption*{$(t_0,t_1)=(\frac{1}{1700},\frac{1}{8})\in\mathcal F_2$}
 \end{minipage}
 \begin{minipage}[c]{0.5\textwidth}
 \vspace{0.5cm}
 \centering
  \begin{overpic}[scale=1]
  {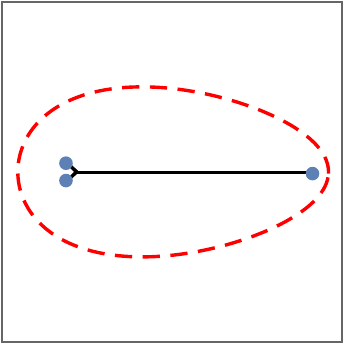}
 \end{overpic}
 \caption*{$(t_0,t_1)=(\frac{1}{500},\frac{1}{8})\in\mathcal F_1$}
\end{minipage}%
\begin{minipage}[c]{0.5\textwidth}
\vspace{0.5cm}
\centering
 \begin{overpic}[scale=1]
  {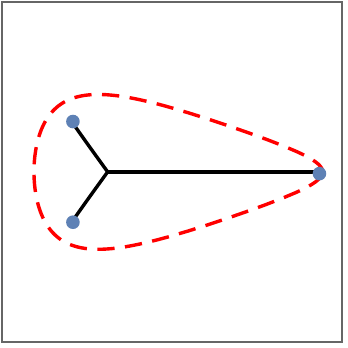}
 \end{overpic}
 \caption*{$(t_0,t_1)=(\frac{1}{50},\frac{1}{8})\in\mathcal F_1$}
 \end{minipage}
 \begin{minipage}[c]{1\textwidth}
 \vspace{0.5cm}
\centering
 \begin{overpic}[scale=1]
  {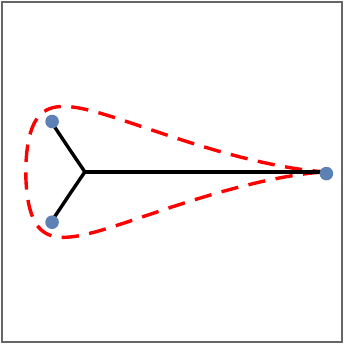}
 \end{overpic}
 \caption*{$(t_0,t_1)=(\frac{5}{128},\frac{1}{8})\in\Gamma_c$}
 \end{minipage}
 \caption{For the given values of $(t_0,t_1)$, the boundary $\partial \Omega$ (dashed contour), the support of the measure $\mu_*$ (solid lines) and the branch points $z_0$, 
$z_1$, $z_2$ (dots) are shown (the figures are scaled differently - compare with Figure~\ref{chain_polynomial_curves}). Note the transition when we move from $\mathcal F_2$ 
to $\mathcal F_1$. We stress that the contours of $\supp\mu_*$ outside the real line are not straight line segments. Numerical outputs.} \label{figure_support_mu_star}
\end{figure}

As we mentioned at the introduction, the construction of the measure $\mu_*$ in Theorem~\ref{theorem_limiting_support_zeros} is our main technical contribution, and it is quite 
involved. The first step, carried out in Section \ref{section_topology_spectral_curve}, is to construct the Riemann surface $\mathcal R$ for \eqref{spectral_curve}; along the way 
we also collect several results that are used later on. The sheet structure of $\mathcal R$ depends, in the terminology of Theorem \ref{theorem_limiting_support_zeros}, on whether 
we are in the three-cut or one-cut cases, and are explicitly given in Sections~\ref{section_sheet_structure_1} and \ref{section_sheet_structure_2}, 
respectively. In Section \ref{section_quadratic_differential}, we introduce a certain quadratic differential $\varpi$ that encodes $\mu_*$ on some of its trajectories. The main 
goal of Section \ref{section_quadratic_differential} is to describe the critical graph $\mathcal G$ of $\varpi$. This is done by first computing $\mathcal G$ for $t_1=0$ by 
``brute force'', and then deforming the parameter $t_1>0$, keeping track of all fundamental domains of $\mathcal G$. It turns out that the three-cut--to--one-cut phase 
transition for $\mu_*$ can be interpreted in terms of a phase transition for $\mathcal G$, and furthermore $\mathcal G$ also plays a fundamental role in the asymptotic analysis of 
the multiple orthogonal polynomials that will be introduced in a moment. In Section \ref{section_proof_s_property} we use this critical graph $\mathcal G$ to recover the measure 
$\mu_*$. 

It is a consequence of this analysis that the support $\Sigma_*$ of $\mu_*$, and in particular the point $z_*$, is determined as the projection of certain trajectories of 
$\varpi$. 
Roughly speaking, there are a number of points $y\in (-\infty,z_0)$ which are solutions to
$$
\re\int_{z_2}^{y}(\xi_1(s)-\xi_2(s))ds=0,
$$
where $\xi_1,\xi_2$ are (appropriate analytic continuations of) the functions in \eqref{asymptotics_xi}; the value $y=z_*$ is the largest of those solutions.
We refer the reader to Sections~\ref{section_quadratic_differential} and \ref{section_proof_s_property} for details.

The measure $\mu_*$ is connected to the normal matrix model by the following theorem.

\begin{thm}\label{theorem_balayage_relation}
The measure $\mu_*$ in \eqref{limiting_measure_zeros} relates to the limiting measure of eigenvalues $\mu_0$ \eqref{limiting_measure_eigenvalues} through the conditions
\begin{align}
U^{\mu_0}(z)=U^{\mu_*}(z), & \quad z\in \C\setminus \Omega, \label{mother_body_equation} \\
U^{\mu_0}(z)<U^{\mu_*}(z), & \quad z\in \Omega. \label{mother_body_inequality} 
\end{align}

\end{thm}

The proof of Theorem~\ref{theorem_balayage_relation} is given in Section~\ref{section_proof_s_property}.

Given a domain $G\subset \C$, a measure $\nu$ is called a {\it mother body} (or also {\it potential-theoretic skeleton}) of $G$ if \cite{gustafsson_lectures_balayage}
\begin{enumerate}
 \item[(M1)] $\supp\nu$ has null area measure;
 \item[(M2)] $\C\setminus \supp\nu$ is connected;
 \item[(M3)] $\supp\nu\subset G$;
 \item[(M4)] for $\mu_G$ the normalized area measure of $G$, it holds true 
 \begin{align*}
 & U^{\mu_G}(z)\leq U^{\nu}(z), \quad z\in \C, \\
 & U^{\mu_G}(z)= U^{\nu}(z), \quad z\in \C\setminus G.
 \end{align*}
\end{enumerate}

Conditions (M1)--(M4) are, generally speaking, quite demanding, and consequently domains with mother bodies are somewhat rare. Cases where the existence of the mother body is 
known include discs, convex polyhedra \cite{gustafsson_polyhedra} and ellipses \cite{sjodin_mother_body_algebraic}, and mother bodies have also been numerically obtained for 
certain oval shapes \cite{savina_shatalov_sternin}. 
Mother body measures appear in the context of quadrature domains \cite{gustafsson_quadrature}, inverse problems in geophysics \cite{zidarov_book}, zero distribution of orthogonal 
polynomials \cite{mhaskar_saff_zeros_extremal_polynomials,gustafsson_et_al_bergman}, among others. We refer the reader to the lecture notes 
\cite{gustafsson_lectures_balayage} by Gustafsson for more details.

Conditions (M1), (M2) and (M3) are satisfied for $\nu=\mu_*$ and $G=\Omega$. Equation~\eqref{mother_body_equation} is the same as \eqref{mother 
body_problem_general}, and together with \eqref{mother_body_inequality} they give (M4) and thus express that $\mu_*$ is a 
mother body for the domain $\Omega$. For this reason, we call the ``one-cut--to--three-cuts'' phase transition for 
$\supp\mu_*$ the {\it mother body phase transition}.

When the (boundary of the) domain $G$ displays some topological transition, it is natural to expect that its mother body measure displays some phase transition as well. 
In fact, this has already been described in several previous works in the context of random normal matrices
\cite{balogh_bertola_et_al_op_planar_measure,bleher_kuijlaars_normal_matrix_model,kuijlaars_tovbis_supercritical_normal_matrix_model,balogh_grava_merzi}, although in some cases 
without explicit mention. However, in our situation 
it is very interesting that the transition for $\mu_*$ occurs before any transition for $\Omega$. In other words, the limiting domain for the eigenvalues of the normal 
matrix model \eqref{normal_matrix_ensemble} does not feel the mother body phase transition: as it is assured by Theorem~\ref{theorem_rational_parametrization_polynomial_curve}, 
the 
boundary $\Omega$ depends analytically on the parameters $(t_0,t_1)$ even across the critical curve $\gamma_c$. 

To our knowledge, this is the first time a phase transition for the mother body, without any phase transition on the boundary of the underlying domain, is described in the 
literature. A somewhat related situation has already appeared in the work of Gustafsson and Lin \cite[Example 5.2]{gustafsson_lin}, where the authors identified, 
indirectly and without any detailed analysis, a phase transition for the branch points of the Schwarz function for a curve moving analytically in time: the 
branch points of the Schwarz function thus play the role of the end points of the mother body. Another similar phenomenon has been described in a previous work of 
the first author and Liechty \cite[Section~X]{bleher_liechty_PDWBC}, where they identified a phase transition for the zero distribution of the underlying orthogonal 
polynomials that is not reflected in the asymptotics of the partition function studied therein. In virtue of this latter work, it is also natural to expect that the partition 
function 
of \eqref{normal_matrix_ensemble} for the cubic potential \eqref{generic_cubic_potential} should not ``feel'' the mother body phase transition.

\subsection{Associated multiple orthogonality}\label{section_associated_MOP}

We follow \cite{bleher_kuijlaars_normal_matrix_model,balogh_bertola_et_al_op_planar_measure} and replace the planar orthogonality to an orthogonality over contours.

Construct a piecewise smooth curve 
\begin{equation}\label{definition_extension_skeleton}
\Sigma = {\Sigma}_0 \cup \Sigma_1 \cup \Sigma_2,
\end{equation}
where each set $\Sigma_j$ is a smooth oriented arc, starting at a common point $a_*\in\R$ and ending at the critical point $\hat z_j$ given by 
Theorem~\ref{theorem_singular_points_spectral_curve}, $j=0,1,2$. Recalling the notations introduced in 
\eqref{notation_conjugate_set}, we assume
\begin{equation}\label{conditions_contour_sigma}
a_*<\hat z_0,\quad {\Sigma}_0 = [a_*,\hat z_0], \quad \Sigma_1\subset \C_- \cup \{ a_* \},\quad \Sigma_2\subset \C_+\cup \{ a_* \}.
\end{equation}
we refer for instance to Figure~\ref{cut_off} for a possible configuration of $\Sigma$. At this moment the contour $\Sigma$ is rather arbitrary; in fact the conditions in 
\eqref{conditions_contour_sigma} are made for simplicity and could even be loosen. But later the contour $\Sigma$ will be chosen in an optimal way.

\begin{figure}[t]
 \centering
 \begin{overpic}[scale=1]
{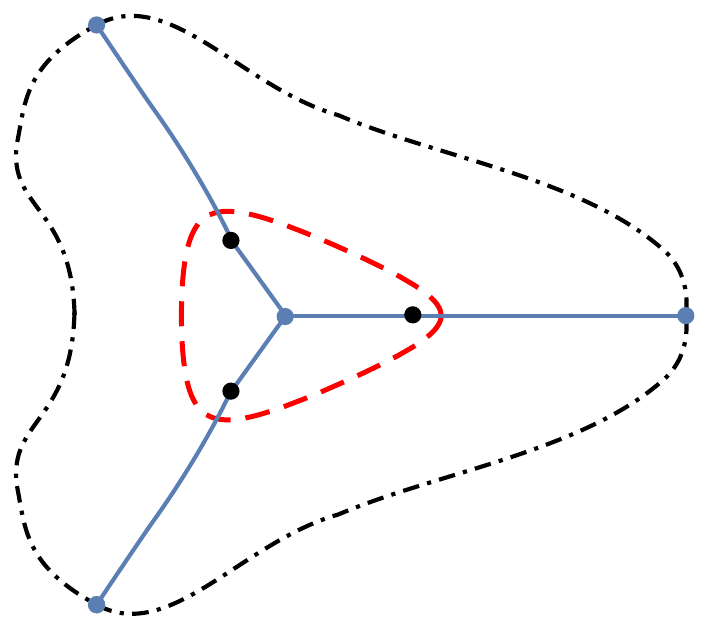}
\put(99,43){$\hat z_0$}
\put(10,-1){$\hat z_1$}
\put(9,87){$\hat z_2$}
\put(65,69){$\gamma_1$}
\put(65,25){$\gamma_2$}
\put(5,61.5){$\gamma_0$}
\put(21,75){$\Sigma_2$}
\put(29,22){$\Sigma_1$}
\put(77,46){$\Sigma_0$}
\put(50,54){$\partial \Omega$}
\put(41,46){$z_*$}
 \end{overpic}
 \caption{The dashed curve is $\partial \Omega$ and the dots inside it are the branch points $z_0, z_1$ and $z_2$ and the point $z_*$. The dot-dashed line is a possible 
configuration for $\partial D=\gamma_0\cup\gamma_1\cup\gamma_2$, and the points on $\partial D$ are the critical points $\hat z_0$, $\hat z_1$ and $\hat z_2$. The solid contour is 
a possible extension $\Sigma=\Sigma_0\cup\Sigma_1\cup \Sigma_2$ of the support $\Sigma_*$ of the mother body measure $\Sigma_*$.}\label{cut_off}
\end{figure}

As a general notational convention, for the rest of the paper we use cyclic notation mod $3$ without further mention when clear from the context. So for instance
$$
\hat z_3=\hat z_0,\quad \hat z_4=\hat z_1,\quad \hat z_5=\hat z_2, 
$$
and similarly for other quantities appearing later on.

Consider a compact, connected set $D\subset \C$, assuming in addition
$$
 \Sigma,\Omega\subset D,\quad \hat z_l\in \partial D, \ l=0,1,2,
$$
where $\Sigma$ is a curve as just explained above and $\Omega$ is the domain given by Theorem~\ref{theorem_rational_parametrization_polynomial_curve}. The domain $D$ should be 
interpreted as the cut off region in \eqref{planar_orthogonality_conditions}.

The points $\hat z_0,\hat z_1, \hat z_2$ split the boundary $\partial D$ into three curves $\gamma_0,\gamma_1,\gamma_2$. For $l=0,1,2$ (and again with cyclic notation mod $3$), 
the curve $\gamma_l$ is the oriented sub arc of $\partial D$ going from $\hat z_{l+2}$ to $\hat z_{l+1}$ that does not contain $\hat z_l$, see Figure~\ref{cut_off} for 
an example of a configuration of $D$, $\Sigma$ and $\Omega$. 

For the third root of unity $\omega$ and the directions at infinity $\infty_l$ given by
\begin{equation}\label{definition_infinities}
\omega=e^{\frac{2\pi i}{3}},\quad \infty_l=-\omega^{-l}\infty,\quad l=0,1,2,
\end{equation}
let $\Gamma_l$ be any unbounded oriented contour from $\infty_{l+1}$ to $\infty_{l+2}$, $l=0,1,2$. Also for $l=0,1,2$ and any non negative integer $k$, define
\begin{align}
w_{l,k,n}(z)        & = \int_{\Gamma_l} s^k e^{-\frac{n}{t_0}\left( sz -V(s)-V(z) \right)}ds,\quad z\in \C, \label{definition_nonhermitian_weights}\\
\widetilde w_{l,k,n}(z) & = \int_{\infty_l}^{\overline{z}} s^k e^{-\frac{n}{t_0}\left( sz-V(s)-V(z) \right)} ds,\quad z\in \C. \nonumber
\end{align}
where we recall that $V$ is given in \eqref{generic_cubic_potential}.

For $l=0,1,2$, the union of contours $\gamma_l\cup  \Sigma_{l+1}\cup \Sigma_{l+2}$ is the boundary of a domain $D_l\subset D$. An application of Green's Theorem to each of the 
pieces $D_0$, $D_1$, $D_2$ yields
\begin{prop}\label{proposition_2d_to_contour}
For any polynomial $Q$ and any integer $k\geq 0$, it is valid
\begin{multline}\label{green1}
2i\iint_D Q(z)\overline{z}^k e^{-\frac{n}{t_0}\left(|z|^2-2\re V(z) \right)} dA(z) \\ =\sum_{l=0}^2 \int_{ \Sigma_l}Q(z)w_{l,k,n}(z) dz + \sum_{l=0}^2 \int_{\gamma_l} 
Q(z)\widetilde w_{l,k,n}(z)dz.
\end{multline}
\end{prop}

The analogous to Proposition~\ref{proposition_2d_to_contour} for the monomial case $V(z)=\frac{z^d}{d}$, $d\geq 3$, is treated in 
\cite[Proposition~1.1]{kuijlaars_lopez_normal_matrix_model}. For the benefit of the reader we include the proof for $V$ as in \eqref{generic_cubic_potential}, which is 
follows the same arguments presented in \cite{kuijlaars_lopez_normal_matrix_model}.

\begin{proof} Let us apply the complex Green's theorem
\begin{equation*}
2i\iint_S \frac{\partial f}{\partial \overline {z}}\, dA=\int_{\partial S} f\,dz
\end{equation*}
with $S=D_l$ and $f=Q(z)\widetilde w_{l,k,n}(z)$. Then we obtain that
\begin{equation*}
2i\iint_{D_l} Q(z)\overline{z}^k e^{-\frac{n}{t_0}\left( |z|^2-V(\overline{z})-V(z) \right)}\, dA(z)
=\int_{\partial D_l} Q(z)\widetilde w_{l,k,n}(z)\,dz.
\end{equation*}
Summing over $l=0,1,2$, we get that
\begin{equation}\label{green2}
2i\iint_{D} Q(z)\overline{z}^k e^{-\frac{n}{t_0}\left( |z|^2-V(\overline{z})-V(z) \right)}\, dA(z)
=\sum_{l=0}^2 \int_{\partial D_l} Q(z)\widetilde w_{l,k,n}(z)\,dz.
\end{equation}
The sum on the right can be partitioned as
\begin{multline*}
\sum_{l=0}^2 \int_{\partial D_l} Q(z)\widetilde w_{l,k,n}(z)\,dz
=\sum_{l=0}^2 \left[\int_{ \ga_l} Q(z)\widetilde w_{l,k,n}(z)\,dz\right. \\
\left.+\int_{ \Sigma_{l+2}} Q(z)\widetilde w_{l,k,n}(z)\,dz 
-\int_{ \Sigma_{l+1}} Q(z)\widetilde w_{l,k,n}(z)\,dz\right].
\end{multline*}
Combining integrals over $\Sigma_l$, we obtain that
\begin{multline*}
\sum_{l=0}^2 \int_{\partial D_l} Q(z)\widetilde w_{l,k,n}(z)\,dz
=\sum_{l=0}^2 \int_{ \ga_l} Q(z)\widetilde w_{l,k,n}(z)\,dz\\
+\sum_{l=0}^2 \int_{ \Sigma_l} Q(z)\left[\widetilde w_{l+1,k,n}(z)-\widetilde w_{l+2,k,n}(z)\right]\,dz.
\end{multline*}

We now observe that
\begin{equation*}
\begin{aligned}
\widetilde w_{l+1,k,n}(z)-\widetilde w_{l+2,k,n}(z)  & = \int_{\infty_{l+1}}^{\overline{z}} s^k e^{-\frac{n}{t_0}\left( sz-V(s)-V(z) \right)} ds \\
						     & \quad -\int_{\infty_{l+2}}^{\overline{z}} s^k e^{-\frac{n}{t_0}\left( sz-V(s)-V(z) \right)} ds \\
						     & = \int_{\Ga_l} s^k e^{-\frac{n}{t_0}\left( sz-V(s)-V(z) \right)} ds \\
						     & = w_{l,k,n}(z),
\end{aligned}
\end{equation*}
hence
\begin{equation}\label{green3}
\begin{aligned}
\sum_{l=0}^2 \int_{\partial D_l} Q(z)\widetilde w_{l,k,n}(z)\,dz
=\sum_{l=0}^2 \int_{ \ga_l} Q(z)\widetilde w_{l,k,n}(z)\,dz
+\sum_{l=0}^2 \int_{ \Sigma_l} Q(z)w_{l,k,n}(z)\,dz.
\end{aligned}
\end{equation}
Equation \eqref{green1} then follows from \eqref{green2} and \eqref{green3}.
\end{proof}

In particular, if $Q=q_{j,n}$ is the planar orthogonal polynomial \eqref{planar_orthogonality_conditions}, we conclude
\begin{equation}\label{contour_integral_with_boundary}
 \sum_{l=0}^2 \int_{ \Sigma_l}q_{j,n}(z)w_{l,k,n}(z) dz \; + \; \sum_{l=0}^2 \int_{\gamma_l}q_{j,n}(z) \widetilde 
w_{l,k,n}(z)dz=0,\quad k=0,\hdots,j-1.
\end{equation}

Led by the fact that the integrals coming from $\partial D=\cup \gamma_l$ should be negligible compared to the integrals over $ \Sigma$ when $n\to \infty$, we follow 
\cite{bleher_kuijlaars_normal_matrix_model,kuijlaars_lopez_normal_matrix_model} and neglect the integrals over $\partial D$ in 
\eqref{contour_integral_with_boundary}. This motivates the introduction of the following family of polynomials.

\begin{definition}\label{definition_mop}
We define $P_{j,n}$ to be the monic polynomial of degree $j$, if it exists, that satisfies
\begin{equation}\label{oc1}
\sum_{l=0}^2 \int_{ \Sigma_l}P_{j,n}(z) w_{l,k,n}(z)dz=0,\quad k=0,\hdots,j-1,
\end{equation}
where the weights $w_{l,k,n}$, $k,n\in\N$, $l=0,1,2$, are given in \eqref{definition_nonhermitian_weights}.
\end{definition}

\begin{prop}\label{proposition_multiple_orthogonality}
The polynomial $P_{j,n}$, if it exists, fulfills the non-hermitian multiple orthogonality conditions
\begin{equation}\label{multiple_orthogonality_conditions_j}
\begin{aligned}
\sum_{l=0}^2\int_{ \Sigma_l} P_{j,n}(z)z^k w_{l,0,n}(z)dz=0, & \quad k=0,\hdots, \left\lceil \frac{j}{2}  \right\rceil -1, \\
\sum_{l=0}^2\int_{ \Sigma_l} P_{j,n}(z)z^k w_{l,1,n}(z)dz=0, & \quad k=0,\hdots, \left\lfloor \frac{j}{2} \right\rfloor -1.
\end{aligned}
\end{equation}
\end{prop}

For a proof when $t_1=0$, we refer to \cite[Lemma~5.1]{bleher_kuijlaars_normal_matrix_model}. 
The case $t_1\neq 0$ is treated similarly. For the sake of completeness we give the proof.

\begin{proof}
Equation \eqref{oc1} reads
\begin{equation}\label{oc2}
\sum_{l=0}^2 \int_{ \Sigma_l}P_{j,n}(z)\int_{\Gamma_l} s^k e^{-\frac{n}{t_0}\left( sz -V(s)-V(z) \right)}dsdz=0,\quad k=0,\hdots,j-1.
\end{equation}
Let $Q(s)$ be any polynomial. Integration by parts gives the identity 
\begin{multline*}
\int_{\Gamma_l} Q'(s) e^{-\frac{n}{t_0}\left( sz -V(s)-V(z) \right)}ds \\
=\frac{n}{t_0}\int_{\Gamma_l} Q(s)\left( z-V'(s)\right)e^{-\frac{n}{t_0}\left( sz -V(s)-V(z) \right)}ds,
\end{multline*}
hence
\begin{multline}\label{oc4}
z\int_{\Gamma_l} Q(s) e^{-\frac{n}{t_0}\left( sz -V(s)-V(z) \right)}ds\\
=\int_{\Gamma_l} \left( V'(s)Q(s)+\frac{t_0 }{n}\,Q'(s)\right)e^{-\frac{n}{t_0}\left( sz -V(s)-V(z) \right)}ds. 
\end{multline}

Introduce the linear differential operator
\begin{equation}\label{oc5}
\mathcal A: Q(s)\mapsto  V'(s)Q(s)+\frac{t_0 }{n}\,Q'(s).
\end{equation}
Then the identity \eqref{oc4} can be written as
\begin{equation*}
z\int_{\Gamma_l} Q(s) e^{-\frac{n}{t_0}\left( sz -V(s)-V(z) \right)}ds
=\int_{\Gamma_l} \mathcal A(Q)(s)e^{-\frac{n}{t_0}\left( sz -V(s)-V(z) \right)}ds. 
\end{equation*}
Applying it $k$ times, we obtain that
\begin{equation*}
z^k\int_{\Gamma_l} Q(s) e^{-\frac{n}{t_0}\left( sz -V(s)-V(z) \right)}ds
=\int_{\Gamma_l} \mathcal A^k(Q)(s)e^{-\frac{n}{t_0}\left( sz -V(s)-V(z) \right)}ds, 
\end{equation*}
thus
\begin{multline*}
\sum_{l=0}^2\int_{ \Sigma_l} P_{j,n}(z)z^k \int_{\Gamma_l} Q(s) e^{-\frac{n}{t_0}\left( sz -V(s)-V(z) \right)}dsdz\\
=\sum_{l=0}^2\int_{ \Sigma_l} P_{j,n}(z)\int_{\Gamma_l} \mathcal A^k(Q)(s)e^{-\frac{n}{t_0}\left( sz -V(s)-V(z) \right)}dsdz. 
\end{multline*}
Observe that if $Q(s)$ is a polynomial, then $\mathcal A^k(Q)(s)$ is a polynomial as well, and since $\deg V'=2$, we obtain from \eqref{oc5} that
\begin{equation*}
\deg \mathcal A^k(Q)=\deg Q+2k,
\end{equation*}
hence if $Q(s)\equiv 1$ and $2k<j$, then orthogonality condition \eqref{oc2} implies that
\begin{multline*}
\sum_{l=0}^2\int_{ \Sigma_l} P_{j,n}(z)z^k \int_{\Gamma_l} e^{-\frac{n}{t_0}\left( sz -V(s)-V(z) \right)}dsdz\\
=\sum_{l=0}^2\int_{ \Sigma_l} P_{j,n}(z)\int_{\Gamma_l} \mathcal A^k(1)(s)e^{-\frac{n}{t_0}\left( sz -V(s)-V(z) \right)}dsdz=0. 
\end{multline*}
This proves the first identity in \eqref{multiple_orthogonality_conditions_j}. To prove the second one, we take $Q(s)\equiv s$
and any $k$ such that $2k+1<j$. Then the orthogonality condition \eqref{oc2} implies that
\begin{multline*}
\sum_{l=0}^2\int_{ \Sigma_l} P_{j,n}(z)z^k \int_{\Gamma_l} s e^{-\frac{n}{t_0}\left( sz -V(s)-V(z) \right)}dsdz\\
=\sum_{l=0}^2\int_{ \Sigma_l} P_{j,n}(z)\int_{\Gamma_l} \mathcal A^k(z)(s)e^{-\frac{n}{t_0}\left( sz -V(s)-V(z) \right)}dsdz=0. 
\end{multline*}
This proves \eqref{multiple_orthogonality_conditions_j}.
\end{proof}

We are mostly interested in the {\it diagonal} polynomials $P_{n,n}$. From Proposition \ref{proposition_multiple_orthogonality}
we obtain that the polynomial $P_{n,n}$ fulfills the non-hermitian multiple orthogonality conditions
\begin{equation}\label{multiple_orthogonality_conditions}
\begin{aligned}
\sum_{l=0}^2\int_{ \Sigma_l} P_{n,n}(z)z^k w_{l,0,n}(z)dz=0, & \quad k=0,\hdots, \left\lceil \frac{n}{2}  \right\rceil -1, \\
\sum_{l=0}^2\int_{ \Sigma_l} P_{n,n}(z)z^k w_{l,1,n}(z)dz=0, & \quad k=0,\hdots, \left\lfloor \frac{n}{2} \right\rfloor -1.
\end{aligned}
\end{equation}

We remark that, as a consequence of the construction of the functions $w_{l,k,n}$, the orthogonality conditions above do not depend on the precise choice of the contour 
$\Sigma$ as in \eqref{definition_extension_skeleton}, but only on its endpoints $\hat z_0,\hat z_1$ and $\hat z_2$.

As a consequence of \eqref{multiple_orthogonality_conditions}, $P_{n,n}$ is additionally characterized through a Riemann-Hilbert problem (shortly RHP). We apply the 
nonlinear Deift/Zhou steepest descent analysis \cite{deift_book,deift_its_zhou_riemann_hilbert_random_matrices_integrable_statistical_mechanics} to this RHP and obtain 
the existence and asymptotic information for the polynomial $P_{n,n}$ for $n$ sufficiently large. This analysis is 
carried out in Sections~\ref{section_riemann_hilbert_analysis_precritical} and 
\ref{section_riemann_hilbert_analysis_postcritical} for the three-cut and one-cut cases, respectively, and we refer the reader to these sections for more details. We also 
stress that several trajectories of the quadratic differential $\varpi$ constructed in Section~\ref{section_quadratic_differential} play a fundamental role in this asymptotic 
analysis. One of the outcomes of it is regarding the zero counting measure
\begin{equation}\label{counting_measures}
d\mu_n(z)=\frac{1}{n}\sum_{P_{n,n}(w)=0}\delta(z-w)dA(z).
\end{equation}

\begin{thm}\label{theorem_limiting_counting_measure}
 Suppose $(t_0,t_1)\in\mathcal F\setminus \gamma_c$. The sequence of zero counting measures $(\mu_n)$ converges weakly to the measure $\mu_*$ given by 
\eqref{limiting_measure_zeros}.
\end{thm}

By Theorem \ref{theorem_rational_parametrization_polynomial_curve}, we know that the restriction of $h$ to $\overline\C\setminus\D$ admits an inverse 
$\psi_1:\overline\C\setminus \Omega\to \overline \C \setminus \D$. The function $\psi_1$ is alternatively characterized as the conformal map from $\C\setminus 
\overline\Omega$ to $\C \setminus \overline\D$ that is uniquely determined by the conditions
\begin{equation*}
\lim_{z\to \infty}\psi_1(\infty)=\infty,\quad \lim_{z\to \infty} \psi_1'(z)=\frac{1}{r},
\end{equation*}
where $r$ is as in Proposition \ref{proposition_definition_r}.

In addition, define the multivalued analytic function
\begin{equation}\label{definition_factor_G}
G(z)=\int_{z_0}^z \xi_1(s)ds,\quad z\in \C\setminus \Sigma_*,
\end{equation}
where $\xi_1$ is as in \eqref{asymptotics_xi} and the path of integration is taken in $\C\setminus\Sigma_*$. As can be seen from the expansion \eqref{asymptotics_xi}, the residue 
of the function $\xi_1$ at $\infty$ is $-t_0$. In particular, this implies that $G$ is well defined modulo $2\pi i t_0$.

As for the uniform asymptotics for $P_{n,n}$, one of the consequences of our asymptotic analysis is given by the following result.
\begin{thm}\label{thm_wave_factor}
Suppose $(t_0,t_1)\in \mathcal F\setminus \gamma_0$. The map $\psi_1$ admits an analytic continuation to $\C\setminus \Sigma_*$, and for a certain constant $c$, the asymptotic 
formula
 \begin{equation}\label{uniform_asymptotics_polynomial}
 P_{n,n}(z)=\sqrt{r \psi_1'(z)}e^{\frac{n}{t_0}(G(z)-V(z)+c)}(1+\Boh(n^{-1})),
 \end{equation}
 holds true uniformly on compacts of $\C\setminus \Sigma_*$, where the branch of the square root is chosen with branch cut on $\Sigma_*$ and so that $\sqrt{r\psi_1'(z)}\to 1$ as 
$z\to\infty$.
\end{thm}
Theorems \ref{theorem_limiting_counting_measure} and \ref{thm_wave_factor} are proven in Section \ref{section_proof_theorem_limiting_counting_measure}, after the conclusion of the 
steepest descent analysis.

Relation \eqref{mother_body_equation} is comparable with results by Elbau~\cite[Lemma~5.1 and 
Theorem~5.3]{elbau_thesis}. As mentioned in the introduction, Elbau shows that any weak limit $\nu$ of the sequence of zero counting measures for 
the polynomials $(q_{n,n})$ in \eqref{planar_orthogonality_conditions} should also 
satisfy 
\eqref{mother_body_equation} (with $\mu_*$ replaced by $\nu$). Hence, what Theorem~\ref{theorem_limiting_counting_measure} says is that, at the level of weak asymptotics, our 
multiple orthogonal polynomials $(P_{n,n})$ have the same behavior as expected for the planar orthogonal polynomials $(q_{n,n})$. We also expect the formula 
\eqref{uniform_asymptotics_polynomial} to hold true if we replace $P_{n,n}$ by the polynomial $q_{n,n}$.

The restriction $(t_0,t_1)\in \mathcal F\setminus \gamma_c$ in Theorems \ref{theorem_limiting_counting_measure} and \ref{thm_wave_factor} is of technical nature. For these values 
of $(t_0,t_1)$, the density of $\mu_*$ vanishes as square root at the endpoints of its support, and as a consequence the required local parametrices in the steepest descent 
analysis are constructed out of solutions to the Airy differential equation. However, in the critical case $(t_0,t_1)\in \mathcal \gamma_c$ the density of $\mu_*$ vanishes with 
order $1/3$ at the endpoint $z_1=z_2$, and a different parametrix is required to complete the steepest descent analysis. This order of vanishing indicates that the local 
parametrix near $z_1=z_2$ should be constructed in terms of solutions to the Pearcey differential equation (see for instance \cite{bleher_kuijlaars_external_source_gaussian_III} 
and the references therein), but the rest of the analysis carried out here should work, after cosmetic changes, in this critical case as well. In particular, the conclusions 
of Theorems~\ref{theorem_limiting_counting_measure} and \ref{thm_wave_factor} should also hold true for $(t_0,t_1)\in \gamma_c$.

\subsection{Behavior at the boundary of the phase diagram}\label{section_behavior_critical_time}

When $(t_0,t_1)\in \Gamma_c$, the branch point $z_0$ and the double point $\hat z_0$ described in Theorem~\ref{theorem_singular_points_spectral_curve} come together at the 
boundary $\partial \Omega$. This corresponds to a branching of the Schwartz function $\xi_1$ on $\partial \Omega$, and explains the emerging of the cusp in $\partial 
\Omega$, as can be seen, for instance, in Figures~\ref{phase_diagram_complete}, \ref{chain_polynomial_curves} and \ref{figure_support_mu_star}.

Consequently, when $(t_0,t_1)\in \Gamma_c$ the support of $\supp\mu_*$ comes to the boundary of $\Omega$ as well, and the density of the measure $\mu_*$ vanishes with order $3/2$ 
at the endpoint $z_0=\hat z_0$. This coalescence (for $t_1=0$) is already described in the literature (see for instance 
\cite{kuijlaars_tovbis_supercritical_normal_matrix_model,lee_teodorescu_wiegmann,bleher_kuijlaars_normal_matrix_model}), and the local behavior of 
the polynomial $P_{n,n}$ (after suitable regularization) near this point is expected to be given in terms of solutions to the Painlevé I equation, see for instance 
\cite{duits_kuijlaars_painleve_I,bleher_deano_painleve_I} for related works. When $t_1>0$, nothing special happens near the other endpoints $z_1,z_2$ of 
$\supp\mu_*$.

\subsection{The S-property}\label{section_s_property}

Thanks to the precise manner we constructed the weights for the integrals in each contour $\Sigma_l$, the sum of integrals in \eqref{multiple_orthogonality_conditions} neither 
depends on the precise choice of contours $\Sigma_1,\Sigma_2$ and $\Sigma_3$ nor on their common point $a_*$. 
This freedom is reflected in Theorem~\ref{theorem_limiting_counting_measure}, which says that the zeros of $P_{n,n}$ accumulate on $\supp\mu_*$ 
regardless of the precise choice of the contour $\Sigma$. Furthermore, although the orthogonality conditions \eqref{multiple_orthogonality_conditions} do depend on the choice of 
endpoints $\hat z_0,\hat z_1$ and $\hat z_2$ for $\Sigma$, it becomes clear from our RH analysis that there is some flexibility in the choice of these points: 
in the large $n$ limit the behavior of $P_{n,n}$ does not depend on these endpoints, as long as they are selected within some regions determined by certain critical trajectories 
of the underlying quadratic differential. 

This freedom in the choice of the contour $\Sigma$ is characteristic for non-hermitian orthogonality, and it is reflected in the behavior of the zeros of the respective orthogonal 
polynomials. Among all possible choices of contours, the zeros, in the large $n$ limit, accumulate in a very particular one, determined by the so called {\it S-property}, as we 
discuss next.

Given a contour $\Sigma$ for the orthogonality \eqref{multiple_orthogonality_conditions}, construct 
three other oriented contours $L_j$, $j=0,1,2$, starting at the point $a_*$ in the inner sector between $\Sigma_{j-1}$ and $\Sigma_{j+1}$, and extending to $\infty$ along the 
directions $\infty_j$, $j=0,1,2$, as defined in \eqref{definition_infinities}. Assume in addition $L_j\cap L_k=L_j\cap \Sigma=\{a_*\}$ for $j\neq k$ and set $L=L_0\cup L_1\cup 
L_2$. We refer to Figure~\ref{figure_contour_x} in Section~\ref{section_transformation_Y_X} for an example of the configuration of $L$ and $\Sigma$.

For any pair of contours $(\Sigma,L)$ as above, we associate a class of pairs of measures $\mathcal M(\Sigma,L)=\{(\nu_1,\nu_2)\}$, where each pair 
$(\nu_1,\nu_2)$ satisfies the constraints
$$
|\nu_1|=2|\nu_2|=1,\quad \supp\nu_1\subset \Sigma,\quad \supp\nu_2\subset L.
$$

For a pair $(\nu_1,\nu_2)\in\mathcal M(\Sigma,L)$, we denote by
$$
I(\nu_1,\nu_2)=\iint \log\frac{1}{|s-z|}d\nu_1(s)d\nu_2(s)
$$
their mutual logarithmic energy and define the {\it vector energy}
$$
E(\nu_1,\nu_2)=I(\nu_1,\nu_1)+I(\nu_2,\nu_2)-I(\nu_1,\nu_2)-\frac{1}{t_0}\int \re(V(x)-\Psi(x))d\nu_1(x),
$$
where $\Psi(z)$ is an appropriate branch of $\frac{2}{3}(z-t_1)^{3/2}$.

The {\it vector equilibrium problem} for the pair $(\Sigma,L)$ and the energy $E(\cdot)$ asks for minimizing this vector energy on $\mathcal M(\Sigma,\Gamma)$. That is, asks for 
finding a pair $(\lambda_1,\lambda_2)\in \mathcal M(\Sigma,L)$, the so-called {\it vector equilibrium measure}, satisfying
$$
E(\lambda_1,\lambda_2)=\inf_{(\mu_1,\mu_2)\in \mathcal M(\Sigma,L)} E(\mu_1,\mu_2).
$$

We stress that the vector equilibrium measure $(\lambda_1,\lambda_2)$ depends on the pair of contours $(\Sigma,L)$. Finally, the {\it S-property problem} asks for finding a pair 
of 
contours $(\Sigma,L)$ for which the respective vector equilibrium measure $(\lambda_1,\lambda_2)$ satisfies the {\it S-properties}
\begin{multline*}
 \frac{\partial }{\partial n_+}\left( 2U^{\lambda_1}(z)-U^{\lambda_2}(z)-\frac{1}{t_0}\re(V(z)-\Psi(z))\right)= \\ \frac{\partial }{\partial n_-}\left( 
2U^{\lambda_1}(z)-U^{\lambda_2}(z)-\frac{1}{t_0}\re(V(z)-\Psi(z))\right),\quad z\in \supp\lambda_1
\end{multline*}
and
\begin{equation*}
 \frac{\partial }{\partial n_+}\left( 2U^{\lambda_2}(z)-U^{\lambda_1}(z)\right)= \frac{\partial }{\partial n_-}\left( 
2U^{\lambda_2}(z)-U^{\lambda_1}(z))\right),\quad z\in \supp\lambda_2,
\end{equation*}
where $n_\pm$ are the normal vectors to $\Sigma\cup L$ and $U^{\lambda_j}$ is defined in \eqref{definition_log_potential}. If the pair $(\Sigma,L)$ has the 
S-property as above, we call it a pair of {\it S-contours}

The S-property has been originally introduced in the context of Padé approximants by H. Stahl 
\cite{stahl_orthogonal_polynomials_complex_weight_function,stahl_extremal_domains,stahl_structure_extremal_domains} and further extended by A. Gonchar and E. Rakhmanov 
\cite{gonchar_rakhmanov_rato_rational_approximation} to non-hermitian orthogonality with varying weights, see also 
\cite{rakhmanov_orthogonal_s_curves,baratchart_yattselev_extremal_domains,kuijlaars_silva_s_curves,martinez_rakhmanov,huybrechs_kuijlaars_lejon,bertola_boutroux} for a more recent 
account of results when dealing with non-hermitian orthogonality.

However, the S-property for multiple orthogonality is much less clear. To our knowledge, the cases which have been studied so far are either restricted to multiple orthogonality 
with fixed (non varying) weights \cite{aptekarev_kuijlaars_van_assche_hermite_pade} or make strong symmetry assumptions for the weights 
\cite{aptekarev_bleher_kuijlaars_external_source_gaussian_II,bleher_kuijlaars_external_source_gaussian_I,bleher_delvaux_kuijlaars_external_source,
bleher_kuijlaars_normal_matrix_model}. For all of those, the S-property followed directly from the symmetry at hand.

In the present setting, given a pair of contours $(\Sigma,L)$ as above, the vector equilibrium energy exists, is unique and can be further characterized in terms of certain 
(Euler-Lagrange) variational conditions \cite{beckermann_et_al_equilibrium_problems,hardy_kuijlaars_vector_equilibrium}. However, finding the pair of S-contours is a much more 
delicate matter. To recover the S-property in our setting, we recall the condition \eqref{s_property}, which allows us to define the {\it positive} measure $\mu_*$ as in 
\eqref{limiting_measure_zeros}. From our 
analysis of the underlying quadratic differential in Section~\ref{section_quadratic_differential}, it is possible to construct two contours $(\Sigma,L)$ such 
that the respective vector equilibrium measure is of the form $(\mu_*,\lambda_*)$, where the measure $\mu_*$ is the one given by 
Theorem~\ref{theorem_limiting_support_zeros}. The contour $L$ satisfies the auxiliary condition
\begin{equation}\label{s_property_second_measure}
(\xi_{2+}(z)-\xi_{3+}(z))dz\in i\R,\quad z\in L,
\end{equation}
and the measure $\lambda_*$ can be constructed from this condition. It then follows that the S-property for this pair of contours $(\Sigma,L)$ is actually equivalent to the 
conditions \eqref{s_property} and \eqref{s_property_second_measure}. Consequently, the measure $\mu_*$ can also be interpreted in terms of the S-property above, and conditions 
\eqref{s_property} and \eqref{s_property_second_measure} can be regarded as {\it algebraic} S-properties. As we will see later, \eqref{s_property} and 
\eqref{s_property_second_measure} also play a fundamental role in the construction of the g-functions used in the forthcoming RH analysis.

\subsection{Statement of Results - $t_1<0$}\label{section_negative_t}

As we mentioned in the introduction, after appropriate modifications the results of Sections~\ref{section_phase_diagram}--\ref{section_s_property} are also valid for 
$t_1\in (-3/4,0)$, as it is discussed next. Although the proofs of these results for $t_1<0$ go very much along the same lines as their respective 
results for $t_1>0$, some technical results collected on the way would have to be proved again taking into account $t_1<0$, what would make this already lengthy paper much 
longer. So for the sake of presentation we decided to not provide proof of these results.

The curve
$$
\Gamma_c^-: \quad t_0=-16s^6+6s^4,\quad t_1=-12s^4+6s^2-\frac{3}{4},\quad 0\leq s \leq\frac{1}{2},
$$
together with the horizontal and vertical axes determine a bounded domain $\mathcal F^-$ on the $(t_0,t_1)$-plane, see Figure~\ref{figure_phase_diagram_negative_t_1}. Note in 
particular that $t_1<0$, whenever $(t_0,t_1)\in \mathcal F^-$.

\begin{figure}[t]
\begin{overpic}[scale=1]
 {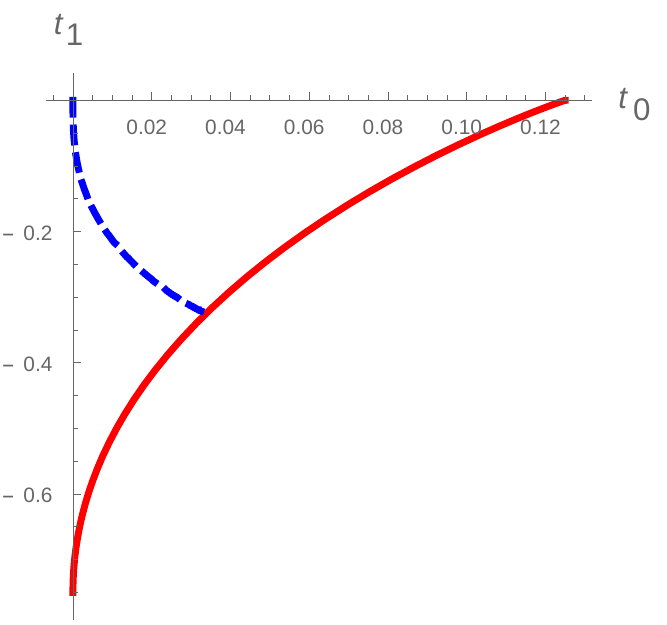}
 \put(22,54){$\gamma_c^-$}
 \put(52,58){$\Gamma_c^-$}
\put(35,65){$\mathcal F_1^-$}
 \put(13,40){$\mathcal F_2^-$}
 \end{overpic}
\caption{Phase diagram for $t_1$ negative. The dashed curve is $\gamma_c^-$, whereas the solid curve is $\Gamma_c^-$. The region to the right of $\gamma_c^-$ is $\mathcal F_1^-$, 
whereas the region to the left of 
$\gamma_c^-$ is $\mathcal F_2^-$.}\label{figure_phase_diagram_negative_t_1}
\end{figure}

When $(t_0,t_1)\in \mathcal F^-$, Proposition~\ref{proposition_definition_r} still holds true. That is, the polynomial $p$ still has a smallest positive root, always simple, that 
we 
keep denoting by $r=r(t_0,t_1)$. The coefficient $a_0$ in \eqref{definition_a_0} is well-defined, but now it becomes negative. Nevertheless, the rational function 
$h$ in \eqref{rational_parametrization} is also well-defined, and Theorems \ref{theorem_rational_parametrization_polynomial_curve}, \ref{thm_equilibrium_measure} and 
\ref{theorem_density_eigenvalues} hold true without any modification in their statements. However, we emphasize that the behavior of the roots of 
$p$ and of the coefficients $r$ and $a_0$, as functions of $(t_0,t_1)$, change when compared to $t_1>0$, so all the auxiliary results needed in their proofs have to be modified. 
We refer the reader to Figure~\ref{figure_chain_polynomial_curves_negative_t_1} where, for a certain choice of $t_1<0$, the evolution in time $t_0>0$ of the boundary 
$\partial \Omega$ is displayed. 

\begin{figure}[t]
\begin{overpic}[scale=1]{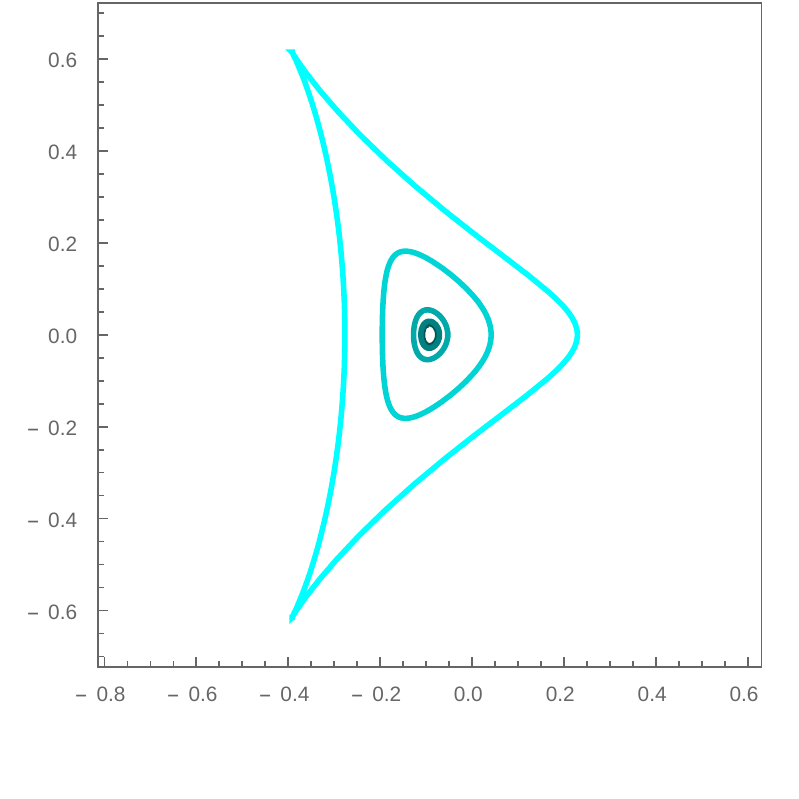}
\end{overpic}
\caption{The boundary $\partial \Omega$ corresponding to $t_1=-\frac{1}{10}$ and $t_0=\frac{1}{2500}, \frac{1}{1700}, \frac{1}{500}, \frac{1}{50}$ and 
$\frac{35+4\sqrt{30}}{1800}$. 
(respectively from dark to bright color). For the pair $(t_0,t_1)=\frac{35+4\sqrt{30}}{1800},-\frac{1}{10})$, which belongs to the critical curve $\Gamma_c^-$, cusps are 
created at the boundary. Numerical output.}\label{figure_chain_polynomial_curves_negative_t_1}
\end{figure}

The quantities $B=B(t_0,t_1)$, $A=A(t_0,t_1)$, given respectively in equations \eqref{equation_B} and \eqref{equation_A}, are still meaningful, and so is the spectral 
curve in \eqref{spectral_curve}, and Theorem~\ref{theorem_schwarz_function} also holds true for $(t_0,t_1)\in\mathcal F^-$. As for the branch points and critical points of the 
spectral curve, we have the following result.

\begin{thm}\label{theorem_singular_points_spectral_curve_negative_t}
 For $(t_0,t_1)\in \mathcal F^-$, the spectral curve \eqref{spectral_curve} has three simple branch points $z_0,z_1,z_2\in \mathcal \C$ with 
$$
z_0\in \R,\quad z_1\in \C_-,\quad z_2=\overline z_1, \quad z_1,z_2\in \Omega,
$$
and three singular points $\hat z_0,\hat z_1, \hat z_2\in \C\setminus \overline \Omega$ which satisfy
$$
\hat z_0\in \R,\quad \hat z_1\in \C_-,\quad z_2=\overline z_1.
$$
\end{thm}

Comparing Theorem~\ref{theorem_singular_points_spectral_curve_negative_t} with the three-cut case $(t_0,t_1)\in \mathcal F_1$ in 
Theorem~\ref{theorem_singular_points_spectral_curve}, the essential difference here is that the point $z_0$ is not always in the domain $\Omega$. For $t_1$ negative and small, the 
three 
points $z_0,z_1$ and $z_2$ are branch points of the Schwarz function $\xi_1$ of $\partial \Omega$. However, when $t_1$ decreases (while $t_0>0$ is kept fixed), this point $z_0$ 
becomes a branch point of the other two solutions $\xi_2,\xi_3$ of the spectral curve \eqref{spectral_curve}, but it is not anymore a branch point of $\xi_1$. If we keep 
decreasing $t_1$ the branch point $z_0$ might leave the domain $\Omega$, even before $t_1$ reaches the critical value $-3/4$. This phenomenon is reflected in the mother body 
measure $\mu_*$, as will be explained in a moment.

\begin{thm}
Suppose $(t_0,t_1)\in \mathcal F^-$ and $z_0$ and $z_2$ are the branch points given by Theorem~\ref{theorem_singular_points_spectral_curve_negative_t}. The implicit equation
\begin{equation}\label{implicit_equation_critical_curve_negative_t}
\re \int_{z_0}^{z_2}(\xi_1(s)-\xi_2(s))ds=0
\end{equation}
defines an analytic curve $\gamma_c^-$ on the $(t_0,t_1)$-plane, which connects the boundary $\Gamma_c^-$ to the origin.  
\end{thm}

A plot of the curve $\gamma_c^-$ can be seen in Figure~\ref{figure_phase_diagram_negative_t_1}. This curve determines two regions $\mathcal 
F_1^-,\mathcal F_2^-\subset \mathcal F^-$, labeled such that $\mathcal F_1^-$ ($\mathcal F_2^-$) consists of the points in $\mathcal F^-$ that are above (below) $\gamma_c^-$.

Proposition~\ref{proposition_change_of_coordinates} holds true for $(t_0,t_1)\in \mathcal F_-$ without any modification, and Theorem~\ref{theorem_change_coordinates_phase_diagram} 
assumes the following form.

\begin{thm}
On the $(r,a_0)$-plane, the critical curve $\Gamma_c^-$ is expressed as
$$
r=s, \quad a_0=-\frac{1-4s^2}{2},\qquad 0<s<\frac{1}{2},
$$
and the curve $\gamma_c^-$ is implicitly given by
\begin{multline}\label{critical_curve_zeros_r_a_plane_negative_t}
(r+w_0^3)\left[ (30a_0^2r^2 +4a_0^3)w_0^3+ra_0(r^2(12-8a_0)+10a_0)w_0^2\right. \\ \left. +6a_0r^2(5-r^2)w_0 +12r^3+3r^5\right]+3w_0^6r^2(1-4a_0^2-2r^2)\log\frac{w_0^3}{r}=0,
\end{multline}
where $w_0=w_0(r,a_0)$ is the unique real solution to $h'(w)=0$.
\end{thm}

The critical curves $\Gamma_c^-$ and $\gamma_c^-$ on the $(r,a_0)$-plane are displayed in Figure \ref{figure_phase_diagram_a_r_plane_negative_t1}. 

\begin{figure}[t]
 \centering
 \includegraphics[scale=1]{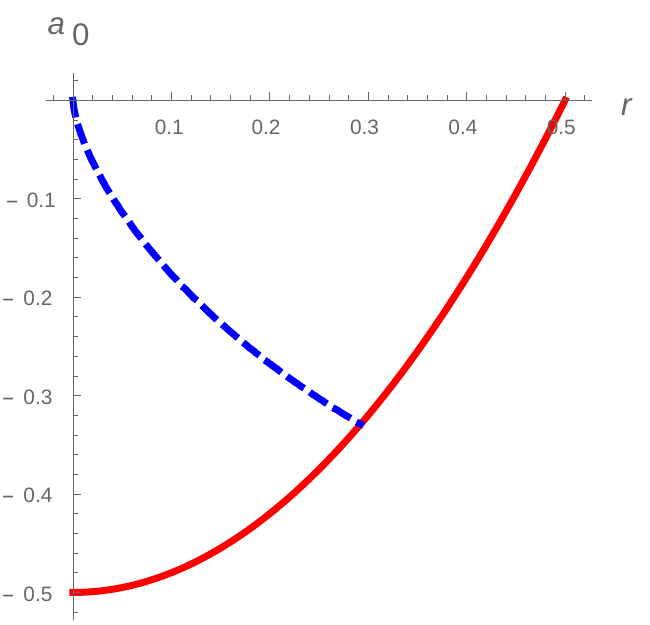}
 \caption{Image of the phase diagram in Figure~\ref{phase_diagram} through the change of variables $(t_0,t_1)\mapsto (r,a_0)$. The solid curve is the image of $\Gamma_c^-$ and the 
dashed curve is the image of $\gamma_c^-$. The region determined by the solid and dashed curves and the $r$-axis corresponds to $\mathcal F_1^-$, whereas the region between the 
solid and dashed curves and the $a_0$-axis corresponds to $\mathcal F_2^-$.}\label{figure_phase_diagram_a_r_plane_negative_t1}
\end{figure}

The function $w_0$ is algebraic in $r$ and $a_0$, and the presence of the $\log$ term in \eqref{critical_curve_zeros_r_a_plane_negative_t} indicates that the curve $\gamma_c^-$ is 
transcendental. In particular, $\gamma_c^-$ is not the analytic continuation of $\gamma_c$ (see also Section~\ref{section_phase_transition_motherbody_curve} below). For an 
indication on how to get \eqref{implicit_equation_critical_curve_negative_t}--\eqref{critical_curve_zeros_r_a_plane_negative_t}, we refer to Remark 
\ref{remark_critical_curve_negative_t1}.

We are ready to state the result equivalent to Theorem~\ref{theorem_limiting_support_zeros}.

\begin{thm}\label{theorem_construction_mother_body_negative_t}
 For $(t_0,t_1)\in \mathcal F^-$, there exists a contour $\Sigma_*\subset \Omega$ for which the function $\xi_1$ in \eqref{asymptotics_xi} admits an analytic 
continuation to $\C\setminus \Sigma_*$ that satisfies the property
 $$
 (\xi_{1+}(s)-\xi_{1-}(s))ds \in i\R,\quad z\in \Sigma_*.
 $$
 
 The contour $\Sigma_*$ is symmetric with respect to the real axis, and
 
\begin{enumerate}[label=(\roman*)]
  \item (Three-cut case) For $(t_0,t_1)\in \mathcal F_1^-$, the contour $\Sigma_*$ decomposes as
$$
\Sigma_*=\Sigma_{*,0}\cup \Sigma_{*,1}\cup \Sigma_{*,2},
$$
where $\Sigma_{*,l}$ is a smooth oriented contour from a common point $z_*\in (-\infty,z_0)$ to the branch point $z_l$ and
$$
\Sigma_{*,0}=[z_*,z_0],\quad (\Sigma_{*,2})^*=\Sigma_{*,1}\subset \overline\C_-.
$$

 \item (One-cut case) For $(t_0,t_1)\in \mathcal F_1^-$, the contour $\Sigma_*$ is a single analytic arc which connects the branch points $z_1$ and $z_2$ and intersects the real 
axis at a point $z_*>z_0$.
\end{enumerate}

Furthermore, the measure
\begin{equation*}
d\mu_*(z)=\frac{1}{2\pi i t_0}(\xi_{1-}(z)-\xi_{1+}(z))dz,\quad z\in \Sigma_*
\end{equation*}
is a probability measure on $\Sigma_*$.
\end{thm}

The arc $\Sigma_*$ for various choices of the parameter $(t_0,t_1)$ is displayed in Figure~\ref{figure_support_mu_star_one_cut}.

\begin{figure}[t]
\begin{minipage}[c]{0.5\textwidth}
\centering
  \begin{overpic}[scale=1]
  {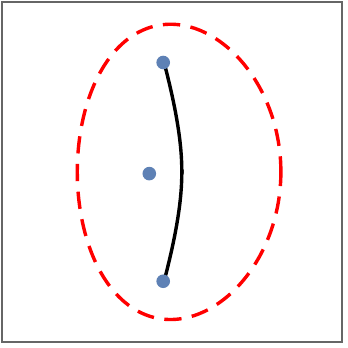}
 \end{overpic}
 \caption*{$(t_0,t_1)=(\frac{1}{2500},-\frac{1}{10})\in\mathcal F_2^-$}
\end{minipage}%
\begin{minipage}[c]{0.5\textwidth}
\centering
 \begin{overpic}[scale=1]
  {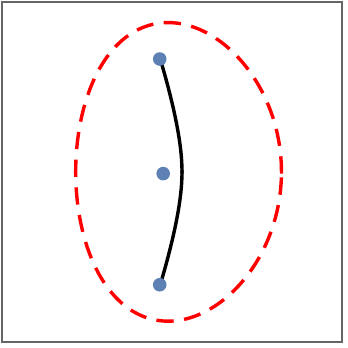}
 \end{overpic}
 \caption*{$(t_0,t_1)=(\frac{1}{1700},-\frac{1}{10})\in\mathcal F_2^-$}
 \end{minipage}
 \begin{minipage}[c]{0.5\textwidth}
 \vspace{0.5cm}
 \centering
  \begin{overpic}[scale=1]
  {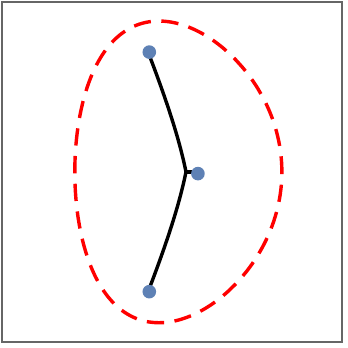}
 \end{overpic}
 \caption*{$(t_0,t_1)=(\frac{1}{500},-\frac{1}{10})\in\mathcal F_1^-$}
\end{minipage}%
\begin{minipage}[c]{0.5\textwidth}
\vspace{0.5cm}
\centering
 \begin{overpic}[scale=1]
  {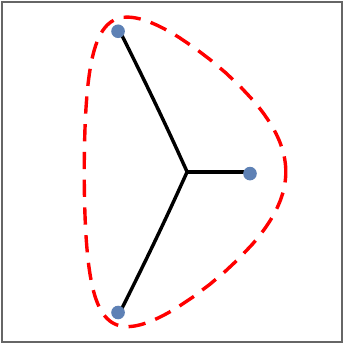}
 \end{overpic}
 \caption*{$(t_0,t_1)=(\frac{1}{50},-\frac{1}{10})\in\mathcal F_1^-$}
 \end{minipage}
 \begin{minipage}[c]{1\textwidth}
 \vspace{0.5cm}
\centering
 \begin{overpic}[scale=1]
  {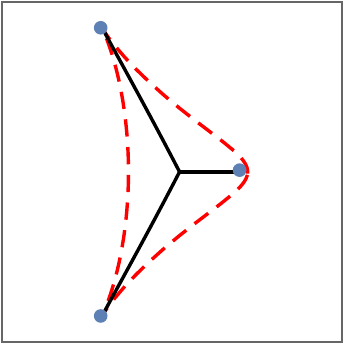}
 \end{overpic}
 \caption*{$(t_0,t_1)=(\frac{135+4\sqrt{30}}{1800},-\frac{1}{10})\in\Gamma_c^-$}
 \end{minipage}
 \caption{For the given values of $(t_0,t_1)$, the boundary $\partial \Omega$ (dashed contour), the support of the mother body measure $\mu_*$ (solid lines) and the branch points 
$z_0$, $z_1$, $z_2$ (dots) are shown (the figures are scaled differently - compare with Figure~\ref{chain_polynomial_curves}). Note the transition when we move from $\mathcal 
F_2^-$ to $\mathcal F_1^-$. We stress that the contours of $\supp\mu_*$ outside the real line are not straight line segments. Numerical outputs.} 
\label{figure_support_mu_star_one_cut}
\end{figure}

In the three-cut cases for $(t_0,t_1)\in \mathcal F_1$ and $(t_0,t_1)\in \mathcal F_1^-$ the geometry of the support $\Sigma_*$ of $\mu_*$ is essentially the same. However, when 
$(t_0,t_1)$ crosses $\gamma_c^-$ and we move to the one-cut case for $t_1<0$, the part of $\Sigma_*$ on the real line disappears, and we are only left with a single analytic 
arc, which is symmetric with respect to the real line. This is in contrast with the one-cut case for $t_1>0$, which corresponds to the shrinking of the arc outside the real line 
and reduction of $\Sigma_*$ to a real interval.

The point $z_*$ in Theorem~\ref{theorem_construction_mother_body_negative_t} moves continuously with $(t_0,t_1)$ and it is determined as one of the finitely many real solutions 
$y=z_*$ to the implicit equation
\begin{equation}\label{equation_determination_intersection_point_negative_t}
\int_{z_1}^{y}(\xi_1(s)-\xi_2(s))ds=0.
\end{equation}
The critical curve $\gamma_c^-$ is then determined by the condition $z_*=z_0$ (compare \eqref{equation_determination_intersection_point_negative_t} with 
\eqref{implicit_equation_critical_curve_negative_t}).

To conclude this section, we remark that Theorem~\ref{theorem_balayage_relation} also extends to $(t_0,t_1)\in \mathcal F^-$ for the respective measure $\mu_0$ and the measure 
$\mu_*$ given by Theorem~\ref{theorem_construction_mother_body_negative_t}.

\subsection{Phase transition along the mother body critical curve}\label{section_phase_transition_motherbody_curve}

As we mentioned above, the critical curves $\gamma_c$ and $\gamma_c^-$ are not analytical continuation of each other. More precisely, if we consider them as functions $t_0(t_1)$, 
then they are both analytic at $t_1=0$ and their values and their first two derivatives are equal to $0$ at $t_1=0$, but their third derivatives at $t_1=0$ are different. This can 
be characterized as a phase transition of the third order. To see this phase transition, let us evaluate the asymptotic behavior of the function $t_0(t_1)$ as $t_1\to +0$ and 
$t_1\to -0$.

We start with $\gamma_c$. When we approach the origin, we take into account the leading terms in \eqref{critical_curve_limiting_zero_distribution} in order to get
$$
t_0=s^6(1+\Boh(s^4)),\quad t_1=\frac{3}{8} s^2(1+\Boh(s^2)),\qquad \mbox{ as } s\to 0.
$$
Hence we have the approximation
\begin{equation}\label{first_order_approximation_positive_mother_body_critical_curve}
t_0 = \frac{8}{27} t_1^3(1+\Boh(t_1)),
\end{equation}
as we approach the origin along $\gamma_c$, where the implicit term is analytic in $t_1$.

The similar analysis for $\gamma_c^-$ is more involved. The relation $h'(w_0)=0$ gives us that
\begin{equation}\label{change_variables_a_w_r}
a_0=\frac{w_0^3-2r}{w_0}.
\end{equation}
Replacing this expression into \eqref{critical_curve_zeros_r_a_plane_negative_t} we arrive at
\begin{multline}\label{alternative_expression_negative_mother_body_critical_curve}
\left(r+w_0^3\right)\left(-24 r^3+ 2 r^5+36 r^4 w_0+\left(10 r^4+22 r^2\right) w_0^3 \right. \\ \left. -48 r^3 w_0^4-\left(4 r^3+r\right) w_0^6+15 r^2 w_0^7+w_0^9\right) 
\\ -6 r^2 w_0^{4} \left(4 r^2+\left(2 r^2-1\right) w_0^2-4 r w_0^3+w_0^6\right)\log \frac{w_0^3}{r}=0.
\end{multline}
We now make the ansatz, to be verified in a moment, that $w_0$ can be expressed as
\begin{equation}\label{change_variables_w_r}
w_0=(C r)^{1/3}
\end{equation}
along $\gamma_c^-$, where $C=C(r)$ is a positive function, to be determined later, which remains bounded away from $0$ and $\infty$ when $r\to 0$. Using the change of variables 
\eqref{change_variables_w_r} in \eqref{alternative_expression_negative_mother_body_critical_curve}, we get

\begin{multline*}
r^4(1+C) \left(-24+22 C-C^2+C^3+\left(2+10 C -4 C^2\right) r^2 \right. \\ \left. +C^{1/3}\left(36  -48 C+15 C^{2} \right) r^{4/3}\right)  
\\ + r^4 \left( 6 C^2 -C^{1/3}\left(24C-24 C^{2}+6 C^{3}\right) r^{4/3}-12 C^2 r^2\right) \log C=0,
\end{multline*}
that is, $C=C(r)$ should satisfy
\begin{equation}\label{implicit_equation_constant_c}
(1+C)(-24+22 C-C^2+C^3) + 6C^2\log C=-Q(r,C(r))
\end{equation}
where
\begin{multline*}
Q(r,c):=\left(2+10 c -4 c^2\right) r^2  +c^{1/3}\left(36  -48 c+15 c^{2} \right) r^{4/3} \\ - \left( c^{1/3}\left(24c-24 c^{2}+6 c^{3}\right) r^{4/3}+12 c^2 r^2\right) 
\log C
\end{multline*}
A simple application of the Implicit Function Theorem tells us that there exists a function $C=C(r)$ satisfying \eqref{implicit_equation_constant_c}. Consequently, we get that 
along $\gamma_c^-$, we can express $w_0$ as a function of $r$ as in \eqref{change_variables_w_r}. Furthermore, when $r\to 0$ it is easily seen that $Q(r,c)\to 0$, thus the 
constant $C_0=C(0)$ solves
$$
(1+C_0)(-24+22 C_0-C_0^2+C_0^3) + 6C_0^2\log C_0=0.
$$
Numerically, we see that
\begin{equation}\label{numerical_approximation_c0}
C_0=1.075\hdots .
\end{equation}
Hence, in virtue of \eqref{change_variables_w_r} and \eqref{change_variables_a_w_r} we get the first order approximation
\begin{equation}\label{first_order_approximation_w_a}
w_0= C_0^{1/3}r^{1/3}(1+\Boh(r)),\quad a_0= \frac{(C_0-2)}{2C_0^{1/3}}r^{2/3}(1+\Boh(r)),
\end{equation}
valid as $r\to 0$ along $\gamma_c^-$, and where the implicit terms are analytic in $r$. By \eqref{system_change_coordinates_a}--\eqref{system_change_coordinates_b}, we know that
\begin{align*}
t_0 & = -2r^4+r^2-\frac{(C_0-2)^2}{C_0^{2/3}}r^{10/3}(1+\Boh(r))=r^2(1+\Boh(r)), \\
t_1 & = a_0(1-a_0-4r^2)=\frac{(C_0-2)}{2C_0^{1/3}}r^{2/3}(1+\Boh(r)),
\end{align*}
so by the second equation in \eqref{first_order_approximation_w_a} we get the first order approximation
\begin{equation}\label{first_order_approximation_negative_mother_body_critical_curve}
t_0= \frac{8C_0}{(C_0-2)^3} t_1^3(1+\Boh(t_1)),
\end{equation}
as $(t_0,t_1)\to (0,0)$ along $\gamma_c^-$, and the implicit term is analytic in $t_1$. Using \eqref{numerical_approximation_c0} we also get
$$
\frac{8}{27}<\left|\frac{8C_0}{(C_0-2)^3}\right|.
$$
Thus, comparing \eqref{first_order_approximation_positive_mother_body_critical_curve} 
with \eqref{first_order_approximation_negative_mother_body_critical_curve}, we see that the tangent vector and the curvature of $\gamma_c$ and $\gamma_c^-$ coincide at the origin, 
but the derivative of their curvatures do not coincide. In the terminology of statistical mechanics, the origin $(t_0,t_1)=(0,0)$ determines a third order phase transition along 
the critical curve $\gamma_c\cup \gamma_c^-$.

\subsection{Setup for the remainder of the paper}

The rest of the paper is organized as follows. 

In Section \ref{section_algebraic_quantities} we derive several technical results on the functions $r$ and $a_0$, which are 
extensively used in the rest of the paper. Propositions \ref{proposition_definition_r} and \ref{proposition_change_of_coordinates} and Theorems 
\ref{theorem_rational_parametrization_polynomial_curve}, \ref{theorem_schwarz_function} and \ref{theorem_change_coordinates_phase_diagram} are proved in Section 
\ref{section_algebraic_quantities}. 

In Section \ref{section_topology_spectral_curve} we study the spectral curve \eqref{spectral_curve} and construct its associated Riemann 
surface $\mathcal R$ as a three-sheeted cover of the complex plane. This sheet structure depends on whether we are in the three-cut (Section \ref{section_sheet_structure_1}) or 
one-cut  (Section \ref{section_sheet_structure_2}) cases. Along the way, we also prove Theorem \ref{theorem_singular_points_spectral_curve} in Section 
\ref{section_topology_spectral_curve}.

In Section \ref{section_quadratic_differential} we introduce the quadratic differential $\varpi$ on the Riemann surface $\mathcal R$ (which was already mentioned after Theorem 
\ref{theorem_limiting_support_zeros}), and describe its critical graph. Using its critical graph, in Section 
\ref{section_proof_s_property} we prove Theorems \ref{thm_equilibrium_measure}, \ref{theorem_density_eigenvalues}, \ref{theorem_limiting_support_zeros} and  
\ref{theorem_balayage_relation}. 

In Sections \ref{section_riemann_hilbert_analysis_precritical} and \ref{section_riemann_hilbert_analysis_postcritical} we carry out the asymptotic analysis of the Riemann-Hilbert 
problem characterizing the multiple orthogonal polynomial $P_{n,n}$ in Proposition \ref{proposition_multiple_orthogonality}, in the three-cut and one-cut cases, respectively. 
This analysis also heavily relies on the critical graph of the quadratic differential $\varpi$. The final ingredient in the analysis of this Riemann-Hilbert problem is the 
so-called global parametrix, whose construction is provided in Section \ref{section_global_parametrix}. 

In Section \ref{section_proof_theorem_limiting_counting_measure} we use the outcome of the asymptotic analysis in order to prove Theorems~\ref{theorem_limiting_counting_measure} 
and \ref{thm_wave_factor}. 

Finally, in the Appendix \ref{appendix_widths} we study the width parameters used in Section \ref{section_quadratic_differential} to perform the deformation of the critical graph 
of $\varpi$.

\section{Limiting boundary of eigenvalues. Proof of Propositions \ref{proposition_definition_r} and \ref{proposition_change_of_coordinates} and Theorems 
\ref{theorem_rational_parametrization_polynomial_curve}, \ref{theorem_schwarz_function} and \ref{theorem_change_coordinates_phase_diagram}}\label{section_algebraic_quantities}

In order to prove the main results of this section we need some technical lemmas, which are also used in the next sections.

\subsection{Proof of Proposition \ref{proposition_definition_r}}\label{analysis_polynomial_p}

It is convenient to change variables for the polynomial $p$ in \eqref{definition_polynomial_p} and instead consider
\begin{multline}\label{definition_p_tilde}
\widetilde p(x)= p(\sqrt x)   =  128 x^{5}-124 x^4+ (-16 t_1+64 t_0 +36)x^3 \\ 
                     + \left(16 t_1^2+8 t_1-28 t_0 -3\right)x^2 +  t_0(2-8 t_1)  x + t_0^2. 
\end{multline}

With the help of Mathematica, the discriminant of $\widetilde p$ with respect to $x$ is computed
\begin{equation}\label{discriminant_p}
\disc(\widetilde p;x)=16384\; t_0^2 \; p_1(t_0)p_2(t_0)p_3(t_0),
\end{equation}
where
\begin{equation}\label{decomposition_discriminant_p}
\begin{aligned}
p_1(s) & = 8192 s^3+192 (64 t_1-7)s ^2-48 (1-4 t_1)^2s +(108 t_1-11) (4 t_1-1)^3  \\
p_2(s) & = 1728 s^2 -432(1+4t_1)s +(3+4t_1)^2(3+16t_1), \\
p_3(s) & = (8s-8t_1-1)^2.
\end{aligned}
\end{equation}

\begin{lem}\label{lemma_discriminant_p}
For $t_1>0$, the polynomials $p_2,p_3$ in \eqref{decomposition_discriminant_p} do not have zeros on $(0,1/8)$.
\end{lem}
\begin{proof}
The discriminant of $p_2$ in $s$ is
$$
-1769472 \; t_1^3,
$$
which is clearly negative for $t_1>0$, and hence $p_2$ does not have real zeros. The Lemma for $p_3$ follows trivially 
from the inequality
$$
8s-8t_1-1<-8t_1<0,\quad 0<s<1/8.
$$
\end{proof}

\begin{lem}\label{lemma_roots_p_1}
The three roots $s_1,s_2,t_{0,crit}$ of the polynomial $p_1$ in \eqref{decomposition_discriminant_p} satisfy
\begin{align}
s_1<0<s_2<t_{0,crit}, & \quad 0<t_1<\frac{11}{108}; \label{lemma_roots_p_1_equation_1_a}\\
s_1<s_2=0<t_{0,crit}, & \quad t_1=\frac{11}{108}; \label{lemma_roots_p_1_equation_1_b}\\
s_1 \leq s_2<0<t_{0,crit}, & \quad \frac{11}{108}<t_1<\frac{1}{4}; \quad \mbox{ and } s_1=s_2 \mbox{ only for } t_1=\frac{1}{8}  \label{lemma_roots_p_1_equation_1_c}
\end{align}

Moreover, the function 
$$
t_1\mapsto t_{0,crit}=t_{0,crit}(t_1),\quad 0<t_1<1/4,
$$ 
is decreasing and
\begin{equation}\label{lemma_roots_p_1_equation_7}
t_{0,crit}(0)=1/8, \quad t_{0,crit}(1/4)=0.
\end{equation}

Finally, the curve $\Gamma_c$ in \eqref{definition_critical_curve_phase_transition} is parameterized by $(t_{0,crit},t_1)$, that is,
\begin{equation}\label{lemma_roots_p_1_equation_2}
\Gamma_c: \ (t_0,t_1)=(t_{0,crit}(t_1),t_1),\quad 0<t_1<\frac{1}{4}.
\end{equation}
\end{lem}

\begin{proof}
For a positive constant $c$, the discriminant of $p_1$ (with respect to $s$) is given by

$$
\disc(p_1;s)=c \; t_1^3(1-8t_1)^2(1-4t_1)^3.
$$

In particular, $\disc(p_1;s)\geq 0$ for $0<t_1<1/4$, so $p_1$ has always three real roots for $0<t_1<1/4$, and these are all distinct if $t_1\neq 1/8$. For 
$t_1=1/8$, $p_1$ simply factors as
\begin{equation}\label{lemma_roots_p_1_equation_3}
p_1(s)=\frac{1}{12}(1+32s)^2(128s -5),
\end{equation}
so the largest positive root $t_{0,crit}$ of $p_1$ is always simple if $0<t_1<1/4$, and also $s_1=s_2=-\frac{1}{32}<0$ for $t_1=1/8$. Evaluating explicitly,
$$
p_1(0)=(1-4t_1)^3(11-108t_1),
$$
which means that $p_1(0)<0$ for $t_1<\frac{11}{108}$, and $p_1(0)\geq 0$ otherwise. Therefore, combining continuity and what we already have, we get the inequalities claimed in 
\eqref{lemma_roots_p_1_equation_1_a}--\eqref{lemma_roots_p_1_equation_1_c}.

We now verify that $t_{0,crit}<1/8$. For $t_1\in (11/108,1/4)$, the value $t_{0,crit}$ is the unique positive root of $p_1$. Since also 
$$
p_1(1/8)=64 t_1^2(108 t_1^2 -92t_1 +27)>0,\quad 0<t_1<\frac{1}{4},
$$
and the leading coefficient of $p_1$ is positive, we use continuity to conclude that the inequality $t_{0,crit}<1/8$ always holds true.

The derivative of $t_{0,crit}$ is computed via the chain rule,
\begin{equation}\label{lemma_roots_p_1_equation_5}
\frac{\partial }{\partial t_1}t_{0,crit}= - \frac{\frac{\partial p_1}{\partial t_1}}{\frac{\partial p_1}{\partial s}}.
\end{equation}

Since $t_{0,crit}$ is the largest root of $p_1$,
\begin{equation}\label{lemma_roots_p_1_equation_6}
\frac{\partial p_1}{\partial s}(t_{0,crit})>0,\quad 0<t_1<\frac{1}{4}.
\end{equation}

The derivative of $p_1$ with respect to $t_1$ is explicitly computed to be
$$
\frac{\partial p_1}{\partial t_1}(s)=48(256s^2 + s (8-32t_1) - (1-4t_1)^2(5-36t_1)).
$$

For $t_1=1/8$, \eqref{lemma_roots_p_1_equation_3} gives
\begin{equation}\label{lemma_roots_p_1_equation_4}
t_{0,crit}=\frac{5}{128},\quad \frac{\partial p_1}{\partial t_1}(t_{0,crit})\Big|_{t_1=1/8}=\frac{81}{24}.
\end{equation}

We claim that $\frac{\partial p_1}{\partial t_1}(t_{0,crit})$ is never zero. To the contrary, the discriminant $\disc(p_1;t_1)$ with respect to $t_1$ would 
be zero, and it is thus enough to show that $\disc(p_1;t_1)$ is different from zero. For a positive constant $c$,
$$
\disc(p_1;t_1)=c \; s^2(8s-1)^3(1+32s)^2;
$$
and since $t_{0,crit}<1/8$ we conclude
$$
\disc(p_1;t_1)\Big|_{s=t_{0,crit}}\neq 0,
$$
so $\frac{\partial p_1}{\partial t_1}(t_{0,crit})$ is never zero for $0<t_1<1/4$. From \eqref{lemma_roots_p_1_equation_4} and continuity,
$$
\frac{\partial p_1}{\partial t_1}(t_{0,crit})>0,\quad 0<t_1<\frac{1}{4}.
$$

Combining this last equation with \eqref{lemma_roots_p_1_equation_5} and \eqref{lemma_roots_p_1_equation_6}, we finally get that $t_{0,crit}$ is a decreasing function of $t_1$.

The limits \eqref{lemma_roots_p_1_equation_7} now follow directly from a combination of \eqref{lemma_roots_p_1_equation_1_a}--\eqref{lemma_roots_p_1_equation_1_c} and the 
explicitly expressions
$$
p_1(s)\Big|_{t_1=0}=(1-8s)^2(11+128s),\quad p_1(s)\Big|_{t_1=1/4}=64s^2(27+128s).
$$

To conclude, we prove \eqref{lemma_roots_p_1_equation_2}. Plugging the parametrization \eqref{definition_critical_curve_phase_transition} into the definition of $p_1$, one can 
easily verify
\begin{equation*}
p_1(t_0)=0,\quad (t_0,t_1)\in \Gamma_c,
\end{equation*}
so $t_0$ is always a root of $p_1$ if $(t_0,t_1)\in \Gamma_c$. For $s=1/4$ in \eqref{definition_critical_curve_phase_transition}, we compute explicitly 
$(t_0,t_1)=(5/128,1/8)=(t_{0,crit},t_1)$, 
hence 
$\Gamma_c$ intersects $(t_0,t_{0,crit})$ at $(5/128,1/8)$. Since $t_{0,crit}$ is always a simple root and the pair $(t_0,t_1)\in \Gamma_c$ must always give rise to a root of 
$p_1$, by continuity we conclude \eqref{lemma_roots_p_1_equation_2}
\end{proof}

As as consequence of Lemma~\ref{lemma_roots_p_1}, the parameter region $\mathcal F$ is alternatively described as
$$
\mathcal F= \left\{ (t_0,t_1) \; | \: 0<t_1<1/4,\quad 0<t_0<t_{0,crit}(t_1) \right\}.
$$

\begin{lem}\label{lemma_triple_root}
 For $(t_0,t_1)\in \mathcal F\cup \Gamma_c$, $t_1\neq 0,1/4$, the polynomial $\widetilde p$ in \eqref{discriminant_p} never has a triple root.
\end{lem}
\begin{proof}
 If $\widetilde p$ has a triple root, then $\widetilde p,\widetilde p',\widetilde p''$ share a common root, say $x_0$. This means that $t_0\mapsto 
\widetilde p(x_0),\widetilde p'(x_0),\widetilde p''(x_0)$ all share a common root. We compute two of their resultants with the help of Mathematica. Their full 
expressions are rather long, but we exhibit their first coefficients,
\begin{align*}
& \result(p(x),p''(x);t_0) = - 44040192 \; x^7+ 56033280 \; x^6+\hdots, \\
& \result(p'(x),p''(x);t_0) = 245760 \; x^5 - 202752 \; x^4+\hdots,
\end{align*}
and since $t_0\mapsto \widetilde p(x_0),\widetilde p'(x_0),\widetilde p''(x_0)$ all share a common root,
$$
\result(p'(x_0),p''(x_0);t_0)=0=\result(p(x_0),p''(x_0);t_0)
$$

We now see $\result(p(x),p''(x);t_0),\result(p'(x),p''(x);t_0)$ as functions of $x$. The equation above says that these polynomials share a common root $x_0$, so their resultant 
with respect to $x$ must be zero. Again with the help of Mathematica, we compute their resultant with respect to $x$ to get
$$
 -c \; t_1^6 (1+t_1)^3(1-4t_1)^4(3+4t_1)^3(27+4t_1)^2(775+864t_1+13824 t_1^2)^2
$$
for some large positive constant $c$. It is then not hard to see that this last expression is never zero for $0<t_1<1/4$, and the Lemma follows.
\end{proof}

\begin{lem}\label{lemma_tilde_p}
 For $(t_0,t_1)\in \mathcal F$, the polynomial $\widetilde p$ defined in \eqref{discriminant_p} has a smallest positive root $\widetilde r$, which is simple. When $(t_0,t_1)\to 
\Gamma_c$, $\widetilde r$ becomes a root of higher multiplicity, given explicitly as $\widetilde r=s^2$, where $s$ is the parameter on 
\eqref{definition_critical_curve_phase_transition}.
\end{lem}

\begin{proof}

When $t_0\to 0$, the polynomial $\widetilde p$ factors into
\begin{equation}\label{lemma_tilde_p_eq_1}
\widetilde p(x)=x^2q(x),\quad q(x):=128x^3 - 124x^2 +(36-16t_1)x + 16t_1^2+8t_1-3.
\end{equation}

The discriminant and value at $x=0$ of $q$ are respectively given by
\begin{align}
\disc(q;x) & = - 1024(1+8t_1)^2(3+16t_1)(108t_1-11), \label{lemma_tilde_q_eq_2a}\\
\restr{ q(0)}{t_0=0} & = 16t_1^2+8t_1-3<0,\quad 0<t_1<\frac{1}{4} \label{lemma_tilde_q_eq_2b}
\end{align}

For the sake of clarity, we split the rest of the proof into three parts, the last of those being a limiting case of the others.

{\it \nth{1} Case:} $0<t_1<\frac{11}{108}$.

In this case, the discriminant \eqref{lemma_tilde_q_eq_2a} is positive, so $q$ has three real roots. For $t_1=1/16$, these are given by
$$
\left\{\frac{3}{8}, \frac{1}{64}(19-3\sqrt{17}),\frac{1}{64}(19+3\sqrt{17})\right\},
$$
so in particular they are all positive. Since $q(0)$ is never zero, see \eqref{lemma_tilde_q_eq_2b}, we conclude that these three roots are all positive for any choice $t_1\in 
(0,11/108)$. This is the same as saying that in the present situation and $t_0=0$, the polynomial $\widetilde p$ has a double root $x_1=\widetilde r=0$ and three simple positive 
roots 
$x_2<x_3<x_4$.

Recall that $s_2$ is a root of $p_1$ as in Lemma~\eqref{lemma_roots_p_1}. For $0<t_0<s_2$, a combination of Equation~\eqref{discriminant_p}, 
Lemma~\ref{lemma_discriminant_p} and Equation~\eqref{lemma_roots_p_1_equation_1_a} assures $\widetilde p$ has only simple roots. By continuity from $t_0$ and the further 
observations $\widetilde p(0)=t_0^2>0$ and $\widetilde p(x)<0$ when $x\to -\infty$, we learn that for $0<t_0<s_2$, the double root of $\widetilde p$ at $x=0$ splits into two 
simple 
roots $x_1<0<\widetilde r$ and the remaining roots still satisfy $\widetilde r<x_2<x_3<x_4$. 

We now approach $t_0\nearrow s_2$. In this situation, two roots of $\widetilde p$ collide, because $p_1$ - hence $\disc(\widetilde p;x)$, see \eqref{discriminant_p} - is zero for 
$t_0=s_2$. For $t_1=1/16$, $t_0=s_2$, the roots of $\widetilde p$ are computed numerically
$$
s_2=0.0512061,\quad \{x_1,\widetilde r,x_2,x_3,x_4\}\approx \{ -0.0169, 0.0607,0.1235,0.4007,0.4007\},
$$
so in this case $x_3=x_4$ and $\widetilde r$ is a simple root. By continuity and Lemma~\ref{lemma_triple_root}, we conclude $x_3=x_4$ whenever $t_0=s_2$, and 
hence $\widetilde r$ is always a simple root for $t_0\leq s_2$. Again by Equation~\eqref{discriminant_p}, Lemma~\ref{lemma_discriminant_p} 
and Equation~\eqref{lemma_roots_p_1_equation_1_a}, we know that $\widetilde p$ does not have multiple roots for $s_2<t_0<t_{0,crit}$, and we finally conclude that $\widetilde r$ 
is always a simple root in the present case.

{\it \nth{2} Case:} $\frac{11}{108}<t_1<\frac{1}{4}$.

This case is somewhat simpler than the previous one. In the present case, the discriminant \eqref{lemma_tilde_q_eq_2a} is negative, so $q$ has one real root and two non real 
roots. From \eqref{lemma_tilde_q_eq_2b} we see that this real root is positive. Similarly as before, it means $\widetilde p$ has a double root $\widetilde r=x_1=0$, and simple 
roots $x_2>0$, $x_3,x_4\in \C\setminus\R$. 

As before, we compute $p(0)=t_0^2$, $ \widetilde p(x)<0$ for $x\to -\infty$ and conclude that for $t_0>0$ and small, the double root splits into two simple 
roots $x_1<0<\widetilde r$ and the remaining roots still satisfy $x_2>0$, $x_3,x_4\in \C\setminus\R$. 

From Equation~\eqref{discriminant_p}, Lemma~\ref{lemma_discriminant_p} and Equation~\eqref{lemma_roots_p_1_equation_1_a} we know $\widetilde p$ has no multiple roots for 
$0<t_0<t_{0,crit}$, so we conclude the smallest positive root $\widetilde r$ is always simple.

{\it \nth{3} Case:} $t_1=\frac{11}{108}$. 

In this case the polynomial $q$ simply factors as
$$
q(x)=\frac{4}{729}(9x-4)^2(288x-23),
$$
and it clearly has three positive roots. The rest follows the same lines as \nth {2} Case.

In either of the cases above, we consider the limit $(t_0,t_1)\to (t_{0,crit},t_1)\in \Gamma_c$ and plug the parametrization \eqref{definition_critical_curve_phase_transition} 
into the expression 
for $\tilde p$ in \eqref{definition_p_tilde}, arriving at
\begin{equation*}
\widetilde p(x)=4(s^2-x)^2\left(32x^3+(64s^2-31)x^2+(48s^3-50s^2+8)+9s^4-12s^3+4s^2\right),
\end{equation*}
so $s^2$ is a root of $\tilde p$ for $(t_0,t_1)\in \Gamma_c$. The discriminant (in $x$) of the polynomial inside brackets above is
$$
144(1-2s)^3(2-3s)^2(8s^2-4s-1)(64s^2+35s+7)< 0,\quad 0<s<1/4,
$$
so that polynomial has only one real root. Its value at $x=0$ is $9s^4-12s^3+4s^2>0$, hence this root is negative. 

Comparing to what we proved before, it means that for $(t_0,t_1)\in \Gamma_c$, the root $\widetilde r$ collides with the root $x_2$, becoming the double root $\widetilde 
r=s^2$, as we want.

\end{proof}

\begin{lem}\label{lemma_tilde_p_2}
 The function 
 $$
 t_1 \mapsto \widetilde r=\widetilde r(t_0,t_1),\quad (t_0,t_1)\in \mathcal F,
 $$
 is increasing and satisfies the inequalities
 \begin{align}
 & \widetilde r <\frac{1}{4}, \label{fundamental_inequalities_tilde_r_a}\\
 & 4t_1<(1-4\widetilde r)^2. \label{fundamental_inequalities_tilde_r_b}
 \end{align}
\end{lem}

\begin{proof}
The chain rule tells us
\begin{equation}\label{lemma_tilde_p_eq_3}
\frac{\partial }{\partial t_1}\widetilde r=-  \frac{\frac{\partial \widetilde p}{\partial t_1}}{\frac{\partial \widetilde p}{\partial x}}.
\end{equation}

As it followed from the calculations in the proof of Lemma~\ref{lemma_tilde_p}, $\widetilde r$ is the second smallest real root of $\widetilde p$, hence
\begin{equation}\label{lemma_tilde_p_eq_4}
\frac{\partial \widetilde p}{\partial x}(\widetilde r)<0.
\end{equation}

Moreover,
$$
\frac{\partial \widetilde p}{\partial t_1}=-8x(2x^2-(4t_1+1)x+t_0).
$$

For $t_1=1/8$, $t_0=1/32<t_{0,crit}=5/128$, we compute
$$
\widetilde r\approx 0.040736,\quad \frac{\partial \widetilde p}{\partial t_1}(\widetilde r)\approx 0.00864782>0,
$$
thus the monotonicity follows from \eqref{lemma_tilde_p_eq_3}, \eqref{lemma_tilde_p_eq_4} and the previous inequality once we prove that $\frac{\partial \widetilde p}{\partial 
t_1}(\widetilde r)$ is never zero.

If $\frac{\partial \widetilde p}{\partial t_1}(\widetilde r)$ were zero, then the polynomials $\widetilde p$, $\frac{\partial \widetilde p}{\partial t_1}$ would 
share a zero, hence their resultant with respect to $x$ would be zero. But
$$
\result\left(\widetilde p, \frac{\partial \widetilde p}{\partial t_1}(\widetilde r);x \right)= c\; t_1^2 t_0^4 (1+8t_1-8t_0)^2,
$$
and this last expression is never zero for $0<t_1<1/4$, $0<t_0<t_{0,crit}<1/8$.

We now prove \eqref{fundamental_inequalities_tilde_r_a}. In the proof of Lemma~\ref{lemma_tilde_p}, we already observed that $\widetilde r\to 0$ when $t_0\to 0$, see 
\eqref{lemma_tilde_p_eq_1}. In particular \eqref{fundamental_inequalities_tilde_r_a} is valid for $t_0$ very small. In addition
$$
\widetilde p(1/4)=\frac{1}{64}(1+8t_1-8t_0)^2\neq 0,
$$
because $0<t_1<1/4$, $0<t_0<t_{0,crit}<1/8$, implying that $\widetilde r \neq 1/4$. By continuity from the case $t_0= 0$, we get 
\eqref{fundamental_inequalities_tilde_r_a}.

For \eqref{fundamental_inequalities_tilde_r_b}, define
$$
f(x):=(1-4x)^2-4t_1,\quad x\in \C.
$$

We want to prove that $f(\widetilde r)>0$. Using \eqref{fundamental_inequalities_tilde_r_a} and the monotonicity of $\widetilde r$
$$
\frac{\partial }{\partial t_1}(f(\tilde r))=-8(1-4\widetilde r)\frac{\partial \widetilde r}{\partial t_1}-4<0,
$$
hence $t_1\mapsto f(\tilde r)$ is decreasing, so it is enough to prove $f(\widetilde r)\geq 0$ for $(t_0,t_1)\in \Gamma_c$. But for this choice, we know that $\widetilde 
r=s^2$, where $s$ is the parameter in 
\eqref{definition_critical_curve_phase_transition}, and in this case
$$
f(\widetilde r)=4s^2(1-2s^2)^2>0,\quad (t_0,t_1)\in \Gamma_c,\quad t_1\neq 0,1/4,
$$
as desired.
\end{proof}

As a consequence of Lemmas~\ref{lemma_tilde_p} and \ref{lemma_tilde_p_2}, we get the following refinement of Proposition \ref{proposition_definition_r}.

\begin{thm}\label{theorem_monotonicity_r}
 For $(t_0,t_1)\in \mathcal F$, the polynomial $p$ in \eqref{definition_polynomial_p} has a smallest positive root $r$. The function $t_1\mapsto r$ is increasing and satisfies
 \begin{equation}\label{equation_theorem_monotonicity_r}
 0<r<\frac{1}{2},\quad 4t_1<(1-4r^2)^2,\quad (t_0,t_1)\in \mathcal F.
 \end{equation}
 
 In the limit $(t_0,t_1)\to (t_{0,crit},t_1)\in\Gamma_c$, $r$ becomes a root of higher multiplicity of $p$, explicitly given as $r=s$, where $s$ is the parameter in 
\eqref{definition_critical_curve_phase_transition}.

Still for $(t_0,t_1)\in \mathcal F$, the quantity $a_0=a_0(t_0,t_1)$ in \eqref{definition_a_0} is well defined and positive, and the function $t_1\mapsto a_0$ is 
increasing.

Finally, for $(t_0,t_1)\in \mathcal F$, it is valid
 \begin{equation}\label{inequality_r_a_0_cusp}
 0<2a_0+2r<1,
 \end{equation}
 and the equality
 \begin{equation}\label{equality_r_a_0_cusp}
 2a_0+2r=1
 \end{equation}
 is attained when $(t_0,t_1)\in \Gamma_c$.
\end{thm}

\begin{proof}
 The first two claims about $r$ follow directly as a consequence of Lemmas~\ref{lemma_tilde_p} and \ref{lemma_tilde_p_2}, having in mind the identification 
\eqref{definition_p_tilde}.

 The positivity of $a_0$ follows directly from \eqref{definition_a_0} and \eqref{equation_theorem_monotonicity_r}. 

Simple computations show that $a_0$ is the smallest root of
$$
q(x):=x^2+x(4r^2-1)+t_1.
$$

This implies
$$
\frac{\partial a_0}{\partial t_1} =-\frac{\frac{\partial q}{\partial t_1}}{\frac{\partial q}{\partial x}}=-\frac{8 a_0 r \frac{\partial r}{\partial t_1}+1}{\frac{\partial 
q}{\partial x}(a_0)}.
$$

Note that the discriminant of $q$ is given by $(1-4r^2)^2-4t_1$, which is strictly positive due to \eqref{equation_theorem_monotonicity_r}. So $a_0$ is always a simple root, 
and since it is the smallest one,
$$
\frac{\partial q}{\partial x}(a_0) <0.
$$

Moreover, since $t_1\mapsto r$ is increasing by Theorem~\ref{theorem_monotonicity_r}, we also know that $ \frac{\partial r}{\partial t_1}> 0$. Recalling that we already proved 
$a_0>0$, we conclude from the last two equations
$$
\frac{\partial a_0}{\partial t_1}> 0,
$$
so $t_1\mapsto a_0$ is indeed increasing.

 The positivity of $2a_0+2r$ is trivial since both $a_0,r$ are positive for $(t_0,t_1)\in \mathcal F$. From the monotonicity of $r$ and $a_0$, it is enough to prove 
\eqref{equality_r_a_0_cusp} in order to conclude \eqref{inequality_r_a_0_cusp}. But \eqref{equality_r_a_0_cusp} follows immediately from the definition 
of $a_0$ in \eqref{definition_a_0}, Equation \eqref{definition_critical_curve_phase_transition} and the value $r=s$ given by Theorem~\ref{theorem_monotonicity_r}.
\end{proof}

\subsection{Proof of Theorems \ref{theorem_rational_parametrization_polynomial_curve}, \ref{theorem_schwarz_function} and \ref{theorem_change_coordinates_phase_diagram} and 
Proposition \ref{proposition_change_of_coordinates} }\label{section_analysis_rational_parametrization}

We proceed to the analysis of the rational function $h$ given in \eqref{rational_parametrization}. 

Denote by
\begin{equation*}
R_0=R_0(h)=\sup\{|w| \mid h'(w)=0\},
\end{equation*}
the {\it critical radius} of $h$. The relevance of $R_0$ comes from the fact that the rational function $h$ is injective on $\overline \C\setminus \overline{D}_{R_0}$, where 
$D_{R}$ denotes the open disc centered at $0$ and radius $R$.

\begin{lem}\label{lemma_zeros_derivative_h}
For the rational function $h$ in \eqref{rational_parametrization} and $(t_0,t_1)\in \mathcal F$, the inequality $R_0<1$ holds true.
\end{lem}
\begin{proof}
 Note that
 $$
 h'(w)=r-\frac{2a_0r}{w^2}-\frac{2r^2}{w^3},
 $$
 
 so the zeros of $h'$ are solutions to the equation
 $$
 w^3-2a_0w-2r=0.
 $$
 
 Using \eqref{inequality_r_a_0_cusp}, we get
 $$
 |(w^3-2a_0w-2r)-w^3|\leq 2a_0+2r<1=|w^3|,\quad w\in \partial \D.
 $$
 
 By Rouch\'{e}'s Theorem we conclude that all the roots of $h'$ are on $\D$.
\end{proof}

\begin{cor}\label{corollary_h_parametrization}
 For $(t_0,t_1)\in \mathcal F$, $h$ is a biholomorphism from $\overline \C\setminus \overline\D$ to $ h(\overline \C\setminus\overline \D)$.
\end{cor}
\begin{proof}
From Lemma~\ref{lemma_zeros_derivative_h}, we know that $R_0<1$. In particular $\overline \C\setminus \overline \D\subset \overline \C \setminus \overline D_{R_0}$, and the result 
follows.
\end{proof}

As a consequence of Corollary~\ref{corollary_h_parametrization}, the set $h(\partial \D)$ is an analytic closed curve that splits $\overline \C$ into two simply 
connected domains; only one of which is bounded and henceforth denoted by $\Omega$.

\begin{proof}[Proof of Theorem~\ref{theorem_schwarz_function}]
A straightforward computation shows that
$$
F(h(w^{-1}),h(w))=0,\quad w\in \C,
$$
where $F$ is as in \eqref{spectral_curve}. For a rational function $\chi=\frac{\chi_1}{\chi_2}$, its {\it degree} is defined as 
$$
\deg \chi=\max\{\deg\chi_1,\deg\chi_2\}.
$$

Since
$$
\max\{\deg_\xi F,\deg_z F\}=3=\max\{\deg h(w),\deg h(w^{-1})\},
$$
it follows from \cite[Theorem~4.21]{book_rational_curves} that $(h(w^{-1}),h(w))$ is a proper parametrization.

From Corollary \ref{corollary_h_parametrization} we know that $h$ maps $V:=\overline\C\setminus \overline{D}_{R_0}$ biholomorphically to an open simply connected set $G\subset 
\overline\C$ with $(\overline\C\setminus \Omega)\subset G$. This means that $h$ admits a meromorphic inverse $g:G\to V$, so that
\begin{equation}\label{definition_omega}
h(g(z))=z,\quad z\in G.
\end{equation}
From its definition, it follows that $g$ maps $\partial \Omega$ to $\partial \D$, and we conclude
\begin{equation}\label{conjugate_property_function_omega}
\overline{g (z)}=\frac{1}{g(z)},\quad z\in \partial \Omega.
\end{equation}
Furthermore,
$$
F\big(h(1/g(z)),h(g(z))\big)=0,
$$
so that the meromorphic function
\begin{equation}\label{schwarz_function_S}
S(z)=h\left(\frac{1}{g(z)}\right),\quad z\in G,
\end{equation}
is a solution to the algebraic equation \eqref{spectral_curve}. Since $\xi_1$ is the only solution in \eqref{asymptotics_xi} that is not branched at $\infty$, we conclude that $S$ 
has to be a meromorphic continuation of $\xi_1$ to the neighborhood $G$ of $\overline \C\setminus \Omega$. Thus using the fact that $h$ has 
real coefficients and \eqref{definition_omega}--\eqref{conjugate_property_function_omega}
$$
\overline z= \overline{h(g(z))}=h(\overline{g(z)})=h\left(\frac{1}{g(z)}\right)=S(z),\quad z\in \partial \Omega,
$$
which shows that the meromorphic continuation $S$ of $\xi_1$ is the Schwarz function of $\partial \Omega$.
\end{proof}

\begin{cor}\label{corollary_genus_0}
 The Riemann surface $\mathcal R$ defined by \eqref{spectral_curve} has genus $0$.
\end{cor}

\begin{proof}
 From Theorem~\ref{theorem_schwarz_function}, we learn the rational parametrization $h$ defines a biholomorphism between $\mathcal R$ and $\overline 
\C$. Since the genus of $\overline \C$ is $0$, the same holds true for $\mathcal R$ 
\cite[pg.~90, Remark~(2)]{book_rational_curves}.
\end{proof}

\begin{proof}[Proof of Theorem~\ref{theorem_rational_parametrization_polynomial_curve}]
 After Corollary~\ref{corollary_h_parametrization}, it only remains to compute the area and harmonic moments \eqref{harmonic_moments}. The computations are very much the same as 
in \cite[pg.~1290]{bleher_kuijlaars_normal_matrix_model}. We include them here for completeness.
 
 Using Green's formula on $\Omega$,
 \begin{equation}\label{equation_area_omega}
 2i\area(\Omega)=2i\iint_\Omega dA(z)=\int_{\partial \Omega} \overline z \; dz.
 \end{equation}
 
 By \eqref{equation_schwarz_function},
 $$
 2i \area(\Omega)=\int_{\partial \Omega} \xi_1(z) \; dz.
 $$
 
We now deform the above integral to $z=\infty$ and use the expansion \eqref{asymptotics_xi} and the Residues Theorem to get that this last integral is equal to 
$2\pi i t_0$.
 
Having in mind the identity
 $$
 \frac{1}{2\pi i}\int_{\partial \Omega}\frac{\overline z}{(z-z_0)^k} \; dz=\frac{1}{2\pi i}\int_{\partial \Omega}\frac{\xi_1(z)}{(z-z_0)^k}\; dz,
 $$ 
the equalities in \eqref{harmonic_moments} follow in a similar fashion.
\end{proof}

\begin{proof}[Proof of Proposition~\ref{proposition_change_of_coordinates}]
 Equation~\eqref{system_change_coordinates_b} follows directly from identity \eqref{definition_a_0}. In fact, it was already used in the proof of 
Theorem~\ref{theorem_monotonicity_r}.

Recall Equation~\ref{equation_area_omega},
$$
2\pi it_0=2i\area(\Omega)=\int_{\partial\Omega}\overline z dz.
$$

We now change coordinates $z=h(w)$ in the integral above, and the formula becomes
\begin{align*}
2\pi i t_0 & = \int_{\partial \D} \overline{h(w)}h'(w)dw \\ 
           & = \int_{\partial\D}h(\overline w)h'(w)dw \\
           & = \int_{\partial \D} h(w^{-1})h'(w) dw
\end{align*}

Expanding the integrand $h(w^{-1})h'(w)$ and using Residues Theorem, we get
$$
2\pi i t_0 =2\pi i \res(h(w^{-1})h'(w),w=0)=2\pi i(-4 a_0^2 r^2-2 r^4+r^2),
$$
which is equivalent to \eqref{system_change_coordinates_a}
\end{proof}

Recall the curve $\gamma_c$ splitting our phase diagram $\mathcal F$ into two parts $\mathcal F_1$, $\mathcal F_2$, see \eqref{critical_curve_limiting_zero_distribution}. 
Lemma~\ref{lemma_zeros_derivative_h} assures us the critical points of $h$ are on the unit disc, but it is important for later to have a better control on the position of these 
points, as it is stated in the next two lemmas.

\begin{lem}\label{lemma_location_zeros_derivative_h}
The zeros $w_0,w_1,w_2$ of the function $h'$ satisfy
 \begin{itemize}
  \item For $(t_0,t_1)\in \mathcal F_1$,
  $$
  w_0\in (0,1),\quad \overline w_2=w_1,\quad w_1\in \D\setminus \R.
  $$
  
  \item For $(t_0,t_1)\in \mathcal F_2$,
  $$
  w_0\in (0,1),\quad -1<w_1<w_2<0.
  $$
  
  \item For $(t_0,t_1)\in \gamma_c$, $h'$ has a simple root $w_0\in (0,1)$ and a double root $w_1=w_2\in (-1,0)$.
  
 \end{itemize} 
 
Furthermore,
\begin{equation}\label{branched_solution_non_vanishing}
 h\left(\frac{1}{w_j}\right)\neq 0,\quad j=0,1,2.
\end{equation}

\end{lem}
\begin{proof}
The zeros of $h'$ are the same as the zeros of the polynomial $\tilde h$ given by
 \begin{equation}\label{equation_zeros_derivative_h}
 \tilde h(w)=\frac{w^3}{r}h'(w)=w^3-2a_0w-2r
 \end{equation}
 
 For $t_1=0$, $\tilde h$ reduces to $w^3-2r$, which has three simple zeros, only one of them real. The discriminant of $\tilde h$ is given by
 \begin{equation}\label{aux_equation_2}
 \disc(\tilde h;w)=4(8a_0^3-27r^2)
 \end{equation}
 so $h'$ has a zero with multiplicity iff $a_0=\frac{3}{2}r^{2/3}$. Plugging this into \eqref{system_change_coordinates_a}--\eqref{system_change_coordinates_b}, we get that 
$t_0,t_1$ are given by 
\eqref{critical_curve_limiting_zero_distribution} for $s=r^{1/3}$. In particular, this is only possible for $0<r<1/8$. 

Furthermore, $\disc(\tilde h;w)$ does not change sign in each of the sets $\mathcal F_1,\mathcal F_2$. When we keep $t_0$ fixed and send $t_1\to 0$, we know from 
\eqref{definition_a_0} that $a_0\to 0$, whereas $r$ remains positive, so $\disc(\tilde h;w)<0$ on $\mathcal F_1$.  On another hand, if we keep $t_1$ fixed and send 
$t_0\to 0$, it follows from \eqref{system_change_coordinates_a}--\eqref{system_change_coordinates_b} that $a_0$ remains positive, whereas $r\to 0$, and hence $\disc(\tilde h;w)>0$ 
on $\mathcal F_2$. 

In virtue of \eqref{aux_equation_2}, the discussion above means that $\tilde h$ - and thus $h'$ - has exactly one real zero for $(t_0,t_1)\in \mathcal F_1$, and three real zeros 
for $(t_0,t_1)\in \mathcal F_2$.

From Lemma~\ref{lemma_zeros_derivative_h}, we already know that all the zeros of $h'$ belong to $\D$.
Since $a_0\geq 0,r>0$ (see Theorem~\ref{theorem_monotonicity_r}), from Descartes' Rule of sign we learn that $h'$ has exactly 
one positive real zero for $(t_0,t_1)\in \mathcal F$.

It only remains to prove \eqref{branched_solution_non_vanishing}. Suppose to the contrary that $h(w_j^{-1})=0$ for some $j$. This means that the polynomial
$$
\hat h(w)=w^2h(w^{-1})=r^2w^3+2a_0rw^2+a_0w+r
$$
has the common root $w_j$ with the polynomial $\tilde h$ in \eqref{equation_zeros_derivative_h}. But,
$$
\result(\tilde h, \hat h)=-r^3[1-4a_0^2+r^2(6-16a_0^2)+r^4(12-16a_0^2)+8r^6+32a_0^3r^2],
$$
and because $0<r<1/2$ and $0\leq a_0 <1/2$ (see Theorem~\ref{theorem_monotonicity_r}), it is not hard to see that the term between brackets above is always positive, thus 
the resultant above is nonzero and consequently $\tilde h$ and $\hat h$ do not have common roots. The proof is complete.
\end{proof}

\begin{proof}[Proof of Theorem \ref{theorem_change_coordinates_phase_diagram}]
Equation \eqref{cusp_critical_curve_positive_t1} follows directly from \eqref{equality_r_a_0_cusp}.

Equation~\eqref{aux_equation_2} and the arguments thereafter immediately show that in the coordinate system $(r,a_0)$, the critical curve $\gamma_c$ is described as in 
\eqref{critical_curve_limiting_zero_distribution_a_r_plane}. 
\end{proof}

\section{Geometry of the spectral curve. Proof of Theorem \ref{theorem_singular_points_spectral_curve}}\label{section_topology_spectral_curve}

An important role for the analysis of the spectral curve \eqref{spectral_curve} is played by the discriminant
\begin{equation}\label{discriminant_spectral_curve}
\mathcal D(z)=\disc(F(\xi,z);\xi),
\end{equation}

$\mathcal D$ is a polynomial of degree $9$ in $z$. It can be computed with the help of 
Mathematica, but its explicit expression is rather complicated to be dealt with directly. Its first coefficients are
$$
\mathcal D(z)=4z^9-4t_1z^8 + \left(4 B+16 t_1\right)z^7+\cdots,\quad z\in \C,
$$
where $B$ is as in \eqref{equation_B}.

Theorem~\ref{theorem_singular_points_spectral_curve} can be restated in terms of the discriminant $\mathcal D$.

\begin{thm}\label{theorem_discriminant_spectral_curve}
For $(t_0,t_1)\in \mathcal F$ the discriminant $\mathcal D$ in \eqref{discriminant_spectral_curve} has three double zeros $\hat z_0$, 
$\hat z_1$, $\hat z_2 \in \C\setminus \overline \Omega$ satisfying
\begin{equation}\label{properties_nodes}
\hat z_0>0,\quad \im \hat z_1<0,\quad \overline{\hat z}_2=\hat z_1.
\end{equation}

In addition, $\mathcal D$ always has a simple zero $z_0>0$, $z_0\in \Omega$. Its remaining zeros $z_1,z_2$ also belong to $\Omega$ and are located as follows.
\begin{enumerate}[label=(\roman*)]
\item For $(t_0,t_1)\in \mathcal F_1$,
$$
\im z_1<0,\quad \overline{z}_2=z_1.
$$

\item For $(t_0,t_1)\in \mathcal F_2$,
$$
z_2<z_1<z_0.
$$

\item For $(t_0,t_1)\in \gamma_c$,
$$
z_1=z_2<z_0
$$
that is, $\mathcal D$ has a double zero at $z_1=z_2$.
\end{enumerate}
\end{thm}

Assuming Theorem~\ref{theorem_discriminant_spectral_curve}, we now prove Theorem~\ref{theorem_singular_points_spectral_curve}. The proof of 
Theorem~\ref{theorem_discriminant_spectral_curve} is given in Section~\ref{subsection_spectral_curve_t_positive}

\begin{proof}[Proof of Theorem~\ref{theorem_singular_points_spectral_curve}] The zeros $z_0,z_1,z_2$ of $\mathcal D$ given by 
Theorem~\ref{theorem_discriminant_spectral_curve} correspond to the branch points given by Theorem~\ref{theorem_singular_points_spectral_curve}. The branch point at $\infty$ 
follows from the asymptotics~\eqref{asymptotics_xi} for the solutions $\xi_1,\xi_2,\xi_3$. In virtue of Lemma~\ref{lemma_nodes} below, the double zeros $\hat z_0,\hat z_1,\hat 
z_2$ of $\mathcal D$ are singular points of \eqref{spectral_curve}. 
\end{proof}

The spectral curve \eqref{spectral_curve} can be seen as a (branched) three-sheeted cover $\mathcal R$ of the Riemann sphere $\overline \C$. The main goal of the rest 
of the present section is to describe the three sheets $\mathcal R_1$, $\mathcal R_2$, $\mathcal R_3$ of $\mathcal R$. As an analytic counterpart, we ultimately prove 
Theorem~\ref{theorem_discriminant_spectral_curve}. During this Section, $t_0$ is always considered to be a fixed parameter, and every deformation is taken with respect to the 
parameter $t_1$.

\subsection{The spectral curve for $t_1=0$}\label{subsection_spectral_curve_t=0}

We briefly discuss the case $t_1=0$ studied by Bleher and Kuijlaars \cite{bleher_kuijlaars_normal_matrix_model}, describing their results in a suitable form for our needs. Besides 
being instructive both to fix notation and keep in mind the main lines of the rest of the section, this particular case is also used later.

For $t_1=0$ the quantities $r,a_0,A,B$ appearing in Proposition \ref{proposition_definition_r}, \eqref{definition_a_0}, \eqref{equation_B} and \eqref{equation_A}, respectively, 
are explicitly given by
\begin{align}
 & r=\frac{\sqrt{1-\sqrt{1-8t_0} }}{2},\label{r_for_t=0}\\
 & a_0=0,  \nonumber \\
 & A=\frac{1+20t_0-8t_0^2-(1-8t_0)^{3/2}}{32},\label{A_for_t=0}\\
 & B=0. \nonumber
\end{align}

The spectral curve \eqref{spectral_curve} is invariant under rotation $(\xi,z)\mapsto (\omega^2 \xi,\omega z)$, $\omega=e^{2\pi i/3}$, and this symmetry is carried over to 
all related quantities. For instance, in this situation it is easy to see that $h(\omega w)=\omega h(w)$.

The discriminant $\mathcal D$ in \eqref{discriminant_spectral_curve} is a cubic polynomial in $z^3$, thus reflecting the aforementioned three-fold rotational symmetry. It has a 
simple real zero $z_0$ and a double real zero $\hat z_0>z_0$ given by
\begin{equation}\label{branchpoints_nodes_t_1=0-1}
z_0=\frac{3}{4}\left(1-\sqrt{1-8t_0}\right)^{2/3},\quad \hat z_0=\frac{3+\sqrt{1-8t_0}}{4},
\end{equation}
and the remaining zeros are
\begin{equation}\label{branchpoints_nodes_t_1=0-2}
z_j=\omega^{-j} z_0,\quad \hat z_j=\omega^{-j} \hat z_0,\quad j=1,2.
\end{equation}

In terms of the rational parametrization $h$, the branch points $z_0,z_1,z_2$ are obtained as
\begin{equation}\label{zeros_h'_t_1=0}
z_j=h(w_j),\quad j=0,1,2,
\end{equation}
where $w_0\in (0,1),w_j=\omega^{2j}w_0$, $j=1,2$, are the zeros of $h'$ as in Lemma~\ref{lemma_location_zeros_derivative_h}.

Regarding Theorem~\ref{theorem_limiting_support_zeros}, we are always in the three-cut situation and
\begin{equation}\label{star_t=0}
\Sigma_{*}=\Sigma_{*,0}\cup\Sigma_{*,1}\cup\Sigma_{*,2},\quad \Sigma_{*,j}=[0,z_j],\quad j=0,1,2.
\end{equation}

For 
\begin{equation}\label{sheet_structure_t=0_no_symmetry}
\mathcal R_1=\overline\C \setminus \Sigma_*,\quad \mathcal R_2=\overline \C \setminus ([-\infty,0]\cup \Sigma_{*,1}\cup \Sigma_{*,2}),\quad \mathcal R_3=\overline 
\C\setminus [-\infty,z_0],
\end{equation}
the Riemann Surface $\mathcal R$ is the resulting surface after gluing $\mathcal R_1$ to $\mathcal R_2$ along \newline ${\Sigma_{*,1}\cup\Sigma_{*,2}}$, $\mathcal R_1$ to 
$\mathcal R_3$ along $\Sigma_{*,0}$ and $\mathcal R_2$ to $\mathcal R_3$ along $[-\infty,0]$, always in the usual crosswise manner. 

Each function $\xi_j$ in \eqref{asymptotics_xi} has an analytic continuation to the whole sheet $\mathcal R_j$, and for the values of $z_0,z_1,z_2$ as above, they satisfy
\begin{align}
&\xi_1(z_0)=\frac{r^{2/3} \left(2 r^2+1\right)}{2^{1/3}}=\xi_3(z_0), \qquad \xi_1(z_j)=\omega^{-j}\xi_1(z_0)=\xi_2(z_j),\quad j=1,2, \label{branchpoints_nodes_t_1=0-3}\\
&\xi_1(\hat z_0)= \hat z_0 = \xi_3(\hat z_0), \qquad \xi_1(\hat z_j)=\omega^{-j}\hat z_j=\xi_2(\hat z_j),\quad j=1,2.   \label{branchpoints_nodes_t_1=0-4}
\end{align}

Careful readers may notice this sheet structure differs from the one given in \cite{bleher_kuijlaars_normal_matrix_model}, where the authors construct the sheets respecting the 
underlying three-fold symmetry. But for us it is more convenient to do it this way, because for $t_1\neq 0$ the symmetry is unavoidably broken. 

The preimage of a point $z\in \overline \C$ through the canonical projection $\pi:\mathcal R\to \overline \C$ on the sheet $\mathcal R_j$ is denoted by $z^{(j)}$, $j=1,2,3$. 
Equivalently, a point $z^{(j)}\in \mathcal R_j$ can be seen as the pair $(\xi_j(z),z)$ - we use both representations without further explanation. 

At the branch points of $\mathcal R$, two of the preimages coincide. More precisely, $\mathcal R$ is branched at the points
$$
z_0^{(1)}=z_0^{(3)},\quad z_1^{(1)}=z_1^{(2)},\quad z_2^{(1)}=z_2^{(2)},\quad \infty^{(2)}=\infty^{(3)},
$$
and for $z\in\C\setminus \{z_0,z_1,z_2\}$ the points $z^{(1)},z^{(2)},z^{(3)}$ are all distinct.

\subsection{The spectral curve for $t_1>0$. Proof of Theorem \ref{theorem_discriminant_spectral_curve}}\label{subsection_spectral_curve_t_positive}

We now focus on the case when $t_1$ is positive.

Consider the system of equations
\begin{align}
&F(\xi,z)= 0, \label{system_nodes_1a}\\
&\frac{\partial F}{\partial \xi}(\xi,z)= 3\xi^2-2z^2\xi-2t_1\xi-(1+t_0)z+B+t_1=0,\label{system_nodes_1b}\\
&\frac{\partial F}{\partial z}(\xi,z)= 3z^2-2z\xi^2-2t_1 z-(1+t_0)\xi+B+t_1=0. \label{system_nodes_1c} 
\end{align}

\begin{lem}\label{lemma_nodes}
 If $z\in \C$ is a zero of $\mathcal D$ but the point $z^{(j)}\in \mathcal R$ is not a branch point of $\mathcal R$, then $(\xi_j(z),z)$ satisfies the system 
\eqref{system_nodes_1a}--\eqref{system_nodes_1c}.
\end{lem}
\begin{proof}
 For a generic point $(\xi,z)=(h(w^{-1}),h(w))$ satisfying \eqref{system_nodes_1a} it is true
$$
-\frac{h'(w^{-1})}{w^2}\pder{F}{\xi} (h(w^{-1}),h(w))+h'(w)\pder{F}{z}(h(w^{-1}),h(w))=0.
$$

If $z^{(j)}$ is a zero of the discriminant $\mathcal D$, then $(\xi_j(z),z)$ satisfies \eqref{system_nodes_1b}, hence from the previous equation
$$
h'(w)\pder{F}{z}(h(w^{-1}),h(w))=0.
$$

Since $z^{(j)}$ is not a branch point, $h'(w)\neq 0$, and the Lemma follows.
\end{proof}

Subtracting Equation~\eqref{system_nodes_1c} from Equation~\eqref{system_nodes_1b}, we get
$$
(\xi-z)(3\xi+3z+ 2z\xi+1+t_0-2t_1)=0,
$$
and as a corollary of Lemma~\ref{lemma_nodes}

\begin{cor}\label{corollary_nodes}
 If $z\in \C$ is a zero of $\mathcal D$ but the point $z^{(j)}\in \mathcal R$ is not a branch point of $\mathcal R$, then the pair $(\xi_j(z),z)$ satisfies at least one of the 
equations below,
\begin{align}
& \xi-z=0, \label{equations_singular_points_1}\\
& 3\xi+3z+ 2z\xi+1+t_0-2t_1=0. \label{equations_singular_points_2}
\end{align}
\end{cor}

Corollary~\ref{lemma_nodes} gives us an analytic tool for studying the dynamics of the singular points $\hat z_0,\hat z_1, \hat z_2$ when deforming $t_1$, namely through the 
system \eqref{equations_singular_points_1}--\eqref{equations_singular_points_2}. When $t_1=0$, $(\xi_1(\hat z_0),\hat z_0),(\xi_2(\hat z_0),\hat z_0)$ satisfy 
\eqref{equations_singular_points_1}, whereas $(\xi_1(\hat z_j),\hat z_j),(\xi_3(\hat z_j),\hat z_j)$, $j=1,2$, satisfy \eqref{equations_singular_points_2}, as can be verified by 
simply plugging the values \eqref{branchpoints_nodes_t_1=0-1},\eqref{branchpoints_nodes_t_1=0-2},\eqref{branchpoints_nodes_t_1=0-4} into 
\eqref{equations_singular_points_1}--\eqref{equations_singular_points_2}. 

Writing $(\xi,z)=(h(w^{-1}),h(w))$, equations~\eqref{equations_singular_points_1}--\eqref{equations_singular_points_2} respectively become
\begin{align}
& \left(w-\frac{1}{w}\right)g(w)=0, \label{equations_singular_points_w_1}\\
& f(w)=0, \label{equations_singular_points_w_2}
\end{align}
where
\begin{align}
 g(w) & = 1-2a_0-r\left(w+\frac{1}{w}\right), \label{equation_function_g} \\
 f(w) & = \widetilde f\left(w+\frac{1}{w}\right), \label{equation_polynomial_f_original}\\
 \widetilde f(w) & =  2r^3 w^3+3r^2(1+2a_0)w^2+r(3+8a_0+4a_0^2+4a_0r^2-6r^2)w \nonumber \\
      &  \quad + 6a_0+2a_0^2-4r^2-12a_0r^2+8a_0^2r^2+2r^4+1+t_0-2t_1. \label{equation_polynomial_f} 
\end{align}

\begin{lem}\label{lemma_zeros_g}
For $(t_0,t_1)\in \mathcal F$, the function $g$ in \eqref{equation_function_g} has exactly one zero $\hat w_0$ on $(1,+\infty)$ and exactly one zero $\hat w_0^{-1}$ on $(0,1)$.
\end{lem}
\begin{proof}
We notice that $w\mapsto w+1/w$ is a bijection from $(1,+\infty)$ to $(2,+\infty)$. Since $(1-2a_0)/r>2$ (see \eqref{inequality_r_a_0_cusp}), the equation
$$
w+\frac{1}{w}=\frac{1-2a_0}{r}
$$
has precisely one solution $\hat w_0$ on $(1,+\infty)$.
\end{proof}

For $t_1=0$, $g$ simplifies to
$$
g(w)=1-r\left(w+\frac{1}{w}\right),
$$
so the point $\hat w_0$ and its images $h(\hat w_0),h(\hat w_0^{-1})$ in this case are given by
\begin{equation}\label{aux_equation_3}
\hat w_0=\frac{\sqrt{1-4 r^2}+1}{2 r},\quad h(\hat w_0^{-1})=h(\hat w_0)=1-r^2=\hat z_0,
\end{equation}
where $\hat z_0$ is given in \eqref{branchpoints_nodes_t_1=0-1} and we used the explicit expression for $r$ in \eqref{r_for_t=0}.

\begin{lem}\label{lemma_zeros_tilde_f}
 For $(t_0,t_1)\in \mathcal F$, the polynomial $\widetilde f$ in \eqref{equation_polynomial_f} has exactly one root on $(-\infty,0)$ and no other real roots.
\end{lem}

\begin{proof}
 The function $(\widetilde f)'$ is a polynomial of degree $2$, whose discriminant is given by
$$
\disc((\widetilde f)')=12r^2\left(12 r^2-3-4 a_0 \left(1+2 r^2-a_0\right)\right).
$$
 
Using the upper bound $r<1/2$ given by Theorem~\ref{theorem_monotonicity_r} and the inequality \eqref{inequality_r_a_0_cusp}, we get
$$
12r^2-3<0, \quad 1+2r^2-a_0>1-a_0\geq 2r+a_0>0,
$$
so $\disc((\widetilde f)')<0$ and hence $(\widetilde f)'$ has no real zeros. This implies that $\widetilde f$ has exactly one real root.

Moreover,
\begin{align}
\widetilde f(0) & = a_0^2 \left(8 r^2+2\right)+a_0 \left(6-12 r^2\right)+1+2 r^4-4 r^2+t_0-2 t_1 \nonumber \\
     & \geq 1+2 r^4-4 r^2+t_0-2 t_1, \label{aux_equation_1}
\end{align}
where we used $a_0\geq 0$ and $0<r<1/2$ (see Theorem~\ref{theorem_monotonicity_r}). The derivative of the right-hand side is 
given 
by
$$
\pder{}{t_1}(1+2 r^4-4 r^2+t_0-2 t_1)=-8r(1-r^2)\pder{r}{t_1}-2<0,
$$
where again we used Theorem~\ref{theorem_monotonicity_r}. This implies the minimum of the right-hand side in \eqref{aux_equation_1} is attained when $(t_0,t_1)\in \Gamma_c$, hence 
using \eqref{definition_critical_curve_phase_transition} with parameter $s=r$ given by Theorem~\ref{theorem_monotonicity_r},
$$
1+2 r^4-4 r^2+t_0-2 t_1>-4 r^4-4 r^3+2 r^2+\frac{1}{2}>0,\quad 0<r<\frac{1}{2}.
$$

Returning this conclusion back to \eqref{aux_equation_1}, we get $\widetilde f(0)>0$, so the only real root of $\widetilde f$ needs to be negative.
\end{proof}

As a consequence,
\begin{cor}\label{corollary_zeros_f}
 For $(t_0,t_1)\in \mathcal F$, the function $f$ in \eqref{equation_polynomial_f_original} has two simple zeros $\hat w_1$, $\hat w_2=\overline{\hat w_1}$, $\im \hat w_1<0$, on 
$\D\setminus \R$, two simple zeros $\hat w_1^{-1},\hat w_2^{-1}$ on $\C\setminus(\D\cup \R)$ and two zeros on $(-\infty,0)\cup \partial \D$, the latter not necessarily distinct.
\end{cor}
\begin{proof}
 The Corollary is a direct consequence of Lemma~\ref{lemma_zeros_tilde_f} and the $2$-to-$1$ correspondence between the zeros of $f$ and the zeros of $\tilde f$ induced by the 
map 
$w\mapsto w+w^{-1}$.
\end{proof}

\begin{cor}
 The functions $f$, $g$ given by \eqref{equation_function_g}--\eqref{equation_polynomial_f} do not have a common root.
\end{cor}
\begin{proof}
This result follows directly from Lemma~\ref{lemma_zeros_g} and Corollary~\ref{corollary_zeros_f}.
\end{proof}

For $t_0=0$, $f$ reduces to
\begin{multline*}
f(w)=2 r^4+2 r^3 w^3+\frac{2 r^3}{w^3}+3 r^2 w^2+\frac{3 r^2}{w^2}+2 r^2 \\ +\left(6 r^3+\left(3-6 r^2\right) r\right) w+\frac{6 r^3+\left(3-6 r^2\right) r}{w}+t_0+1
\end{multline*}
and the root $\hat w_1$ and the values $h(\hat w_1),h(\hat w_1^{-1})$ are given by
\begin{equation}\label{aux_equation_4}
\hat w_1=\omega^2 \hat w_0,\quad h(\hat w_1)=\omega^2 h(\hat w_0)=\omega^2\hat z_0=\hat z_1,\quad h(\hat w_1^{-1})=\omega h(\hat w_0^{-1})=\omega \hat z_0=\omega^{2} 
\hat z_1
\end{equation}
where we recall $\omega=e^{\frac{2\pi i}{3}}$, $\hat z_0,\hat z_1$ and $\hat w_0$ are given in \eqref{branchpoints_nodes_t_1=0-1}, \eqref{branchpoints_nodes_t_1=0-2} and
\eqref{aux_equation_3}. 

To get started to the deformation argument used for the proof of Theorem~\ref{theorem_discriminant_spectral_curve}, we state the following weak form of 
Theorem~\ref{theorem_discriminant_spectral_curve} as a Lemma.

\begin{lem}\label{lemma_zeros_discriminant_small_t_1}
 For $t_1>0$ small, the discriminant $\mathcal D$ has three simple zeros, which are branch points of $\mathcal R$, and three double zeros, which are not branch points of $\mathcal 
R$.
\end{lem}

\begin{proof}
As explained in Section~\ref{subsection_spectral_curve_t=0}, the result is true for $t_1=0$. As a general fact that follows by continuity arguments, for small perturbations of 
$t_1$, the multiplicity of a zero of $\mathcal D$ cannot increase. That is, simple zeros of $\mathcal D$ for $t_1=0$ stay simple for small $t_1$, and double zeros of $\mathcal D$ 
for $t_1=0$ either keep being double zeros or else split into two distinct simple zeros.
 
In particular, we get that for small perturbation of $t_1$, $\mathcal D$ still has {\it at least} three simple zeros.
 
 By Riemann-Hurwitz formula and the genus $0$ condition given by Corollary~\ref{corollary_genus_0}, we know that $\mathcal R$ has four branch points. One of those is the 
point $z=\infty$, see \eqref{asymptotics_xi}, and each simple zero of $\mathcal D$ is also a branch point. Hence, $\mathcal D$ must have {\it at most} three simple zeros, and by 
the remarks above it follows that $\mathcal D$ has exactly three simple zeros and exactly three double zeros for small perturbations of $t_1$.
\end{proof}

\begin{remark}\label{remark_zeros_discriminant}
 The discriminant $\mathcal D$ can be expressed as
 \begin{equation*}
 \mathcal D(z)=(\xi_1-\xi_2)^2(\xi_1-\xi_3)^2(\xi_2-\xi_3)^2,
 \end{equation*}
 where $\xi_j$'s are the solutions to \eqref{spectral_curve}. In particular, at least two of the $\xi_j$'s coincide at each zero of $\mathcal D$. By continuity from the case 
$t_1=0$, it follows that for $t_1$ small, the solution $\xi_1$ coincides with one of the other two solutions at each zero of $\mathcal D$.
\end{remark}

After this preparation, we can proceed to

\begin{proof}[Proof of Theorem~\ref{theorem_discriminant_spectral_curve}]
For $w_j, \hat w_j$, $j=0,1,2$, the points given by Lemmas~\ref{lemma_location_zeros_derivative_h}, \ref{lemma_zeros_g} and Corollary~\ref{corollary_zeros_f}, define
\begin{equation}\label{equation_definition_z_hat_z}
z_j=h(w_j),\quad \hat z_j=h(\hat w_j),\quad j=0,1,2.
\end{equation}

Our goal is to prove that these points satisfy the conclusions of Theorem~\ref{theorem_discriminant_spectral_curve}. As a first step we prove their geometric properties {\it 
(i)--(iii)}, and afterwards we prove that these points are zeros of $\mathcal D$ with the claimed multiplicities. At the end, we prove that $z_0,z_1,z_2\in \Omega$ in the 
three-cut case and $z_0,z_1\in \Omega$ in the one-cut case.

First note that $\hat z_0,\hat z_1,\hat z_2\in \C\setminus \overline \Omega$, because $|\hat w_j|>1$ and $h$ maps $\C\setminus \D$ to $\C\setminus \Omega$, see 
Corollary~\ref{corollary_h_parametrization} and also Lemma~\ref{lemma_zeros_g} and Corollary~\ref{corollary_zeros_f}.

We verify \eqref{properties_nodes}. For $t_1=0$, \eqref{properties_nodes} follows from \eqref{aux_equation_3} and \eqref{aux_equation_4}. Since $\hat w_0\in (1,+\infty)$ 
(Lemma~\ref{lemma_zeros_g}) and $h$ maps $(1,+\infty)$ to $(0,+\infty)\setminus\Omega$ (Corollary~\ref{corollary_h_parametrization}), it follows that $\hat 
z_0\in(0,+\infty)\setminus\Omega$, and in particular $z_0>0$. Because $\hat w_1$ is never real and $|\hat w_1|>1$, see Corollary~\ref{corollary_zeros_f}, the point $\hat z_1$ 
cannot be real neither, again due to Corollary~\ref{corollary_h_parametrization}. Thus by continuity with respect to $t_1$ we get $\im \hat z_1<0$. The equality $\hat 
z_2=\overline{\hat z_1}$ is trivial from \eqref{equation_definition_z_hat_z} and the definition of $\hat w_1,\hat w_2$ given by Corollary~\ref{corollary_zeros_f}.

We now prove that {\it (i)-(iii)} are satisfied by the points $z_0,z_1,z_2$. 

For $t_1=0$, {\it (i)} follows from \eqref{zeros_h'_t_1=0}. When we tune up $t_1$, the point $z_1$ cannot become real for $(t_0,t_1)\in \mathcal F_1$. Indeed, the value $w_j$ is a 
double zero of 
$$
 z_j=h(w).
$$

If $ z_1$ becomes real for $(t_0,t_1)\in \mathcal F_1$ - hence also $ z_2=\overline{z_1}$ - then the equation $z_1=h(w)$ has two distinct solutions $ w_1$, $w_2$ with 
multiplicity 
two, see Lemma~\ref{lemma_zeros_derivative_h}, which cannot occur because of the explicit form of $h$. By continuity, we get {\it (i)}.

When $(t_0,t_1)\in \gamma_c$, we know that $w_1=w_2\in (-1,0)$, see Lemma~\ref{lemma_zeros_derivative_h}. This automatically implies $z_1=z_2$. Choosing $s=1/4$ in 
\eqref{critical_curve_limiting_zero_distribution_a_r_plane}, we compute explicitly
$$
w_0=1/2, \quad w_1=w_2=\frac{1}{4},\quad z_0=\frac{111}{1024}>\frac{21}{256}=z_1=z_2,
$$
so for the respective choice of parameters $(t_0,t_1)\in \gamma_c$ {\it (iii)} holds true. Since $h'$ never has triple roots, see Lemma~\ref{lemma_zeros_derivative_h}, by 
continuity we conclude {\it (iii)} holds true for every choice of parameters $(t_0,t_1)\in \gamma_c$.

When $(t_0,t_1)$ enters $\mathcal F_2$, we learn from {\it (iii)} and continuity that $z_1,z_2<z_0$, so to get {\it (ii)} it only remains to prove that in this situation 
$z_2<z_1$, or equivalently $h(w_2)<h(w_1)$.

The function $w\mapsto h(w)$ goes to $-\infty$ when $w\to -\infty$. On the negative axis its derivative has two simple zeros $w_1<w_2$ and no others. This implies $h$ is 
increasing in $(-\infty,w_1)$ and decreasing in $(w_1,w_2)$, and hence $h(w_1)>h(w_2)$, so finally {\it(ii)} is proven.

For $(t_0,t_1)\in \mathcal F\setminus \gamma_c$, we now prove $z_0$, $z_1$, $z_2$ are simple zeros of $\mathcal D$ and $\hat z_0$, $\hat z_1$, $\hat z_2$ are double zeros 
of $\mathcal D$. 

The points $z_0,z_1,z_2$ are the only branch points of $\mathcal R$, so surely they are zeros of $\mathcal D$. The remaining zeros of $\mathcal D$ must be of multiplicity at least 
two, because they are not branch points. Since we already know the points $\hat z_1,\hat z_2,\hat z_3$ are pairwise distinct, a total counting of zeros (according to multiplicity) 
shows it is enough to prove $\hat z_1,\hat z_2,\hat z_3$ are always zeros of $\mathcal D$, and their multiplicity properties will follow.

For $t_1=0$, $\mathcal D$ is given by
\begin{multline*}
\mathcal D(z)=4z^9+(t_0^2+4A+12t_0-8)z^6 \\ +(4t_0^3+18At_0+12t_0^2-36A+12t_0+4)z^3-27A^2,
\end{multline*}
where $A$ is given in \eqref{A_for_t=0}. Using \eqref{aux_equation_3}, \eqref{aux_equation_4}, after a lengthy calculation it follows that $\hat z_j$ is a double zero of 
$\mathcal D$, $j=0,1,2$. 

From Lemma~\ref{lemma_zeros_discriminant_small_t_1}, we know $\mathcal D$ has three double zeros for $t_1$ small, and these are not branch points. Let 
$h(\tilde w)=\tilde z$ be one of these zeros. We can assume $h(\tilde w^{-1})=\xi_1(\tilde z)$, see Remark~\ref{remark_zeros_discriminant}, so in particular
\begin{equation}\label{aux_equation_5}
|\tilde w|>1.
\end{equation}

From Corollary~\ref{corollary_nodes} we know the pair $(\xi,\tilde z)=h((\tilde w^{-1}),h(\tilde w))$ satisfies one of 
Equations~\eqref{equations_singular_points_1}--\eqref{equations_singular_points_2}, and hence 
$\tilde w$ must satisfy at least one of equations \eqref{equations_singular_points_w_1}--\eqref{equations_singular_points_w_2}. Combining Equation~\eqref{aux_equation_5} with 
Lemma~\ref{lemma_zeros_g} and Corollary~\ref{corollary_zeros_f}, we thus get that $\tilde w_j$ must be one of the points $\hat w_0,\hat w_1,\hat w_2$, so $\tilde z$ must be one of 
the points $\hat z_1,\hat z_2,\hat z_3$. 

Although carried out for $t_1$ small, the argument above works as long as none of the points $z_j,\hat z_j$ pairwise coincide. But we already proved \eqref{properties_nodes}, 
{\it (i)-(iii)} are valid, so the only coalescence that can happen is $z_1=z_2$, and only for $(t_0,t_1)\in \gamma_c$. By continuity it means that in this case $\mathcal D$ has a 
unique simple zero $z_0$ and the other points $\hat z_0,\hat z_1,\hat z_2$, $z_1=z_2$ are double zeros of $\mathcal D$ - the double zero $z_1=z_2$ is still a branch point of 
$\mathcal R$. When we move beyond $\gamma_c$, the double zero $z_1=z_2$ 
splits into the two simple zeros $z_1,z_2$ and the point $z_1$ is still a simple zero. The genus $0$ constraint then guarantees that the remaining zeros should still be of 
multiplicity at least two, and by analytic continuation these must be given by $\hat z_0,\hat z_1,\hat z_2$.

We now verify that $z_0,z_1,z_2\in \Omega$ for $(t_0,t_1)\in \mathcal F_1$. For $t_1=0$ we know from \eqref{branchpoints_nodes_t_1=0-3} that
$$
\xi_1(z_0)=\xi_3(z_0),\qquad \xi_1(z_j)=\xi_2(z_j),\quad j=1,2,
$$
where $\xi_1,\xi_2,\xi_3$ are (analytic continuations of) the solutions to \eqref{spectral_curve} as in \eqref{asymptotics_xi}. Since we already proved that the points 
$z_0,z_1$ and $z_2$ do not pairwise coincide for $(t_0,t_1)\in \mathcal F_1$, we conclude that these equalities are valid for every choice $(t_0,t_1)\in \mathcal F_1$. Hence 
$\xi_1$ is branched at each of the points $z_0,z_1$ and $z_2$. But from Theorem~\ref{theorem_schwarz_function} we know that $\xi_1$ is the Schwarz function of $\partial \Omega$, 
hence $\xi_1$ is meromorphic on $\overline \C\setminus \Omega$ and consequently its branch points $z_0,z_1$ and $z_2$ have to belong to $\Omega.$

It only remains to prove that $z_0,z_1\in \Omega$ in the one-cut case $(t_0,t_1)\in \mathcal F_2$. When we cross $\gamma_c$, the branch point $z_0$ does not coalesce with any 
other zero of $\mathcal D$; thus by continuity we get that $\xi_1(z_0)=\xi_3(z_0)$ for every $(t_0,t_1)\in \mathcal F_2$. The discriminant $\mathcal D$ is a polynomial of degree 
$9$ with positive leading coefficient, and for $(t_0,t_1)\in \mathcal F_2$ we already know that its only real zeros are $z_2,z_1,z_0$ and $\hat z_0$, the first three of 
multiplicity one, and the last one of multiplicity two. Hence 
\begin{equation}\label{inequalities_discriminant}
\begin{aligned}
\mathcal D(z)<0, & \quad z\in (-\infty,z_2),\\
\mathcal D(z)>0, & \quad z\in (z_2,z_1), \\
\mathcal D(z)<0, & \quad z\in (z_1,z_0), \\
\mathcal D(z)>0, & \quad z\in (z_0,\hat z_0)\cup (\hat z_0,+\infty).
\end{aligned}
\end{equation}

We already know that $\xi_1(z_2)=\xi_3(z_2)$. The third inequality in \eqref{inequalities_discriminant} then implies that the boundary values of the analytic continuations 
of $\xi_1$ and $\xi_2$ are complex conjugate of each other in the interval $(z_1,z_0)$. Combining with the second inequality in \eqref{inequalities_discriminant} we get that 
$\xi_1(z_1)=\xi_3(z_1)$, that is, the function $\xi_1$ is branched at $z_1$ as well. Since we already know from Theorem~\ref{theorem_schwarz_function} that $\xi_1$ is the Schwarz 
function of $\partial \Omega$, the branch points of $\xi_1$ have to be in $\Omega$, that is, $z_0,z_1\in \Omega$, concluding the proof.
\end{proof}

\begin{remark}
We stress that the proof of Theorem~\ref{theorem_discriminant_spectral_curve} also shows that the points $z_j$ and $\hat z_j$, $j=0,1,2$, given by 
Theorem~\ref{theorem_singular_points_spectral_curve} can be obtained through the equalities
$$
h(w_j)=z_j,\quad h(\hat w_j)=\hat z_j,\qquad j=0,1,2,
$$
where $w_0$, $w_1$ and $w_2$ are the zeros of $h'$ as in Lemma~\ref{lemma_location_zeros_derivative_h}, and $\hat w_0$, $\hat w_1$ and $\hat w_2$ are given by Lemma 
\ref{lemma_zeros_g} and Corollary \ref{corollary_zeros_f}. We will use this fact extensively in the next sections.
\end{remark}

\subsection{Sheet structure for $\mathcal R$}\label{section_sheet_structure_general}

To construct the sheet structure of $\mathcal R$, we make use of the following proposition.

\begin{prop}\label{proposition_branch_points_xi_functions}
For $(t_0,t_1)\in \mathcal F$, the (meromorphic continuation of) the functions $\xi_1,\xi_2$ and $\xi_3$ in \eqref{asymptotics_xi} satisfy
$$
\xi_1(\hat z_0)=\xi_3(\hat z_0),\qquad \xi_1(\hat z_j)=\xi_2(\hat z_j),\quad j=1,2.
$$
Furthermore,
 \begin{enumerate}[label=(\roman*)]
\item for $(t_0,t_1)\in \mathcal F_1$,
$$
\xi_1(z_0)=\xi_3(z_0),\qquad \xi_1(z_j)= \xi_2(z_j),\quad j=1,2;
$$

\item for $(t_0,t_1)\in \mathcal F_2$,
$$
\xi_1(z_j)=\xi_3(z_j),\quad j=0,1, \qquad \xi_2(z_2)= \xi_3(z_3);
$$

\item for $(t_0,t_1)\in \gamma_c$,
$$
\xi_1(z_0)=\xi_3(z_0), \qquad \xi_1(z_1)=\xi_2(z_1)=\xi_3(z_1).
$$
\end{enumerate}
\end{prop}
\begin{proof}
 The first equality in {\rm (ii)} was explicitly verified in the final part of the proof of Theorem~\ref{theorem_discriminant_spectral_curve}. The other properties claimed by the 
proposition follow from similar continuity arguments. We skip the details.
\end{proof}

Using Proposition~\ref{proposition_branch_points_xi_functions} we are ready to construct the sheet structure of the Riemann surface $\mathcal R$ associated with the 
algebraic equation \eqref{spectral_curve}. We see $\mathcal R$ as a branched three-sheeted cover of $\C$ and denote its sheets by $\mathcal R_1$, $\mathcal R_2$ and $\mathcal 
R_3$, 
so that
$$
\mathcal R=\mathcal R_1\cup\mathcal R_2\cup\mathcal R_3.
$$

The explicit construction of $\mathcal R_1,\mathcal R_2,\mathcal R_3$ is carried over below and depends on whether we are in the three-cut or 
one-cut case. In both situations, for $j=1,2,3$, the function $\xi_j$ in \eqref{asymptotics_xi} admits a meromorphic continuation to the whole sheet $\mathcal R_j$. As 
usual, these functions $\xi_1,\xi_2,\xi_3$ are regarded as branches of the same meromorphic function 
$$
\xi:\mathcal R \to \overline \C,\quad \xi\equiv \xi_j \mbox{ on } \mathcal R_j,
$$
which is the global solution to \eqref{spectral_curve}.

Moreover, given the sheet structure $\mathcal R=\mathcal R_1\cup\mathcal R_2\cup\mathcal R_3$, we denote the restriction of the canonical projection $\pi:\mathcal R\to \overline 
\C$ to $\mathcal R_j$ by $\pi_j$, $j=1,2,3$. With this notation, the function $\pi_j$ is invertible outside the branch cuts connecting the sheets, and its inverse 
$\pi_j^{-1}$ extends continuously to the branch cuts if one considers appropriate limiting boundary values. As at the end of Section~\ref{subsection_spectral_curve_t=0}, we denote 
by $p^{(j)}$ the inverse image of a point $p\in\overline\C$ through $\pi_j$. That is, $p^{(j)}$ denotes the point in $\mathcal R_j$ which is uniquely defined through the relation
$$
\{p^{(j)}\}=\pi_j^{-1}(\{p\}),\quad j=1,2,3.
$$

The point $p^{(j)}$ is also well defined at branch points. However, if $p$ belongs to the projection of the open arcs constituting the branch cuts of 
$\mathcal R_j$, then the set $\pi_j^{-1}(\{p\})$ contains two points on $\mathcal R_j$, one on each side of the branch cut. We denote these two points by $p^{(j)}_+, 
p^{(j)}_-\in\mathcal R_j$, labeled according to
$$
\{p^{(j)}_+ \}=\pi_{j+}^{-1}(\{p\}),\quad \{ p^{(j)}_- \}=\pi_{j-}^{-1}(\{p\}).
$$
In particular, if $p$ belongs to the branch cut connecting two sheets $\mathcal R_j$ and $\mathcal R_k$, then $p_\pm^{(j)}=p_\mp^{(k)}$.

\subsubsection{Sheet structure in the three-cut case}\label{section_sheet_structure_1}

Consider a Jordan arc $\gamma_0$ connecting $z_1$ to $z_2$ and intersecting $\R$ exactly once, say at the point $z_*$. Assume in addition $\gamma_0^*=\gamma_0$ and $z_*<z_0$. Set 
$$
\Sigma_*=\gamma_0\cup [z_*,z_0]
$$
and define
$$
\mathcal R_1=\overline \C\setminus\Sigma_*,\quad \mathcal R_2=\overline \C\setminus (\gamma_0\cup[-\infty,z_0]),\quad \mathcal R_3=\overline \C\setminus 
[-\infty,z_0].
$$

We construct the three-sheeted Riemann surface 
$$
\mathcal R=\mathcal R_1\cup\mathcal R_2\cup \mathcal R_3
$$
connecting $\mathcal R_1$ to $\mathcal R_2$ along $\gamma_0$, $\mathcal R_1$ to $\mathcal R_3$ along $[z_*,z_0]$ and 
$\mathcal R_2$ to $\mathcal 
R_3$ along $[-\infty,z_*]$, always in the usual crosswise manner. For $t_1=0$ and the choice $\gamma_0=[0,z_1]\cup [0,z_2]$, this is in agreement with the sheet 
structure for $t_1=0$ carried over in Section~\ref{subsection_spectral_curve_t=0}. This sheet structure is illustrated in Figure~\ref{figure_sheet_structure}.

The Riemann surface $\mathcal R$ is branched at the points
$$
z_0^{(1)}=z_0^{(3)},\quad z_1^{(1)}=z_1^{(2)},\quad z_2^{(1)}=z_2^{(2)},\quad \infty^{(2)}=\infty^{(3)}.
$$

\begin{figure}[t]
 \centering
\includegraphics[scale=1]{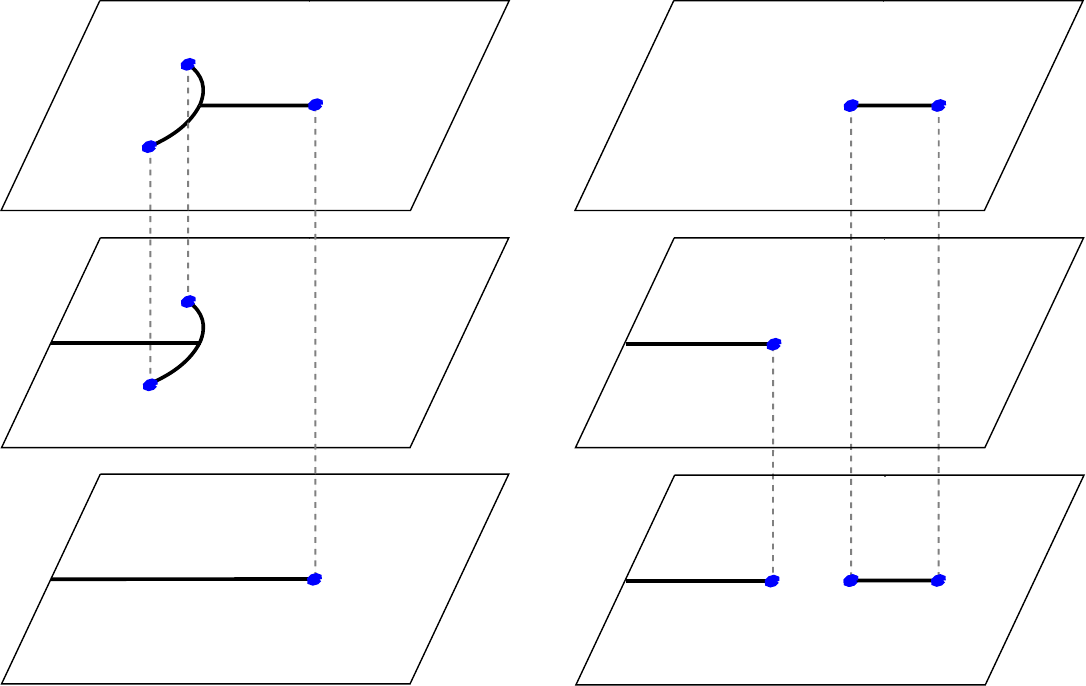}
\caption{Sheet structure of $\mathcal R$ for $(t_0,t_1)\in \mathcal F_1$ (left) and $(t_0,t_1)\in\mathcal F_2$ (right).}\label{figure_sheet_structure}
\end{figure}

We emphasize here the freedom in the choice of $\gamma_0$. This freedom is exploited later. We also remark that this sheet structure preserves the equalities
\begin{equation}\label{zeros_difference_xi_functions_F_1}
\begin{aligned}
\xi_1(\hat z_0)=\xi_3(\hat z_0),\quad \xi_1(\hat z_j)=\xi_2(\hat z_j),\ j=1,2, \\
\xi_1(z_0)=\xi_3(z_0),\quad \xi_1(z_j)=\xi_2(z_j),\ j=1,2,
\end{aligned}
\end{equation}
claimed by Proposition~\ref{proposition_branch_points_xi_functions}. In addition, the following 
properties hold true
\begin{align}
& \xi_2(x)<\xi_3(x)<\xi_1(x),\quad & x>\hat z_0,\label{equalities_inequalities_real_line_xi_functions_precritical_1}\\
& \xi_2(x)<\xi_1(x)<\xi_3(x),\quad & z_0<x<\hat z_0, \label{equalities_inequalities_real_line_xi_functions_precritical_2}\\
& \xi_{1\pm}(x)=\overline{\xi_{3\pm}(x)}=\xi_{3\mp}(x),\quad \im \xi_2(x)=0,\quad & z_*<x<z_0, \label{equalities_inequalities_real_line_xi_functions_precritical_3}\\
& \xi_{2\pm}(x)=\overline{\xi_{3\pm}(x)}=\xi_{3\mp}(x),\quad \im \xi_{1}(x)=0, \quad & x<z_*,\label{equalities_inequalities_real_line_xi_functions_precritical_4}
\end{align}
as it follows from the asymptotic behavior \eqref{asymptotics_xi} and an analysis of the sign of the discriminant $\mathcal D$ as in \eqref{inequalities_discriminant}. We skip the 
details.

Furthermore, from the construction of the Riemann surface, it also holds true
\begin{align*}
& \xi_{1\pm}(z) = \xi_{3\mp}(z),\quad z\in [z_*,z_0],\\
& \xi_{1\pm}(z) = \xi_{2\mp}(z), \quad z\in \gamma_0
\end{align*}
%

\subsubsection{Sheet structure in the one-cut case}\label{section_sheet_structure_2}

For $(t_0,t_1)\in \mathcal F_2$, we define
$$
\mathcal R_1=\overline \C\setminus [z_1,z_0],\quad \mathcal R_2=\overline \C \setminus [-\infty,z_2],\quad \mathcal R_3=\overline\C\setminus( [-\infty,z_2]\cup [z_1,z_0] ).
$$

In the present case, $\mathcal R_1$ is connected to $\mathcal R_3$ along $[z_1,z_0]$ and $\mathcal R_2$ is connected to $\mathcal R_3$ along $[-\infty,z_2]$, in the usual 
crosswise way. This sheet structure is shown in Figure~\ref{figure_sheet_structure}.

The branch points of $\mathcal R$ are given by
$$
z_0^{(1)}=z_0^{(3)},\quad z_1^{(1)}=z_1^{(3)},\quad z_2^{(2)}=z_2^{(3)},\quad \infty^{(2)}=\infty^{(3)},
$$
see Figure~\ref{figure_sheet_structure}. In the same spirit as in \eqref{zeros_difference_xi_functions_F_1}--\eqref{equalities_inequalities_real_line_xi_functions_precritical_4}, 
we also have the following equalities,
\begin{equation}\label{zeros_difference_xi_functions_F_2}
\begin{aligned}
\xi_1(\hat z_0)=\xi_3(\hat z_0),\quad \xi_1(\hat z_j)=\xi_2(\hat z_j),\ j=1,2, \\
\xi_2(z_2)=\xi_3(z_2),\quad \xi_1(z_j)=\xi_3(z_j),\ j=0,1,
\end{aligned}
\end{equation}
and the relations
\begin{equation}\label{equalities_inequalities_real_line_xi_functions_supercritical}
\begin{aligned}
& \xi_2(x)<\xi_3(x)<\xi_1(x),\quad & x>\hat z_0,\\
& \xi_2(x)<\xi_1(x)<\xi_3(x),\quad & z_0<x<\hat z_0, \\
& \xi_{1\pm}(x)=\overline{\xi_{3\pm}(x)}=\xi_{3\mp}(x),\quad \im \xi_2(x)=0,\quad & z_1<x<z_0, \\
& \xi_{2\pm}(x)=\overline{\xi_{3\pm}(x)}=\xi_{3\mp}(x),\quad \im \xi_{1}(x)=0, \quad & x<z_2.
\end{aligned}
\end{equation}

\section{Meromorphic quadratic differential on $\mathcal R$}\label{section_quadratic_differential}

Given any point $p\in \mathcal R$ which is not a branch point, we define the following function element in a neighborhood of $p$
\begin{equation}\label{definition_function_germ_Q}
Q(z)=
\left\{
\begin{aligned}
 & \xi_2(z)-\xi_3(z), & \mbox{ if } p\in\mathcal R_1, \\
 & \xi_1(z)-\xi_3(z), & \mbox{ if } p\in\mathcal R_2, \\
 & \xi_1(z)-\xi_2(z), & \mbox{ if } p\in\mathcal R_3.
\end{aligned}
\right.
\end{equation}

The function element $Q$ cannot be extend to a (single-valued) meromorphic function on the whole Riemann surface $\mathcal R$, but it admits an analytic extension along any path 
on $\mathcal R$. Hence given any path $\gamma$ on $\mathcal R$, it is meaningful to talk about contour integrals of $Q$ along $\gamma$. More important, the square $Q^2$ 
extends to a (single-valued!) meromorphic function on the whole Riemann surface $\mathcal R$, as it is shown in a general framework in \cite[Theorem~1.8]{martinez_silva}. Due to 
our explicit sheet structure, this can also be verified directly, but we skip the details. 

We are interested in the associated quadratic differential
\begin{equation}\label{quadratic_differential}
\varpi=-(Q(z))^2dz^2.
\end{equation}

Zeros and poles of $\varpi$ are the zeros and poles of $Q^2$, along with its multiplicities. Simple poles and zeros are called finite critical points, whereas poles of order at 
least $2$ are called infinite critical points. An arc $\gamma\subset\mathcal R$ is said to be an {\it arc of trajectory} of $\varpi$ if
\begin{equation}\label{definition_trajectory_integral}
\re \int^z \sqrt{-\varpi} = \re \int^z Q(s)ds=\const,\quad z\in \gamma.
\end{equation}

A trajectory is a maximal arc of trajectory, and it is called critical if it extends to a finite critical point of $\varpi$ on at least one of its ends. Two trajectories can only 
intersect at the critical points. The union of all critical trajectories of $\varpi$ is denoted by $\mathcal G=\mathcal G(\varpi)$ and is called the {\it critical graph} of 
$\varpi$.

The main goal of this Section is to describe the critical graph $\mathcal G$. As we will see later on, the critical graph plays a 
substantial role in the Riemann-Hilbert/Steepest Descent analysis carried over in Sections~\ref{section_riemann_hilbert_analysis_precritical} and 
\ref{section_riemann_hilbert_analysis_postcritical}. Some of its trajectories also encode Equation~\eqref{s_property}: Theorem~\ref{theorem_limiting_support_zeros} will follow 
almost immediately once we describe the critical graph $\mathcal G$.

To describe the trajectories of $\varpi$, we follow the methodology of the recent work \cite{martinez_silva} by Martínez-Finkelshtein and the second author. The main conclusion of 
this analysis is that the topology of the critical graph of $\varpi$ only depends on whether $(t_0,t_1)$ belongs to $\mathcal F_1$ or $\mathcal F_2$. 

In our setting, the analysis works as follows. We first describe the trajectories for $t_1=0$, for which the underlying rotational symmetry plays a fundamental role. It turns out 
that in this case $\mathcal R\setminus \mathcal G$ consists only of strip and half plane domains, and no short trajectories. When we increase $t_1$, the critical graph is 
deformed, and we control its dynamics by means of analyzing the widths of the strip domains, showing that they do not vanish on $\mathcal F_1$, and thus the critical graph is 
preserved for values of 
$(t_0,t_1)$ in this domain. When we cross $\gamma_c$, moving from $\mathcal F_1$ to $\mathcal F_2$, we are able to identify the phase transitions for the trajectories, and 
describe 
the critical graph for values of $(t_0,t_1)\in \mathcal F_2$ that are sufficiently close to $\gamma_c$. Once again the critical graph consists only of strip and half plane 
domains, 
but now there are also short trajectories. We again analyze the widths of the strip domains, and also the short trajectories, and prove that the critical graph remains unchanged 
on $\mathcal F_2$.

The standard references on quadratic differentials are the books by Strebel~\cite{strebel_book} and by Jenkins~\cite{jenkins_book}. We follow closely the aforementioned work 
\cite{martinez_silva}, where the reader can also find a discussion on the general theory of quadratic differentials in a form suitable for our needs. 

The present Section is organized in the following manner. In Sections~\ref{subsection_technical_computations_precritical} and 
\ref{subsection_technical_computations_supercritical} we derive some technical lemmas that are needed for the computation of the critical graph, first for the three-cut case 
$(t_0,t_1)\in \mathcal F_1$ and then for the one-cut case $(t_0,t_1)\in \mathcal F_2$. In Section~\ref{subsection_critical_graph_general_principles} we compute the 
zeros and poles of $\varpi$, and discuss some general principles that are used for the computation of the critical graph. Finally, in 
Sections~\ref{subsection_critical_graph_precritical} and \ref{subsection_critical_graph_supercritical} we derive the critical graph in the three-cut and one-cut cases, 
respectively. 

When describing the trajectories and dynamics of the critical graph of $\varpi$, instead of a precise formulation of the behavior of each trajectory we opt for a more ``reader 
friendly'' approach, with visual description and illustration of the results by a number of pictures. And of course, we always provide rigorous proofs of the results.

\subsection{Technical computations for the three-cut case}\label{subsection_technical_computations_precritical}

When $t_1=0$, the sheet structure constructed in Section~\ref{section_sheet_structure_1} is consistent with the sheet structure in 
Section~\ref{subsection_spectral_curve_t=0}. However, for the analysis of the trajectories of $\varpi$ when $t_1=0$, it is more convenient (although, strictly speaking, not 
necessary) to construct the sheets in a different way that better reflects the underlying discrete rotational symmetry.

According to Theorem~\ref{theorem_singular_points_spectral_curve}, $z_j,\hat z_j$, $j=1,2,3$, denote the branch points and the singular points of \eqref{spectral_curve}, 
respectively. In the case $t_1=0$, these points are explicitly given in \eqref{branchpoints_nodes_t_1=0-1}--\eqref{branchpoints_nodes_t_1=0-2}. 

We set 
\begin{equation}\label{definition_cuts_L_t=0}
L_j= [0,\infty e^{\frac{(2j+3)\pi i}{3}}], \quad j=0,1,2, \qquad L=\bigcup_{j=0}^2L_j,
\end{equation}
and recalling the set $\Sigma_*$ given explicitly for $t_1=0$ in \eqref{star_t=0}, we define
\begin{equation}\label{new_sheet_structure}
\widetilde{\mathcal R}_1=\overline \C\setminus \Sigma_*,\quad \widetilde{\mathcal R}_2=\overline\C\setminus \left( \Sigma_*\cup L \right),\quad \widetilde{\mathcal R}_3=\overline 
\C\setminus L.
\end{equation}

We then connect the sheets $\widetilde{\mathcal R}_1$ and $\widetilde{\mathcal R}_2$ along $\Sigma_*$ and the sheets $\widetilde{\mathcal R}_2$ and $\widetilde{\mathcal R_3}$ 
along $L$, always in the crosswise manner, and denote by $\widetilde{\mathcal R}$ the resulting three-sheeted Riemann surface,
$$
\widetilde{\mathcal R}=\widetilde{\mathcal R}_1\cup\widetilde{\mathcal R}_2\cup\widetilde{\mathcal R}_3.
$$

This construction can be compared to \eqref{sheet_structure_t=0_no_symmetry} through the identities
\begin{equation}\label{new_sheet_structure_2}
\widetilde {\mathcal R}_1=\mathcal R_1,
\qquad 
 \widetilde{\mathcal R}_2=
 \begin{cases}
  \mathcal R_3,& -\frac{\pi}{3}<\arg z<\frac{\pi}{3},\\
  \mathcal R_2, & \mbox{otherwise},
 \end{cases}
\qquad
 \widetilde{\mathcal R}_3=
 \begin{cases}
  \mathcal R_2,& -\frac{\pi}{3}<\arg z<\frac{\pi}{3},\\
  \mathcal R_3, & \mbox{otherwise}.
 \end{cases}
\end{equation}
That is, we interchange the sectors $-\pi/3<\arg z<\pi /3$ between the sheets $\mathcal R_2$ and $\mathcal R_3$. We refer the reader to Figure~\ref{figure_regluing_t_0} for a 
comparison of these sheets structures. 

\begin{figure}[t]
\centering
\includegraphics[scale=1]{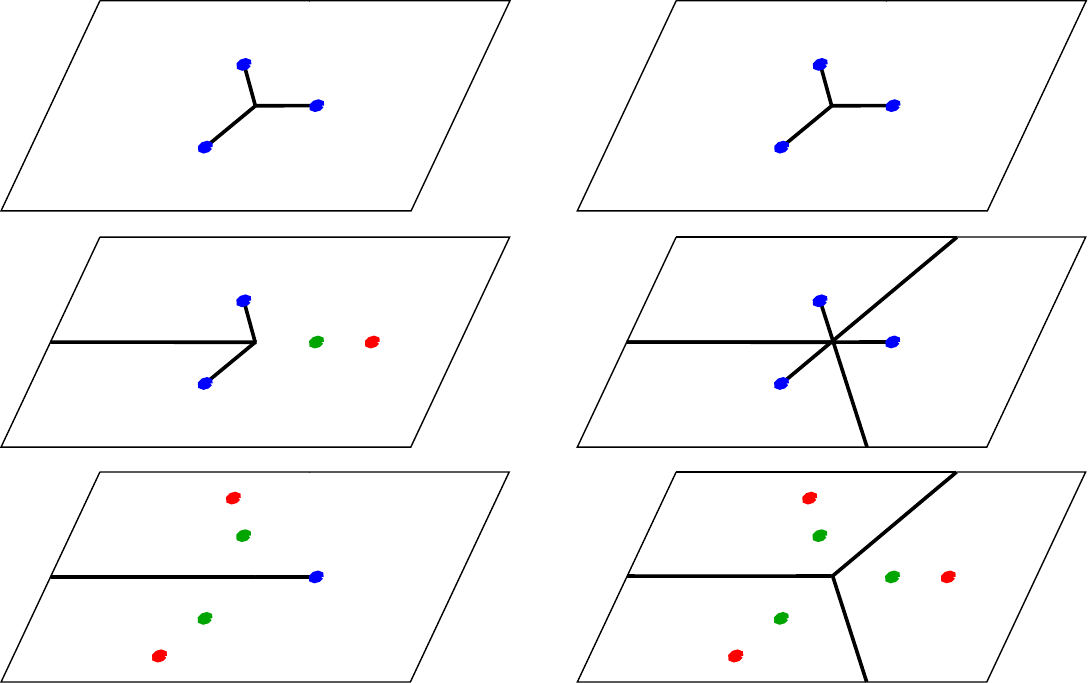}
\caption{For $t_1=0$, the cut structure and the zeros of $\varpi$ before (left) and after (right) the regluing $\mathcal R_j\mapsto \widetilde{\mathcal 
R}_j$.\label{figure_regluing_t_0}}
\end{figure}

Clearly this new sheet structure also affects the analytic continuation of the function germs in \eqref{asymptotics_xi}. Each function germ in \eqref{asymptotics_xi} admits 
an analytic continuation to the whole sheet $\widetilde{\mathcal R}_j$, and these analytic continuations satisfy the equalities
$$
\xi_1(z_j)=\xi_2(z_j),\quad \xi_1(\hat z_j)=\xi_2(\hat z_j),\quad j=0,1,2,
$$
and the identities
\begin{align}
& \xi_3(x)<\xi_2(x)<\xi_1(x),\quad & x>\hat z_0, \label{equalities_inequalities_real_line_xi_functions_t=0_1}\\
& \xi_3(x)<\xi_1(x)<\xi_2(x),\quad & z_0<x<\hat z_0, \label{equalities_inequalities_real_line_xi_functions_t=0_2}\\
& \xi_{1\pm}(x)=\overline{\xi_{2\pm}(x)}=\xi_{2\mp}(x),\quad \im \xi_3(x)=0,\quad & 0<x<z_0, \label{equalities_inequalities_real_line_xi_functions_t=0_3}\\
& \xi_{2\pm}(x)=\overline{\xi_{3\pm}(x)}=\xi_{3\mp}(x),\quad \im \xi_{1}(x)=0, \quad & x<0, \label{equalities_inequalities_real_line_xi_functions_t=0_4}
\end{align}
which are compatible with \eqref{zeros_difference_xi_functions_F_1}--\eqref{equalities_inequalities_real_line_xi_functions_precritical_4}, keeping in 
mind \eqref{new_sheet_structure_2}. During Section~\ref{section_quadratic_differential}, for $t_1=0$ we always use the analytic continuations of the functions $\xi_1,\xi_2,\xi_3$ 
in accordance to the sheet structure for $\widetilde{\mathcal R}$, unless otherwise stated. For $t_1>0$ we keep using the sheet structure constructed in 
Sections~\ref{section_sheet_structure_1} and \ref{section_sheet_structure_2}.

For $t_1=0$, define
\begin{equation}\label{definition_function_h_j}
h_j(x,y) =\int_{x}^y (\re\xi_{j+}(s)-\re\xi_{3+}(s))ds,\quad x,y\in \R, \quad j=1,2,
\end{equation}
where the integrals above are computed along the real axis. We emphasize that the functions $\xi_1$, $\xi_2$, $\xi_3$ in the expression above correspond to the 
sheet structure $\widetilde R_1,\widetilde R_2,\widetilde R_3$.

\begin{lem}\label{lemma_h_j_t=0}
Suppose that $t_1=0$. For $j=1,2$, the following properties hold true for $h_j$.
\begin{itemize} 
 \item[(i)] If $x,y\in [0,z_0]$, $x\neq y$, then $h_j(x,y)\neq 0$.
 \item[(ii)] If $h(x_j,y_j)=0$, for $x_j<y_j\leq 0$, $j=1,2$, then $(x_1,y_1)\cap (x_2,y_2)\neq \emptyset$.
\end{itemize}
\end{lem}
\begin{proof}
We first prove (i). If $h_j(x,y)=0$ for $x,y\in [0,z_0]$, then we conclude that there exists $u_0$ between $x$ and $y$ for which
$$
\re \xi_{j+}(u_0)=\re \xi_{3+}(u_0).
$$

It follows from the first equation in \eqref{branchpoints_nodes_t_1=0-1} that $z_0<1$. In particular, this implies
\begin{equation}\label{aux_equation_6}
 0<u_0<z_0<1.
\end{equation}

From \eqref{equalities_inequalities_real_line_xi_functions_t=0_3} we know that $\overline{\xi_{1+}(u_0)}=\xi_{2+}(u_0)$, so
$$
\re \xi_{2+}(u_0)=\re \xi_{1+}(u_0)=\xi_3(u_0).
$$

For $t_1=0$, the coefficient of $\xi^2$ in \eqref{spectral_curve} is $-z^2$. Using Vieta's relations we hence conclude 
$$
u_0^2=\xi_{1+}(u_0)+\xi_{2+}(u_0)+\xi_3(u_0)=3\xi_3(u_0).
$$

Plugging in the pair $(\xi_3(u_0),u_0)=(u_0^2/3,u_0)$ back to \eqref{spectral_curve}, we see that $u_0$ must be a root of
\begin{equation*}
\phi(u)=\frac{2}{27}u^6-\left( \frac{2}{3}-t_0 \right)u^3-A=0,
\end{equation*}
where the coefficient $A$, given explicitly in \eqref{A_for_t=0}, is positive because $t_0\in (0,1/8)$. Due to the rotational symmetry $\phi(\omega u)=\phi(u)$, the 
polynomial $\phi$ has at most two real roots. Straightforward computations show
$$
\phi(0)=-A<0,\qquad \phi(1)=t_0-A-\frac{16}{27}<0.
$$
Since the degree of $\phi$ is even and its leading coefficient is positive, the inequalities above imply that $\phi$ does not have roots on the interval $[0,1]$, 
thus $u_0\in \R\setminus [0,1]$. 
But this is in contradiction with \eqref{aux_equation_6}.

To get (ii), we first note that arguments similar as the analysis above show that $\phi$ has a zero in each of the intervals $(x_1,y_1)$ and $(x_2,y_2)$. On the other hand, the 
inequality and comments above also show that $\phi$ has exactly one zero on $(-\infty,0)$, so this zero must belong to both intervals $(x_1,y_1)$ and $(x_2,y_2)$.
\end{proof}

For $(t_0,t_1)\in \mathcal F_1$ and recalling the definition of the function germ $Q$ in \eqref{definition_function_germ_Q} and the sheet structure described in 
Section~\ref{section_sheet_structure_2}, the quantities
\begingroup
\allowdisplaybreaks
\begin{align}
\tau_1 & = \re \int_{z_2^{(3)}}^{z_0^{(3)}}Q(s)ds =\re \int_{z_2}^{z_0}(\xi_1(s)-\xi_2(s))ds, \label{width_tau_1}\\
\tau_2 & = \re \int_{z_0^{(2)}}^{z_2^{(2)}}Q(s)ds =\re \int_{z_0}^{z_2}(\xi_1(s)-\xi_3(s))ds, \label{width_tau_2}\\
\tau_3 & = \re \int_{z_2^{(3)}}^{\hat z_2^{(3)}}Q(s)ds = \re \int_{z_2}^{\hat z_2}(\xi_1(s)-\xi_2(s))ds, \label{width_tau_3}\\
\tau_4 & = \re \int_{z_2^{(3)}}^{z_1^{(2)}}Q(s)ds = \re \int_{z_2}^{x_*}(\xi_1(s)-\xi_2(s))ds+\re \int_{x_{*}}^{z_1}(\xi_1(s)-\xi_3(s))ds, \label{width_tau_4}\\
\tau_5 & = \int_{z_0^{(2)}}^{\hat z_0^{(2)}}Q(s)ds = \int_{z_0}^{\hat z_0}(\xi_1(s)-\xi_3(s))ds, \label{width_tau_5}
\end{align}%
\endgroup
are of interest for what comes later. In the formulas above, the paths of integration are taken in $\C\setminus (-\infty,z_*]\cup\Sigma_*$ and $x_*$ is any point in the 
interval $(-\infty,z_*)$: the value $\tau_4$ does not depend on the precise choice of $x_*$, as it is indicated by the first integral defining it. 

What is important here is that $\tau_j$, $j=1,\hdots,5$, does not vanish for $(t_0,t_1)\in \mathcal F_1$. The analysis of these quantities is carried out in the 
Appendix~\ref{appendix_widths_precritical}.

\subsection{Technical computations for the one-cut case}\label{subsection_technical_computations_supercritical}

It is a simple observation that if you fix $(\tilde t_0,\tilde t_1)\in \gamma_c$, then any pair of the form $(\tilde t_0,t_1)\in \mathcal F$, with $t_1$ larger 
than $\tilde t_1$, actually belongs to $\mathcal F_2$, as can be seen in Figure~\ref{phase_diagram}. In other words, $\mathcal F_2$ consists of points in $\mathcal F$ of the form 
$(\tilde t_0,t_1)$, where $(\tilde t_0,\tilde t_1)\in \gamma_c$  for some $\tilde t_1$ and $t_1>\tilde t_1$. We use this fact without further mention.

Recall that $w_0,w_1,w_2$ denote the zeros of $h'$ (see Lemma~\ref{lemma_location_zeros_derivative_h}) and, moreover, $w_1=w_2$ for $(t_0,t_1)\in\gamma_c$. For 
$(t_0,t_1)\in\mathcal F_2$, denote additionally by $\tilde w_j$ the simple root of $h(w)-z_j=h(w)-h(w_j)$, $j=0,1,2$, and extend $\tilde w_j$ for values 
$(t_0,t_1)\in \gamma_c$ by continuity. 

When $(t_0,t_1)\in\gamma_c$, we recall that $r,a_0$ are given in terms of $s\in (0,1/2)$ by \eqref{critical_curve_limiting_zero_distribution_a_r_plane}, so that 
$$
h'(w)=r-\frac{2 a_0 r}{w^2}-\frac{2 r^2}{w^3}=-\frac{s^3 (2 s-w) (s+w)^2}{w^3}.
$$
In particular $w_0=2s$, so
$$
h(w)-h(w_0)=\frac{s^3 (2 s-w)^2 (s+4 w)}{4 w^2}.
$$

Summarizing, the last two equations tell us that when $(t_0,t_1)\in \gamma_c$, the quantities $w_0, w_1,w_2$ and $\tilde w_0$ mentioned above are 
given in terms of the parameter $s\in(0,1/2)$ by
\begin{equation}\label{critical_zeros_critical_curve}
w_1=-s=w_2,\quad \tilde w_0=-\frac{s}{4},\quad w_0= 2s.
\end{equation}

For the choice $a_0=\frac{1}{4}$, $r=\frac{1}{32}$, corresponding to the pair $(t_0,t_1)=(\frac{639}{524288},\frac{191}{1024})\in\mathcal F_2$, we compute numerically
\begin{equation*}
\begin{array}{lll}
z_2\approx 0.18352, & w_2\approx -0.12932, & \tilde w_2 \approx -1.86844, \\
z_1\approx 0.20797, & w_1\approx -0.63351, & \tilde w_1 \approx -0.07786, \\
z_0\approx 0.29599, & w_0\approx 0.76284,  & \tilde w_0 \approx -0.05370.
\end{array}
\end{equation*}
so for this choice
\begin{equation}\label{inequalities_w_j}
 \tilde w_2 < w_1 < w_2 <\tilde w_1 <\tilde w_0 < w_0.
\end{equation}

If $(t_0,t_1)\in \mathcal F_2$, then none of the points $w_j,\tilde w_j$'s can pairwise coincide. Indeed, we already know from Lemma~\ref{lemma_location_zeros_derivative_h} that 
the points $w_0,w_1$ and $w_2$ do not pairwise coincide. If $\tilde w_j=\tilde w_k$ (or also $\tilde w_j=w_k$) for some pair $j,k$, then consequently
$$
z_j=h(w_j)=h(\tilde w_j)=h(\tilde w_k)=h(w_k)=z_k,
$$
which cannot occur on $\mathcal F_2$ (see Theorem~\ref{theorem_discriminant_spectral_curve}). Hence by continuity we conclude that \eqref{inequalities_w_j} holds for every pair 
$(t_0,t_1)\in\mathcal F_2$.

\begin{lem}\label{lemma_G_3}
 Fix $(\tilde t_0,\tilde t_1)\in \gamma_c$. Then
 $$
 \lim_{t_1\to \tilde t_1+} \frac{\partial w_1}{\partial t_1}=-\infty,\quad \lim_{t_1\to \tilde t_1+} \frac{\partial w_2}{\partial t_1}=+\infty
 $$
\end{lem}
\begin{proof}
 We deal with the equality for $w_1$. The case $w_2$ is analogous.
 Let $(\tilde t_0,t_1)\in\mathcal F_2$, so $t_1>\tilde t_1$. In this situation, $w_1$ is a simple zero of $h'$, so
 \begin{equation}\label{aux_equation_17}
 \frac{\partial w_1}{\partial t_1}=-\frac{\frac{\partial h'}{\partial t_1}(w_1)}{h''(w_1)}.
 \end{equation}
 
We see $(t_0,t_1)\mapsto (r,a_0)$ as a change of coordinates, so that from the chain rule 
\begin{equation}\label{change_coordinates_gradient}
\left(\frac{\partial h'}{\partial t_0},\frac{\partial h'}{\partial t_1}\right)=\left(\frac{\partial h'}{\partial r},\frac{\partial h'}{\partial a_0} \right)\frac{D (r,a_0)}{D 
(t_0,t_1)},
\end{equation}
where $D(r,a_0)/D(t_0,t_1)$ is the Jacobian of the change of coordinates. From the Inverse Function Theorem,
\begin{align*}
\frac{D(r,a_0)}{D(t_0,t_1)} & =\left(\frac{D(t_0,t_1)}{D(r,a_0)}\right)^{-1} \\
			    & = \det\left(\frac{D(t_0,t_1)}{D(r,a_0)}\right)^{-1}
			    \left(
			    \begin{array}{ll}
			     \phantom{-}\partial t_1/\partial a_0 & -\partial t_0/\partial a_0 \\
			     -\partial t_1/\partial r & \phantom{-}\partial t_0/\partial r
			    \end{array}
			    \right).
\end{align*}
We use this expression to compute the second component in \eqref{change_coordinates_gradient}, arriving at
\begin{equation}\label{deformation_h_t}
\frac{\partial h'}{\partial t_1}=\det\left(\frac{D(t_0,t_1)}{D(r,a_0)}\right)^{-1} \left( \frac{\partial t_0}{\partial r}\frac{\partial h'}{\partial a_0} -\frac{\partial 
t_0}{\partial a_0}\frac{\partial h'}{\partial r}\right).
\end{equation}
 
Using \eqref{system_change_coordinates_a}--\eqref{system_change_coordinates_b}, we compute explicitly 
\begin{align*}
 \frac{\partial t_0}{\partial r} & = 2 \left(1-4 a_0^2\right) r-8 r^3, & \frac{\partial t_0}{\partial a_0} &= -8 a_0 r^2, \\
 \frac{\partial t_1}{\partial r} & = -8a_0r, & \frac{\partial t_1}{\partial a_0} & = 1-2 a_0-4 r^2,
\end{align*}
and after a lengthy calculation
$$
\det\left(\frac{D(t_0,t_1)}{D(r,a_0)}\right)=2 r \left(8 a_0^3-4 a_0^2 \left(4 r^2+1\right)+a_0 \left(8 r^2-2\right)+\left(1-4 r^2\right)^2\right).
$$

Similarly,
$$
\frac{\partial t_0}{\partial r}\frac{\partial h'}{\partial a_0} -\frac{\partial 
t_0}{\partial a_0}\frac{\partial h'}{\partial r}=-\frac{4 r^2 \left(8 a r-2 a w^3-4 r^2 w+w\right)}{w^3}
$$

We are interested in the values of $\partial h'/\partial t_1$ for parameters on $\gamma_c$, so using \eqref{critical_curve_limiting_zero_distribution_a_r_plane} and 
\eqref{critical_zeros_critical_curve} in the last two equations, we get 
$$
\frac{D(t_0,t_1)}{D(r,a_0)}=2 s^3 \left(1-s^2\right)^3 \left(2 s^2+1\right)^2 \left(1-4 s^2\right)>0,\quad s\in (0,1/2)
$$
and
$$
\frac{\partial t_0}{\partial r}\frac{\partial h'}{\partial a_0} -\frac{\partial t_0}{\partial a_0}\frac{\partial h'}{\partial r}=-4 s^4 (1 - 15 s^4 - 4 s^6)<0, \quad s\in (0,1/2).
$$
In virtue of \eqref{deformation_h_t}, we conclude
$$
\frac{\partial h'}{\partial t_1}(w_1)<0,\quad (\tilde t_0,\tilde t_1)\in \gamma_c.
$$

Additionally, since $w_1$ is the smallest root of the continuous function $h'$ on $(-\infty,0)$ and $h'(w)\to r>0$ when $w\to -\infty$, we conclude
$$
h''(w_1)<0,\quad (t_0,t_1)\in\mathcal F_2.
$$

When $t_1$ approaches $\tilde t_1$, $w_1$ becomes a double root of $h'$, and hence $h''(w_1)$ approaches zero. The result then follows by combining the last two inequalities with 
\eqref{aux_equation_17}.
\end{proof}

Having in mind \eqref{system_change_coordinates_b}, set 
\begin{align*}
G(w) & = h\left(\frac{1}{w}\right)-\frac{t_1+h(w)^2}{3} \\
     & = \frac{2 r^2 w^2}{3}+\frac{4 a_0 r w}{3}+\frac{2 a_0}{3}+\frac{r(3-4a_0^2-2r^2)}{3 w} \\
     & \qquad -\frac{2a_0r^2(1+2a_0)}{3 w^2}-\frac{4 a_0 r^3}{3 w^3}-\frac{r^4}{3 w^4}, \quad w\in \C.
\end{align*}

When $(t_0,t_1)\in\gamma_c$ we use \eqref{critical_curve_limiting_zero_distribution_a_r_plane} and \eqref{critical_zeros_critical_curve} to compute
\begin{equation}\label{aux_equation_20}
G(w_1)=0=G(w_2),\quad G(\tilde w_0)=-\frac{3}{8}s^2 \left(7 s^6+12 s^4+8\right)<0.
\end{equation}

\begin{lem}\label{lemma_G_1}
 If $( \tilde t_0, \tilde t_1)\in \gamma_c$, then $G'(w)<0$ on $(-\infty,0)$.
\end{lem}
\begin{proof}
On $\gamma_c$, the quantities $a_0$ and $r$ are explicitly given in terms of $s\in(0,1/2)$ by \eqref{critical_curve_limiting_zero_distribution_a_r_plane}. Substituting 
these values in the definition of $G$, we can rewrite
\begin{align*}
G'(w)& =\frac{s^3}{3w^5}p(w),\\
p(w)& =4s^3w^6 + 6s^2w^5+(2s^6+9s^4-3)w^3+6s^5(1+3s^2)w^2+18s^8w+4s^9.
\end{align*}

The discriminant of $p$ can be decomposed as
$$
\disc(p;w)=c\; (s^2-1)^3s^{24}(s^{30}+\hdots),
$$
where the polynomial between parentheses has rational coefficients and is symbolically computed with Mathematica. Its zeros are computed numerically and verified to 
not belong to $(0,1/2)$.

As a consequence, we conclude that $p$ - and hence $G'$ - does not have zeros with multiplicity on $(-\infty,0)$. For $s=1/2$,
$$
p(w) = \frac{(w-1) \left(64 w^5+256 w^4+256 w^3-52 w^2-10 w-1\right)}{1024}
$$
and the set of zeros (with nonnegative imaginary part) of the polynomial between parentheses above is numerically computed to be
$$
\{-4.14697+1.16757i , -0.144709+0.158865 i, 0.58336 \}.
$$
so in particular $p(w)$ has no negative zeros for $s=1/2$, and the same holds true for $G'$. Since $G'$ is never zero at $w=0$ and, as we just observed, $G'$ does not have zeros 
with multiplicity, by continuity of the zeros of $G'$ with respect of $s$ we get that $G'$ is never zero on $(-\infty,0)$. The result then follows from the extra observation that 
$G'(w)\to -\infty$ when $w\to -\infty$. 
\end{proof}

\begin{lem}\label{lemma_G_4}
 Fix $(\tilde t_0,\tilde t_1)\in\gamma_c$. There exists $\varepsilon >0$ such that for any pair $(\tilde t_0,t_1)\in\mathcal F_2$ with $0<t_1-\tilde t_1<\varepsilon$, the 
following inequalities hold true
 $$
 G(w_2)<0<G(w_1).
 $$
\end{lem}

\begin{proof}
 We will prove
 \begin{equation}\label{aux_equation_18}
 \lim_{t_1\to \tilde t_1+}\frac{\partial }{\partial t_1}(G(w_1))=+\infty,\quad \lim_{t_1\to \tilde t_1+}\frac{\partial }{\partial t_1}(G(w_2))=-\infty.
 \end{equation}
 
 By continuity, it then follows that $t_1\mapsto G(w_1)$ ($t_1\mapsto G(w_2)$) is increasing (decreasing) for $t_1$ sufficiently close to $\tilde t_1$. Since 
$G(w_1)=0=G(w_2)$ for $t_1=\tilde t_1$ (see \eqref{aux_equation_20}), the result will follow.
 
 From the definition of $G$ and $j=1,2$,
 \begin{equation}\label{aux_equation_19}
 \frac{\partial }{\partial t_1}(G(w_j))=-\frac{1}{w_j^2}h'\left(\frac{1}{w_j}\right)\frac{\partial w_j}{\partial t_1}+\frac{\partial h}{\partial 
t_1}\left(\frac{1}{w_j}\right)-\frac{2}{3}h(w_j)\frac{\partial h}{\partial t_1}(w_j)-\frac{1}{3},
 \end{equation}
where we also used $h'(w_j)=0$. Additionally, we use \eqref{critical_zeros_critical_curve} to get 
$$
h'\left(\frac{1}{w_j}\right)\to h'\left(-\frac{1}{s}\right)=2s^9-3s^7+s^3>0,\quad \mbox{ as } t_1\to \tilde t_1.
$$

From Lemma~\ref{lemma_G_3} and the inequality above, we know that the first term on the right hand side of \eqref{aux_equation_19} goes to $(-1)^{j+1}\infty$ when $t_1\searrow 
\tilde t_1$. Since $w_j\to -s\neq 0$ in the limit $t_1\searrow \tilde t_1$, the remaining terms remain bounded, concluding the proof of \eqref{aux_equation_18}.
\end{proof}

\begin{prop}\label{proposition_G}
 Suppose $(\tilde t_0,\tilde t_1)\in\gamma_c$. There exists $\varepsilon>0$ such that for every choice $(\tilde t_0,t_1)\in\mathcal F_2$ with $0<t_1-\tilde t_1<\varepsilon$, the 
function $G$ 
has no zeros on the intervals $(-\infty,w_1]$ and $[w_2,\tilde w_0]$.
\end{prop}
\begin{proof}
Since $\tilde w_0<0$, it follows from continuity and Lemma~\ref{lemma_G_1} that
$$
 G'(w)<0, \quad w\in (-\infty,\tilde w_0], \quad 0<t_1-\tilde t_1<\varepsilon,
$$
so if $0<t_1-\tilde t_1<\varepsilon$, then $G$ is strictly decreasing on the interval $(-\infty,\tilde w_0]$, and as a consequence $G$ has 
at most one zero in this interval. From Lemma~\ref{lemma_G_4}, this zero has to be on the subinterval $(w_1,w_2)$, thus $G$ is never zero on 
$(-\infty,w_1]\cup[w_2,\tilde w_0]$.
\end{proof}

\begin{cor}\label{corollary_zeros_xi_G}
 Suppose $(\tilde t_0,\tilde t_1)\in\gamma_c$. There exists $\varepsilon>0$ such that for every choice $(\tilde t_0,t_1)\in\mathcal F_2$ with $0<t_1-\tilde t_1<\varepsilon$, the 
functions
 $$
 z\mapsto \xi_1(z)-\frac{t_1+z^2}{3},\quad z\mapsto \xi_2(z)-\frac{t_1+z^2}{3},
 $$
 do not vanish, respectively, on the intervals $(-\infty,z_2], [z_1,z_0]$.
\end{cor}

\begin{proof}
Under the mapping $z=h(w)$, the functions 
$$z\mapsto \xi_1(z)-\frac{t_1+z^2}{3},\quad  z\in (-\infty,z_2], \qquad z\mapsto 
\xi_2(z)-\frac{t_1+z^2}{3}, \quad z\in [z_1,z_0]
$$
transform to 
$$
w\mapsto G(w),\quad w\in (-\infty,\tilde w_2], \qquad w\mapsto G(w),\quad w\in [\tilde w_1,\tilde w_0]
$$
respectively, where we recall that $\tilde w_j$, $j=1,2$, is the simple root of $h(w)=z_j$. From \eqref{inequalities_w_j} we know $(-\infty,\tilde w_2]\subset (-\infty,w_1], 
[\tilde w_1,\tilde w_0]\subset [w_2,\tilde w_0]$, and the result 
follows from Proposition~\ref{proposition_G}.
\end{proof}

Similarly as it is done in \eqref{definition_function_h_j} for $t_1=0$, we now define
\begin{equation*}
h_j(x,y) =\int_{x}^y (\re\xi_{j+}(s)-\re\xi_{3+}(s))ds,\quad x,y\in \R, \quad j=1,2,\quad (t_0,t_1)\in\mathcal F_2.
\end{equation*}

The next result is the analogous of Lemma~\ref{lemma_h_j_t=0} for the one-cut case, and 
its proof is also similar.

\begin{lem}\label{lemma_zeros_h_j}
 Suppose $(\tilde t_0,\tilde t_1)\in \gamma_c$. There exists $\varepsilon>0$ such that for any pair $(\tilde t_0,t_1)\in\mathcal F$ with $0<t_1-\tilde t_1<\varepsilon$, the 
following properties hold true.
\begin{itemize}
 \item[(i)] If $x,y\in (-\infty,z_2]$, $x\neq y$, then $h_1(x,y)\neq 0$.
 \item[(ii)] If $x,y\in [z_1,z_0]$, $x\neq y$, then $h_2(x,y)\neq 0$.
\end{itemize}

\end{lem}
\begin{proof}
 For simplicity, denote $J_1=(-\infty,z_2]$, $J_2=[z_1,z_0]$. If $h_j(x,y)=0$ for some points $x,y\in 
J_j$, $x\neq y$, then there exists a point $u_0\in J_j$ for which
 $$
 \re \xi_{j+}(u_0)= \re\xi_{3+}(u_0).
 $$
 
Using the equalities in \eqref{equalities_inequalities_real_line_xi_functions_supercritical} we conclude
 $$
 \re \xi_{1+}(u_0)= \re\xi_{2+}(u_0)=\re\xi_{3+}(u_0).
 $$

According to the sheet structure constructed in Section~\ref{section_sheet_structure_2}, the function $\xi_j$ is real on $J_j$ and continuous across $J_j$, so the equality above 
implies
$$
 \re \xi_{1+}(u_0)+ \re\xi_{2+}(u_0)+\re\xi_{3+}(u_0)=3\xi_j(u_0).
$$

On the other hand, the sum $\xi_1+\xi_2+\xi_3$ is equal to minus the coefficient of $\xi^2$ in \eqref{spectral_curve}, that is,
$$
\xi_j(u_0)=\frac{1}{3}(\re \xi_{1+}(u_0)+ \re\xi_{2+}(u_0)+\re\xi_{3+}(u_0))=\frac{t_1+u_0^2}{3}.
$$

This last equality implies that the function 
$$
z\mapsto \xi_j(z)-\frac{t_1+z^2}{3}
$$
is zero at the point $u_0\in J_j$, contradicting Corollary~\ref{corollary_zeros_xi_G}.
\end{proof}

In the same spirit as at the end of Section~\ref{subsection_technical_computations_supercritical}, for $(t_0,t_1)\in \mathcal F_2$ we introduce the quantities
\begin{align}
\tau_1 & = \re \int_{z_2^{(3)}}^{z_1^{(3)}}Q(s)ds = \re \int_{z_2}^{z_1}(\xi_1(s)-\xi_2(s))ds, \label{width_tau_1_post}\\
\tau_2 & = \re \int_{z_2^{(2)}}^{z_1^{(2)}}Q(s)ds = \re \int_{z_2}^{z_1}(\xi_1(s)-\xi_3(s))ds, \label{width_tau_2_post}\\
\tau_3 & = \re \int_{z_2^{(3)}}^{\hat z_2^{(3)}}Q(s)ds = \re \int_{z_2}^{\hat z_2}(\xi_1(s)-\xi_2(s))ds, \label{width_tau_3_post}\\
\tau_4 & = \re \int_{z_2^{(1)}}^{z_1^{(1)}}Q(s)ds = \re \int_{z_2}^{z_1}(\xi_2(s)-\xi_3(s))ds, \label{width_tau_4_post}\\
\tau_5 & = \re\int_{z_0^{(2)}}^{\hat z_0^{(2)}}Q(s)ds = \re\int_{z_0}^{\hat z_0}(\xi_1(s)-\xi_3(s))ds, \label{width_tau_5_post}\\
\tau_6 & = -\re\int_{z_1^{(1)}}^{z_0^{(1)}}Q(s)ds = \re\int_{z_1}^{z_0}(\xi_3(s)-\xi_2(s))ds, \label{width_tau_6_post}
\end{align}%

The analysis of these quantities is carried over in the Appendix~\ref{appendix_widths_supercritical}. The important fact for what comes later is that these quantities never vanish 
for $(t_0,t_1)\in\mathcal F_2$.

\subsection{Quadratic differential on the spectral curve: general principles}\label{subsection_critical_graph_general_principles}

For a point $p\in \overline\C$, recall that $p^{(j)}$ denotes its preimage under the canonical projection $\pi:\mathcal R \mapsto \overline \C$ that lies on 
the sheet $\mathcal R_j$. Additionally, the points $z_j,\hat z_j$, $j=0,1,2,$ are given by 
Theorem~\ref{theorem_discriminant_spectral_curve}.

We use Theorem~\ref{theorem_discriminant_spectral_curve}, equations \eqref{zeros_difference_xi_functions_F_1}, \eqref{zeros_difference_xi_functions_F_2} and the asymptotics 
\eqref{asymptotics_xi} to find all critical points of $\varpi$. For 
$(t_0,t_1)\in\mathcal F_1$ they are as follows
\begin{itemize}
 \item Double zeros at the branch points $z_0^{(1)}=z_0^{(3)}$, $z_1^{(1)}=z_1^{(2)}$, $z_2^{(1)}=z_2^{(2)}$;
 
 \item Double zeros at $\hat z_0^{(2)}$, $\hat z_1^{(3)}$, $\hat z_2^{(3)}$;
 
 \item Simple zeros at $z_0^{(2)}$, $z_1^{(3)}$, $z_2^{(3)}$;
 
 \item A pole of order $5$ at $\infty^{(1)}$ and a pole of order $14$ at $\infty^{(2)}=\infty^{(3)}$,
\end{itemize}
whereas for $(t_0,t_1)\in \mathcal F_2$ the critical points are given by
\begin{itemize}
 \item Double zeros at the branch points $z_0^{(1)}=z_0^{(3)}$, $z_1^{(1)}=z_1^{(3)}$, $z_2^{(2)}=z_2^{(3)}$;
 
 \item Double zeros at $\hat z_0^{(2)}$, $\hat z_1^{(3)}$, $\hat z_2^{(3)}$;
 
 \item Simple zeros at $z_0^{(2)}$, $z_1^{(2)}$, $z_2^{(1)}$;
 
 \item A pole of order $5$ at $\infty^{(1)}$ and a pole of order $14$ at $\infty^{(2)}=\infty^{(3)}$.
\end{itemize}

These critical points are shown in Figure~\ref{figure_critical_points_quadratic_differential}. From the general theory, it is known that there are $n+2$ trajectories emanating 
from a given zero of order $n$, and any two consecutive trajectories form an angle $2\pi/(n+2)$ at the zero.

\begin{figure}[t]
 \centering
 \begin{overpic}[scale=1]
{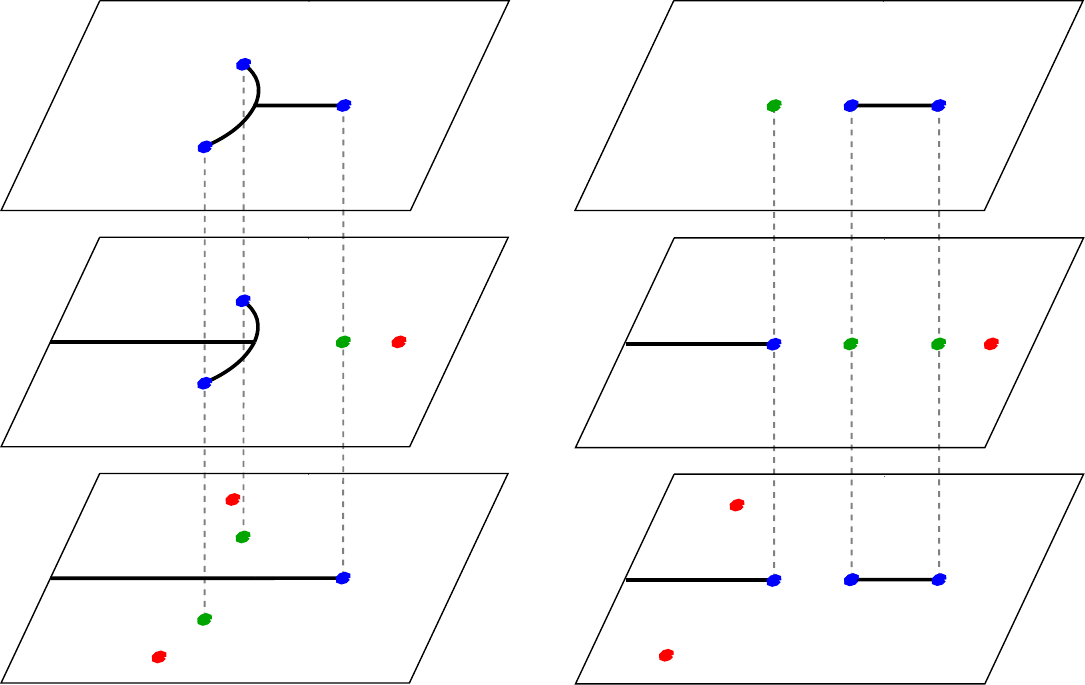}
 \put(33,53){$z_0^{(1)}$}
 \put(15,45){$z_1^{(1)}$}
 \put(23,58){$z_2^{(1)}$}
 \put(30.5,33){$z_0^{(2)}$}
 \put(13,25){$z_1^{(2)}$}
 \put(22,36.5){$z_2^{(2)}$}
 \put(33,9){$z_0^{(3)}$}
 \put(20,5){$z_1^{(3)}$}
 \put(23,13){$z_2^{(3)}$}
 \put(38,32){$\hat z_0^{(2)}$}
 \put(10,2){$\hat z_1^{(3)}$}
 \put(22.5,17){$\hat z_2^{(3)}$}
 \put(87.5,53){$z_0^{(1)}$}
 \put(78,54.5){$z_1^{(1)}$}
 \put(68,54){$z_2^{(1)}$}
 \put(86,32.5){$z_0^{(2)}$}
 \put(77.5,32.5){$z_1^{(2)}$}
 \put(71,32.5){$z_2^{(2)}$}
 \put(88,9){$z_0^{(3)}$}
 \put(77.5,11){$z_1^{(3)}$}
 \put(71,11){$z_2^{(3)}$}
 \put(92.5,31){$\hat z_0^{(2)}$}
 \put(62.5,2){$\hat z_1^{(3)}$}
 \put(69,16){$\hat z_2^{(3)}$}
 \end{overpic}
 \caption{The critical points of $\varpi$: on the left for $(t_0,t_1)\in\mathcal F_1$ and on the right for $(t_0,t_1)\in\mathcal 
F_2$. Blue dots represent the zeros that are also branch points, red dots represent the remaining double zeros, and green dots 
represent simple zeros. In addition, there are also poles at the points at infinity.}\label{figure_critical_points_quadratic_differential}
\end{figure}

It is time to fix some notation concerning the critical trajectories. From a given zero $p^{(j)}\in \mathcal R_j$ emanates a number of critical trajectories on $\mathcal 
R_j$, which we denote by $\gamma_0(p^{(j)}),\gamma_1(p^{(j)}),\hdots$; we convention that these trajectories are labeled in such a way that their canonical projections 
$\pi(\gamma_0(p^{(j)})),\pi(\gamma_1(p^{(j)})),\hdots$, are enumerated in the anti-clockwise direction, starting on the positive horizontal direction. This is well 
defined as long as there are no trajectories emanating along branch cuts, situation that will not occur. We also note that if $p^{(k)}=p^{(j)}$ is a branch point joining 
two 
sheets $\mathcal R_j$ and $\mathcal R_k$, the trajectories $\gamma_l(p^{(j)}), \gamma_l(p^{(k)})$ are different, because they emanate from different sheets. We refer the reader 
to Figure~\ref{figure_label_critical_trajectories} for an example.

Similar notation is adopted at the poles $\infty^{(1)}$, $\infty^{(2)}=\infty^{(3)}$. In this case, there are certain directions, henceforth called {\it critical 
directions} and denoted $\theta^{(k)}_{0}$, $\theta^{(k)}_{1},\hdots$, along which any trajectory extending to $\infty^{(k)}$ has to do so along one of the critical directions 
$\theta^{(k)}_{0},\theta^{(k)}_{1},\hdots$. These directions are easily computed to be given by the angles
\begin{equation*}
\theta^{(1)}_{j}=\frac{\pi}{3}+\frac{2\pi}{3}j,\quad j=0,1,2,\qquad \theta^{(2)}_{j}=\theta^{(3)}_{j}=\frac{\pi}{6}+\frac{\pi}{3}j,\quad j=0,\hdots,5.
\end{equation*}

Note that $\mathcal R$ is branched at $\infty^{(2)}=\infty^{(3)}$, so in the same spirit as for finite branch points, although the numerical values for $\theta^{(2)}_{j}$ and 
$\theta^{(3)}_{j}$ are the same, we refer to $\theta^{(2)}_{j}$ as the critical direction on the sheet $\mathcal R_2$, whereas $\theta^{(3)}_{j}$ as the critical direction on 
the sheet $\mathcal R_3$.

\begin{figure}[t]
\centering
 \begin{overpic}[scale=1]
  {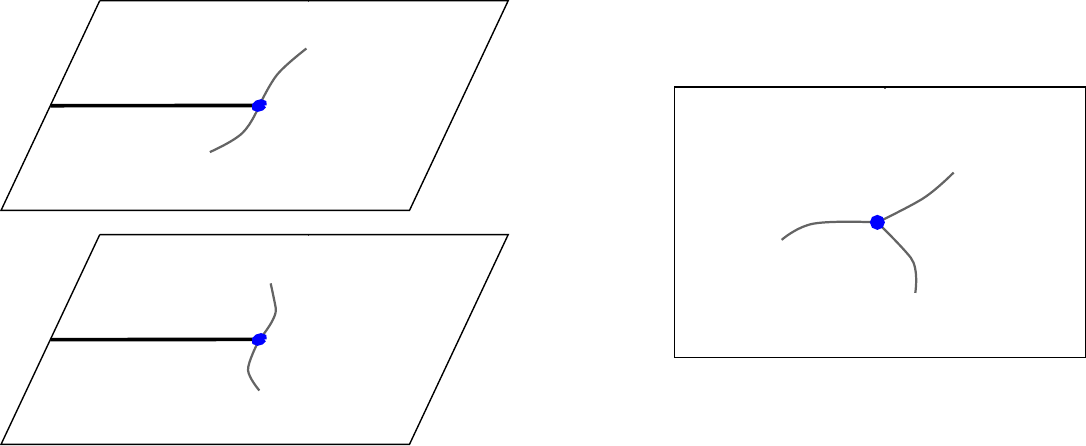}
  \put(79,18){$p$}
  \put(25,30){$p^{(1)}=p^{(2)}$}
  \put(25,9){$p^{(2)}=p^{(1)}$}
  \put(28,37){$\gamma_0(p^{(1)})$}
  \put(8,26){$\gamma_1(p^{(1)})$}
  \put(25,15){$\gamma_0(p^{(2)})$}
  \put(17,2){$\gamma_1(p^{(2)})$}
  \put(88,25){$\gamma_0(p)$}
  \put(69,21){$\gamma_1(p)$}
  \put(85,14){$\gamma_2(p)$}
 \end{overpic}
 \caption{Example of critical trajectories' labeling around critical points. On the left-hand side, around the double zero $p^{(1)}=p^{(2)}$ connecting two sheets of the Riemann 
surface. On the right, around the simple zero $p$ which is not a branch point.}\label{figure_label_critical_trajectories}
\end{figure}

Given $p\in \mathcal R$, we define the number $\eta(p)$ as follows.
\begin{itemize}
 \item If $p$ is a regular point of $\varpi$, then $\eta(p)=0$.
 \item If $p$ is a zero of $\varpi$, then $\eta(p)$ is its order as a zero.
 \item If $p$ is a pole of $\varpi$, then $\eta(p)$ is minus its order as a pole.
\end{itemize}

Clearly $\eta(p)>0$ iff $p$ is a zero, and $\eta(p)<0$ iff $p$ is a pole.

A domain $D\subset \mathcal R$ is called a {\it $\varpi$-polygon} if its boundary $\partial D$ is the union of a finite number of critical trajectories. For a $\varpi$-polygon 
$D$, we define the quantities
$$
\lambda(D)=\sum_{p\in D}\eta(p)
$$
and
$$
\kappa(D)=\sum_{p\in \partial D}\beta(p),
$$
where the summation for $\kappa$ is over all critical points on the boundary of $D$,
$$
\beta(p)=1-\theta(p)\frac{\eta(p)+2}{2\pi}
$$
and $\theta(p)=\theta(D;p)$ is the inner angle of $D$ at $p$.

If $D$ is simply connected, the formula
\begin{equation}\label{teichmuller_formula}
 \kappa(D)=2+\lambda(D)
\end{equation}
holds true. It is a consequence of the argument principle and is known as {\it Teichm\"{u}ller Lemma} or {\it Teichm\"{u}ller formula} \cite[Theorem~14.1]{strebel_book}.

Recall that $\mathcal G$ denotes the critical graph of $\varpi$ (see \eqref{definition_trajectory_integral} and the discussion thereafter). The critical graph $\mathcal G$ can be 
decomposed as a finite union
\begin{equation}\label{canonical_decomposition_critical_graph}
\mathcal G=\bigcup \mathcal U,
\end{equation}
where the sets $\mathcal U$'s are pairwise disjoint $\varpi$-polygons without critical points. For a quadratic 
differential without recurrent trajectories, we can classify the sets on the decomposition above into four distinct groups, namely: {\it strip } domains, {\it half plane} (or {\it 
end}) domains, {\it circle} domains and {\it ring} domains. This classification is known as the {\it Basic Structure Theorem} \cite[Theorem~3.5]{jenkins_book}. For our quadratic 
differential $\varpi$ in \eqref{quadratic_differential}, only the first two types of domain appear, and we explain their definition and basic properties next, in a form suitable 
for our needs.

A $\varpi$-polygon $\mathcal U$ in the decomposition \eqref{canonical_decomposition_critical_graph} is called a half plane domain if the primitive
\begin{equation}\label{primitive_uniformization}
\Upsilon(z)=\int^z Q(s)ds,\quad z\in \mathcal U
\end{equation}
is a conformal map from $\mathcal U$ to a half plane on $\C$ of the form
$$
\{z\in \C \; \mid \; \delta \re z>c \},
$$
for some real constant $c$ and some $\delta\in \{1,-1\}$. The boundary of a half plane domain contains exactly one infinite critical point $p$ and at least one finite critical 
point of $\varpi$. Locally at $p$, $\partial \mathcal U$ consists of two arcs of trajectories that end at $p$ along two consecutive critical directions. 

Conversely, any two consecutive critical directions at a given infinite critical point uniquely determine a half plane domain in the decomposition 
\eqref{canonical_decomposition_critical_graph}. 

In the same spirit, $\mathcal U$ is called a strip domain if the primitive $\Upsilon$ in \eqref{primitive_uniformization} is a conformal map from $\mathcal U$ to a strip on $\C$ 
of the form
$$
\{z\in \C \; \mid \; c_1<\re z <c_2\}
$$
for some real constants $c_1,c_2$. The boundary of $\mathcal U$ contains two infinite critical points $p_1,p_2$ (possibly $p_1=p_2$). The set $\partial U\setminus\{p_1,p_2\}$ has 
two connected components, and each of them contains at least one finite critical point. Moreover, the inner angle of $\partial U$ at each finite critical point is zero (that is, 
$\beta(p_1)=\beta(p_2)=1$). 

The positive number 
$$
\sigma(\mathcal U)=c_2-c_1
$$
is well defined regardless of the exact choice of the integration constant in \eqref{primitive_uniformization}, and it is called the {\it width} of the strip domain. For the 
quadratic differential $\varpi$ in \eqref{quadratic_differential}, it can be explicitly computed by
\begin{equation}\label{computation_width_parameter}
\sigma(\mathcal U)=|\Upsilon(q_1)-\Upsilon(q_2)|=\left|\re\int_{q_2}^{q_1}Q(s)ds \right|,
\end{equation}
where $q_1,q_2$ are any two points on different connected components of $\partial U\setminus \{p_1,p_2\}$, and the integral is computed along any path connecting $q_1$ and $q_2$ 
that lies entirely in $\overline{\mathcal U}$.

Recall that $\pi_j:\mathcal R_j\to \overline\C$ denotes the restriction of the canonical projection $\pi:\mathcal R\to \overline\C$ to the sheet $\mathcal R_j$. We follow 
\cite{martinez_silva} and list some general principles regarding trajectories of quadratic differentials. These principles will be extensively used later on.

\begin{enumerate}
	\item[\textbf{P.1}] The quadratic differential $\varpi$ does not have recurrent trajectories, as it follows from Jenkins' 
Three Poles Theorem. Consequently, any trajectory $\gamma$ of $\varpi$ has a well defined limiting point along its two directions (possibly the same, in case $\gamma$ is closed).

	\item[\textbf{P.2}] If $\gamma\subset \mathcal R_j$ is an arc of trajectory of $\varpi$, then the arc $\gamma^*=\pi_j^{-1}(\pi(\gamma)^*)$, obtained as the lift of the 
complex conjugate of $\pi(\gamma)$ to $\mathcal R_j$, is also an arc of trajectory.
	\item[\textbf{P.3}] If a $\varpi$-polygon does not have poles on its interior, then it has to have poles on its boundary. This is a consequence of 
\eqref{teichmuller_formula}. 

	\item[\textbf{P.4}] The function $Q^2$ in \eqref{quadratic_differential} is analytic on the parameters $(t_0,t_1)$. Consequently the trajectories of $\varpi$ change 
continuously (in any reasonable topology) with the parameters $(t_0,t_1)$. In our setting, it is enough to have the following observation.
Choose a critical point $p$ of $\varpi$, varying continuously with the parameter $(t_0,t_1)$. Fix $(\tilde t_0,\widetilde t_1)$ and assume that for small 
perturbations of $(\tilde t_0,\tilde t_1)$ the order of the critical point $p$ is 
preserved and that a given open set $U\subset \mathcal R$ does not contain critical points of $\varpi$. If for $(\tilde t_0,\tilde t_1)$ the critical 
trajectory emanating from $p$ along a given direction intersects the open set $U$, then the same holds true for small perturbations of $(\tilde t_0,\tilde t_1)$.

	\item[\textbf{P.5}] If for a given value $(\tilde t_0,\tilde t_1)$, a point $p$ belongs to the half plane domain for the 
pole $\infty^{(k)}$ determined by the 
critical angles $\theta^{(k)}_j,\theta^{(k)}_{j+1}$, then the same holds true for parameters on a small neighborhood of $(\tilde t_0,\tilde t_1)$. The 
point $p$ can depend continuously on $(t_0,t_1)$.

	\item[\textbf{P.6}] If for a given value $(\tilde t_0,\tilde t_1)$, an arc of trajectory emerging from a certain point $p$ intersects the real 
line at a {\it regular } point, then the same holds true for small perturbations of $(\tilde t_0,\tilde t_1)$. As before, the point $p$ can depend continuously on the parameters.
\end{enumerate} 

When $t_1=0$ we will make use of one more principle, which will be enunciated in Section~\ref{subsection_critical_graph_precritical}.

\subsection{Critical graph in the three-cut case}\label{subsection_critical_graph_precritical}

This section is devoted to the description of the critical graph in the three-cut case $(t_0,t_1)\in\mathcal F_1$. We will use extensively the principle {\bf P.1} without 
further mention. More precisely, in our situation this principle assures us that every trajectory has a limiting endpoint in its both directions, and our goal will be to, starting 
at a given critical trajectory emanating from a critical point, find its other endpoint.

\subsubsection{Critical graph for $t_1=0$}
 
 We first describe the critical graph for $t_1=0$. To this end, we use the sheet structure in \eqref{new_sheet_structure}. 
 
For $t_1=0$, the algebraic equation \eqref{spectral_curve} is invariant under the action
$$
(\xi,z)\mapsto (\omega^2 \xi,\omega z),
$$
where we recall that $\omega=e^{2\pi i/3}$. With regard to the sheet structure \eqref{new_sheet_structure}, this discrete rotational symmetry is reflected on the trajectories 
of $\varpi$, and it leads to the following principle.

\begin{enumerate}
 \item[\textbf{P.7}] Suppose $t_1=0$. If $\gamma\subset \widetilde{\mathcal R}_j$ is an arc of trajectory, then the arc $\tilde \gamma=\pi_j^{-1}(\omega\pi(\gamma))$, obtained 
as the lift of $\omega\pi(\gamma)$ to $\widetilde{\mathcal R}_j$, is also an arc of trajectory.
\end{enumerate}

Some of the trajectories of $\varpi$ are straightforward to describe. On the interval $(-\infty,0)$, the solutions $\xi_2$ and $\xi_3$ are complex conjugate of each other (see 
\eqref{equalities_inequalities_real_line_xi_functions_t=0_4}), and the same is 
true for $\xi_1$ and $\xi_2$ on $(0,z_0)$ (see \eqref{equalities_inequalities_real_line_xi_functions_t=0_3}). It thus follows from the definition 
of a trajectory in \eqref{definition_trajectory_integral} that
\begin{equation}\label{aux_equation_trajectories_2}
\gamma_1(z_0^{(3)})=\pi_1^{-1}(( -\infty ,0])\cup \pi_3^{-1}([0,z_0]),
\end{equation}
that is, the trajectory $\gamma_1(z_0^{(3)})$ starts to the left of $z_0^{(3)}$, moves to the sheet $\widetilde{\mathcal R}_1$ and goes to $\infty^{(1)}$ along the negative axis. 
Using the principle {\bf P.7}, we also get the trajectories $\gamma_0(z_1^{(3)})$ and $\gamma_2(z_2^{(3)})$. These are depicted in Figure~\ref{figure_description_trajectories_0}.

\begin{figure}[t]
 \centering
 \includegraphics[scale=1]{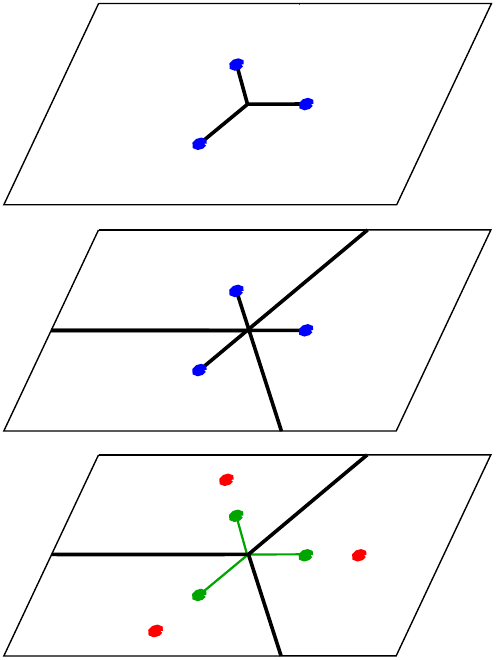}
 \caption{In thin lines the trivial trajectories of $\varpi$ in the large are depicted; these are the trajectories described by 
equations~\eqref{aux_equation_trajectories_2}--\eqref{aux_equation_trajectories_2}. In thick lines the branch cuts for the Riemann 
surface.}\label{figure_description_trajectories_0}
\end{figure}

We now focus on the trajectory $\gamma_0(z_0^{(1)})$. 

\begin{lem}\label{lemma_trajectory_1}
 The trajectory $\gamma_0(z_0^{(1)})$ goes to $\infty^{(1)}$ with angle $\theta^{(1)}_0$ and it is entirely contained in $\widetilde{\mathcal R}_1$
\end{lem}
\begin{proof}
We first claim that the trajectory $\gamma_0(z_0^{(1)})$ cannot intersect the interval $\pi^{-1}_1((z_0,+\infty))$. Indeed, to the contrary we use the principle {\bf P.2} to get 
the equality $\gamma_0(z_0^{(1)})=\gamma_1(z_0^{(1)})$. In particular, this implies that $\gamma_0(z_0^{(1)})$ is the boundary of a $\varpi$-polygon without 
poles on its closure, contradicting {\bf P.3}. 

The trajectory $\gamma_0(z_0^{(1)})$ cannot intersect the interval $\pi_{1+}^{-1}([0,z_0])$ neither. To see this, suppose it does, say at a point which projects to $x_0\in 
[0,z_0]$. It then follows from \eqref{definition_trajectory_integral} and the definition of $\varpi$ \eqref{quadratic_differential} that
\begin{equation}\label{aux_equation_7}
 \re \int_{x_0}^{z_0} (\xi_{2}(z)-\xi_3(z))dz  = \int_{x_0}^{z_0}(\re \xi_{2+}(x)-\re \xi_{3+}(x))dx =0,
\end{equation}
where the first integral is computed along $\pi(\gamma_0(z_0^{(1)}))$, and the second integral, obtained after deformation of the contour for the first one, is performed on 
$\R$. This is the same as saying that $h_2(x_0,z_0)=0$, contradicting Lemma~\ref{lemma_h_j_t=0} {\rm (i)}.

In summary, we proved that the trajectory $\gamma_0(z_0^{(1)})$ does not intersect the arc $\pi_1^{-1}([0,+\infty))$. Since it cannot intersect the segment $\pi_1^{-1}([0,\infty 
e^{\pi i/3})$, because this is an arc of the trajectory $\gamma_0(z_1^{(3)}))$, the only possibility left is that it goes to $\infty^{(1)}$ with angle $\theta^{(1)}_0$, as we want.
\end{proof}

Using the principles {\bf P.2} and {\bf P.7}, we also get the behavior of the trajectories $\gamma_1(z_0^{(1)})$, $\gamma_j(z_k^{(1)})$, $k=1,2$, $j=0,1$. The outcome we have 
so far is displayed in Figure~\ref{figure_description_trajectories_1}.

\begin{figure}[t]
 \centering
 \includegraphics[scale=1]{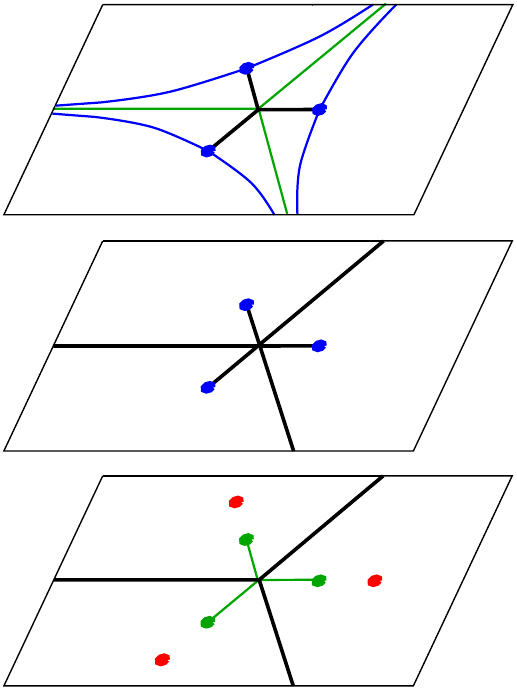}
 \caption{Trajectories of $\varpi$ in the large after equations~\eqref{aux_equation_trajectories_2}--\eqref{aux_equation_trajectories_2} and 
Lemma~\ref{lemma_trajectory_1}.}\label{figure_description_trajectories_1}
\end{figure}

Recall that $L_0,L_1$ and $L_2$ denote the segments connecting the sheets $\widetilde{\mathcal R}_2$ and $\widetilde{\mathcal R}_3$ (see \eqref{definition_cuts_L_t=0}).

\begin{lem}\label{lemma_trajectory_2}
The trajectory $\gamma_0(z_0^{(3)})$, intersects the cut $L_2$ exactly once, moves to the sheet $\widetilde{\mathcal R}_2$ and then 
extends to $\infty^{(2)}$ along the critical angle $\theta^{(2)}_{1}$.
\end{lem}
\begin{proof}
We first prove that $\gamma_0(z_0^{(3)})$ moves to $\widetilde{\mathcal R}_2$.

The trajectory $\gamma_0(z_0^{(3)})$ cannot intersect $\pi_3^{-1}([z_0,+\infty))$, otherwise the principle {\bf P.2} tells us that $\gamma_0(z_0^{(3)})=\gamma_2(z_0^{(3)})$, and 
this trajectory then determines a $\varpi$-polygon without poles on its closure, contradicting the principle {\bf P.3}. So if we assume it does not move to $\widetilde{\mathcal 
R}_2$, it has to extend to $\infty^{(3)}$ with angle $\theta^{(3)}_{0}$, and as a consequence of the principle {\bf P.2} we further get that $\gamma_2(z_0^{(3)})$ ends up at 
$\infty^{(3)}$ with angle $\theta^{(5)}_{3}$. We then apply Teichmüller's formula \eqref{teichmuller_formula} to the $\varpi$-polygon determined by 
these two trajectories that contains $\hat z_0^{(3)}$. On the boundary there are two critical points, namely $z_0^{(3)}$ and $\infty^{(3)}$, and the sum on the left hand side of 
\eqref{teichmuller_formula} is equal to $2$. On the interior, there is only the double zero $\hat z_0^{(3)}$, so the sum on the right hand side of \eqref{teichmuller_formula} is 
$4$, a contradiction. 

It is a conclusion of the last paragraph that the trajectory $\gamma_0(z_0^{(3)})$ goes to $\widetilde{\mathcal R}_2$. Combining {\bf P.2} and {\bf P.7}, we further get  
that the trajectories 
$\gamma_2(z_0^{(3)}),\gamma_{j}(z_1^{(3)})$ and $\gamma_{j-1}(z_2^{(3)})$ also move to $\widetilde{\mathcal R}_2$, $j=1,2$.

We now prove that $\gamma_0(z_0^{(3)})$ intersects $L_2$ exactly once. Again in virtue of the principles {\bf P.2} and {\bf P.7}, it is enough to verify that $\gamma_1(z_2^{(3)})$ 
intersects $L_0$ exactly once.

We already showed that $\gamma_1(z_2^{(3)})$ intersects $L_0$ at least once, say at a point projecting to $x\in (-\infty,0)$. Suppose now that $\gamma_1(z_2^{(3)})$ intersects 
$L_0$ at another point, say projecting to $y\in (-\infty,0)$. For simplicity, assume $x<y$, the case $x>y$ is analogous. Proceeding as in \eqref{aux_equation_7}, we conclude that 
$$
h(x,y)=h(y,0)=0.
$$  
From Lemma~\ref{lemma_h_j_t=0} {\rm (ii)}, we learn that $(x,y)\cap (y,0)\neq \emptyset$, which is certainly not true. We also remark that a combination of the principles {\bf 
P.2} and {\bf P.7} yield that $\gamma_0(z_2^{(3)})$ intersects the cut $L_2$ exactly once.

As a final step, we prove that $\gamma_0(z_0^{(3)})$ extends to $\infty^{(2)}$ with critical angle $\theta^{(2)}_{1}$. Again due to {\bf P.2} and {\bf P.7}, it is enough 
to prove that $\gamma_0(z_2^{(3)})$ extends to $\infty^{(2)}$ with angle $\theta_1^{(2)}$. The latter fact follows if we prove that $\gamma_0(z_2^{(3)})$ stays in the sector of 
$\widetilde{\mathcal R}_2$ determined by the positive real axis and $L_2$. Since we already noticed that $\gamma_0(z_2^{(3)})$ intersects $L_0$ exactly once and then enters this 
sector, it is enough to verify that $\gamma_0(z_2^{(3)})$ does not intersect the positive real axis on $\widetilde{\mathcal R}_2$.

Suppose the latter happens, say at a certain point projecting to $x\in [0,\infty)$. There are two possibilities, namely either $0\leq x \leq z_0$ or $x>z_0$.

For the first situation, we note that in this case the integral
$$
\int_{0^{(3)}}^{x^{(2)}}Q(s)ds,
$$
computed along the contour $\gamma_2(z_2^{(3)})\cup\gamma_0(z_2^{(3)})$, is purely imaginary. Deforming this integral to the interval $\pi_{2+}^{-1}([0,x])$ we conclude
$$
\int_0^x(\re \xi_{1+}(s)-\re \xi_{3+}(s))ds=0,
$$
and the equation above contradicts Lemma~\ref{lemma_h_j_t=0} {\rm (i)}.

Now let us assume that $x>z_0$. Using {\bf P.2}, this assumption implies that $\gamma_0(z_2^{(3)})=\gamma_2(z_1^{(3)})$. The trajectories 
$\gamma_2(z_2^{(3)})$, $\gamma_0(z_2^{(3)})$, $\gamma_2(z_1^{(3)})$ then determine a $\varpi$-polygon whose only critical point in its interior is the double zero 
$z_0^{(1)}=z_0^{(2)}$ (this is the shaded domain in Figure~\ref{figure_description_trajectories_2}). The right hand side of \eqref{teichmuller_formula} is then $4$, whereas the 
left hand side is $2$, which is a contradiction.
\end{proof}

\begin{figure}[t]
 \includegraphics[scale=1]{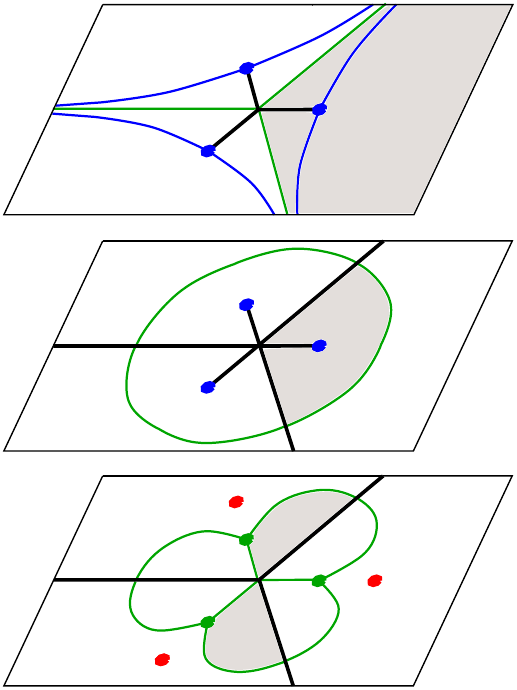}
 \caption{Figure for the hypothetical situation considered in the proof of Lemma~\ref{lemma_trajectory_2}. The shaded region is the $\varpi$-polygon contradicting Teichmüller's 
formula~\eqref{teichmuller_formula}.}\label{figure_description_trajectories_2}
\end{figure}

Again after an application of the principles {\bf P.2} and {\bf P.7}, we get the partial configuration in Figure~\ref{figure_description_trajectories_3}

\begin{figure}[t]
 \centering
 \includegraphics[scale=1]{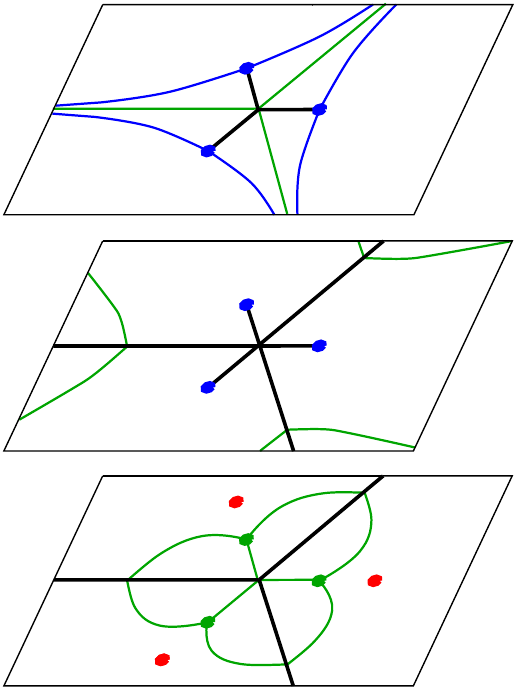}
 \caption{Trajectories of $\varpi$ in the large after Lemma~\ref{lemma_trajectory_2}. Compare with the previous stage in 
Figure~\ref{figure_description_trajectories_1}}\label{figure_description_trajectories_3}
\end{figure}

\begin{lem}\label{lemma_trajectory_3_1}
 The trajectories $\gamma_0(z_0^{(2)})$ and $\gamma_1(z_0^{(2)})$ extend to $\infty^{(2)}$ along critical angles 
$\theta^{(2)}_{0}$ and $\theta^{(2)}_{5}$, respectively.
\end{lem}
\begin{proof}

The trajectory $\gamma_0(z_0^{(2)})$ cannot intersect $\pi_{2+}^{-1}([0,z_0])$ neither. Otherwise, if it does so at a point $x_+^{(2)}$, then 
$$
0=\int_0^{x_+^{(2)}}Q(s)ds=\int_0^x (\re\xi_{2+}(s)-\re\xi_{3+}(s))ds,
$$
where for the first integral we integrate over $\gamma_0(z_0^{(2)})$, and then we deform this contour of integration to the real line to get the second equality. Since $x\in 
(0,z_0]$, this is in contradiction with Lemma~\ref{lemma_h_j_t=0} {\rm (i)}. 

Having in mind {\bf P.2}, a consequence of the discussion above is that the trajectories $\gamma_0(z_0^{(2)})$ and $\gamma_1(z_0^{(2)})$ have to stay in the region on 
$\widetilde{\mathcal R}_2\cup \widetilde{\mathcal R}_3$ determined by the contour $\gamma_0(z_2^{(3)})\cup\gamma_2(z_2^{(3)})\cup \gamma_0(z_1^{(3)})\cup\gamma_2(z_1^{(3)})$. 
Additionally, the trajectory $\gamma_0(z_0^{(2)})$ cannot intersect $\pi_2^{-1}([z_0,\infty))$, otherwise using {\bf P.2} we conclude that 
$\gamma_0(z_0^{(2)})=\gamma_1(z_0^{(2)})$, and this trajectory then determines a $\varpi$-polygon without poles on its closure, contradicting {\bf P.3}. The only possibility left 
is the description claimed by the Lemma.
\end{proof}

Again using {\bf P.2} and {\bf P.7}, we translate the previous Lemma to the trajectories emanating from $z_1^{(2)}$, $z_2^{(2)}$. The outcome is displayed in Figure~
\ref{figure_description_trajectories_3_1}.

\begin{figure}[t]
 \centering
 \includegraphics[scale=1]{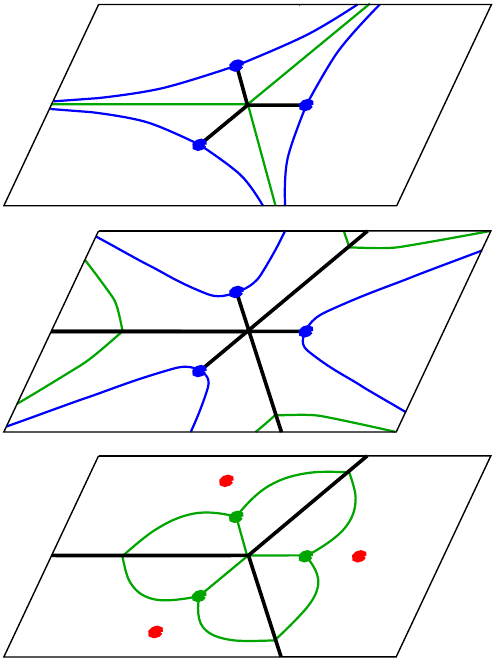}
 \caption{Trajectories of $\varpi$ in the large after Lemma~\ref{lemma_trajectory_3_1}. Compare with the previous stage in 
Figure~\ref{figure_description_trajectories_3}.}\label{figure_description_trajectories_3_1}
\end{figure}

\begin{lem}\label{lemma_trajectory_3}
 The trajectory $\gamma_0(\hat z_0^{(3)})$ never leaves $\widetilde{\mathcal 
R}_3$ and extends to $\infty^{(3)}$ with critical angle $\theta^{(3)}_{0}$.  The trajectory $\gamma_1(\hat z_0^{(3)})$ intersects the cut $L_0$ on $\widetilde{\mathcal R}_3$ 
exactly once, moves to the sheet $\widetilde{\mathcal R}_2$ and extends to $\infty^{(2)}$ along the critical direction $\theta^{(2)}_{1}$. 
\end{lem}
\begin{proof}
An integral deformation argument similar to the one presented in the proof of Lemma~\ref{lemma_trajectory_2} shows that any contour of the form $\gamma_k(\hat z_0^{(3)})\cup 
\gamma_j(\hat z_0^{(3)})$, $j,k\in \{1,2,3,4\}$, $j\neq k$, intersects each of the cuts $L_0$ and $L_2$ at most once. Moreover, these contours cannot be bounded closed loops in 
$\widetilde {\mathcal R}_3$, otherwise we would get a contradiction to {\bf P.3}.

There are two critical directions in the sector $-\frac{\pi}{3}<\arg z<\frac{\pi}{3}$ in $\widetilde{\mathcal R}_3$, namely $\theta^{(3)}_{0}$ and $\theta^{(3)}_{5}$, so from the 
local behavior of trajectories near the pole $\infty^{(3)}$ we know that at most two trajectories from $\hat z_0^{(3)}$ can extend to $\infty^{(3)}$.

Combining the last two paragraphs, the only possibility left is that $\gamma_0(\hat z_0^{(3)})$ and $\gamma_3(\hat z_0^{(3)})$ stay in $\widetilde{\mathcal R}_3$ and go to 
$\infty^{(3)}$ along the critical directions $\theta^{(3)}_{0}$ and $\theta^{(3)}_{5}$, respectively, and the trajectories $\gamma_1(\hat z_2^{(3)})$ and $\gamma_2(\hat 
z_2^{(3)})$ intersect the cuts $L_0$ and $L_2$, respectively, and move to $\widetilde{\mathcal R}_2$.

Consequently, we use {\bf P.2} and {\bf P.7} to get that $\gamma_1(\hat z_1^{(3)})$ has to intersect the branch cut $L_0$ on $\widetilde{\mathcal R}_3$ and move to the sheet 
$\widetilde{\mathcal R}_3$. Similarly as before, we also get that $\gamma_1(\hat z_1^{(3)})$ intersects $L_0$ at exactly one point. Indeed, due to the constrained geometry we 
have, this trajectory has to end up at a critical point on $\widetilde{\mathcal R}_1\cup\widetilde{\mathcal R}_2$. Hence it has to cross $L_0$ an odd number of times. If 
$x_1<y_1$ are any two points of intersection of $\gamma_1(\hat z_1^{(3)})$, we get $h(x_1,y_1)=0$, as it follows from deformation of 
integrals as we did above. From Lemma~\ref{lemma_h_j_t=0} {\it (ii)} we then conclude that there can be at most two of such points $x_1,y_1$. Since the number of intersections is 
odd, we conclude that actually there is exactly one intersection point. 

Using again {\bf P.2} and {\bf P.7}, we translate the outcome of the previous paragraph to the trajectory $\gamma_0(\hat z_0^{(3)})$, concluding that it intersects $L_2$ 
exactly once, and thus it has to stay in $\widetilde{\mathcal R}_2$. Consequently, the only possibility left for its behavior in the large is that it has to extend to 
$\infty^{(2)}$ along the critical direction $\theta^{(2)}_1$, concluding the proof.
\end{proof}

Using the principles {\bf P.2} and {\bf P.7}, it is straightforward to get from Lemma~\ref{lemma_trajectory_3} the behavior of the trajectories emanating from the double zeros 
$\hat z_j^{(3)}$, $j=0,1,2$, hence concluding the description of the critical graph of $\varpi$ for $t_1=0$, as can be seen in Figure~\ref{figure_description_trajectories_5}. For 
later convenience, we also reverse the regluing $\widetilde{\mathcal R}_j\mapsto \mathcal R_j$, and the result is also displayed in Figure~\ref{figure_description_trajectories_5}.

\begin{figure}[t]
 \centering
 \includegraphics[scale=1]{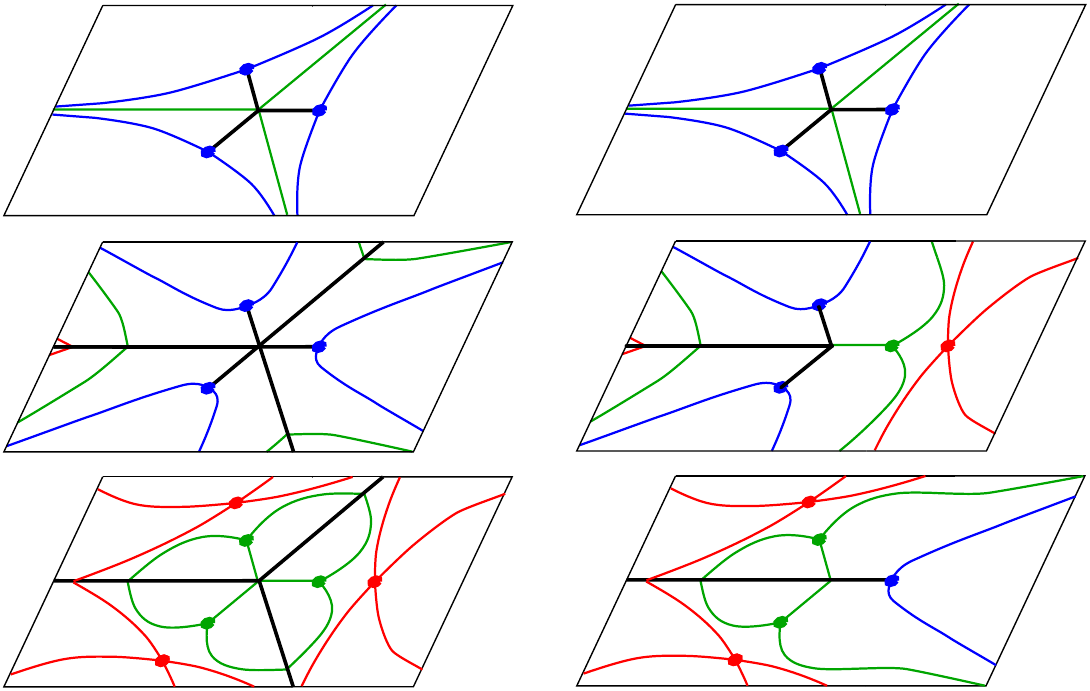}
 \caption{Critical graph of $\varpi$ for $t_1=0$ before (left) and after (right) the regluing $\widetilde{\mathcal R}_j\mapsto \mathcal 
R$.}\label{figure_description_trajectories_5}
\end{figure}

\subsubsection{Planar realization of the critical graph for $t_1=0$}\label{planar_critical_graph}

A planar realization of $\mathcal G$ for $t_1=0$ can be seen in Figure~\ref{figure_planar_graph}. It consists of a rectangle whose top and bottom are identified, and whose left 
and right rims correspond, respectively, to the poles $\infty^{(2)}=\infty^{(3)}$ and $\infty^{(1)}$. On each of these rims, there are a number of marked points, corresponding to 
the critical directions at the respective poles.

From its planar realization, it is easy to identify the strip and end domains of $\mathcal G$. We denote the strip domains by $\mathcal S_j$, $j=1,\cdots, 9$, labeling them 
according to Figure~\ref{figure_planar_graph}.

\begin{figure}[t]
 \centering
 \begin{overpic}[scale=1]
  {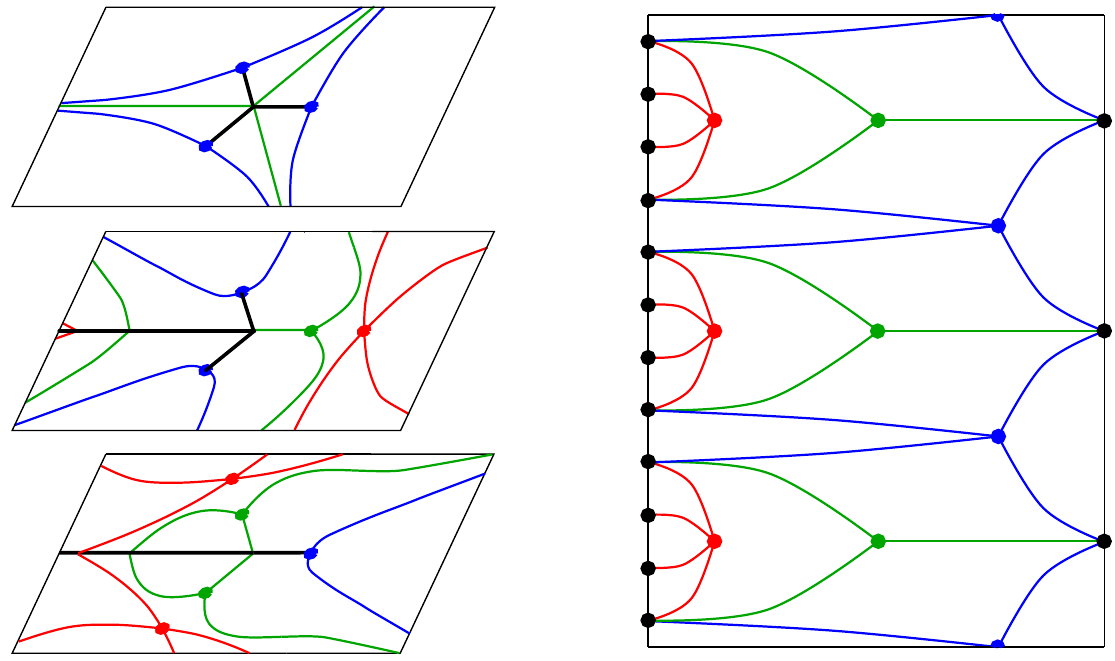}
  \put(52,55){$\theta^{(3)}_{0}$}
  \put(52,50.5){$\theta^{(3)}_{1}$}
  \put(52,46){$\theta^{(3)}_{2}$}
  \put(52,41){$\theta^{(2)}_{3}$}
  \put(52,36.5){$\theta^{(2)}_{4}$}
  \put(52,31.7){$\theta^{(2)}_{5}$}
  \put(52,26.9){$\theta^{(2)}_{0}$}
  \put(52,22){$\theta^{(2)}_{1}$}
  \put(52,17.5){$\theta^{(2)}_{2}$}
  \put(52,12.5){$\theta^{(3)}_{3}$}
  \put(52,8){$\theta^{(3)}_{4}$}
  \put(52,3.5){$\theta^{(3)}_{5}$}   
  \put(100.5,47.5){$\theta^{(1)}_{2}$}
  \put(100.5,28.5){$\theta^{(1)}_{1}$}
  \put(100.5,10){$\theta^{(1)}_{0}$}
  \put(67,47.5){$3$}
  \put(67,28.5){$5$}
  \put(67,10){$8$}
  \put(83,52){$1$}
  \put(83,33){$7$}
  \put(83,14){$4$}
  \put(83,43){$9$}
  \put(83,24){$2$}
  \put(83,5){$6$}
  \put(18,47.5){$\scriptstyle 7$}
  \put(19.5,50.5){$\scriptstyle 2$}
  \put(21,46){$\scriptstyle 9$}
  \put(24.5,47){$\scriptstyle 1$}
  \put(26,50.5){$\scriptstyle 6$}
  \put(23,51.5){$\scriptstyle 4$}
  \put(27,32){$\scriptstyle 2$}
  \put(23,25){$\scriptstyle 7$}
  \put(14,31){$\scriptstyle 4$}
  \put(13,27){$\scriptstyle 9$}
  \put(6,27){$\scriptstyle 3$}
  \put(8,31){$\scriptstyle 8$}
  \put(30,29){$\scriptstyle 5$}
  \put(25,12){$\scriptstyle 1$}
  \put(23,5){$\scriptstyle 6$}
  \put(21,14){$\scriptstyle 3$}
  \put(16,3.5){$\scriptstyle 8$}
  \put(17,10.5){$\scriptstyle 9$}
  \put(16,7){$\scriptstyle 4$}
 \end{overpic}
 \caption{On the left panel the critical graph of $\varpi$ in the three-cut case is displayed, and on the right panel its planar realization. The marked black dots on the rims of 
the rectangle correspond to the critical directions at the poles $\infty^{(1)}$ (right rim) and $\infty^{(2)}=\infty^{(3)}$ (left rim). In both pictures, the number $j$ indicates 
the 
strip domain $\mathcal S_j$.}\label{figure_planar_graph}
\end{figure}

\subsubsection{Deformation of the critical graph in the three-cut case}\label{deformation_critical_graph_precritical}

We now prove that the critical graph depicted in Figure~\ref{figure_planar_graph} is always valid in the three-cut case, that is, for $(t_0,t_1)\in\mathcal F_1$. According to 
the principle {\bf P.4}, this is the case as long as

\begin{enumerate}
 \item[\rm (i)] Existing zeros of $\varpi$ do not coalesce and there are no new zeros appearing. This is the case as the discussion at the beginning of 
Section~\ref{subsection_critical_graph_general_principles} assures us.

\item[\rm (ii)] No new domains appear. Taking into account {\rm (i)} above, this can only happen if a short trajectory changes its behavior. Since we do not have short 
trajectories for $t_1=0$, we are safe.

\item[\rm (iii)] The existing strip domains do not shrink. This amounts to showing that their widths do not vanish. For $j=1,\hdots,5$, the identity $\sigma(\mathcal 
S_j)=|\tau_j|$ follows directly from \eqref{computation_width_parameter} and the definition of $\tau_j$ given in \eqref{width_tau_1}--\eqref{width_tau_5}. Consequently, 
$\sigma(\mathcal S_j)\neq 0$ for $j=1,\cdots,5$ and $(t_0,t_1)\in\mathcal F_1$. Moreover, the symmetry under conjugation shows that $\sigma(\mathcal S_j)=\sigma(\mathcal 
S_{j-5})$ for $j=6,\hdots,9$, thus $\sigma(\mathcal S_j)\neq 0$ for $j=6,\hdots,9$ as well.
\end{enumerate}

The considerations {\rm (i)--(iii)} above hence show that the critical graph displayed in Figure~\ref{figure_planar_graph} is valid for every choice $(t_0,t_1)\in\mathcal F_1$.

\subsection{Critical graph in the one-cut case}\label{subsection_critical_graph_supercritical}

We now describe the critical graph in the one-cut case $(t_0,t_1)\in\mathcal F_2$. We first describe 
the critical graph of $\varpi$ for values of the parameter $(t_0,t_1)\in\mathcal F_2$ that are close to the critical curve $\gamma_c$ and then prove that 
the critical graph remains unchanged when we deform the parameters $(t_0,t_1)$ within $\mathcal F_2$.

\subsubsection{Critical graph in the one-cut case - short range}

When $(t_0,t_1)$ approaches the critical curve $\in\gamma_c$ from $\mathcal F_1$, the critical points $z_1^{(j)}$ and $z_2^{(j)}$, $j=1,2$, of $\varpi$ coalesce, thus the behavior 
of the critical trajectories emanating from these points can possibly (in fact, will) change. Choosing $(t_0,t_1)$ sufficiently close to $\gamma_c$, the remaining trajectories 
emanating from the critical points $z_0^{(1)}$, $z_0^{(2)}$, $\hat z_0^{(2)}$, $\hat z_1^{(3)}$, $\hat z_2^{(3)}$ inherit their behavior already described for parameters on 
$\mathcal F_1$. This is true because when we approach $\gamma_c$ from $\mathcal F_1$, each of such trajectories do not coalesce with other zeros of $\varpi$, and from principle 
{\bf P.4} they have to preserve their behavior. 

Taking into account the relations in the last two equations in \eqref{equalities_inequalities_real_line_xi_functions_supercritical}, we get 
$$
\gamma_1(z_2^{(1)})=\pi_1^{-1}((-\infty,z_2]), \quad \gamma_1(z_0^{(2)})=\pi_2^{-1}([z_1,z_0])=\gamma_0(z_1^{(2)}),
$$
so the starting configuration is the one displayed on Figure~\ref{figure_description_trajectories_6}, which at this moment we know it is valid whenever $(t_0,t_1)\in\mathcal F_2$ 
is chosen sufficiently close to $\gamma_c$.

\begin{figure}[t]
 \centering
 \includegraphics[scale=1]{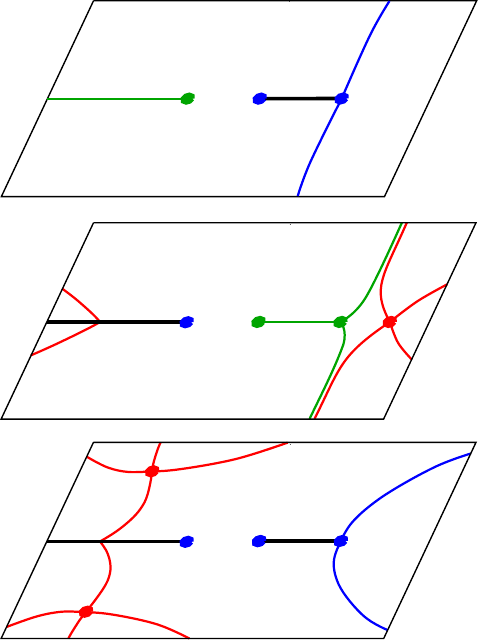}
 \caption{Trajectories that are preserved on the one-cut case when $(t_0,t_1)$ is close to $\gamma_c$.}\label{figure_description_trajectories_6}
\end{figure}

\begin{lem}\label{lemma_intersection_trajectories_postcritical}
 The trajectories $\gamma_0(z_1^{(1)}),\gamma_1(z_1^{(1)}), \gamma_0(z_1^{(3)})$ and $\gamma_1(z_1^{(3)})$ do not cross the branch cut $[z_1,z_0]$ connecting $\mathcal R_1$ and 
$\mathcal R_3$. 

The trajectories $\gamma_0(z_2^{(2)}),\gamma_1(z_2^{(2)}), \gamma_0(z_2^{(3)})$ and $\gamma_1(z_2^{(3)})$ do not cross the branch cut $(-\infty,z_2]$ connecting $\mathcal R_2$ and 
$\mathcal R_3$.
\end{lem}
\begin{proof}
 We deal with $\gamma_0(z_1^{(1)})$. The result for the remaining trajectories follow analogously.
 
 Suppose $\gamma_0(z_1^{(1)})$ crosses the branch cut $[z_1,z_0]$ at a point projecting to $x_0$. Deforming the integral over $\gamma_0(z_1^{(1)})$ to the real line, we 
conclude
 $$
 h_2(z_1,x_0)=\int_{z_1}^{x_0} \re(\xi_{2}(x)-\xi_{3+}(x))dx=0,
 $$
 but this cannot occur, as it follows from Lemma~\ref{lemma_zeros_h_j} {\rm (ii)} and our assumption that $(t_0,t_1)$ is sufficiently close to $\gamma_c$.
\end{proof}

\begin{lem}
The trajectories $\gamma_0(z_1^{(1)})$ and $\gamma_0(z_2^{(1)})$ extend to $\infty^{(1)}$ along the critical angle $\theta^{(1)}_{0}$. 

The trajectories $\gamma_1(z_1^{(1)})$ and $\gamma_2(z_2^{(1)})$ extend to $\infty^{(1)}$ along the critical angle $\theta^{(1)}_{2}$. 
\end{lem}
\begin{proof}
 From Lemma~\ref{lemma_intersection_trajectories_postcritical}, we know that $\gamma_0(z_1^{(1)})$ must stay on the sheet $\mathcal R_1$, so either $\gamma_0(z_1^{(1)})$ 
intersects $\pi_1^{-1}([z_2,z_1])$ or it goes to $\infty^{(1)}$. If the former occurs, then we use {\bf P.2} to get that $\gamma_0(z_1^{(1)})=\gamma_3(z_1^{(1)})$, and 
consequently it determines a bounded $\varpi$ polygon without poles on its closure, contradicting {\bf P.3}. Hence we conclude that $\gamma_0(z_1^{(1)})$ extends to 
$\infty^{(1)}$, 
either along $\theta^{(1)}_{0}$ or $\theta^{(1)}_{1}$. 

Consequently, the conclusion above forces $\gamma_0(z_2^{(1)})$ to stay in $\mathcal R_1$, so it also has to go to $\infty^{(1)}$ either along $\theta^{(1)}_{0}$ or 
$\theta^{(1)}_{1}$. We then apply \eqref{teichmuller_formula} to the $\varpi$-polygon $D\subset \mathcal R_1$ determined by the trajectories $\gamma_0(z_2^{(1)})$ and 
$\gamma_1(z_2^{(1)})$. There are no critical points on $D$, so the right-hand side of \eqref{teichmuller_formula} is equal to $2$. Moreover, a simple computation shows that 
$\beta(z_2^{(1)})=0$, so
$$
\beta(\infty^{(1)})=2-\beta(z_2^{(1)})=2,
$$
and then $\theta(\infty^{(1)},D)=\frac{\pi}{3}$. Consequently, $\gamma_0(z_2^{(1)})$ extends to $\infty^{(1)}$ along $\theta^{(1)}$, and the same has to hold for 
$\gamma_0(z_1^{(1)})$. 

Finally, we get the behavior of the trajectories $\gamma_1(z_1^{(1)})$ and $\gamma_2(z_2^{(1)})$ by simply applying the principle {\bf P.2}.
\end{proof}

Very similar arguments as for the previous proof also work to describe the behavior of the trajectories emanating from $z_1^{(j)}$, $z_2^{(j)}$, $j=2,3$. We skip the details. 

The final outcome is the critical graph displayed in Figure~\ref{figure_planar_graph_2}, where we remind the reader that we are assuming $(t_0,t_1)\in\mathcal F_2$ is 
sufficiently close to the critical curve $\gamma_c$. In this figure, the planar version of the critical graph (as explained in Section~\ref{planar_critical_graph}) is also 
displayed. We denote the strip domains by $\mathcal S_j$, $j=1,\hdots,8$, labeled as displayed in Figure~\ref{figure_planar_graph_2}.

\begin{figure}[t]
 \centering
 \begin{overpic}[scale=1]
  {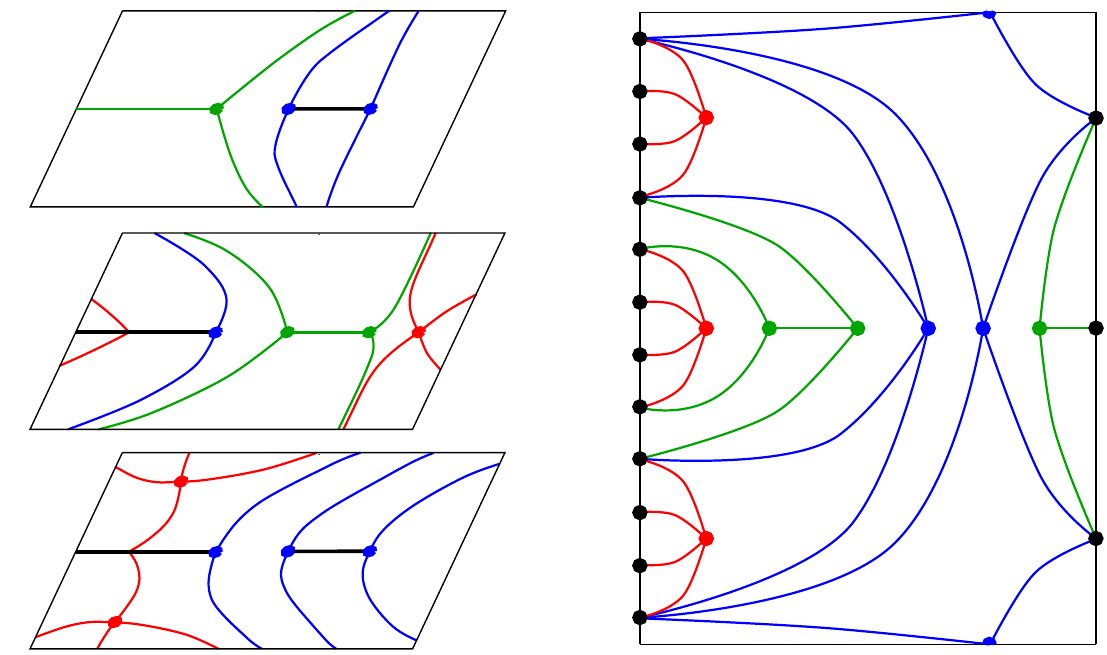}
  \put(52,55){$\theta^{(3)}_{0}$}
  \put(52,50.5){$\theta^{(3)}_{1}$}
  \put(52,46){$\theta^{(3)}_{2}$}
  \put(52,41){$\theta^{(2)}_{3}$}
  \put(52,36.5){$\theta^{(2)}_{4}$}
  \put(52,31.7){$\theta^{(2)}_{5}$}
  \put(52,26.9){$\theta^{(2)}_{0}$}
  \put(52,22){$\theta^{(2)}_{1}$}
  \put(52,17.5){$\theta^{(2)}_{2}$}
  \put(52,12.5){$\theta^{(3)}_{3}$}
  \put(52,8){$\theta^{(3)}_{4}$}
  \put(52,3.5){$\theta^{(3)}_{5}$}   
  \put(100.5,47.5){$\theta^{(1)}_{2}$}
  \put(100.5,28.5){$\theta^{(1)}_{1}$}
  \put(100.5,10){$\theta^{(1)}_{0}$}
  \put(30,51){$\scriptstyle 6$}
  \put(22,48.5){$\scriptstyle 4$}
  \put(27,45){$\scriptstyle 7$}
  \put(13,32){$\scriptstyle 8$}
  \put(12,25){$\scriptstyle 3$}
  \put(35,28.5){$\scriptstyle 5$}
  \put(22,28.5){$\scriptstyle 2$}
  \put(17,12){$\scriptstyle 3$}
  \put(15,5){$\scriptstyle 8$}
  \put(22,9){$\scriptstyle 1$}
  \put(32,12){$\scriptstyle 7$}
  \put(29,6){$\scriptstyle 6$}
  \put(70,45){$3$}
  \put(87,47){$7$}
  \put(70,12){$8$}
  \put(87,7){$6$}
  \put(65.5,28.5){$5$}
  \put(79.5,28.5){$2$}
  \put(85.5,28.5){$1$}
  \put(91,28.5){$4$}
 \end{overpic}
 \caption{On the left the critical graph of $\varpi$ for $(t_0,t_1)\in\mathcal F_2$, and on the right its planar realization. The marked black dots on the rims of 
the rectangle correspond to the critical directions at the poles $\infty^{(1)}$ (right rim) and $\infty^{(2)}=\infty^{(3)}$ (left rim). In both pictures, the number $j$ indicates 
the strip domain $\mathcal S_j$.}\label{figure_planar_graph_2}
\end{figure}

\subsubsection{Deformation of the critical graph in the one-cut case}

We now prove that the critical displayed in Figure~\ref{figure_planar_graph_2} is always valid in the one-cut case, that is, for the whole range $(t_0,t_1)\in \mathcal R_2$.
Similarly as in Section~\ref{deformation_critical_graph_precritical}, we use the principle {\bf P.4} to conclude that this is the case as long as

\begin{enumerate}
 \item[(i)] The order of every critical point of $\varpi$ is preserved and no new critical points appear. This is indeed the case, as discussed at the beginning of 
Section~\ref{subsection_critical_graph_general_principles}.

\item[(ii)] No new domains appear. Taking into account {\rm (i)} above, this can only occur if a short trajectory changes its behavior. In the present situation, the only critical 
trajectory is $\gamma_0(z_1^{(2)})=\gamma_2(z_0^{(2)})$, which lives on the real line, so from principle {\bf P.2} it is clear that this trajectory does not change its behavior.

\item[(iii)] The strip domains do not shrink. As before, it is enough to verify that the widths $\sigma(\mathcal S_j)$, $j=1,\hdots, 8$, do not vanish for $(t_0,t_1)\in\mathcal 
F_2$. For $j=1,\hdots,6$, these are given by $\sigma(\mathcal S_j)=|\tau_j|$, where $\tau_j$ is as in \eqref{width_tau_1_post}--\eqref{width_tau_6_post}, so they do not vanish. 
Due to the symmetry under conjugation, the remaining widths satisfy $\sigma(\mathcal S_{j})=\sigma(\mathcal S_{j-6})$, $j=7,8$, so these do not vanish as well.
\end{enumerate}

From {\rm (i)}--{\rm (iii)} above, we finally conclude that the critical graph depicted in Figure~\ref{figure_planar_graph_2} is valid for $(t_0,t_1)\in \mathcal R_2$.

\section{Proof of Theorems \ref{thm_equilibrium_measure}, \ref{theorem_density_eigenvalues}, 
\ref{theorem_limiting_support_zeros} and  \ref{theorem_balayage_relation}}\label{section_proof_s_property}

In Section~\ref{section_sheet_structure_general} the Riemann surface $\mathcal R$ was constructed as a three-sheeted cover of $\overline\C$ with sheets 
$\mathcal R_1$, $\mathcal R_2$, $\mathcal R_3$. Up to this moment, the cut $\gamma_0$ used in the construction of $\mathcal R_1$ and $\mathcal R_2$ in the three-cut case 
(carried out in Section~\ref{section_sheet_structure_1})
was quite arbitrary, but in what follows it is important to choose it in an optimal way.

From the analysis of the quadratic differential carried over in Section~\ref{subsection_critical_graph_precritical} (see in particular Figure~\ref{figure_planar_graph}), we know 
that the arc of trajectory $\gamma_0(z_1^{(3)})\cap \mathcal R_3$ connects $z_1^{(3)}$ to a point $z_*^{(3)}$, where $z_*\in (-\infty, z_0)$. We then 
define $\Sigma_{*,1}$ to be the contour on the complex plane obtained by the projection of this arc of trajectory, that is,
\begin{equation}\label{projection_property_sigma_1_three_cut}
\Sigma_{*,1}=\pi\left( \gamma_0(z_1^{(3)})\cap \mathcal R_3\right).
\end{equation}
In this way $\Sigma_{*,1}$ is an arc with endpoints $z_1$ and $z_*$, and it is contained on the lower half plane. We additionally set
\begin{equation}\label{projection_property_sigma_2_three_cut}
\Sigma_{*,2}=(\Sigma_{*,1})^*=\pi\left( \gamma_2(z_2^{(3)})\cap \mathcal R_3\right),
\end{equation}
so that $\Sigma_{*,2}$ is an arc with endpoints $z_2$ and $z_*$ that is contained on the upper half plane. Furthermore, denote
\begin{equation}\label{projection_property_sigma_0_three_cut}
\Sigma_{*,0}=[z_*,z_0],
\end{equation}
where $z_*$ is the point of intersection of $\Sigma_{*,1}$ with the real axis as above.

Following the construction carried out in Section~\ref{section_sheet_structure_1}, we then set 
$$
\gamma_0=\Sigma_{*,1}\cup \Sigma_{*,2}
$$
and
\begin{equation}\label{special_cut}
\Sigma_*=\Sigma_{*,0}\cup \Sigma_{*,1}\cup \Sigma_{*,2}
\end{equation}
so that the cut for the sheet $\mathcal R_1$ is simply given by
$$
\gamma_0\cup [z_*,z_0]=\Sigma_*.
$$

Note also that with this sheet structure, we can further characterize the interval $\Sigma_{*,0}$ by
$$
\Sigma_{*,0}=\pi \left( \gamma_1(z_0^{(2)}) \cap \mathcal R_2\right).
$$

For the rest of this paper, whenever we refer to the sheet structure $\mathcal R_1\cup\mathcal R_2\cup\mathcal R_3$ for $\mathcal R$ in the three-cut case, we always 
assume the cut used in the sheet $\mathcal R_1$ to be given by $\Sigma_*$ as in \eqref{special_cut}. Furthermore, we orient the arcs of $\Sigma_{*}$ outwards, that is, 
$\Sigma_{*,j}$ is oriented from $z_*$ to $z_j$, $j=0,1,2$.

In the one-cut case, we keep denoting
$$
\Sigma_*=[z_1,z_0].
$$
In this case it follows from the analysis in Section~\ref{subsection_critical_graph_supercritical} that $\Sigma_*$ can be alternatively expressed through the identity
\begin{equation}\label{projection_property_sigma_one_cut}
\Sigma_*=\pi\left( \gamma_0(z_1^{(2)}) \right)=\pi\left( \gamma_1(z_0^{(2)}) \right),
\end{equation}
see in particular Figure~\ref{figure_planar_graph_2}.

Theorem~\ref{theorem_schwarz_function} assures us that the function $w\mapsto h(w)$ is a conformal map from $\overline \C$ to $\mathcal R$. Standard arguments on conformal 
mappings allow us to recover the inverse image of each of the sheets $\mathcal R_1$, $\mathcal R_2$ and $\mathcal R_3$ under $h$. The outcome can be seen in 
Figure~\ref{w_partition}.

\begin{figure}[t]
\begin{overpic}[scale=1]
{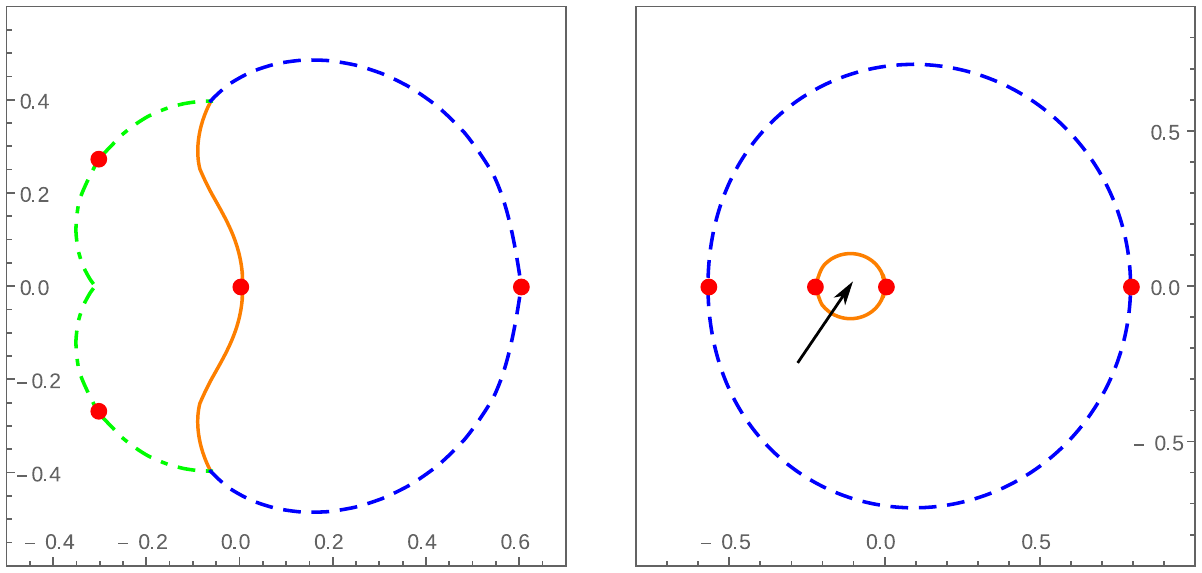}
\put(10,43){$\scriptstyle h^{-1}(\mathcal R_1)$}
\put(25,25){$\scriptstyle h^{-1}(\mathcal R_3)$}
\put(9,25){$\scriptstyle h^{-1}(\mathcal R_2)$}
\put(21,23){$\scriptstyle 0$}
\put(9,33){$\scriptstyle w_2$}
\put(8.5,14.5){$\scriptstyle w_1$}
\put(40,23){$\scriptstyle w_0$}
\put(60,43){$\scriptstyle h^{-1}(\mathcal R_1)$}
\put(80,25){$\scriptstyle h^{-1}(\mathcal R_3)$}
\put(63,15){$\scriptstyle h^{-1}(\mathcal R_2)$}
\put(75,23){$\scriptstyle 0$}
\put(60,23){$\scriptstyle w_1$}
\put(64.5,23){$\scriptstyle w_2$}
\put(90.5,23){$\scriptstyle w_0$}
\end{overpic}
\caption{The partition of the $w$-plane into the inverse images of the sheets $\mathcal R_1$, $\mathcal R_2$ and $\mathcal R_3$ under the conformal map $h$. The dots are the 
points $w_0,w_1,w_2$ and $0$, which are the inverse images of the branch points $z_0$, $z_1$, $z_2$ and $\infty^{(2)}$, respectively. In the left panel the three-cut case 
$\mathcal F_1$ is considered: the dashed line is the inverse image of the cut $\Sigma_{*,0}$ connecting $\mathcal R_1$ and $\mathcal R_3$, the dot-dashed line is the inverse image 
of the cut $\Sigma_{*,1}\cup\Sigma_{*,2}$ connecting $\mathcal R_1$ and $\mathcal R_2$ and the solid line is the inverse image of the cut $(-\infty,z_*)$ connecting $\mathcal R_2$ 
and $\mathcal R_3$. In the right panel the one-cut case $\mathcal F_2$ is considered: the dashed line is the inverse image of the cut $(z_1,z_0)$ connecting $\mathcal R_1$ and 
$\mathcal R_3$ and the solid line is the inverse image of the cut $(-\infty,z_2)$ connecting $\mathcal R_2$ and $\mathcal R_3$. Numerical output for the choices 
$r=1/20$ and $a_0=1/10$ (three-cut) and $a_0=1/4$ (one-cut).}\label{w_partition}
\end{figure}

With these definitions, the set $\Sigma_*$ satisfies all the geometric properties claimed by Theorem~\ref{theorem_limiting_support_zeros}, except that we still 
have to prove the inclusion \eqref{inclusion_star_Omega}. Furthermore, the function $\xi_1$ in \eqref{asymptotics_xi} admits a meromorphic continuation (that we keep 
denoting by $\xi_1$) to the whole sheet $\mathcal R_1=\overline\C\setminus \Sigma_*$. 

Recalling \eqref{schwarz_function_S} {\it et seq.}, we also know that the function germ $\xi_1$ in \eqref{asymptotics_xi} admits another meromorphic continuation $S$ to a 
neighborhood $G$ of $\overline\C\setminus \Omega$, which is also the Schwarz function of $\Omega$. However, we emphasize that we do not know yet whether 
\eqref{inclusion_star_Omega} is valid, and consequently we cannot be sure if $\xi_1\equiv S$ in a full neighborhood of $\overline\C\setminus \Omega$.

Hence, for a moment we reserve the notation $S$ for the meromorphic continuation of $\xi_1$ as the Schwarz function as above, and we use $\xi_1$ to denote the solution to 
\eqref{spectral_curve} that is meromorphic in $\overline\C\setminus \Sigma_*$. Once we obtain \eqref{inclusion_star_Omega}, we can then conclude that these two meromorphic 
continuations coincide in a full neighborhood of $\overline\C\setminus \Omega$.

\begin{proof}[Proof of Theorem~\ref{theorem_limiting_support_zeros}] 
We now prove all the conclusions of Theorem~\ref{theorem_limiting_support_zeros}, except for the inclusion \eqref{inclusion_star_Omega}.

From the construction of the set $\Sigma_*$ as above, it follows that $\xi_1$ is meromorphic on $\overline\C\setminus \Sigma_*$, with pole only at $\infty$. Hence $\xi_1$ is 
analytic in $\C\setminus \Sigma_*$.

Suppose first that $(t_0,t_1)\in \mathcal F$. From the sheet structure constructed in Section \ref{section_sheet_structure_1} (and the cut $\Sigma_*$ as in \eqref{special_cut}), 
it holds true
\begin{align*}
\xi_{1-}(s)=\xi_{3+}(s), & \quad s\in \Sigma_{*,0}, \\
\xi_{1-}(s)=\xi_{2+}(s), & \quad s\in \Sigma_{*,1}\cup\Sigma_{*,2}.
\end{align*}
Since $\Sigma_{*,0}$, $\Sigma_{*,1}$ and $\Sigma_{*,2}$ are obtained as projections of arcs of the trajectories $\gamma_1(z_0^{(2)})$, $\gamma_{0}(z_1^{(3)})$ and 
$\gamma_{2}(z_2^{(3)})$, respectively (see \eqref{projection_property_sigma_1_three_cut}--\eqref{projection_property_sigma_0_three_cut}), we get 
\begin{align*}
(\xi_{1+}(s)-\xi_{1-}(s))ds & =(\xi_{1+}(s)-\xi_{3+}(s))ds\in i\R, \quad s\in \Sigma_{0,*}, \\
(\xi_{1+}(s)-\xi_{1-}(s))ds & =(\xi_{1+}(s)-\xi_{2+}(s))ds\in i\R, \quad s\in \Sigma_{1,*}\cup\Sigma_{2,*}.
\end{align*}
This proves \eqref{s_property} in the three-cut case. Similarly, we use \eqref{projection_property_sigma_one_cut} to conclude that \eqref{s_property} is satisfied 
in  the one-cut case as well. Furthermore, from its construction it immediately follows that $\Sigma_*$ satisfies the properties claimed in {\rm (i)} (in the three-cut case) and 
{\rm (ii)} (in the one-cut case).

To prove that $\mu_*$ defined in \eqref{limiting_measure_zeros} is a probability measure, first note that \eqref{s_property} automatically implies that it is a real measure. In 
the three-cut case, its density does not vanish on the arcs $\Sigma_{*,0}$, $\Sigma_{*,1}$ and $\Sigma_{*,2}$ of its support, so this measure has constant sign on each of these 
arcs. For $t_1=0$, we know that this density is positive in each of these arcs \cite{bleher_kuijlaars_normal_matrix_model}, so by continuity we conclude that this measure is 
positive for any pair $(t_0,t_1)\in \mathcal F_1$ as well. The total mass of $\mu_*$ is given by
$$
\frac{1}{2\pi i t_0} \int_{\Sigma_*} (\xi_{1-}(z)-\xi_{1+}(z))dz= \frac{1}{2\pi i t_0} \oint_\gamma \xi_1(s) ds,
$$
where $\gamma$ is a contour positively oriented and encircling $\Sigma_*$. We deform $\gamma$ to infinity and use the expansion \eqref{asymptotics_xi} to get 
$$
\frac{1}{2\pi i t_0} \int_{\Sigma_*} (\xi_{1-}(z)-\xi_{1+}(z))dz=-\frac{1}{t_0} \res(\xi_1,\infty)=1,
$$
so we conclude that in the three-cut case $\mu_*$ is indeed a probability measure.

To conclude that $\mu_*$ is positive in the one-cut case as well, note that the density of $\mu_*$ is continuous and does not vanish in the interval $\Sigma_*$, hence $\mu_*$ has 
constant sign along this interval. Calculations very similar as above show that the total mass of $\mu_*$ is $1$, in particular $\mu_*$ has to be positive, and we then conclude 
that $\mu_*$ is a probability measure, as we want.
\end{proof}

We remind that to conclude the proof of Theorem~\ref{theorem_limiting_support_zeros} we still have to verify the inclusion \eqref{inclusion_star_Omega}. To do so, we first 
have to prove some auxiliary results, which are inspired from a similar analysis carried out by Balogh, Bertola, Lee and McLaughlin \cite{balogh_bertola_et_al_op_planar_measure}.

Recall that the principal value of the Cauchy transform $C^\lambda$ of a finite and compactly supported measure $\lambda$ is defined by
$$
C^{\lambda}(z)=\lim_{\varepsilon\to 0}\int_{|s-z|\geq \varepsilon} \frac{d\lambda(s)}{s-z},\quad z\in \C.
$$
$C^{\lambda}$ is analytic on the open sets of $\C\setminus\supp\lambda$ and satisfies the identity
\begin{equation}\label{identity_cauchy_transform_potential}
2\frac{\partial U^{\lambda}}{\partial z}(z)=C^{\lambda}(z), \quad z\in \C\setminus\supp\lambda,
\end{equation}
where $U^\lambda$ is the potential of $\lambda$ as in \eqref{definition_log_potential}.

\begin{lem}
For $(t_0,t_1)\in \mathcal F$ and $V$ the cubic polynomial \eqref{generic_cubic_potential}, the Cauchy transform of the measure $\mu_0$ in \eqref{definition_planar_measure} 
satisfies
\begin{equation}\label{cauchy_transform_planar_measure}
C^{\mu_0}(z)=
\begin{dcases}
-\frac{1}{t_0}(\overline z -V'(z)), & z\in \overline{\Omega},\\
-\frac{1}{t_0}(S(z) -V'(z)), & z\in \C\setminus{\Omega}.
\end{dcases}
\end{equation}
whereas the Cauchy transform of $\mu_*$ in \eqref{limiting_measure_zeros} satisfies
\begin{equation}\label{cauchy_transform_mother_body}
C^{\mu_*}(z)= -\frac{1}{t_0}(\xi_1(z) -V'(z)),\quad z\in \C\setminus \Sigma_*.
\end{equation}
\end{lem}
\begin{proof}
 Suppose $z\in \C\setminus \overline\Omega$. Using Green's Theorem and the definition of the Schwarz function $S$, we can write
 \begin{align*}
  C^{\mu_0}(z) & = \frac{1}{\pi t_0}\iint_\Omega \frac{dA(s)}{s-z} \\
	       & = \frac{1}{2\pi i t_0}\int_{\partial \Omega} \frac{S(s)}{s-z}ds,
 \end{align*}
where $\partial \Omega$ is oriented counterclockwise. The Schwarz function $S$ is analytic on $\C\setminus \Omega$, so we can deform $\partial \Omega$ to $\infty$ in order to 
conclude
\begin{align*}
C^{\mu_0}(z) & = -\frac{1}{t_0}\res\left(\frac{S(s)}{s-z},s=z\right) - \frac{1}{t_0}\res\left( \frac{S(s)}{s-z},s=\infty \right)\\
	     & = -\frac{1}{t_0}\left(S(z)-V'(z)\right),
\end{align*}
where we computed the residue at $\infty$ using the fact that $S$ is an analytic continuation of the function germ $\xi_1$ in \eqref{asymptotics_xi}. This is enough to prove the 
second equality in \eqref{cauchy_transform_planar_measure} for $z\in \C\setminus{\overline \Omega}$. Using continuity, this extends to $\partial \Omega$ as well.

Suppose now that $z\in \Omega$. Using the Cauchy-Green formula,
\begin{align*}
 \overline z & =\frac{1}{2\pi i}\int_{\partial \Omega}\frac{\overline s}{s-z}ds-\frac{1}{\pi}\iint_{\Omega}\frac{dA(s)}{s-z} \\
	     & = \frac{1}{2\pi i}\int_{\partial \Omega}\frac{S(s)}{s-z}ds-t_0 C^{\mu_0}(s).
\end{align*}
Proceeding as before, we can compute the contour integral on the right-hand side above, concluding that
$$
\overline z = V'(z)-t_0C^{\mu_0}(s),
$$
which is equivalent to the first equation in \eqref{cauchy_transform_planar_measure}.

Finally, \eqref{cauchy_transform_mother_body} follows again by a residue calculation, having in mind the identity
$$
C^{\mu_*}(z)=\frac{1}{2\pi i t_0}\int_{\Sigma_*} \frac{\xi_{1-}(s)-\xi_{1+}(s)}{s-z}ds=\frac{1}{2\pi i t_0}\int_{\gamma} \frac{\xi_1(s)}{s-z}ds,
$$
where we recall that $\Sigma_*$ is oriented outwards, and $\gamma$ is any positively oriented contour encircling $\Sigma_*$ and for which $z$ is on the exterior region of $\gamma$.

\end{proof}

Recalling \eqref{generic_cubic_potential}, denote by
$$
\mathcal U_0(z)=2U^{\mu_0}(z)+\mathcal V(z), \quad z\in \C,
$$
the total potential of $\mu_0$, and by
\begin{equation}\label{definition_total_potential_mother_body}
\mathcal U_*(z)=2U^{\mu_*}(z)+\mathcal V(z),\quad z\in \C,
\end{equation}
the total potential of $\mu_*$.

\begin{proof}[Proof of Theorem~\ref{thm_equilibrium_measure}]
Since the partial derivatives of $U^{\mu_0}$ are continuous and $C^{\mu_0}$ is absolutely convergent in $\C$, the identity \eqref{identity_cauchy_transform_potential} for $\mu_0$ 
is actually valid on $\Omega$ as well, and the first equation in \eqref{cauchy_transform_planar_measure} then says
$$
\frac{\partial \mathcal U_0}{\partial z}(z)=0,\quad z\in \Omega.
$$
This identity suffices to conclude that $\mathcal U_0$ is constant, say $l$, on $\overline\Omega$, which is the same as \eqref{variational_equality_planar_measure}.

The second equality in \eqref{cauchy_transform_planar_measure} provides an harmonic extension $\widetilde{\mathcal U}_0$ of $\mathcal U_0$ to a neighborhood of $\partial \Omega$. 
More concretely, there is a neighborhood $G$ of $\C\setminus \Omega$ and a constant $c$ such that the harmonic function
$$
\widetilde{\mathcal U}_0(z)=-\frac{1}{t_0}\re \int^z S(s)ds+\frac{|z|^2}{t_0}+c,\quad z\in G,
$$
coincides with $\mathcal U_0$ in $\C\setminus \Omega$. Note also that the primitive above does not depend on the path of integration chosen within $G$, because the residue of $S$ 
at $\infty$ is real.

We will prove that 
\begin{equation}\label{equation_local_minimum}
\widetilde{\mathcal U}_{0}(z_0+\varepsilon \eta)=l+\frac{2\varepsilon^2}{t_0^2}+\mathcal O(\varepsilon^3), \quad \mbox{ as } \varepsilon\to 0,
\end{equation}
which is enough to conclude \eqref{variational_inequality_planar_measure}.

On $\partial \Omega$, we know that
$$
\frac{\partial \widetilde{\mathcal {U}}_0}{\partial z}(z)=-\frac{1}{t_0}(S(z)-\overline z)=0,
$$
so the (real) gradient of $\widetilde{\mathcal U}_{0}$ vanishes on $\partial \Omega$. Furthermore, the Laplacian of $\widetilde{\mathcal U}_{0}$ is
$$
\Delta\widetilde{\mathcal U}_{0}(z)=4\frac{\partial^2 \widetilde{\mathcal U}_{0}}{\partial z\partial \overline z}(z)=\frac{4}{t_0},\quad z\in \partial \Omega,
$$
and the determinant of the Hessian $H(\widetilde{\mathcal U}_{0})$ of $\widetilde{\mathcal U}_{0}$ is
$$
\det H(\widetilde{\mathcal U}_{0})=4\left( \left(\frac{\partial^2 \widetilde{\mathcal U}_{0}}{\partial z\partial \overline z}\right)^2-\frac{\partial^2 
\widetilde{\mathcal U}_{0}}{\partial z^2}\frac{\partial^2 \widetilde{\mathcal U}_{0}}{\partial \overline{z}^2} \right)=\frac{4}{t_0}(1-|S'(z)|^2),
$$
and the latter vanishes on $\partial \Omega$ because $|S'(z)|=1$ there. 

Hence we conclude that the eigenvalues of $H(\widetilde{\mathcal U}_{0})$ are $4/t_0$ and $0$. Taking into account that $\widetilde{\mathcal U}_{0}$ is constant along 
$\partial\Omega$, we see that the tangent vector to $\partial \Omega$ is an eigenvector for the eigenvalue $0$, and consequently $\eta$ is an eigenvector associated to $4/t_0$. 
This is enough to get
\eqref{equation_local_minimum}.
\end{proof}

\begin{proof}[Proof of Theorem~\ref{theorem_density_eigenvalues}]
The conditions \eqref{variational_equality_planar_measure}--\eqref{variational_inequality_planar_measure} are enough to conclude that $\mu_0$ is the equilibrium measure of $D$ in 
the external field $\mathcal V$ \cite[Theorem~I.3.3]{Saff_book}. The result is then an immediate consequence of 
\cite[Theorem~4.1]{elbau_felder_density_eigenvalues_normal_matrices}.
\end{proof}

To prove the inclusion \eqref{inclusion_star_Omega}, we need two more lemmas.

\begin{lem}\label{lemma_minimum_total_potential}
Suppose $(t_0,t_1)\in \mathcal F$. The total potential $\mathcal U_*$ does not have a local minimum on $\Sigma_*$.
\end{lem}
\begin{proof}
Suppose first $z\notin \{z_*,z_0,z_1,z_2\}$. With respect to the outward orientation of $\Sigma_*$, denote by $n_+$ and $n_-$ the normal 
vectors of $\Sigma_*$ at $z$ pointing to the positive and negative sides of $\Sigma_*$, respectively. Combining \eqref{identity_cauchy_transform_potential} and 
\eqref{cauchy_transform_mother_body}, it follows that 
\begin{align*}
\frac{\partial \mathcal U_*}{\partial n_\pm}(z) & = 2\re \left( n_\pm \frac{\partial \mathcal{U}_{*\pm}}{\partial z}(z) \right) \\
						& = \pm 2\re \left( n_+ \left(C^{\mu_*}_{\pm}(z)+\frac{1}{t_0}(\overline z - V'(z)) \right)\right) \\
						& = \pm \frac{2}{t_0}\re \Big(n_+\left(\overline z -\xi_{1\pm}(z)\right)\Big).
\end{align*}
The vector tangent to $\Sigma_*$ along its positive direction is $\tau_+=-in_+$, and the equalities above then imply
\begin{align}
\frac{\partial \mathcal U_*}{\partial n_+}(z)+\frac{\partial \mathcal U_*}{\partial n_-}(z) & = \frac{2}{t_0}\re \Big(n_+\big(\xi_{1-}(z)-\xi_{1+}(z)\big)\Big) \nonumber \\
    & =- \frac{2}{t_0}\im \Big(\tau_+ \big(\xi_{1-}(z)-\xi_{1+}(z)\big)\Big).\label{normal_derivative_equals_imaginary}
\end{align}
Since the measure $\mu_*$ is positive, we learn from \eqref{limiting_measure_zeros} and the last equality that
$$
\frac{\partial \mathcal U_*}{\partial n_+}(z)+\frac{\partial \mathcal U_*}{\partial n_-}(z)<0,
$$
so at least one of these directional derivatives is negative, thus $z$ cannot be a local minimum. 

Suppose now $z\in \{z_0,z_1,z_2\}$, say $z=z_j$. Using \eqref{cauchy_transform_mother_body} and the relation \eqref{identity_cauchy_transform_potential},
\begin{equation}\label{potential_mother_body_near_branch_points}
U^{\mu_*}(u)=-\frac{1}{t_0}\re\int_{z_j}^u \xi_1(s)ds + \frac{1}{t_0}\re(V(u)-V(z_j)),\quad u\in G\setminus \Sigma_*,
\end{equation}
where $G$ is a sufficiently small neighborhood of $z$ and the path of integration lies in $G\setminus \Sigma_*$. On $G$ the function $\xi_1$ admits an 
expansion of the form
$$
\xi_1(u)=\xi_1(z_j)+\Boh((u-z_j)^{1/2}),\quad u\in G\setminus \Sigma_*.
$$
Using this expansion in \eqref{potential_mother_body_near_branch_points}, we then get
\begin{align*}
\mathcal U_*(u) & = -\frac{2}{t_0}\re \big(\xi_1(z_j)u\big) +\frac{|u|^2}{t_0}+c+\Boh((u-z_j)^{3/2}), \\
		& = -\frac{2}{t_0}\re u \re\xi_1(z_j)+\frac{2}{t_0}\im u\im \xi_1(z_j) \\ & 
		\qquad +\frac{(\re u)^2+(\im u)^2}{t_0} +c+\Boh((u-z_j)^{3/2}), \quad u\in G\setminus \Sigma_*
\end{align*}
for some constant $c$. From \eqref{branched_solution_non_vanishing} we know that $\xi_1(z_j)=h(w_j^{-1})\neq 0$, so the leading contribution in the formula above is linear in $\re 
u$ and $\im u$, and consequently $\mathcal U_*$ cannot have a local minimum at $z_j$.

It remains to verify that $z=z_*$ cannot be a local minimum in the three-cut case $(t_0,t_1)\in \mathcal F_1$. Consider the angle $\theta$ between $\Sigma_{*,2}$ and 
$\Sigma_{*,0}$ on the upper half plane, so that $\theta\leq \pi$. Assume for the moment that $\theta>\pi/2$. In this case, for $\epsilon>0$ sufficiently small, the point 
$z_* +\varepsilon i$ is in between $\Sigma_{*,0}$ and $\Sigma_{*,2}$, and using continuity we can conclude
\begin{equation}\label{identity_intersection_point_positive_side}
\mathcal U_*(z_*+\varepsilon i)=\lim_{\delta\to 0^+}\mathcal U_*(z_*+\varepsilon i+\delta)=U_*(z_*)+\varepsilon \lim_{\delta\to 0^+} \frac{\partial \mathcal 
U_*}{\partial n_+}(z_*+\delta)+\mathcal O(\varepsilon^2),
\end{equation}
where $n_+=i$ is the normal to $\Sigma_{*,0}$, and the error term is uniform in $\delta$. Similarly,
\begin{equation}\label{identity_intersection_point_negative_side}
\mathcal U_*(z_*-\varepsilon i)=\lim_{\delta\to 0^+}\mathcal U_*(z_*-\varepsilon i+\delta)=U_*(z_*)+\varepsilon \lim_{\delta\to 0^+} \frac{\partial \mathcal 
U_*}{\partial n_-}(z_*+\delta)+\mathcal O(\varepsilon^2),
\end{equation}

Proceeding as in \eqref{normal_derivative_equals_imaginary} (and having in mind that in the present case $\tau_+=1$),
\begin{align*}
\lim_{\delta\to 0^+} \left(\frac{\partial \mathcal U_*}{\partial n_+}(z_*+\delta)+\frac{\partial \mathcal U_*}{\partial n_-}(z_*+\delta)\right)
& = -\frac{2}{t_0} \lim_{\delta\to 0^+}\im \big(\xi_{1-}(z_*+\delta)- \xi_{1+}(z_*+\delta) \big).
\end{align*}
As before, we use that the density of $\mu_*$ is positive on $\Sigma_{*,0}$ and does not vanish on $z_*$ to conclude that the limit above is strictly negative. Taking into account 
\eqref{identity_intersection_point_positive_side}--\eqref{identity_intersection_point_negative_side}, this is enough to conclude that $z_*$ is not a point of minimum.

It only remains to prove the inequality $\theta>\pi/2$. When $t_1=0$, then $\theta=3\pi/2$ and we are done. Since $\theta$ varies continuously with $t_1$, it is enough to prove 
that $\theta\neq \pi/2$.

To the contrary, suppose $\theta=\pi/2$. In this case, the vector tangent to $\Sigma_{*,2}$ converges to $i$ as $z\to z_*$ along $\Sigma_{*,2}$, so that from \eqref{s_property} we 
learn
$$
(\xi_{1+}(z_*)-\xi_{1-}(z_*))\in\R,
$$
where the $\pm$ boundary values are with respect to $\Sigma_{*,2}$. But in this case $\xi_{1+}(z_*)\in \R$ as well, so we conclude that $\xi_{1-}(z_*)\in \R$, and consequently all 
the solutions to \eqref{spectral_curve} are real for $z=z_*$. But this cannot occur, because $z_*<z_0$ and we know from Theorem~\ref{theorem_discriminant_spectral_curve} that the 
discriminant of \eqref{spectral_curve} is negative on the interval $(-\infty,z_0)$.
\end{proof}

\begin{lem}\label{lemma_minimum_total_potential_2}
 The total potential $\mathcal U_*$ does not have a local minimum on $\Omega$.
\end{lem}
\begin{proof}
 From Lemma~\ref{lemma_minimum_total_potential}, we know that $\mathcal U_*$ does not attain a minimum on $\Sigma_*$, so if $p\in \Omega$ is a point of minimum of $\mathcal U_*$, 
then the gradient of $\mathcal U_*$ should vanish at $p$. This means that
\begin{equation}\label{zero_gradient_total_potential_mother_body}
0=\frac{\partial \mathcal U_*}{\partial z}(p)=C^{\mu_*}(p)+\frac{1}{t_0}(\overline p - V'(p)) = \frac{1}{t_0}(\overline p-\xi_1(p)),
\end{equation}
where for the first equality we used the definition of $\mathcal U_*$ given in \eqref{definition_total_potential_mother_body} together with the identity 
\eqref{identity_cauchy_transform_potential}, and for the last equality we used \eqref{cauchy_transform_mother_body}. That is, we conclude that the point $p$ where $\mathcal U_*$ 
attains its minimum should satisfy $\xi_1(p)=\overline p$. Let $w$ be such that
\begin{equation}\label{schwarz_points_1}
h(w)=p,\qquad h\left(\frac{1}{w}\right)=\xi_1(p)=\overline p.
\end{equation}
We know from Theorem~\ref{theorem_rational_parametrization_polynomial_curve} that $h$ maps $\partial \D$ to $\partial \Omega$. Since $p\in \Omega$, we learn from this that 
$|w|\neq 
1$, so the point $\tilde w :=1/\overline w$ is different from $w$. Furthermore, because the rational function $h$ has real coefficients,
\begin{equation}\label{schwarz_points_2}
h\left( \tilde w \right)=\overline{h( w^{-1})}=\overline{\xi_1(p)}=p=h(w),
\end{equation}
and also
\begin{equation}\label{schwarz_points_3}
h\left( \frac{1}{\tilde w} \right)=\overline{h(w)}=\overline{p}=\xi_1(p)=h\left(\frac{1}{w}\right).
\end{equation}

Since $\tilde w\neq w$, the pairs $(h(w^{-1}),h(w))$ and $(h({\tilde{w}}^{-1}),h(\tilde w))$ represent distinct points on $\mathcal R$. In virtue 
of the equalities \eqref{schwarz_points_1}--\eqref{schwarz_points_3}, we conclude that $p$ is a zero of the discriminant of \eqref{spectral_curve}. But $p\notin \Sigma_*$, so in 
particular it cannot be a branch point. Thus $p=\hat z_j$ for some $j$, and from Theorem~\ref{theorem_singular_points_spectral_curve} we learn that $p$ cannot be on 
$\Omega$. Hence \eqref{zero_gradient_total_potential_mother_body} cannot hold on $\Omega\setminus\Sigma_*$, and the proof is complete.
\end{proof}

\begin{proof}[Proof of \eqref{inclusion_star_Omega}]
Fix $t_0$. When $t_1=0$, the set $\Sigma_*$ is explicitly given by \eqref{star_t=0}, and \eqref{inclusion_star_Omega} is valid in this case. Suppose now that 
\eqref{inclusion_star_Omega} is not valid for some value of $t_1$. Continuity arguments show that for the smallest of such value, say $\tilde t_1$, it holds true
\begin{equation}\label{partial_inclusion_star_omega}
\Sigma_*\subset \overline{\Omega}\qquad \mbox{and} \qquad \partial \Omega\cap \Sigma_*\neq \emptyset.
\end{equation}

Recalling that $\xi_1$ and $S$ coincide in a neighborhood of $\infty$, the inclusion in \eqref{partial_inclusion_star_omega} implies $\xi_1\equiv S$ on the whole set 
$\C\setminus \Omega$. Thus, still for the given pair $(t_0,\tilde t_1)$, the second identity in \eqref{cauchy_transform_planar_measure} combined with 
\eqref{cauchy_transform_mother_body} then says that
\begin{equation}\label{equality_cauchy_transforms}
C^{\mu_0}(z)=C^{\mu_*}(z),\quad z\in \C\setminus \Omega.
\end{equation}
Taking into account \eqref{identity_cauchy_transform_potential}, we further get
\begin{equation*}
\mathcal U_*(z)-\mathcal U_0(z)=c,\quad z\in \C\setminus \Omega,
\end{equation*}
for some constant $c$. This constant is equal to
\begin{align*}
c & =\lim_{z\to\infty}\left(U^{\mu_*}(z)-U^{\mu_0}(z)\right)\\
  & = \lim_{z\to\infty}\left(\iint \log\left|1-\frac{s}{z}\right|d\mu_0(s)-\int \log\left|1-\frac{s}{z}\right|d\mu_*(s)\right)=0.
\end{align*}
In particular, it follows from Theorem~\ref{thm_equilibrium_measure} that
\begin{equation}\label{minimum_potential_mother_body_1}
\mathcal U_*(z)>l,\quad \mbox{ for }z\in \C\setminus \overline{\Omega} \mbox{ sufficiently close to } \partial\Omega,
\end{equation}
and
\begin{equation}\label{minimum_potential_mother_body_2}
\mathcal U_*(z)=l,\quad z\in \partial \Omega.
\end{equation}
The function $\mathcal U_*$ is continuous on $\overline\Omega$, so it has a minimum in this set. Using Lemma~\ref{lemma_minimum_total_potential_2} and 
\eqref{minimum_potential_mother_body_2}, we thus conclude
\begin{equation}\label{minimum_potential_mother_body_3}
\mathcal U_*(z)>l,\quad z\in \Omega.
\end{equation}

From \eqref{partial_inclusion_star_omega}, we know that there exists $p\in \partial \Omega\cap \Sigma_*$. From 
\eqref{minimum_potential_mother_body_1}--\eqref{minimum_potential_mother_body_3} we learn that $p$ is a local minimum of $\mathcal U_*$, but this is in 
contradiction with Lemma~\ref{lemma_minimum_total_potential}.

Hence \eqref{inclusion_star_Omega} always holds true, and the proof of Theorem~\ref{theorem_limiting_support_zeros} is complete.
\end{proof}

\begin{proof}[Proof of Theorem~\ref{theorem_balayage_relation}] Having in mind \eqref{inclusion_star_Omega}, we now know that $\xi_1\equiv S$ on 
$\C\setminus \Omega$. Proceeding as in \eqref{equality_cauchy_transforms}--\eqref{minimum_potential_mother_body_3}, we conclude
\begin{align}
 & \mathcal U_*(z) = \mathcal U_0(z), \quad z\in \C\setminus\Omega,  \label{identity_total_potentials}\\
 & \mathcal U_*(z) > l, \qquad z\in \Omega, \label{inequality_total_potentials}
\end{align}
where $l$ is the constant from Theorem~\ref{thm_equilibrium_measure}. From the definition of $\mathcal U_0$ and $\mathcal U_*$ it follows that \eqref{identity_total_potentials} is 
equivalent to \eqref{mother_body_equation}. Furthermore, by Theorem~\ref{thm_equilibrium_measure} we know that $\mathcal U_0\equiv l$ on $\Omega$, and 
\eqref{inequality_total_potentials} then becomes \eqref{mother_body_inequality}.
\end{proof}

\section{Riemann-Hilbert analysis in the three-cut case}\label{section_riemann_hilbert_analysis_precritical}

In this section we perform the Riemann-Hilbert/Steepest Descent analysis for the multiple orthogonal polynomial $P_{n,n}$ (given in Definition~\ref{definition_mop}) in the 
three-cut case $(t_0,t_1)\in\mathcal F_1$.

The analysis is similar to the one presented by Bleher and Kuijlaars in \cite{bleher_kuijlaars_normal_matrix_model} for $t_1=0$. The main difference 
here is that we construct the $g$-functions only with the help of the $\xi$-functions, without relying on any vector equilibrium problem. Although our $g$-functions 
could also be given in terms of the solution to a vector equilibrium problem, we do not elaborate in this direction.

For $j=0,1,2$, we extend $\Sigma_{*,j}$ given in Section~\ref{section_proof_s_property} to a contour $\Sigma_{j}$ in the following way. For $j=0$, simply set $\Sigma_j=[z_*,\hat 
z_0]$. For $j=2$, we take $\Sigma_2\setminus\Sigma_{*,2}$ to be an arc from $z_2$ to $\hat z_2$ contained in the projection of the strip domain $\mathcal S_3$ of the 
quadratic differential $\varpi$, that is,
\begin{equation}\label{projection_property_extension_sigma}
 \Sigma_2\setminus(\Sigma_{*,2}\cup \{\hat z_2\})\subset \pi(\mathcal S_3),
\end{equation}
where we recall that the strip domain $\mathcal S_3$ is labeled as in Figure \ref{figure_planar_graph}. Lastly we extend $\Sigma_{*,1}$ imposing the conjugation property
\begin{equation}\label{symmetry_relations_extension_sigma_*}
\Sigma_1\setminus \Sigma_{*,1}=(\Sigma_2\setminus \Sigma_{*,2})^*.
\end{equation}
In this way,
\begin{equation}\label{projection_property_extension_sigma_2}
 \Sigma_1\setminus(\Sigma_{*,1}\cup \{\hat z_1\})\subset \pi(\mathcal S_8),
\end{equation}
where the strip domain $\mathcal S_8$ is given as in Figure \ref{figure_planar_graph}. We refer the reader to Figure \ref{figure_projection_property}.

\begin{figure}[t]
\centering
 \begin{overpic}[scale=1]
 {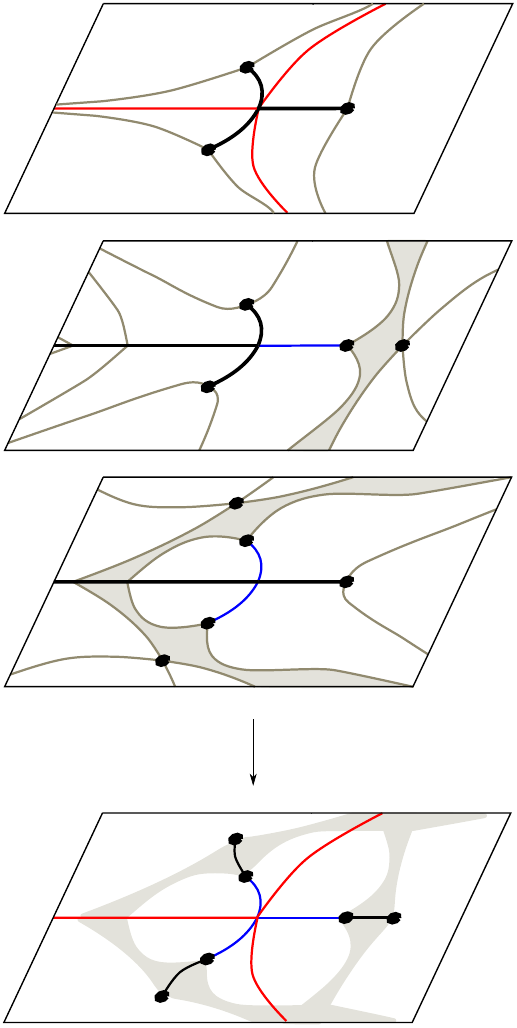}
 \put(50,10){$\C$}
 \put(26,27){$\pi$}
 \put(50,43){$\mathcal R_3$}
 \put(50,66){$\mathcal R_2$}
 \put(50,89){$\mathcal R_1$}
 \put(28,11){$\scriptstyle{\Sigma_{*}}$}
 \put(32,16.5){$\scriptstyle{L}$}
 \end{overpic}
 \caption{Illustration of properties \eqref{projection_property_extension_sigma} and \eqref{definition_L_2}. The trajectories highlighted on $\mathcal R$ in blue and red are 
projected onto $\C$ to the contours $\Sigma_*$, $L$, also respectively represented in blue and red. $\Sigma_*$ consists of the union of the three trajectories projected from 
$\mathcal R_2$, $\mathcal R_3$, whereas $L$ consists of the pieces projected from $\mathcal R_1$. In addition, we extend $\Sigma_*$ to $\Sigma$ by constraining $\Sigma\setminus 
\Sigma_*$ to lie within the shaded region on $\C$, which consists of the projections of the gray strip domains on $\mathcal R$. The arcs of $\Sigma\setminus \Sigma_*$ are depicted 
in black.}\label{figure_projection_property}
\end{figure}

The arc $\Sigma_j$ is oriented from $z_*$ to $\hat z$. We then set 
$$
\Sigma=\Sigma_0\cup\Sigma_1\cup\Sigma_2,
$$
and consider the multiple orthogonal polynomial $P_{n,n}$ in Definition~\ref{definition_mop} with this choice of $\Sigma$.

\subsection{Multiple orthogonality in terms of Airy functions}

The multiple orthogonality conditions \eqref{multiple_orthogonality_conditions} can be stated in terms of solutions to the Airy equation $y''=zy$. The Airy function $\ai$ is the 
special solution to the Airy equation determined by the asymptotic behavior
\begin{align*}
 \ai(z)  & = \frac{z^{-1/4}}{2\sqrt \pi} e^{-\frac{2}{3}z^{3/2}}(1+\Boh(z^{-3/2})), \\
 \ai'(z) & = -\frac{z^{1/4}}{2\sqrt \pi} e^{-\frac{2}{3}z^{3/2}}(1+\Boh(z^{-3/2})),
\end{align*}
valid when $z\to \infty, -\pi< \arg z<\pi$. It admits the integral representation
$$
\ai(z)=\frac{1}{2\pi i}\int_{\Gamma_0} e^{\frac{1}{3}s^3-zs}ds,\quad z\in\C,
$$
where $\Gamma_0$ is a contour as in Section~\ref{section_associated_MOP}, see \eqref{definition_infinities} {\it et seq}. For $\Gamma_1$, $\Gamma_2$, $\omega$ also given as in 
Section~\ref{section_associated_MOP}, set 
\begin{equation}\label{definition_functions_airy_1}
y_j(z)=\omega^{j}\ai(\omega^{j}z)=\frac{1}{2\pi i}\int_{\Gamma_j}e^{\frac{1}{3}s^3-zs},\quad z\in\C,\quad j=0,1,2.
\end{equation}
Note that the integral representations above, in combination with contour deformation, immediately imply that
\begin{equation}\label{equation_sum_airy_functions}
y_0(z)+y_1(z)+y_2(z)=\ai(z)+\omega \ai(\omega z) + \omega^2\ai(\omega^2 z)=0.
\end{equation}

Using the functions $y_0,y_1$ and $y_2$ in \eqref{definition_functions_airy_1} and setting $c_n=(n/t_0)^{2/3}$, the conditions \eqref{multiple_orthogonality_conditions} can 
be rewritten as
\begin{equation}\label{multiple_orthogonality_conditions_airy}
\begin{aligned}
\sum_{l=0}^2\int_{ \Sigma_l} P_{n,n}(z)z^k e^{\frac{n}{t_0}V(z)}y_l(c_n(z-t_1))dz=0, & \quad k=0,\hdots, \left\lceil \frac{n}{2}  \right\rceil -1, \\
\sum_{l=0}^2\int_{ \Sigma_l} P_{n,n}(z)z^k e^{\frac{n}{t_0}V(z)}y'_l(c_n(z-t_1))dz=0, & \quad k=0,\hdots, \left\lfloor \frac{n}{2} \right\rfloor -1.
\end{aligned}
\end{equation}

\subsection{The Riemann-Hilbert problem $Y$}\label{section_rhp_Y}

For simplicity of presentation, we assume for now on $n$ even. The case $n$ odd can be treated similarly, after appropriate (and non essential) modifications.

Consider the following Riemann-Hilbert problem for $Y$.

\begin{itemize}
 \item $Y:\C\setminus \Sigma\to \C^{3\times 3}$ is analytic;
 
 \item $Y_+(z)=Y_-(z)J_Y(z)$, $z\in \Sigma$, where
 $$
 J_Y(z)=
 \begin{pmatrix}
  1 & e^{\frac{n}{t_0}V(z)}y_j(c_n(z-t_1)) & e^{\frac{n}{t_1}V(z)}y'_j(c_n(z-t_1)) \\
  0 & 1 & 0 \\
  0 & 0 & 1
 \end{pmatrix}, \quad z\in \Sigma_j;
 $$
 
 \item $Y(z)=(I+\Boh(z^{-1}))
 \begin{pmatrix}
  z^n & 0 & 0 \\
  0 & z^{-\frac{n}{2}} & 0 \\
  0 & 0 & z^{-\frac{n}{2}}
 \end{pmatrix},\quad z\to \infty;
 $
 
 \item $Y(z)=
 \begin{pmatrix}
  \Boh(1) & \Boh(\log(z-\hat z_j)) & \Boh(\log(z-\hat z_j)) \\
  \Boh(1) & \Boh(\log(z-\hat z_j)) & \Boh(\log(z-\hat z_j))\\
  \Boh(1) & \Boh(\log(z-\hat z_j)) & \Boh(\log(z-\hat z_j))
 \end{pmatrix}, \quad z\to \hat z_j, \quad j=0,1,2.
$

\item $Y(z)$ is remains bounded when $z\to z_*$.
\end{itemize}

Here and in what follows, $\Sigma$ is oriented outwards, that is, towards $\infty$, see Figure~\ref{figure_contour_y}.

It turns out that the polynomial $P_{n,n}$ uniquely exists if, and only if, the Riemann Hilbert problem above has a solution. In such a case, the polynomial $P_{n,n}$ is recovered 
through
\begin{equation}\label{P_nn_Y}
P_{n,n}=Y_{1,1}.
\end{equation} 
As one of the consequences of our analysis, we get that for sufficiently large $n$, the Riemann-Hilbert problem above has a solution, and thus the polynomial $P_{n,n}$ exists for 
$n$ sufficiently large.

\begin{figure}[t]
 \begin{overpic}[scale=1]
  {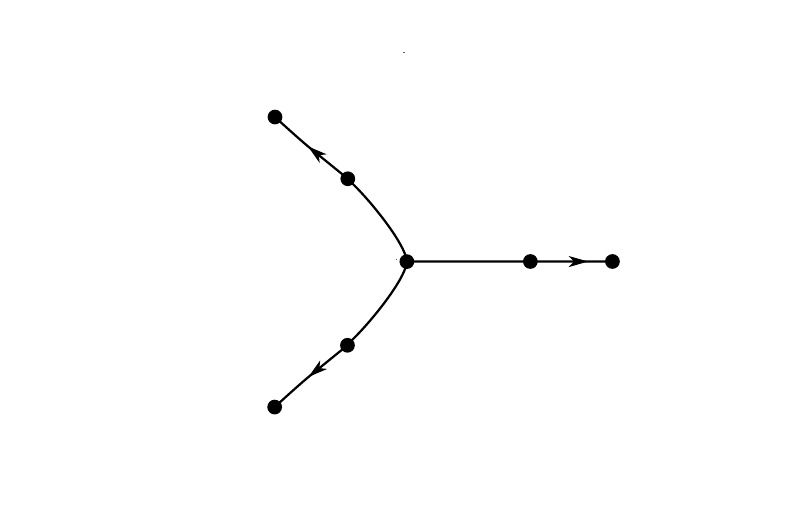}
  \put(79.5,32){$\hat z_0$}
  \put(31,11){$\hat z_1$}
  \put(30,52){$\hat z_2$}
  \put(65,35){$z_0$}
  \put(45,19){$z_1$}
  \put(45,45){$z_2$}
  \put(46,32){$z_*$}
  \put(55,34){$\Sigma$}
 \end{overpic}
 \caption{Contour $\Sigma$ for the RHP for $Y$.}\label{figure_contour_y}
\end{figure}

\subsection{First transformation: $Y\mapsto X$}\label{section_transformation_Y_X}

The first transformation, essentially the same as in \cite{bleher_kuijlaars_normal_matrix_model}, has the goal of reducing the jump matrix of $Y$ to nontrivial $2\times 2$ 
blocks. Define
\begin{align*}
 y_3(z) & = 2\pi i(\omega^2 \ai(\omega z)-\omega \ai(\omega^2z)),\\
 y_4(z) & = \omega y_3(\omega z), \\
 y_5(z) & = \omega^2 y_3(\omega^2 z).
\end{align*}
Using the identities $1-\omega=i\sqrt{3}\omega^2$ and $1-\omega^2=-i\sqrt{3}\omega$, in combination with \eqref{equation_sum_airy_functions}, we also get 
\begin{align*}
y_3( z) & = \frac{2\pi}{\sqrt{3}}((1-\omega)\ai (\omega z)+(1-w^2)\omega^2 \ai(\omega^2 z)) \\
	& = \frac{2\pi}{\sqrt{3}}(-\omega\ai(\omega z)-\omega^2\ai(\omega^2 z) +\ai(\omega z)+\ai(\omega^2 z))\\
	& = \frac{2\pi}{\sqrt{3}} (\ai(z) + \ai(\omega z) + \ai(\omega^2 z)) 
\end{align*}
This last identity immediately implies that $y_3(\omega z)=y_3(z)$, and consequently $y_4(z)=\omega y_3(z)$ and $y_5(z)=\omega^2y_3(z)$.

Additionally, set 
\begin{equation}\label{definition_L_0_projection}
L_0=(-\infty,z_*]=\pi(\mathcal R_1\cap \gamma_1(z_0^{(2)}))
\end{equation}
and
\begin{equation}\label{definition_L_2}
 L_2=\pi(\mathcal R_1\cap \gamma_0(z_1^{(3)})), \quad L_1=L_2^*=\pi(\mathcal R_1\cap \gamma_2(z_2^{(3)})),
\end{equation}
where $\pi:\mathcal R\to \overline \C$ is the canonical projection as before, see Figure~\ref{figure_projection_property}.
$L_1$ and $L_2$ are unbounded analytic arcs with a common finite endpoint $z_*\in\R$, and they extend to $\infty$ along the angles $\frac{\pi}{3}$ and 
$-\frac{\pi}{3}$, respectively. We set the orientation on $L_j$ to be outwards, that is, from $z_*$ to $\infty$.

\begin{remark}
 For ease of presentation, in what follows we assume that $L_j\cap \Sigma_{*}=\{z_*\}$. If $L_j$ intersects $\Sigma_{*}$ in more points, it is still possible to carry out the 
Riemann-Hilbert analysis along the same lines as we present here, with appropriate and non essential modifications. Numerical experiments indicate that the condition $L_j\cap 
\Sigma_{*}=\{z_*\}$ always holds true anyway.
\end{remark}

The union $L=L_0\cup L_1 \cup L_2$ divides the plane into three domains, henceforth denoted $G_0$, $G_1$, $G_2$, where $G_j$ is uniquely defined through the condition that 
$\hat z_j$ is contained in $G_j$, see Figure~\ref{figure_contour_x}.

\begin{figure}[t]
 \begin{overpic}[scale=1]
  {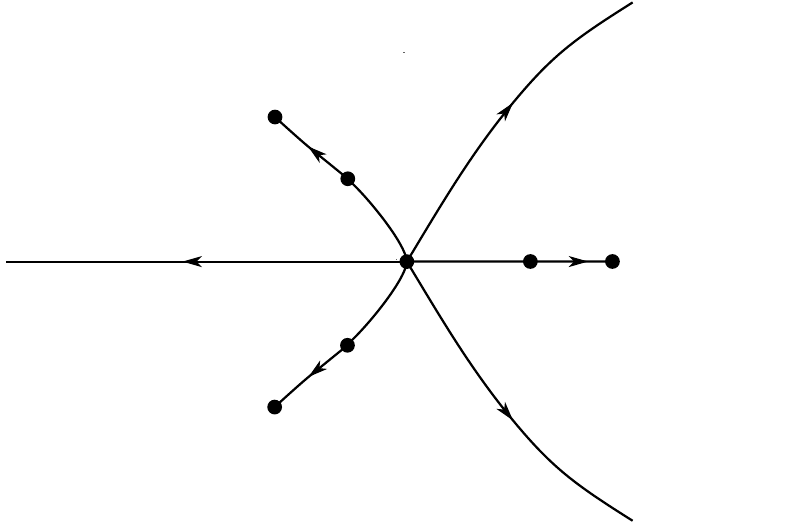}
  \put(79.5,32){$\hat z_0$}
  \put(31,11){$\hat z_1$}
  \put(30,52){$\hat z_2$}
  \put(65,35){$z_0$}
  \put(45,19){$z_1$}
  \put(45,45){$z_2$}
  \put(80,50){$G_0$}
  \put(10,50){$G_2$}
  \put(10,10){$G_1$}
  \put(10,34.5){$L_0$}
  \put(69,10){$L_2$}
  \put(65,58.5){$L_1$}
 \end{overpic}
 \caption{Contour $\Sigma\cup L$, $L=L_0\cup L_1\cup L_2$, for the RHP's for $X$ and $T$, and the sectors $G_0$, $G_1$, $G_2$.}\label{figure_contour_x}
\end{figure}

We make the transformation
\begin{multline}\label{transformation_X}
X(z)=
\begin{pmatrix}
 1 & 0 & 0 \\
 0 & \frac{c_n^{-1/4}}{\sqrt{2\pi}} & 0 \\
 0 & 0 & i\frac{c_n^{1/4}}{\sqrt{2\pi}}
\end{pmatrix}
Y(z)
\begin{pmatrix}
 1 & 0 & 0 \\
 0 &  y'_{j+3}(c_n(z-t_1)) & -y'_{j}(c_n(z-t_1)) \\
 0 & -y_{j+3}(c_n(z-t_1))  & y_{j}(c_n(z-t_1)
\end{pmatrix} \\
\times
\begin{pmatrix}
 1 & 0 & 0 \\
 0 & 1 & 0 \\
 0 & 0 & -2\pi i
\end{pmatrix}, \quad z\in G_j, \quad j=0,1,2.
\end{multline}

It follows as in \cite[pp.~1297--1301]{bleher_kuijlaars_normal_matrix_model} that $X$ satisfies the following RHP 
\begin{itemize}
 \item $X:\C\setminus (\Sigma\cup L)\to \C^{3\times 3}$ is analytic;
 
 \item $X_+(z)=X_-(z)J_X(z)$, $z\in \Sigma\cup L$, where $J_X$ is given by
 $$
 J_X(z)=
 \begin{pmatrix}
  1 & 0 & 0 \\
  0 & \omega^2 & 1 \\
  0 & 0 & \omega
 \end{pmatrix},
 \quad z\in L;
\qquad 
 J_X(z)=
 \begin{pmatrix}
  1 & e^{\frac{n}{t_0}V(z)} & 0 \\
  0 & 1 & 0 \\
  0 & 0 & 1
 \end{pmatrix},
 \quad z\in \Sigma.
$$

 \item $X$ has the same endpoint conditions as $Y$ at $z_*,\hat z_j$, $j=0,1,2$.

\item $
X(z)=(I+\Boh(z^{-1}))A(z)Q(z), \quad z\to \infty,
$

\noindent where
\begin{equation}\label{definition_asymptotic_matrix_A}
A(z)=
\begin{pmatrix}
 1 & 0 & 0 \\
 0 & z^{1/4} & 0 \\
 0 & 0 & z^{-1/4}
\end{pmatrix}
\times
\begin{cases}
\begin{pmatrix}
 1 & 0 & 0 \\
 0 & \frac{1}{\sqrt 2} & -\frac{i}{\sqrt 2} \\
 0 & -\frac{i}{\sqrt 2} & \frac{1}{\sqrt 2}
\end{pmatrix},
& z\in G_0, \\
\begin{pmatrix}
 1 & 0 & 0 \\
 0 & -\frac{i}{\sqrt 2} & -\frac{1}{\sqrt 2} \\
 0 & \frac{1}{\sqrt 2} & \frac{i}{\sqrt 2}
\end{pmatrix},
& z\in G_1, \\
\begin{pmatrix}
 1 & 0 & 0 \\
 0 & \frac{i}{\sqrt 2} & \frac{1}{\sqrt 2} \\
 0 & -\frac{1}{\sqrt 2} & -\frac{i}{\sqrt 2}
\end{pmatrix},
& z\in G_2. \\
\end{cases}
\end{equation}
\noindent and
$$
Q(z)=
\begin{cases}
\begin{pmatrix}
 z^n & 0 & 0 \\
 0 & z^{-\frac{n}{2}}e^{\frac{2n}{3t_0}(z-t_1)^{3/2}} & 0 \\
 0 & 0 & z^{-\frac{n}{2}}e^{-\frac{2n}{3t_0}(z-t_1)^{3/2}}
\end{pmatrix} 
& 
z\in G_0, \\
\begin{pmatrix}
 z^n & 0 & 0 \\
 0 & z^{-\frac{n}{2}}e^{-\frac{2n}{3t_0}(z-t_1)^{3/2}} & 0 \\
 0 & 0 & z^{-\frac{n}{2}}e^{\frac{2n}{3t_0}(z-t_1)^{3/2}}
\end{pmatrix} 
& 
z\in G_1\cup G_2,
\end{cases}
$$
\end{itemize}

\subsection{Second transformation: $X\mapsto T$}\label{section_transformation_X_T}

The second transformation has the goal of removing the $n$-dependence from the asymptotics of $X$. 

Introduce the $g$-functions
\begin{equation}\label{g_functions}
\begin{aligned}
g_1(z) & = \int_{z_0}^z \xi_1(s)ds + c_1 ,    & \quad z\in \C\setminus (\Sigma_*\cup L_0),  \\
g_2(z) & = \int_{z_*}^z \xi_2(s)ds+c_2,  & \quad z\in \C\setminus (\Sigma_{*,1}\cup \Sigma_{*,2}\cup L_0), \\
g_3(z) & = \int_{z_0}^z \xi_3(s)ds+c_3, & \quad  z\in \C\setminus(\Sigma_{*,0}\cup L_0),
\end{aligned}
\end{equation}
where for $g_2$ the path of integration starts along $(z_*,+\infty)$ and
\begin{equation}\label{definition_constants_c_j_1}
c_2=-\pi it_0 -\int_{\Sigma_{*,2}}(\xi_{1+}(s)-\xi_{2+}(s))ds, \quad c_1=c_3=\int_{z_*}^{z_0}\xi_{3-}(s)ds.
\end{equation}

The constant $c_2$ can be alternatively expressed as
\begin{equation}\label{definition_constants_c_j_2}
c_2=-\pi i t_0 +2\pi i t_0\mu_*(\Sigma_{*,2})=-\pi i t_0 +2\pi i t_0\mu_*(\Sigma_{*,1})=-\pi i t_0 \mu_*(\Sigma_{*,0}),
\end{equation}
where $\mu_*$ is given by Theorem~\ref{theorem_limiting_support_zeros}. In particular, $c_2$ is purely imaginary.
 
We could as well express one of the $g$-functions in terms of the other two through $\sum \xi_j=V'$, but we found more convenient to work with three $g$-functions 
instead.

The asymptotics \eqref{asymptotics_xi} can be rewritten as
\begin{align*}
\xi_1(z) & = V'(z)+\frac{t_0}{z}  + \Boh(z^{-2}), \\
\xi_2(z) & = -(z-t_1)^{1/2}-\frac{t_0}{2z} +\Boh(z^{-3/2}), \qquad z\to \infty \\
\xi_3(z) & = (z-t_1)^{1/2}-\frac{t_0}{2z} +\Boh(z^{-3/2}), 
\end{align*}
which in turn give
\begin{equation}\label{asymptotics_g_functions}
\begin{aligned}
g_1(z) & = V(z)+l_1+t_0 \log z + \Boh(z^{-1}), \\
g_2(z) & = -\frac{2}{3}(z-t_1)^{3/2}+l_2-\frac{t_0}{2}\log z +\Boh(z^{-1/2}), \qquad z\to \infty \\
g_3(z) & = \frac{2}{3}(z-t_1)^{3/2}+l_3-\frac{t_0}{2}\log z +\Boh(z^{-1/2}), 
\end{aligned}
\end{equation}
for some constants $l_1,l_2,l_3$.

\begin{lem}\label{lemma_constants_c_j}
 For the constants $l_2,l_3,c_2$ and $c_3$ as in \eqref{g_functions}--\eqref{asymptotics_g_functions}, it is valid
 \begin{equation}\label{equality_variational_constants}
 l_3=l_2,\qquad c_2=i\im c_3.
 \end{equation}
\end{lem}
\begin{proof}
 From the asymptotics \eqref{asymptotics_g_functions},
\begin{align*}
 g_{3+}(z)-g_{2-}(z) & = \frac{2}{3}((z-t_1)^{3/2}_+ + (z-t_1)^{3/2}_-)+l_3-l_2 \\
		     & \quad -\frac{t_0}{2} (\log z)_+ +\frac{t_0}{2}(\log z)_- +\Boh(z^{-1/2}) \\
		     & = l_3-l_2+ \pi i t_0  +\Boh(z^{-1/2}), \quad z\in L_0
\end{align*}

On another hand, from the definition of $g_2$ and $g_3$,
\begin{align*}
 g_{3+}(z)-g_{2-}(z) & = c_3-c_2+\int_{z_0}^{z_*}\xi_{3-}(s)ds-\int_{z_*}^{z_2}(\xi_{1+}(s)-\xi_{2+}(s))ds, \\ 
		     & = -c_2 +\int_{z_*}^{z_2}(\xi_{1-}(s)-\xi_{1+}(s))ds = \pi i t_0,\quad z\in L_0,
\end{align*}
where for the second equality we used the jump condition $\xi_{1\pm}=\xi_{2\mp}$ in $\Sigma_{*,2}$, which follows from the sheet structure constructed in 
Section~\ref{section_sheet_structure_1}.

To get the second equality in \eqref{equality_variational_constants}, just note that
\begin{align*}
 2i\im c_3=c_3+\overline{c_3} & = \int_{z_*}^{z_0}(\xi_{3-}(s)-\overline{\xi_{3-}(s)})ds \\
			      & = \int_{z_*}^{z_0}(\xi_{1+}(s)-\xi_{1-}(s))ds \\
			      & = -2\pi i t_0\mu_{*}(\Sigma_{*,0})=2c_2,
\end{align*}
where for the second equality we used the jump equalities in \eqref{equalities_inequalities_real_line_xi_functions_precritical_3} and for the last equality we used 
\eqref{definition_constants_c_j_2}.

\end{proof}

The second transformation is given by
\begin{multline}\label{transformation_T}
T(z)=
\begin{pmatrix}
  e^{\frac{n}{t_0}l_1} & 0 & 0 \\
  0 & e^{\frac{n}{t_0}l_2} & 0 \\
  0 & 0 & e^{\frac{n}{t_0}l_2}
 \end{pmatrix}
X(z)
\\
\times
\begin{cases}
 \begin{pmatrix}
  e^{\frac{n}{t_0}(V(z)-g_1(z))} & 0 & 0 \\
  0 & e^{-\frac{n}{t_0}g_3(z)} & 0 \\
  0 & 0 & e^{-\frac{n}{t_0}g_2(z)}
 \end{pmatrix},
 &
 z\in G_0 \\
 \begin{pmatrix}
  e^{\frac{n}{t_0}(V(z)-g_1(z))} & 0 & 0 \\
  0 & e^{-\frac{n}{t_0}g_2(z)} & 0 \\
  0 & 0 & e^{-\frac{n}{t_0}g_3(z)}
 \end{pmatrix},
 &
 z\in G_1\cup G_2  
\end{cases}
\end{multline}

It follows that $T$ is the solution to the following RHP,

\begin{itemize}
 \item $T:\C\setminus(\Sigma\cup L)\to \C^{3\times 3}$ is analytic;

 \item $T_+(z)=T_-(z)J_T(z)$, $z\in \Sigma\cup L$, where $J_T$ is given by
\begin{multline*}
J_T(z)= \\
 \begin{cases}
 \begin{pmatrix}
  e^{\frac{n}{t_0}(g_{1-}(z)-g_{1+}(z))} & e^{\frac{n}{t_0}(g_{1-}(z)-g_{3+}(z))} & 0 \\
  0 & e^{\frac{n}{t_0}(g_{3-}(z)-g_{3+}(z))} & 0 \\
  0 & 0 & e^{\frac{n}{t_0}(g_{2-}(z)-g_{2+}(z))}
 \end{pmatrix},
 & z\in \Sigma_0, \\
  \begin{pmatrix}
  e^{\frac{n}{t_0}(g_{1-}(z)-g_{1+}(z))} & e^{\frac{n}{t_0}(g_{1-}(z)-g_{2+}(z))} & 0 \\
  0 & e^{\frac{n}{t_0}(g_{2-}(z)-g_{2+}(z))} & 0 \\
  0 & 0 & e^{\frac{n}{t_0}(g_{3-}(z)-g_{3+}(z))}
 \end{pmatrix},
 & z\in \Sigma_1\cup \Sigma_2, \\
  \begin{pmatrix}
  e^{\frac{n}{t_0}(g_{1-}(z)-g_{1+}(z))} & 0 & 0 \\
  0 &  \omega^2e^{\frac{n}{t_0}(g_{2-}(z)-g_{2+}(z))} & e^{\frac{n}{t_0}(g_{2-}(z)-g_{3+}(z))} \\
  0 & 0 & \omega e^{\frac{n}{t_0}(g_{3-}(z)-g_{3+}(z))}
 \end{pmatrix},
 & z\in  L_0, \\
 \begin{pmatrix}
  1 & 0 & 0 \\
  0 & \omega^2 e^{-\frac{n}{t_0}(g_{2}(z)-g_{3}(z))} & 1 \\
  0 & 0 & \omega e^{-\frac{n}{t_0}(g_{3}(z)-g_{2}(z))}
 \end{pmatrix},
 & z\in L_1, \\
 \begin{pmatrix}
  1 & 0 & 0 \\
  0 & \omega^2 e^{-\frac{n}{t_0}(g_{3}(z)-g_{2}(z))} & 1 \\
  0 & 0 & \omega e^{-\frac{n}{t_0}(g_{2}(z)-g_{3}(z))}
 \end{pmatrix},
 & z\in L_2;
 \end{cases}
\end{multline*}
 
 \item $T$ satisfies the same endpoint conditions as $X$ when $z\to z_*,\hat z_j$, $j=0,1,2$.

 \item $T(z)=(I+\Boh(z^{-1}))A(z)$, as $z\to \infty$.
\end{itemize}
 
Our next goal is to simplify the jump matrix $J_T$ further. For this purpose it is convenient to introduce the functions
\begin{equation}\label{definition_Phi}
\begin{aligned}
\Phi_0(z) & =\frac{1}{t_0} \int_{z_0}^z (\xi_1(s)-\xi_3(s))ds, \\
\Phi_1(z) & =\frac{1}{t_0} \int_{z_1}^z (\xi_1(s)-\xi_2(s))ds, \qquad z\in \C\setminus(\Sigma_*\cup L_0),\\
\Phi_2(z) & =\frac{1}{t_0} \int_{z_2}^z (\xi_1(s)-\xi_2(s))ds. \\
\end{aligned}
\end{equation}
and also
\begin{equation}\label{definition_Psi}
\begin{aligned}
 \Psi_j(z)  & =\frac{1}{t_0} \int_{z_*}^{z}(\xi_2(s)-\xi_3(s))ds, \quad z\in \C\setminus (\Sigma_*\cup L_0),\quad j=0,1, \\
 \Psi_2(z)  & =\frac{1}{t_0} \int_{z_*}^{z}(\xi_3(s)-\xi_2(s))ds, \quad z\in \C\setminus (\Sigma_*\cup L_0).
\end{aligned}
\end{equation}
where the paths of integration for $\Psi_0$, $\Psi_1$ and $\Psi_2$ are as follows.
\begin{enumerate}
 \item[$\bullet$] If $\im z\geq 0$, then the path of integration for $\Phi_0$ emanates from $z_*$ in the sector between $L_0$ and $\Sigma_2$ on the upper half plane. If $\im z<0$, 
then the path of integration for $\Phi_0$ emanates from $z_*$ in the sector between $L_0$ and $\Sigma_1$ on the lower half plane.

\item[$\bullet$] For $j=1,2$, the path of integration for $\Psi_j$ emanates from $z_*$ in the sector between $L_j$ and $\Sigma_{j+1}$ on $G_{j+1}$.
\end{enumerate}

The main properties of the functions $\Phi_j$ and $\Psi_j$ are collected in the next proposition.

\begin{prop}\label{proposition_properties_g_phi_psi}
The functions $g_j, \Phi_j, \Psi_j$, $j=0,1,2$, satisfy
\begin{enumerate}[label=(\Alph*)]
 
 \item \label{sigma_0} For $z\in \Sigma_{*,0}$,
 \begin{center}
 \begin{enumerate}[label=(\roman*)]
  \item $\displaystyle{ g_{1+}(z)-g_{1-}(z)=t_0 \Phi_{0+}(z),}$
  \item $\displaystyle{ g_{2+}(z)-g_{2-}(z)=0, }$
  \item $\displaystyle{ g_{3+}(z)-g_{3-}(z)= t_0 \Phi_{0-}(z), }$
  \item $\displaystyle{ g_{3+}(z)-g_{1-}(z)=0, }$
  \item \label{real_part_phi_0} $\displaystyle{ \re \Phi_{0\pm}(z)=0, }$
  \item $\displaystyle{ \Phi_{0+}(z)+\Phi_{0-}(z)=0, }$
  \item $\displaystyle{ \Psi_{1+}(z)+\Psi_{2-}(z)=\Phi_{0+}(z)+\frac{2c_2}{t_0}. }$
 \end{enumerate}
\end{center}
 
\item \label{sigma_1_2} For $j=1,2$ and $z\in \Sigma_{*,j}$,
 \begin{center}
 \begin{enumerate}[label=(\roman*)]
  \item $\displaystyle{ g_{1+}(z)-g_{1-}(z)=t_0 \Phi_{j+}(z), }$
  \item $\displaystyle{ g_{2+}(z)-g_{2-}(z)=t_0 \Phi_{j-}(z), }$
  \item $\displaystyle{ g_{3+}(z)-g_{3-}(z)=0, }$
  \item $\displaystyle{ g_{2+}(z)-g_{1-}(z)=(-1)^{j+1}\pi i t_0, }$
  \item \label{real_part_phi_1_2} $\displaystyle{ \re \Phi_{j\pm}=0, }$
  \item $\displaystyle{ \Phi_{j+}(z)+\Phi_{j-}(z)=0, }$
  \item $\displaystyle{ \Phi_{j-1,-}(z)+(-1)^{j+1}\Psi_{j+1,+}(z)=\Phi_{j+}(z)-\frac{c_2}{t_0}-\pi i. } $
 \end{enumerate}
\end{center} 

\item For $z\in \Sigma_0\setminus \Sigma_{j,0}$,
$$
g_3(z)-g_1(z)=-t_0\Phi_0(z).
$$

\item \label{jump_relation_extension_sigma} For $j=1,2$, and $z\in \Sigma_j\setminus \Sigma_{*,j}$,
$$
g_2(z)-g_1(z)=-t_0\Phi_j(z)+(-1)^{j+1}\pi it_0.
$$

\item \label{sign_phi_extension_sigma_star}$\re \Phi_j$ is negative on $\Sigma_j\setminus \Sigma_{*,j}$, $j=0,1,2$.

\item \label{monotonicity_phi} For $j=0,1,2$, the functions $\im\Phi_{j+}$ and $\im\Phi_{j-}$ are decreasing and increasing, respectively, along the orientation of 
$\Sigma_{*.j}$.

 \item \label{properties_L_0}For $z\in L_0$,
 \begin{center}
 \begin{enumerate}[label=(\roman*)]
 \item \label{jumps_psi_0} $\displaystyle{ \Psi_{0+}(z) + \Psi_{0-}(z)= 0 , }$
 \item \label{real_part_psi_0} $\displaystyle{ \re \Psi_{0\pm}(z)=0, }$
 \item $\displaystyle{ g_{1+}(z) - g_{1-}(z) = -2\pi i t_0, }$
 \item $\displaystyle{ g_{2+}(z) - g_{2-}(z) = t_0\Psi_{0+}(z) + 2c_2 + 2\pi it_0, }$
 \item $\displaystyle{ g_{3+}(z) - g_{3-}(z) = t_0\Psi_{0-}(z) - 2c_2, }$
 \item $\displaystyle{ g_{3+}(z) - g_{2-}(z) = \pi i t_0. }$
 \item $\displaystyle{ \Phi_{2-}(z)-\Phi_{1+}(z) = \Psi_{0+}(z)+\frac{2c_2}{t_0} + 2\pi i. }$
 \end{enumerate}
\end{center}

 \item \label{properties_L_1_2} For $j=1,2$, and $z\in L_j$,
 \begin{center}
 \begin{enumerate}[label=(\roman*)]
 \item \label{jumps_psi_1_2} $\displaystyle{ \Psi_{j+}(z)=\Psi_{j-}(z), }$
 \item \label{real_part_psi_1_2} $\displaystyle{ \re \Psi_j(z)=0, }$
 \item $\displaystyle{ g_3(z)-g_2(z)=(-1)^jt_0\Psi_j(z)+(-1)^{j+1}c_2 }$
 \item $\displaystyle{ \Phi_{j-1}(z)-\Phi_{j+1}(z)=\Psi_{j}(z)-\frac{c_2}{t_0}+\pi i. }$
 \end{enumerate}
\end{center}

 \item \label{monotonicity_psi_j} For $j=0,1,2$, the function $\im\Psi_{j+}$ is decreasing along the orientation of $L_j$.
 
\end{enumerate}
\end{prop}

\begin{proof}
To see that the conditions \ref{sigma_0}--\ref{real_part_phi_0} and \ref{sigma_1_2}--\ref{real_part_phi_1_2} are true, note that from the sheet structure constructed in 
Section~\ref{section_sheet_structure_1}, it follows that
$$
\Phi_{j\pm}(z)=\int_{z_j}^z(\xi_{1\pm}(s)-\xi_{1\mp}(s))ds,\quad z\in \Sigma_*.
$$
Since $\Sigma_*$ satisfies \eqref{s_property}, the right-hand side above has to be purely imaginary, leading to \ref{sigma_0}--\ref{real_part_phi_0} and 
\ref{sigma_1_2}--\ref{real_part_phi_1_2}.

Similarly, to get \ref{properties_L_0}--\ref{real_part_psi_0} and \ref{properties_L_1_2}--\ref{real_part_psi_1_2} we note that 
\eqref{definition_L_0_projection}--\eqref{definition_L_2} say that $L_j$ coincides with the projection of a trajectory on the first sheet of the quadratic differential 
\eqref{quadratic_differential}. From the definition of $\varpi$ in \eqref{definition_function_germ_Q}--\eqref{quadratic_differential} and of its trajectories 
\eqref{definition_trajectory_integral}, and taking also into account that $z_*\in L$, we thus get
$$
\re \int^z_{z_*} (\xi_{2+}(s)-\xi_{3+}(s))ds=0,\quad z\in L,
$$
which is enough to conclude \ref{properties_L_0}--\ref{real_part_psi_0} and \ref{properties_L_1_2}--\ref{real_part_psi_1_2}.

The remaining conditions claimed in \ref{sigma_0}--\ref{jump_relation_extension_sigma} and \ref{properties_L_0}--\ref{properties_L_1_2} follow in a straightforward manner, once 
one 
has in mind the sheet structure for the spectral curve $\mathcal R$ (see 
Sections~\ref{section_sheet_structure_1} and the beginning of Section~\ref{section_proof_s_property}), and also 
equations~\eqref{definition_constants_c_j_1}--\eqref{definition_constants_c_j_2} and \eqref{equality_variational_constants}. We skip the details for these computations.

Recalling that $\xi_1<\xi_3$ on $(z_0,\hat z_0)$, see \eqref{equalities_inequalities_real_line_xi_functions_precritical_2}, we get \ref{sign_phi_extension_sigma_star} for $j=0$. 

Note that the analytic continuation of the function $\Phi_2$ (that we keep denoting by $\Phi_2$) coincides (up to a multiplicative real factor) with the primitive $\Upsilon$ in 
\eqref{primitive_uniformization} on the strip domain $\mathcal U=\mathcal S_3$. In particular, $\Phi_2$ maps $\mathcal S_3$ to a vertical strip on $\C$ whose one of the boundary 
components is the imaginary axis, and consequently the sign of $\re\Phi_2$ is constant on $\mathcal S_3$. The trajectory $\gamma_0(z_2^{(3)})$ is contained in one of the 
components of $\partial \mathcal S_3$, is mapped by $\Phi_2$ to the imaginary axis and extends to $\infty$ along the angle $\theta_{0}^{(3)}=\pi/6$. Using the expansion 
\eqref{asymptotics_xi} it follows that
$$
\Phi_2(z)=\frac{z^3}{3t_0}+\Boh(z) = i\frac{|z|^3}{3t_0}+\Boh(z),\quad \mbox{ as } z\to\infty \mbox{ along } \gamma_0(z_2^{(3)}),
$$
thus $\gamma_0(z_2^{(3)})$ is mapped by $\Phi_2$ to $i\R_+$. Hence, because $\Phi_2$ is conformal, we conclude that the left-hand side of $\gamma_0(z_2^{(3)})$ in the 
orientation from $z_2^{(3)}$ to $\infty$ (that is, $\mathcal S_3$) is mapped by $\Phi_2$ to the left-hand side of $i\R_+$ in the natural orientation. Since the sign of $\re 
\Phi_2$ is constant on $\mathcal S_3$, this is enough to conclude that 
$\re\Phi_2(z)<0$ on $\mathcal S_3$. In virtue of \eqref{projection_property_extension_sigma_2}, we conclude \ref{sign_phi_extension_sigma_star} for $j=2$. Finally, 
\ref{sign_phi_extension_sigma_star} for $j=1$ follows in a similar fashion, or also noticing the symmetry relations $\overline{\Phi_2(z)}=\Phi_1(\overline z)$ and 
\eqref{symmetry_relations_extension_sigma_*}.

For \ref{monotonicity_phi}, denote by $\gamma_z$ the sub arc of $\Sigma_{*,j}$ from $z$ to $z_j$. It follows from \eqref{definition_Phi} that we can write
$$
\Phi_{j\pm}(z)=\frac{1}{t_0}\int_{z_j}^z(\xi_{1\pm}(s)-\xi_{1\mp}(s))ds=\pm 2\pi i\mu_*(\gamma_z).
$$
The measure $\mu_*$ is positive, so $\mu_*(\gamma_z)$ is decreasing along $\Sigma_{*,j}$, and \ref{monotonicity_phi} follows from the equation above.

We now proceed to prove \ref{monotonicity_psi_j}. We already know that
\begin{equation}\label{imaginary_part_psi_j_L_j}
\im\Psi_{j+}(z)=\Psi_{j+}(z),\quad z\in L_j,
\end{equation}
as it follows from \ref{properties_L_0}--\ref{real_part_psi_0} and \ref{properties_L_1_2}--\ref{real_part_psi_1_2}. The derivative of $\Psi_{j+}$ is, up to a sign, equal to 
$\xi_{2+}-\xi_{3+}$, which does not vanish along $L_j$. Combining with \eqref{imaginary_part_psi_j_L_j} we learn that $\im\Psi_{j+}$ is monotone along $L_j$. Taking into account 
that $L_j$ extends to $\infty$ with angle $\pi -2j\pi/3$ and using the asymptotics \eqref{asymptotics_xi}, we learn
\begin{align*}
\Psi_{j+}(z) & = \epsilon\frac{4}{3}z^{3/2}+\Boh(z^{1/2}) \\
	     & = -\frac{4i}{3}|z|^{3/2}+\Boh(z^{1/2}),\quad \mbox{ as } z\to \infty \mbox{ along } L_j,
\end{align*}
where $\epsilon=-1$ for $j=0,1$ and $\epsilon=1$ for $j=2$. In virtue of the previous comments, this is enough to conclude \ref{properties_L_0}--\ref{real_part_psi_0} and 
\ref{properties_L_1_2}--\ref{real_part_psi_1_2}.
\end{proof}

The jump matrix $J_T$ can then be rewritten as
$$
 J_T(z)=
 \begin{cases}
 \begin{pmatrix}
  e^{-n\Phi_{j+}(z)} & 1 & 0 \\
  0 & e^{-n\Phi_{j-}(z)} & 0 \\
  0 & 0 & 1
 \end{pmatrix},
 & z\in \Sigma_{*,j}, \ j=0,1,2, \\
 \begin{pmatrix}
  1 & e^{n\Phi_j(z)} & 0 \\
  0 & 1 & 0 \\
  0 & 0 & 1
 \end{pmatrix},
 & z\in \Sigma_j\setminus \Sigma_{*,j},\ j=0,1,2, \\
 \begin{pmatrix}
 1 & 0 & 0 \\
 0 & \omega^2 e^{-n(\Psi_{0+}(z)+\alpha_0)} & 1 \\
 0 & 0 & \omega e^{-n(\Psi_{0-}(z)-\alpha_0)}
 \end{pmatrix},
 & z\in L_0, \\
 \begin{pmatrix}
  1 & 0 & 0 \\
  0 & \omega^2 e^{-n(\Psi_{j}(z)-\alpha_j)} & 1 \\
  0 & 0 & \omega e^{n(\Psi_{j}(z)-\alpha_j)}
 \end{pmatrix},
 & z\in L_j, \quad j=1,2. \\
 \end{cases}
 $$
 where we set 
 $$
 \alpha_0=\frac{2c_2}{t_0},\quad \alpha_1=\alpha_2=\frac{c_2}{t_0}.
 $$
 
\begin{remark}
 Our second transformation $X\mapsto T$ should be compared with the sequence of transformations $X\mapsto V\mapsto U\mapsto T$ in \cite{bleher_kuijlaars_normal_matrix_model}.
\end{remark}
 
\subsection{Opening of lenses: $T\mapsto S$}

Based on the properties of the functions $\Phi_j$, $\Psi_j$, we now open lenses around the contours $\Sigma_*$, $L$. We denote by
$$
\mathscr{S}=\mathscr{S}^+\cup \mathscr{S}^-
$$
the (open) lens around the contour $\Sigma_{*}$, with the convention that $\mathscr{S}^+$ and $\mathscr{S}^-$ are the parts of $\mathscr{S}$ lying on the upper 
and lower sides of $\Sigma_{*}$, respectively. Furthermore, $\partial \mathscr{S}^{+}$ and $\partial\mathscr S^-$ denote the parts of the boundary of $\mathscr S$ lying on the 
upper and lower sides of $\Sigma_{*}$, respectively. In addition, we assume that $\partial\mathscr{S}$ intersects $\Sigma_{*,j}$ only at the endpoint $z_j$ 
and, moreover, $\partial\mathscr{S}^\pm$ is chosen so that it intersects the contour $L_{j}$, $j=0,1,2$, at a point other than $z_*$. We also set $\mathscr S_j^{\pm}$ and 
$\partial \mathscr S_j^\pm$ to be the parts of $\mathscr S^\pm$ and $\partial \mathscr{S}^\pm$ on the $\pm$-sides of $\Sigma_{*,j}$, $j=0,1,2$, respectively. We refer the reader 
to Figure~\ref{figure_opening_lenses} for a depiction of this lens.

We claim that the lens $\mathscr S$ can be chosen so that
\begin{equation}\label{re_phi_j_positive}
 \re\Phi_j(z)>0,\quad z\in \partial\mathscr{S}_j^\pm\setminus \{z_j\},\quad j=0,1,2.
\end{equation}
To see this, we use Proposition~\ref{proposition_properties_g_phi_psi} \ref{monotonicity_phi} and the Cauchy-Riemann equations to get that $\re \Phi_j$ is increasing in both 
normal directions to $\Sigma_{*,j}$. Taking into account that $\re\Phi_{j\pm}=0$ along $\Sigma_{*,j}$ (see Proposition~\ref{proposition_properties_g_phi_psi} 
\ref{sigma_0}--\ref{real_part_phi_0} and \ref{sigma_1_2}--\ref{real_part_phi_1_2}) and reducing $\mathscr S$ if necessary, \eqref{re_phi_j_positive} follows.

In the very same spirit, we construct the lens
$$
\mathscr{L}=\mathscr{L}^+\cup\mathscr{L}^-
$$
around $L$, where $\mathscr L^\pm$ denotes the part of $\mathscr L$ on the $\pm$-side of $L$, and as before we use $\partial \mathscr L^\pm$ to denote the part of the boundary of 
$\mathscr L$ that is on the $\pm$-side of $L$, and $\mathscr L_j^\pm$, $\partial \mathscr L_j^\pm$ to denote the respective parts of $\mathscr L^\pm$, $\partial \mathscr L^\pm$ 
on the $\pm$-sides of $L_j$, $j=0,1,2$. We additionally assume that $\partial \mathscr L_j^\pm$ does not intersect $L_j$ and for some $\varepsilon >0$ small,
$$
\arg z \to \pi -\frac{2\pi i}{3}j \pm \varepsilon, \quad \mbox{ as } z\to \infty \mbox{ along } \partial \mathscr L_j^\pm.
$$
In an analogous manner as in \eqref{re_phi_j_positive}, we claim that $\mathscr L$ can be chosen so that 
\begin{equation}\label{re_psi_j_positive}
\begin{aligned}
\re \Psi_{0}(z)>0, & \quad z\in \partial \mathscr L_0^\pm, \\
\pm\re \Psi_j(z)>0, & \quad z\in \partial\mathscr L_j^\pm, \ j=1,2.
\end{aligned}
\end{equation}
Indeed, similarly as before we combine Cauchy-Riemann equations and Proposition~\ref{proposition_properties_g_phi_psi} \ref{monotonicity_psi_j} to conclude that $\re \Psi_{j}$ 
is increasing in the direction normal to $L_j$ (pointing towards the positive side of $L_j$). Taking into account that $\re \Psi_{j+}=0$ along $L_j$ (see 
\ref{properties_L_0}--\ref{real_part_psi_0} and \ref{properties_L_1_2}--\ref{real_part_psi_1_2}), $\Psi_j$ is continuous along $L_j$ for $j=1,2$ (see 
\ref{properties_L_1_2}--\ref{jumps_psi_1_2}) and reducing $\mathscr L$ if necessary, this leads to the inequalities in \eqref{re_psi_j_positive} on $\partial\mathscr L_j^\pm$, 
$j=1,2$, and also on $\partial \mathscr L_0^+$. Finally, the inequality for $\partial \mathscr L_0^-$ then follows from the inequality on $\partial \mathscr L_0^+$ and the jump 
condition \ref{properties_L_0}--\ref{jumps_psi_0}.

The lips $\partial \mathscr S_j^\pm$, $\partial \mathscr L_j^\pm$ of the lenses are oriented outwards, that is, towards $\infty$, see Figure~\ref{figure_opening_lenses}.

\begin{figure}[t]
 \begin{overpic}[scale=1]
  {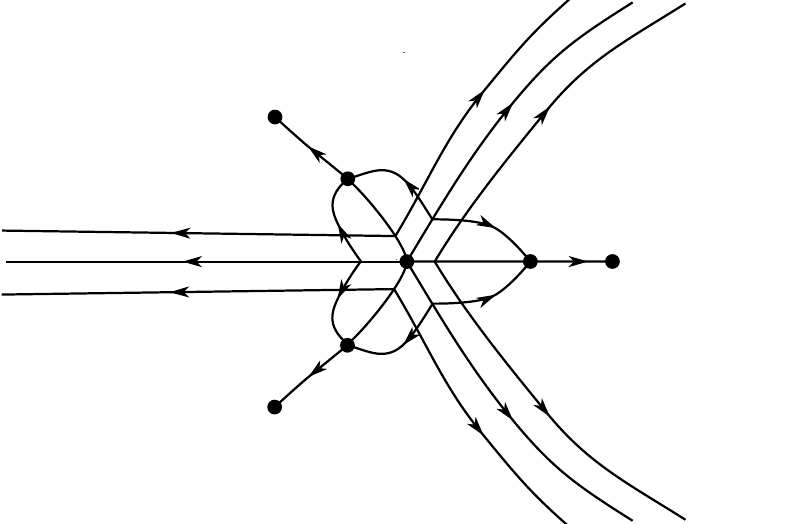}
  \put(80.5,0){$\scriptstyle \mathscr L_2^+$}
  \put(72,0){$\scriptstyle \mathscr L_2^-$}
  \put(2,34){$\scriptstyle \mathscr L_0^-$}
  \put(2,29.8){$\scriptstyle \mathscr L_0^+$}
  \put(79.5,64){$\scriptstyle \mathscr L_1^-$}
  \put(71,64){$\scriptstyle \mathscr L_1^+$}
  \put(45.5,42){$\scriptscriptstyle{\mathscr S_2^-}$}
  \put(42,37.5){$\scriptscriptstyle{ \mathscr S_2^+}$}
  \put(57.5,34.5){$\scriptscriptstyle{\mathscr S_0^+}$}
  \put(58,30){$\scriptscriptstyle{\mathscr S_0^-}$}
  \put(42,26){$\scriptscriptstyle{\mathscr S_1^-}$}
  \put(46,23){$\scriptscriptstyle{\mathscr S_1^+}$}
 \end{overpic}
 \caption{Lenses $\mathscr S$, $\mathscr L$ determining the contour $\Sigma_S$.}\label{figure_opening_lenses}
\end{figure}

We are finally ready to open lenses. Set 
\begin{equation}\label{transformation_S_1}
S(z)=T(z),\quad z\in \C\setminus (\mathscr S\cup \mathscr L),
\end{equation}
and
\begin{multline}\label{transformation_S_2}
S(z)=T(z)\times \\
\begin{cases}
 \begin{pmatrix}
  1 & 0 & 0 \\
  \mp e^{-n\Phi_j}(z) & 1 & 0 \\
  0 & 0 & 1 
 \end{pmatrix},
 & z\in \mathscr S_j^\pm\setminus \mathscr L, \\
   \begin{pmatrix}
  1 & 0 & 0 \\
  0 & 1 & 0 \\
  0 & \mp \omega^{\mp } e^{-n(\Psi_0(z)\pm\alpha_0)} & 1
 \end{pmatrix},
 & z\in \mathscr L_0^\pm\setminus \mathscr S, \\
  \begin{pmatrix}
  1 & 0 & 0 \\
  0 & 1 & 0 \\
  0 & \mp \omega^{\mp } e^{\mp n(\Psi_j(z)-\alpha_j)} & 1
 \end{pmatrix},
 & z\in \mathscr L_j^\pm\setminus \mathscr S, \ j\neq 0. \\
 \begin{pmatrix}
 1 & 0 & 0 \\
 \pm e^{-n\Phi_{\pm 1}(z)} & 1 & 0 \\
 \omega^\pm e^{-n(\Psi_0(z)+\Phi_{\pm 1}(z)\pm \alpha_0)} & \mp \omega^\mp e^{-n(\Psi_0(z)\pm \alpha_0 )} & 1 \\
\end{pmatrix},
& z\in \mathscr L_0^\pm \cap \mathscr S_{\pm 1}^\mp \\
 \begin{pmatrix}
  1 & 0 & 0 \\
  \pm e^{-n\Phi_{j\pm 1}(z)} & 1 & 0 \\
  \omega^\pm e^{\mp n(\Psi_j(z)\pm \Phi_{j\pm 1}(z)-\alpha_j)} & \mp \omega^\mp e^{\mp n (\Psi_j(z)-\alpha_j)} & 1 \\
 \end{pmatrix},
& z\in \mathscr L_j^\pm \cap \mathscr S_{j\pm 1}^\mp,\ j\neq 0 
\end{cases}
\end{multline}
where $\omega^+=\omega$, $\omega^-=\omega^{-1}=\omega^2$ and all indices are understood modulo $3$.

Denote
$$
\Gamma_S=\Sigma\cup L \cup \partial \mathscr S\cup \partial \mathscr L.
$$

The matrix $S$ satisfies the following RHP.

\begin{itemize}
 \item $S:\C\setminus \Gamma_S \to \C^{3\times 3}$ is analytic;
 
 \item $S_+(z)=S_-(z)J_S(z)$, $z\in \Gamma_S$, where the jump matrix $J_S$ is given by
 
 \begin{equation*}
 J_S(z)= 
 \begin{cases}
  \begin{pmatrix}
   0  & 1 & 0 \\
   -1 & 0 & 0 \\
   0  & 0 & 1
  \end{pmatrix},
  & z\in \Sigma_{*} \\
  \begin{pmatrix}
   1 &  0 & 0 \\
   0 &  0 & 1 \\
   0 & -1 & 0 
  \end{pmatrix},
  & z\in L \\
  \begin{pmatrix}
   1 & e^{n\Phi_j(z)} & 0 \\
   0 & 1  & 0 \\
   0 & 0 & 1
  \end{pmatrix},
  & z\in \Sigma_{j}\setminus \Sigma_{*,j},  \\
  \begin{pmatrix}
   1 & 0 & 0 \\
   e^{-n \Phi_j(z)} & 1 & 0 \\
   0 & 0 & 1
  \end{pmatrix},
  & z\in \partial\mathscr S_j\setminus \mathscr L, \\
   \begin{pmatrix}
   1 & 0 & 0 \\
   0 & 1 & 0 \\
   0 & \omega^\mp e^{-n(\Psi_0(z)\pm\alpha_0)}& 1
  \end{pmatrix},
  & z\in \partial \mathscr L_0^\pm\setminus\mathscr S,
  \\
    \begin{pmatrix}
   1 & 0 & 0 \\
   0 & 1 & 0 \\
   0 & \omega^\mp e^{\mp n(\Psi_j(z)+\alpha_j)}& 1
  \end{pmatrix},
  & z\in \partial \mathscr L_j^\pm\setminus\mathscr S, \ j\neq 0, \\
  \begin{pmatrix}
   1 & 0 & 0 \\
   e^{-n\Phi_{\pm 1}(z)} & 1 & 0 \\
   \mp e^{-n(\Psi_0(z)+\Phi_{\pm 1}\pm\alpha_0)} & 0 & 1 
  \end{pmatrix},
  & z\in \partial \mathscr{S}_{\pm 1}^{\mp}\cap \mathscr L \\
  \begin{pmatrix}
   1 & 0 & 0 \\
   e^{-n\Phi_{j\pm 1}(z)} & 1 & 0 \\
   \mp e^{\mp n(\Psi_j(z)\pm \Phi_{j\pm 1}(z)-\alpha_j)} & 0 & 1
  \end{pmatrix},
  & z\in \partial \mathscr{S}_{j\pm 1}^\mp \cap \mathscr L, \ j\neq 0, \\
  \end{cases}
  \end{equation*}
and
 \begin{multline*}
 J_S(z)= \\
 \begin{cases}
  \begin{pmatrix}
   1 & 0 & 0 \\
   0 & 1 & 0  \\
   \mp \omega^\pm e^{-n(\Psi_0(z)+\Phi_{\pm 1}(z)\pm \alpha_j)} & \omega^{\mp}e^{-n(\Psi_0(z)\pm\alpha_0)} & 1
  \end{pmatrix},
  & z\in \partial \mathscr{L}_0^\pm \cap \mathscr S \\
  \begin{pmatrix}
   1 & 0 & 0 \\
   0 & 1 & 0 \\
   \mp\omega^\pm e^{\mp n(\Psi_j(z)\pm \Phi_{j\pm 1}(z)-\alpha_j)} & \omega^\mp e^{\mp n(\Psi_j(z)-\alpha_j)} & 1
  \end{pmatrix},
  & z\in \partial \mathscr{L}_j^\pm \cap \mathscr S,\ j\neq 0,
 \end{cases}
 \end{multline*}
 
 \item $S$ has the same endpoint behavior as $X$ when $z\to z_*,\hat z_j$, $j=0,1,2$,
 
 \item $S(z)=(I+\Boh(z^{-1}))A(z)$, as $z\to \infty$.
\end{itemize}

The jump conditions above can be verified directly from the definition of $S$, once one has in hands 
Lemma~\ref{lemma_constants_c_j} and Proposition~\ref{proposition_properties_g_phi_psi}. The remaining conditions on the RHP above follow directly from the RHP for $T$. We skip the 
details.

\subsection{The global parametrix}\label{section_global_parametrix}

As we will see in a moment, the jump matrix $J_S$ converges to the identity matrix on $\Gamma_S\setminus(\Sigma_*\cup L)$. Hence, neglecting the jumps on
$\Gamma_S\setminus(\Sigma_*\cup L)$, we are led to the Riemann-Hilbert problem for $M$, commonly called the {\it global parametrix}.

\begin{itemize}
 \item $M:\C\setminus(\Sigma_*\cup L)\to \C^{3\times 3}$ is analytic;
 
 \item $M_+(z)=M_-(z)J_M(z)$, $z\in \Sigma_*\cup L$, where
 \begin{equation}\label{jump_matrix_global_parametrix}
 J_M(z)=
 \begin{cases}
  \begin{pmatrix}
   0 & 1 & 0 \\
   -1 & 0 & 0 \\
   0 & 0 & 1
  \end{pmatrix},
  & z\in \Sigma_*, \\
   \begin{pmatrix}
   1 & 0 & 0 \\
   0 & 0 & 1 \\
   0 & -1 & 0
  \end{pmatrix},
  & z\in L;
 \end{cases}
 \end{equation}

 \item $M(z)=\Boh((z- z_j)^{-1/4})$ as $z\to z_j$, $j=0,1,2$;
 
 \item $M$ remains bounded as $z\to z_*$.

 \item $M(z)=(I+\Boh(z^{-1}))A(z)$ as $z\to\infty$;
\end{itemize}

We postpone the construction of the parametrix to Section \ref{section_construction_global_parametrix}.

\subsection{The local parametrices}\label{section_local_parametrix}

Denote by
$$
D_\delta(z_j)=\{z\in\C \mid |z-z_j|<\delta \},\quad j=0,1,2,
$$
the disk of radius $\delta>0$ around $z_j$ and set 
\begin{equation}\label{D_neighborhood}
D_\delta=\bigcup_{j=0}^2 D_\delta(z_j).
\end{equation}

For $\delta>0$ sufficiently small, we search for a matrix $P$, called the {\it local parametrix}, solution to the following RHP.

\begin{itemize}
 \item $P:D_\delta\setminus \Gamma_S\to \C^{3\times 3}$ is analytic;
 
 \item $P_+(z)=P_-(z)J_S(z)$, $z\in  \Gamma_S\cap D_\delta$;
 
 \item $P(z)=(I+\Boh(n^{-1}))M(z)$, as $n\to \infty$ uniformly for $z\in\partial D_\delta$, where $M$ is the global parametrix constructed in 
Section~\ref{section_global_parametrix};
\end{itemize}

Note that the non trivial jumps for $P$ only come on the upper left $2\times 2$ corner of $J_S$, so this is essentially a $2\times 2$ RHP.

As we are in the three-cut case $(t_0,t_1)\in \mathcal F_1$, the function $\Phi_j$ has order of vanishing $3/2$ at $z_j$, that is, 
$$
\Phi_j(z)=\const \times (z-z_j)^\frac{3}{2}(1+\Boh(z-z_j)^{1/2}), \quad \mbox{ as } z\to z_j,\quad j=0,1,2,
$$
and the local parametrix can be constructed out of Airy functions, see for instance \cite{deift_book}. We skip this construction here.

\subsection{Final transformation: $S\mapsto R$}\label{section_transformation_S_R}

We arrived at the final step of our analysis. For $D_\delta$ as in \eqref{D_neighborhood} and $M$ and $P$ the global and local parametrices considered in 
Sections~\ref{section_global_parametrix} and \ref{section_local_parametrix}, respectively, we make the final transformation
\begin{equation}\label{transformation_R}
R(z)=
\begin{cases}
 S(z)M(z)^{-1}, & z\in \C \setminus (\Gamma_S\cup D_\delta),\\
 S(z)P(z)^{-1}, & z\in D_\delta \setminus \Gamma_S.
\end{cases}
\end{equation}

Since the jumps of $M$ and $P$ coincide with the jumps of $S$ on $\Sigma_{*}\cup L$ and $\Gamma_S \cap D_\delta$, respectively, it follows that $R$ satisfies a RHP on the contour
$$
\Gamma_R=\partial D_\delta \cup \left(\Gamma_S\setminus (D_\delta\cup L\cup \Sigma_*) \right),
$$
where each piece of $\partial D_\delta$ is oriented in the clockwise direction, see Figure~\ref{figure_contour_r}. More precisely,
\begin{itemize}
 \item $R:\C\setminus \Gamma_R\to \C^{3\times 3}$ is analytic;
 
 \item $R_+(z)=R_-(z)J_R(z)$, $z\in \Gamma_R$, where
 $$
 J_R(z)=
 \begin{cases}
 M(z)J_S(z)M(z)^{-1}, & z\in \Gamma_R\setminus \partial D_\delta, \\
 P(z)M(z)^{-1}, & z\in \partial D_\delta.
 \end{cases}
 $$
 
 \item $R(z)=I+\Boh(z^{-1}), \quad z\to \infty$.
\end{itemize}

\begin{figure}[t]
 \includegraphics[scale=1]{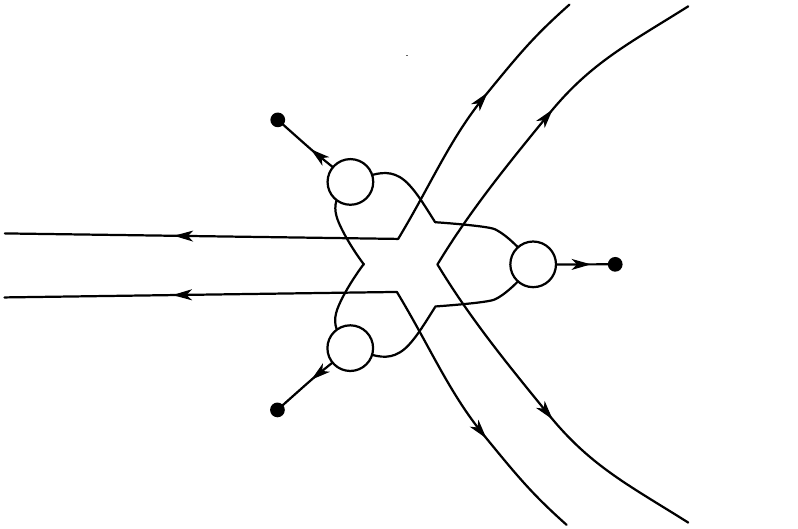}
 \caption{Contour $\Sigma_R$ for the jump of the RHP for $R$.}\label{figure_contour_r}
\end{figure}

It follows from \eqref{re_phi_j_positive}, \eqref{re_psi_j_positive} and the definition of the jump $J_S$ that for some positive constant $c$,
$$
J_R(z)=I+\Boh(e^{-nc}),\quad z\in \Gamma_R\setminus\partial D_\delta,
$$
whereas from the RHP for $P$ it follows that
$$
J_R(z)=I+\Boh(n^{-1}), \quad z\in \partial D_\delta,
$$
where the implicit terms in the last two formulas above are uniform in $z$. As a consequence \cite{deift_book}, we conclude that for $n$ large enough the RHP for $R$ is uniquely 
solvable and
\begin{equation}\label{estimate_R}
R(z)=I+\Boh\left( \frac{1}{n(1+|z|)} \right),\quad n\to\infty,
\end{equation}
uniformly on $\C\setminus \Gamma_R$.

Thus we can invert the series of transformations performed
$$
Y \mapsto X \mapsto T \mapsto S \mapsto R
$$
to conclude that the RHP for $Y$ is uniquely solvable and, moreover, translate \eqref{estimate_R} into asymptotic information about $Y$ (see 
Section~\ref{section_proof_theorem_limiting_counting_measure} below for an example) concluding the steepest descent analysis.

\section{Riemann-Hilbert analysis in the one-cut case}\label{section_riemann_hilbert_analysis_postcritical}

We proceed to the Riemann-Hilbert/Steepest Descent analysis in the one-cut case $(t_0,t_1)\in\mathcal F_2$. We do not give much details, and mostly highlight the main 
differences comparing to the three-cut case carried out in Section~\ref{section_riemann_hilbert_analysis_precritical}. The focus is on the jumps and 
parametrices, the remaining aspects of the steepest descent analysis are the same as in the three-cut case.

Following Theorem~\ref{theorem_limiting_support_zeros}, for $(t_0,t_1)\in \mathcal F_2$ we denote $\Sigma_*=[z_1,z_0]$. The first step is to define the contours $\Sigma$ and $L$ 
in the same spirit as \eqref{projection_property_extension_sigma} and \eqref{definition_L_2}. As before, these are defined taking into account the critical graph of the quadratic 
differential $\varpi$. The parts of $\Sigma$ and $L$ lying on the real line are defined by 
$$
\Sigma_0=[z_2,\hat z_0],\quad L_0=(-\infty,z_2]=\pi\left( \gamma_1(z_2^{(1)})\right).
$$
To construct $\Sigma_1$, $\Sigma_2$, $L_1$ and $L_2$, we rely on the critical graph of $\varpi$. For $\pi:\mathcal R\to\overline \C$ the canonical projection and $\mathcal S_3$ 
the strip domain determined by the condition that $\hat z_2^{(3)}$ and $z_2^{(3)}$ are the critical points on its boundary (see Figure~\ref{figure_planar_graph_2}), we consider an 
oriented contour $\Sigma_2$ from $z_2$ to $\hat z_2$, contained in the upper half plane, and satisfying
\begin{equation}\label{projection_property_extension_sigma_postcritical}
\Sigma_2\setminus\{z_2,\hat z_2\}\subset \pi\left( \mathcal S_3\cap \mathcal R_3 \right),
\end{equation}
and set $\Sigma_1=(\Sigma_2)^*$, $\Sigma=\Sigma_0\cup\Sigma_1\cup\Sigma_2$. Furthermore, define
\begin{equation}\label{definition_L_2_postcritical}
L_1=\pi\left( \gamma_0(z_2^{(1)})\right),\quad L_2=(L_1)^*=\pi\left( \gamma_2(z_2^{(1)}) \right),
\end{equation}
and then set $L=L_0\cup L_1\cup L_2$, see Figure~\ref{figure_projection_property_postcritical}. Choosing $\Sigma_1$ and $\Sigma_2$ appropriately, we can also be sure that 
$L\cap \Sigma=\{z_2\}$.

\begin{figure}[t]
\centering
 \begin{overpic}[scale=1]
 {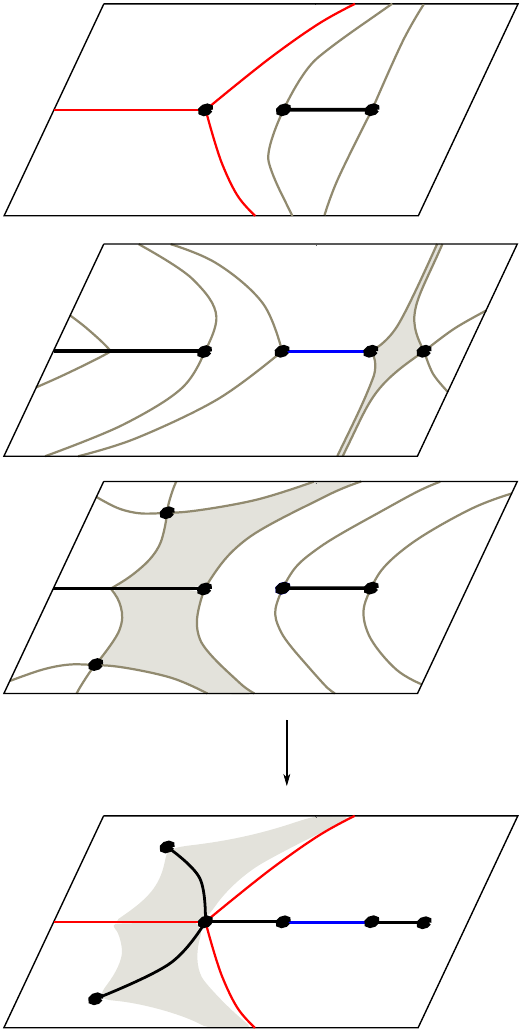}
 \put(50,10){$\C$}
 \put(29,26){$\pi$}
 \put(50,43){$\mathcal R_3$}
 \put(50,66){$\mathcal R_2$}
 \put(50,89){$\mathcal R_1$}
 \put(30,11.5){$\scriptstyle{\Sigma_{*}}$}
 \put(30,16.5){$\scriptstyle{L}$}
 \end{overpic}
 \caption{Illustration of properties \eqref{projection_property_extension_sigma_postcritical} and \eqref{definition_L_2_postcritical}. The trajectories highlighted on $\mathcal R$ 
in blue and red are projected onto $\C$ to the contours $\Sigma_*$, $L$, also respectively represented in blue and red. $\Sigma_*$ is the interval in blue color, whereas $L$ 
consists of the pieces projected from $\mathcal R_1$. In addition to the pieces on the real line, we extend $\Sigma_*$ to $\Sigma$ by constraining $\Sigma\setminus 
(\Sigma_*\cup\R)$ to lie within the shaded region on $\C$, which consists of the projections of the gray strip domains on $\mathcal R$. The arcs of $\Sigma\setminus \Sigma_*$ 
on the complex plane are depicted in black.}\label{figure_projection_property_postcritical}
\end{figure}

For this choice of $\Sigma$, we consider the diagonal sequence of multiple orthogonal polynomials $(P_{n,n})$ in Definition~\ref{definition_mop}. As before, such 
polynomials can be alternatively characterized by \eqref{multiple_orthogonality_conditions_airy}. Furthermore, assuming $n$ even as before, $P_{n,n}$ is alternatively described by 
the Riemann-Hilbert problem $Y$ given in Section~\ref{section_rhp_Y}.

As in the three-cut situation, the contour $L$ defined above splits the complex plane into three regions $G_0,G_1,G_2$, where $G_j$ contains the point $\hat z_j$, $j=0,1,2$. 
The first transformation $Y\mapsto X$ is exactly the same as in Section~\ref{section_transformation_Y_X}, see Figure~\ref{figure_contour_x_postcritical} for a display of the jump 
contours.

\begin{figure}[t]
 \begin{overpic}[scale=1]
  {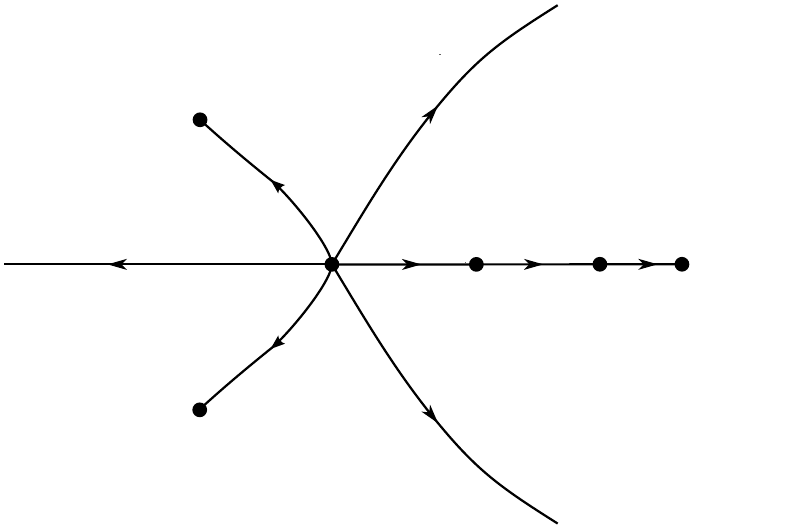}
  \put(87,32){$\hat z_0$}
  \put(25,11){$\hat z_1$}
  \put(25,52.5){$\hat z_2$}
  \put(74,35){$z_0$}
  \put(59,35){$z_1$}
  \put(43,34){$z_2$}
  \put(90,50){$G_0$}
  \put(10,60){$G_2$}
  \put(10,0){$G_1$}
  \put(10,34.5){$L_0$}
  \put(60,10){$L_2$}
  \put(65,58.5){$L_1$}
 \end{overpic}
 \caption{Contour $\Sigma\cup L$ for the RHP's for $X$ and $T$, and the sectors $G_0$, $G_1$, $G_2$.}\label{figure_contour_x_postcritical}
\end{figure}

The second transformation $X\mapsto T$ is also similar as for the three-cut case. The only difference is concerned the starting points of integration in the definition of 
the $g$-functions.

More precisely, we define
\begin{equation}\label{g_functions_postcritical}
\begin{aligned}
g_1(z) & = \int_{z_0}^z \xi_1(s)ds+c_1 , & \quad z\in \C\setminus (-\infty,z_0),  \\
g_2(z) & = \int_{z_2}^z \xi_2(s)ds+c_2,  & \quad z\in \C\setminus (-\infty,z_2), \\
g_3(z) & = \int_{z_0}^z \xi_3(s)ds+c_3,  & \quad  z\in \C\setminus (-\infty,z_0),
\end{aligned}
\end{equation}
where
$$
c_1=c_3=\int_{z_2}^{z_0}\xi_{3-}(s)ds,\quad c_2=-\pi i t_0.
$$

As in \eqref{asymptotics_g_functions}, $g_1,g_2$ and $g_3$ admit the asymptotic expansion
\begin{align*}
g_1(z) & = V(z)+l_1+t_0 \log z + \Boh(z^{-1}),\nonumber \\
g_2(z) & = -\frac{2}{3}(z-t_1)^{3/2}+l_2-\frac{t_0}{2}\log z +\Boh(z^{-1/2}), \qquad z\to \infty \\
g_3(z) & = \frac{2}{3}(z-t_1)^{3/2}+l_3-\frac{t_0}{2}\log z +\Boh(z^{-1/2}), \nonumber
\end{align*}
and the proof of Lemma~\ref{lemma_constants_c_j} carries over without any essential modification, leading to $l_2=l_3$ and $c_2=i\im c_3$.

For the $g$-functions as in \eqref{g_functions_postcritical}, we make the transformation $X\mapsto T$ as in \eqref{transformation_T}. The resulting RH-problem characterizing 
$T$ is similar to the one presented in Section~\ref{section_transformation_X_T}. The jump contour is given by $\Sigma\cup L$, see Figure~\ref{figure_contour_x_postcritical} and, 
after simplifications, its jump matrix $J_T$ reduces to
\begin{multline*}
J_T(z)= \\
\begin{cases}
\begin{pmatrix}
 e^{\frac{n}{t_0}(g_{1-}(z)-g_{1+}(z))} & 1 & 0 \\
 0 & e^{\frac{n}{t_0}(g_{3-}(z)-g_{3+}(z))} & 0 \\
 0 & 0 & 1
\end{pmatrix},
& z\in (z_1,z_0), \\
\begin{pmatrix}
1 & e^{\frac{n}{t_0}(g_{1-}(z)-g_{3+}(z))} & 0 \\
0 & 1 & 0 \\
0 & 0 & 1
\end{pmatrix},
& z\in (z_2,z_1)\cup (z_0,\hat z_0), \\
\begin{pmatrix}
 1 & e^{\frac{n}{t_0}(g_{1}(z)-g_{2}(z))} & 0 \\
 0 & 1 & 0 \\
 0 & 0 & 1
\end{pmatrix},
& z\in \Sigma_1\cup\Sigma_2, \\
\begin{pmatrix}
 1 & 0 & 0 \\
 0 & \omega^2 e^{\frac{n}{t_0}(g_{2-}(z)-g_{2+}(z))} & 1 \\
 0 & 0 & \omega e^{\frac{n}{t_0}(g_{3-}(z)-g_{3+}(z))}
\end{pmatrix},
& z\in L_0, \\
\begin{pmatrix}
 1 & 0 & 0 \\
 0 & \omega^2 e^{\frac{n}{t_0}(g_{3}(z)-g_{2}(z))} & 1 \\
 0 & 0 & \omega e^{\frac{n}{t_0}(g_{2}(z)-g_{3}(z))}
\end{pmatrix},
& z\in L_1, \\
\begin{pmatrix}
 1 & 0 & 0 \\
 0 & \omega^2 e^{\frac{n}{t_0}(g_{2}(z)-g_{3}(z))} & 1 \\
 0 & 0 & \omega e^{\frac{n}{t_0}(g_{3}(z)-g_{2}(z))}
\end{pmatrix},
& z\in L_2.
\end{cases}
\end{multline*}

Analogously to \eqref{definition_Phi}, \eqref{definition_Psi}, we now consider
\begin{equation}\label{definition_Phi_postcritical}
\begin{aligned}
\Phi_0(z) & =\frac{1}{t_0} \int_{z_0}^z (\xi_1(s)-\xi_3(s))ds, &\quad z\in \C\setminus(-\infty,z_0), \\
\Phi_1(z) & =\frac{1}{t_0} \int_{z_1}^z (\xi_1(s)-\xi_3(s))ds, &\quad z\in \C\setminus((-\infty,z_2)\cup (z_1,+\infty)),\\
\Phi_2(z) & =\frac{1}{t_0} \int_{z_2}^z (\xi_1(s)-\xi_2(s))ds+\Phi_1(z_2), &\quad z\in \C\setminus((-\infty,z_2)\cup (z_1,+\infty)),
\end{aligned}
\end{equation}
and also
$$
\Psi(z) = \frac{1}{t_0}\int_{z_2}^z(\xi_2(s)-\xi_3(s))ds,\quad z\in \C\setminus((-\infty,z_2)\cup (z_1,+\infty)).
$$

The jump matrix $J_T$ is then expressed in terms of these functions as
$$
J_T(z)= 
\begin{cases}
\begin{pmatrix}
 e^{-n\Phi_{0+}(z)} & 1 & 0 \\
 0 & e^{-n\Phi_{0-}(z)} & 0 \\
 0 & 0 & 1
\end{pmatrix},
& z\in (z_1,z_0), \\
\begin{pmatrix}
1 & e^{n\Phi_{0}(z)} & 0 \\
0 & 1 & 0 \\
0 & 0 & 1
\end{pmatrix},
& z\in (z_0,\hat z_0), \\
\begin{pmatrix}
1 & e^{n\Phi_{1}(z)} & 0 \\
0 & 1 & 0 \\
0 & 0 & 1
\end{pmatrix},
& z\in (z_2,z_1), \\
\begin{pmatrix}
 1 & e^{n\Phi_{2}(z)}  & 0 \\
 0 & 1 & 0 \\
 0 & 0 & 1
\end{pmatrix},
& z\in \Sigma_1\cup\Sigma_2, \\
\begin{pmatrix}
 1 & 0 & 0 \\
 0 & \omega^2 e^{-n\Psi_+(z)}  & 1 \\
 0 & 0 & \omega e^{-\Psi_-(z)}
\end{pmatrix},
& z\in L_0, \\
\begin{pmatrix}
 1 & 0 & 0 \\
 0 & \omega^2 e^{-n\Psi(z)} & 1 \\
 0 & 0 & \omega e^{n\Psi(z)}
\end{pmatrix},
& z\in L_1, \\
\begin{pmatrix}
 1 & 0 & 0 \\
 0 & \omega^2 e^{n\Psi(z)} & 1 \\
 0 & 0 & \omega e^{-n\Psi(z)}
\end{pmatrix},
& z\in L_2.
\end{cases}
$$

The jump matrices above are in a suitable form for the opening of lenses. We open the lens $\mathscr S$ around $(z_1,z_0)$ and denote by $\mathscr S^{\pm}$ 
the part of $\mathscr S$ on the $\pm$-side of $(z_1,z_0)$, and by $\partial \mathscr S^\pm$ the component of the boundary of $\partial \mathscr S$ on the $\pm$-side of 
$(z_1,z_0)$. Similarly, $\mathscr L_j^\pm$ denotes the part of $\mathscr L$ on the $\pm$-side of $L_j$, and $\partial L_j^\pm$ denotes the component of the boundary of $\mathscr 
L$ on the $\pm$-side of $L_j$. Additionally, we open the lens $\mathscr L$ in such a way that it does not intersect $\Sigma$, see Figure~\ref{figure_opening_lenses_postcritical}.

The functions $\Phi_0$ and $\Phi_1$ satisfy
\begin{equation}\label{decaying_conditions_postcritical_1}
\Phi_0(z)<0,\quad z\in (z_0,\hat z_0],\qquad \Phi_1(z)<0, \quad z\in [z_2,z_1).
\end{equation}

Moreover, due to the construction of $\Sigma_1,\Sigma_2$ and $L$ as in equations~\eqref{projection_property_extension_sigma_postcritical}--\eqref{definition_L_2_postcritical}, we 
can be sure that (after reducing the lenses if necessary)
\begin{equation}\label{decaying_conditions_postcritical_2}
\begin{aligned}
& \re \Phi_0(z)>0,\quad z\in \partial \mathscr S^\pm \setminus \{z_0,z_1\}, \\
& \re \Phi_2(z)<0,\quad z\in \Sigma_1\cup\Sigma_2, \\
& \re \Psi(z)>0, \quad z\in \partial \mathscr L_0^\pm\setminus\{z_2\}, \\
& \pm \re \Psi(z)>0,\quad z\in \partial \mathscr L_1^\pm\setminus\{z_2\}, \\
& \mp \re \Psi(z)>0,\quad z\in \partial \mathscr L_2^\pm\setminus\{z_2\}.
\end{aligned}
\end{equation}
These conditions will assure the jumps for the next transformation have the right decaying properties. We stress that due to the constant $\Phi_1(z_2)$ in the definition of 
$\Phi_2$ in \eqref{definition_Phi_postcritical}, the strict inequality $\re\Phi_2<0$ also holds true at the endpoint $z_2$ of $\Sigma_1$ and $\Sigma_2$.

We then set 
$$
S(z)=T(z),\quad z\mbox{ outside the lenses } \mathscr S\cup \mathscr L,
$$
and on the lenses
$$
S(z)=T(z)\times
\begin{cases}
 \begin{pmatrix}
  1 & 0 & 0 \\
  \mp e^{-n\Phi_0(z)} & 1 & 0 \\
  0 & 0 & 1
 \end{pmatrix},
 & z\in \mathscr S^\pm, \\
 \begin{pmatrix}
  1 & 0 & 0 \\
  0 & 1 & 0 \\
  0 & \mp \omega^\mp e^{-n\Psi(z)} & 1 \\
 \end{pmatrix},
 & z\in \mathscr L_0^\pm \\
 \begin{pmatrix}
  1 & 0 & 0 \\
  0 & 1 & 0 \\
  0 & \mp \omega^\mp e^{\mp n\Psi(z)} & 1 \\
 \end{pmatrix},
 & z\in \mathscr L_1^\pm,\\
  \begin{pmatrix}
  1 & 0 & 0 \\
  0 & 1 & 0 \\
  0 & \mp \omega^\mp e^{\pm n\Psi(z)} & 1 \\
 \end{pmatrix},
 & z\in \mathscr L_2^\pm.
\end{cases}
$$

\begin{figure}[t]
 \begin{overpic}[scale=1]
{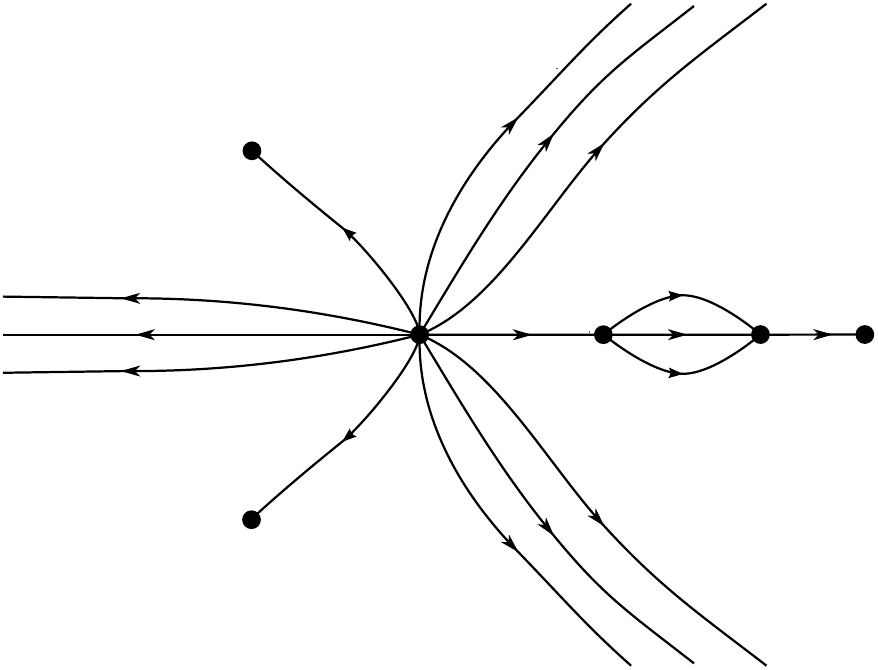}
\put(76,44){$\scriptstyle \partial\mathscr S^+$}
\put(76,31){$\scriptstyle \partial\mathscr S^-$}
\put(10,44){$\scriptstyle \partial\mathscr L_0^-$}
\put(10,31){$\scriptstyle \partial\mathscr L_0^+$}
\put(52,10){$\scriptstyle \partial\mathscr L_2^-$}
\put(74,12){$\scriptstyle \partial\mathscr L_2^+$}
\put(54.5,65){$\scriptstyle \partial\mathscr L_1^+$}
\put(77,65){$\scriptstyle \partial\mathscr L_1^-$}
 \end{overpic}
 \caption{The boundary components of the lenses $\mathscr S$, $\mathscr L$ determining the contour $\Gamma_S$ for $S$.}\label{figure_opening_lenses_postcritical}
\end{figure}

Then $S$ satisfies a Riemann-Hilbert problem on the contour $\Gamma_S$ shown in Figure~\ref{figure_opening_lenses_postcritical}. The jump matrix $J_S$ coincides with $J_T$ 
outside the lenses, and on the remaining parts of $\Gamma_S$ it is given by
$$
J_S(z)=
\begin{cases}
\begin{pmatrix}
 0 & 1 & 0 \\
 -1 & 0 & 0 \\
 0 & 0 & 1
\end{pmatrix},
& z\in (z_1,z_0) \\
\begin{pmatrix}
 1 & 0 & 0 \\
 0 & 0 & 1 \\
 0 & -1 & 0
\end{pmatrix},
& z\in L \\
\begin{pmatrix}
 1 & 0 & 0 \\
 e^{-n\Phi_0(z)} & 1 & 0\\
 0 & 0 & 1
\end{pmatrix},
& z\in \partial \mathscr S^{\pm},\\
\begin{pmatrix}
 1 & 0 & 0 \\
 0 & 1 & 0 \\
 0 & \omega^\mp e^{-n\Psi(z)} & 1 
\end{pmatrix},
& z\in \partial \mathscr L_0^\pm \\
\begin{pmatrix}
 1 & 0 & 0 \\
 0 & 1 & 0 \\
 0 & \omega^\mp e^{\mp n\Psi(z)} & 1
\end{pmatrix},
& z\in \partial \mathscr L_1^\pm \\
\begin{pmatrix}
 1 & 0 & 0 \\
 0 & 1 & 0 \\
 0 & \omega^\mp e^{\pm n\Psi(z)} & 1
\end{pmatrix},
& z\in \partial \mathscr L_2^\pm
\end{cases}
$$

The next step is the construction of the parametrices. In virtue of \eqref{decaying_conditions_postcritical_1}--\eqref{decaying_conditions_postcritical_2}, the jump matrix 
$J_S$ is exponentially small on the lipses of the lenses as well as in $\Sigma\setminus[z_1,z_0]$, as long as we stay away from the endpoints $z_0,z_1,z_2$. Near $z_2$, the jumps 
for $S$ on $\Sigma$ are still exponentially small, so for the local parametrix near $z_2$ we only have to take into account the jumps coming from $L$ and $\partial \mathscr L$.

More concretely, 

\begin{enumerate}
 \item[$\bullet$] The RHP for the global parametrix $M$ is essentially the same as in Section~\ref{section_global_parametrix}, having in mind that $\Sigma_*=(z_1,z_0)$. We refer 
to Section \ref{global_parametrix_one_cut} for details.
 
 \item[$\bullet$] The local parametrices near $z_0,z_1$ are constructed out of Airy functions in exactly the same way as in Section~\ref{section_local_parametrix}. A little more 
care should be taken for the parametrix near $z_2$. As we already observed, the jumps for $S$ on $\Sigma$ near $z_2$ are exponentially small, so we neglect them for the 
construction of the local parametrix near $z_2$. Hence the jump condition on $D_\delta(z_2)$ becomes
$$
P_+(z)=P_-(z)J_S(z),\quad z\in  D_\delta(z_2)\cap (L\cup \partial \mathscr L),
$$
see Figure~\ref{figure_local_parametrix_postcritical} for the jump contours of $P$ near $z_2$. The remaining RHP is 
essentially $2\times 2$. Although there are nine rays emanating from $z_2$ instead of the usual four rays, this parametrix is still constructed out of Airy functions, see for 
instance \cite{kuijlaars_tovbis_supercritical_normal_matrix_model}.
\end{enumerate}

\begin{figure}[t]
\centering
 \begin{overpic}[scale=1]
  {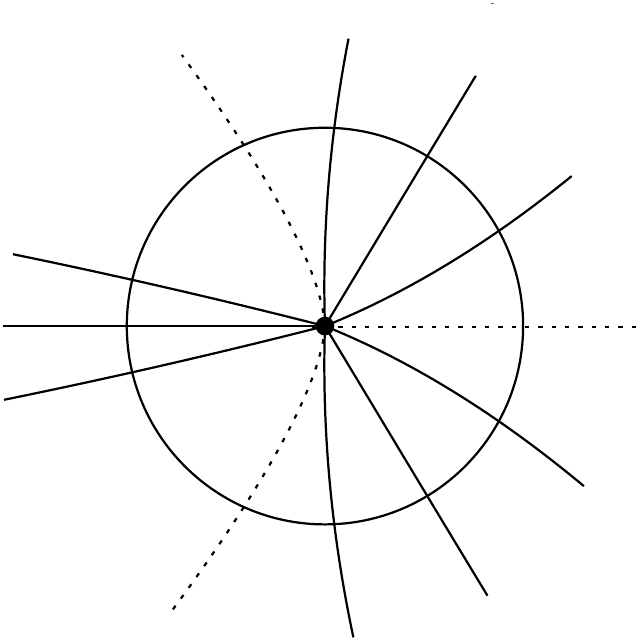}
  \put(90,51){$[z_2,z_1]$}
  \put(30,90){$\Sigma_2$}
  \put(30,5){$\Sigma_1$}
  \put(55,93){$\partial\mathscr L_1^+$}
  \put(90,70){$\partial\mathscr L_1^-$}
  \put(90,25){$\partial\mathscr L_2^+$}
  \put(56,0){$\partial\mathscr L_2^-$}
  \put(0,33){$\partial\mathscr L_0^+$}
  \put(0,62){$\partial\mathscr L_0^-$}
  \put(0,51){$L_0$}
  \put(75,85){$L_1$}
  \put(75,10){$L_2$}
  \put(45,45){$\scriptstyle z_2$}
 \end{overpic}
 \caption{Blow up of the jump contours for $S$ near $z_2$. The solid lines represent $L$ and $\partial\mathscr L$, so these are contours for the jumps of the local parametrix $P$. 
The dashed lines represent $\Sigma$, and since $J_S$ is exponentially small in these contours, they are not taken into account for the construction of $P$ near $z_2$, so $P$ is 
analytic across these contours.}\label{figure_local_parametrix_postcritical}
\end{figure}

The final transformation $S\mapsto R$ is similar as in Section~\ref{section_transformation_S_R}, equation \eqref{transformation_R}. The contour $\Sigma_R$ for $R$ is displayed in 
Figure~\ref{figure_contour_r_postcritical}. 

\begin{figure}[t]
\centering
 \begin{overpic}[scale=1]
  {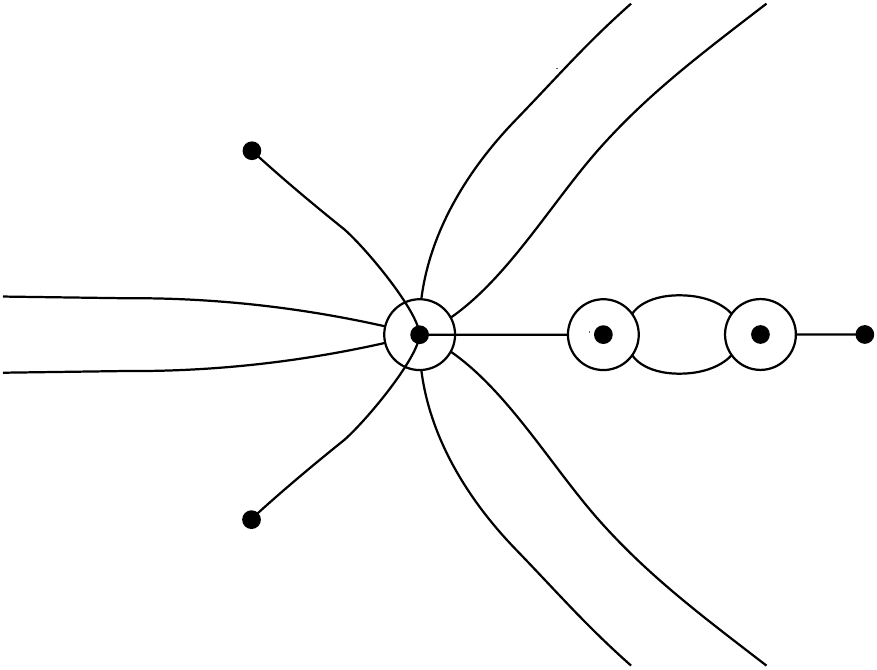}
\put(76,44){$\scriptstyle \partial\mathscr S^+$}
\put(76,31){$\scriptstyle \partial\mathscr S^-$}
\put(10,44){$\scriptstyle \partial\mathscr L_0^-$}
\put(10,31){$\scriptstyle \partial\mathscr L_0^+$}
\put(52,10){$\scriptstyle \partial\mathscr L_2^-$}
\put(74,12){$\scriptstyle \partial\mathscr L_2^+$}
\put(54.5,65){$\scriptstyle \partial\mathscr L_1^+$}
\put(77,65){$\scriptstyle \partial\mathscr L_1^-$}
\put(35,54){$\scriptstyle \Sigma_2$}
\put(35,26){$\scriptstyle \Sigma_1$}
\put(48,39.5){$\scriptstyle z_2$}
\put(68,39.5){$\scriptstyle z_1$}
\put(86,39.5){$\scriptstyle z_0$}
 \end{overpic}
 \caption{Contours for the jumps of $R$.}\label{figure_contour_r_postcritical}
\end{figure}

As in the three-cut case, it turns out that the jump matrix $J_R$ is close to the identity as $n\to\infty$: on the lipses of the lenses $\mathscr{S}$ and $\mathscr{L}$, this is 
true because of \eqref{decaying_conditions_postcritical_1}--\eqref{decaying_conditions_postcritical_2}, whereas on the boundary of $D_\delta$, this is true from the 
construction of the local parametrix. We only have to be careful about the jumps inside $D_\delta(z_2)$ that are not canceled by the local parametrix, which are given by 
$$
J_R(z)=P(z)J_S(z)P(z)^{-1}=P(z)J_T(z)P(z)^{-1},\quad z\in\Sigma\cap D_\delta(z_2).
$$
Since $P$ is bounded near $z_2$, the second inequality in \eqref{decaying_conditions_postcritical_2} together with the identity above assure us that 
$J_R$ is exponentially small for $z\in\Sigma\cap D_\delta(z_2)$. 

As the final outcome, we get that the jump matrix $J_R$ satisfies
$$
J_R(z)=I+\Boh(n^{-1}),\quad n\to\infty,
$$
uniformly in $J_R$, and the analysis is concluded in a similar fashion as in Section~\ref{section_transformation_S_R}.


\section{Construction of the global parametrix}\label{section_construction_global_parametrix}

In this section we prove the existence of the global parametrix in the three-cut and one-cut cases. We also construct its first row explicitly.

It is convenient to perform a regluing of the sheets forming the Riemann surface $\mathcal R$, in much the same spirit as used for $t_1=0$ in Section 
\ref{subsection_technical_computations_precritical}.

To do so, recall the definition of the contours $L_1$ and $L_2$ given in \eqref{definition_L_2} and \eqref{definition_L_2_postcritical} in the three-cut and one-cut cases, 
respectively, which defined the sector $G_0$ containing the point $z_0$, as shown in Figures \ref{figure_contour_x} and \ref{figure_contour_x_postcritical}.

We construct a new Riemann surface
$$
\widetilde{\mathcal R}= \widetilde{\mathcal R}_1\cup \widetilde{\mathcal R}_2\cup \widetilde{\mathcal R}_3,
$$
obtained from the original surface $\mathcal R$ after interchanging the sectors $\pi^{-1}(G_0)\cap \mathcal R_2$ and $\pi^{-1}(G_0)\cap \mathcal R_3$. Thus the sheets 
$\widetilde{\mathcal R}_1$ and $\widetilde{\mathcal R}_2$ are connected crosswise along $\Sigma_*$ and the sheets $\widetilde{\mathcal R}_2$ and $\widetilde{\mathcal R}_3$ are 
connected crosswise along $L$. In the three-cut case, the branch points of $\widetilde R$ are 
$$
z_j^{(1)}=z_{j}^{(2)},\ j=0,1,2,\qquad \infty^{(2)}=\infty^{(3)},
$$
whereas in the one-cut case the branch points are 
$$
z_j^{(1)}=z_j^{(2)},\ j=0,1, \qquad z_2^{(2)}=z_2^{(3)},\qquad \infty^{(2)}=\infty^{(3)}.
$$

We also denote
\begin{equation}\label{partition_reglued_surface}
\widetilde{\mathcal R}_{j,k}=\pi^{-1}(G_k)\cap \mathcal R_j,\quad j=1,2, \; k=0,1,2,
\end{equation}
and refer to Figures \ref{figure_regluing_three_cut} and \ref{figure_regluing_one_cut} for a depiction of the regluing and the sets \eqref{partition_reglued_surface} in the 
three-cut and one-cut cases, respectively.

\begin{figure}[t]
 \centering
 \begin{overpic}[scale=1]
 {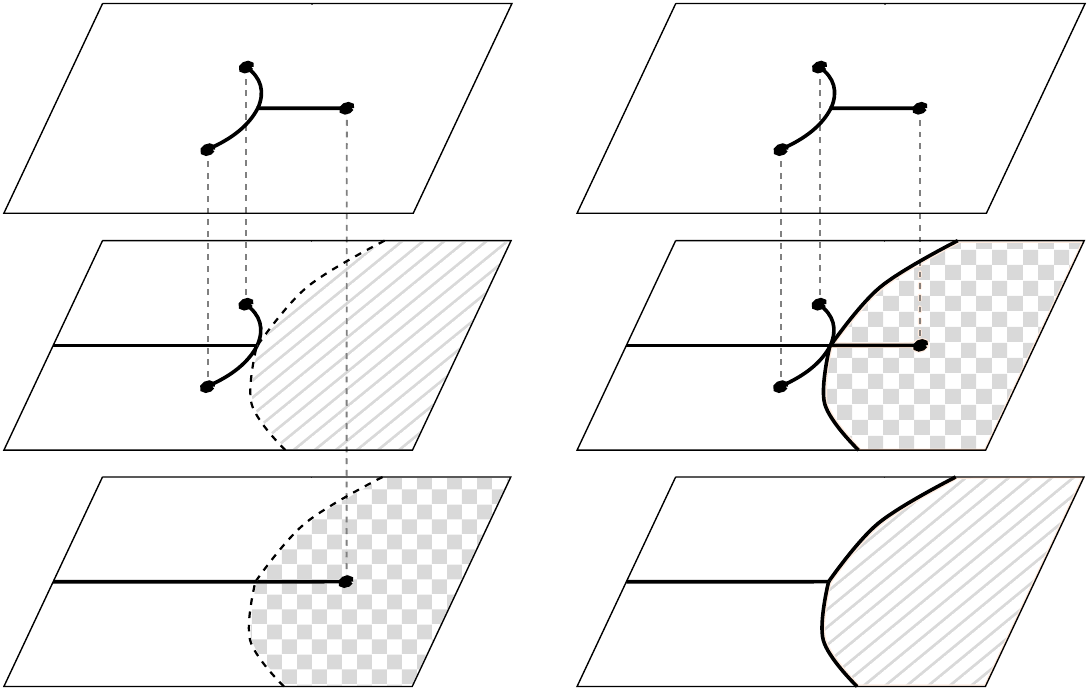}
 \put(0,55){$\mathcal R_1$}
 \put(0,34){$\mathcal R_2$}
 \put(0,13){$\mathcal R_3$}
 \put(100,55){$\widetilde{\mathcal R}_1$}
 \put(100,34){$\widetilde{\mathcal R}_2$}
 \put(100,13){$\widetilde{\mathcal R}_3$}
 \put(82,25){$\widetilde{\mathcal R}_{2,0}$}
 \put(64,36){$\widetilde{\mathcal R}_{2,2}$}
 \put(58,24){$\widetilde{\mathcal R}_{2,1}$}
 \put(82,4){$\widetilde{\mathcal R}_{3,0}$}
 \put(64,15){$\widetilde{\mathcal R}_{3,2}$}
 \put(58,3){$\widetilde{\mathcal R}_{3,1}$}
 \end{overpic}
 \caption{The sheet structure for $\mathcal R$ (left panel) and $\mathcal{\widetilde R}$ (right panel) in the three-cut case. The dashed lines on $\mathcal 
R_2$ and $\mathcal R_3$ are the (preimages through $\pi$ of) the curves $L_1$ and $L_2$. They bound the shaded areas, which are 
interchanged between the sheets to create the new sheets $\mathcal{\widetilde R}_2$ and $\mathcal{\widetilde R}_3$. On the right panel we also distinguish the set 
$\mathcal{\widetilde R}_{j,k}$'s defined in \eqref{partition_reglued_surface}. }\label{figure_regluing_three_cut}
\end{figure}

\begin{figure}[t]
 \centering
 \begin{overpic}[scale=1]
 {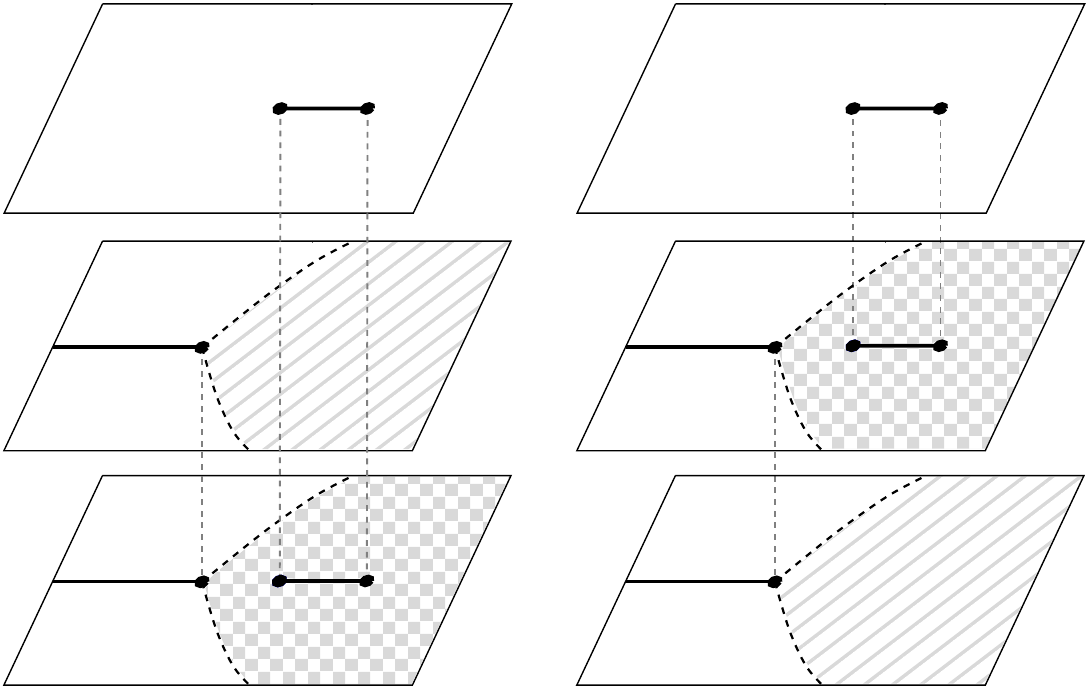}
 \put(0,55){$\mathcal R_1$}
 \put(0,34){$\mathcal R_2$}
 \put(0,13){$\mathcal R_3$}
 \put(100,55){$\widetilde{\mathcal R}_1$}
 \put(100,34){$\widetilde{\mathcal R}_2$}
 \put(100,13){$\widetilde{\mathcal R}_3$}
 \put(82,25){$\widetilde{\mathcal R}_{2,0}$}
 \put(64,36){$\widetilde{\mathcal R}_{2,2}$}
 \put(58,24){$\widetilde{\mathcal R}_{2,1}$}
 \put(82,4){$\widetilde{\mathcal R}_{3,0}$}
 \put(64,15){$\widetilde{\mathcal R}_{3,2}$}
 \put(58,3){$\widetilde{\mathcal R}_{3,1}$}
 \end{overpic}
 \caption{The sheet structure for $\mathcal R$ (left panel) and $\mathcal{\widetilde R}$ (right panel) in the one-cut case. The dashed lines on $\mathcal 
R_2$ and $\mathcal R_3$ are the (preimages through $\pi$ of) the curves $L_1$ and $L_2$. They bound the shaded areas, which are 
interchanged between the sheets to create the new sheets $\mathcal{\widetilde R}_2$ and $\mathcal{\widetilde R}_3$. On the right panel we also distinguish the set 
$\mathcal{\widetilde R}_{j,k}$'s defined in \eqref{partition_reglued_surface}. }\label{figure_regluing_one_cut}
\end{figure}

It is also convenient to denote by $\Sigma_{*,+}^{(k)}$ and $\Sigma_{*,-}^{(k)}$ the positive and negative sides of the cut $\Sigma_*$ on the sheet $\widetilde{\mathcal R}_k$. 
Ditto for the other quantities $\Sigma_{*,j,\pm}^{(k)}$, $L_{\pm}^{(k)}$ and $L_{j,\pm}^{(k)}$. In particular, note that
$$
\Sigma_{*,\pm}^{(1)}=\Sigma_{*,\mp}^{(2)},\quad L_{\pm}^{(2)}=L_{\mp}^{(3)}.
$$
We orient each arc of $\Sigma_{*}^{(k)}$ and $L^{(k)}$ according to the orientation induced from their projection $\Sigma_{*}$ and $L$. Thus, for instance, the 
positive side of $\Sigma_{*,+}^{(1)}$ lies on the sheet $\widetilde{\mathcal R}_1$, whereas the negative side of $\Sigma_{*,+}^{(1)}$ lies on the sheet $\widetilde{\mathcal R}_2$.

\subsection{The inverse of the rational parametrization}

According to Theorem \ref{theorem_schwarz_function}, the rational function $h$ induces the bijection \eqref{rational_function_bijection} between $\overline \C$ and $\mathcal R$, 
and consequently between $\overline \C$ and $\widetilde{\mathcal R}$. This means that there exist three meromorphic functions
\begin{equation}\label{inverse_rational_parametrization_psi1}
\psi_j:\widetilde{\mathcal R}_j\to \overline \C,\quad j=1,2,3,
\end{equation}
for which 
\begin{equation}\label{inverse_rational_parametrization}
\psi:\widetilde{\mathcal R}\to \overline \C,\quad \restr{\psi}{\widetilde{\mathcal R}_j}=\psi_j,\; j=1,2,3,
\end{equation}
is the inverse of $h$. 

Set 
\begin{equation}\label{definition_w_partition}
\begin{aligned}
& \mathcal W_j = \psi(\widetilde{\mathcal R}_j), \quad  j=1,2,3,\\
& \mathcal W_{j,k} = \psi(\widetilde{\mathcal R}_{j,k}), \quad j=2,3,\; k=0,1,2,
\end{aligned}
\end{equation}
and also
\begin{equation}\label{definition_jumps_w_plane}
\begin{aligned}
 & \Xi=\psi(\Sigma_{*,+}^{(1)}),   \\
 & \Lambda_j=\psi(L^{(2)}_{j,+}),   \quad j=0,1,2, \\ 
 & \Lambda=\psi(L_+^{(2)})=\Lambda_0\cup\Lambda_1\cup\Lambda_2.
\end{aligned}
\end{equation}
In the three-cut case, we also define
\begin{equation}\label{definition_jumps_w_plane_three_cut}
\Xi_j=\psi(\Sigma_{*,j,+}^{(1)}), \; j=0,1,2,
\end{equation}
so that $\Xi=\Xi_1\cup\Xi_2\cup\Xi_3$.

Using basic properties of conformal maps, the sets \eqref{definition_w_partition}--\eqref{definition_jumps_w_plane_three_cut} can be described in the $w$-plane. The outcome for 
\eqref{definition_w_partition} can be seen in Figure \ref{figure_mother_body_w_plane_reglued}. The sets 
\eqref{definition_jumps_w_plane}--\eqref{definition_jumps_w_plane_three_cut} are displayed in Figures \ref{figure_mother_body_w_plane_jumps} and 
\ref{figure_mother_body_w_plane_jumps_one_cut} in the three-cut and one cut-cases, respectively.

\begin{figure}[t]
\begin{minipage}[c]{0.5\textwidth}
\centering
  \begin{overpic}[scale=1]
  {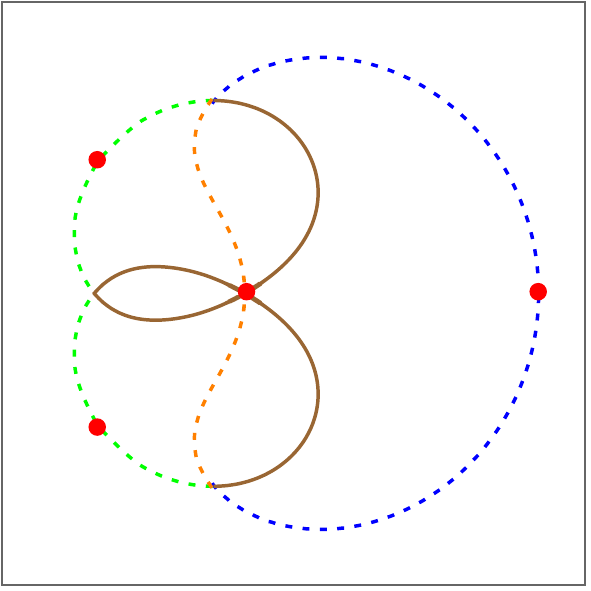}
  \put(85,85){$\scriptstyle \mathcal W_1$}
  \put(65,50){$\scriptstyle \mathcal W_{2,0}$}
  \put(20,63){$\scriptstyle \mathcal W_{2,2}$}
  \put(20,35){$\scriptstyle \mathcal W_{2,1}$}
  \put(20,49){$\scriptstyle \mathcal W_{3,0}$}
  \put(40,70){$\scriptstyle \mathcal W_{3,1}$}
  \put(40,29){$\scriptstyle \mathcal W_{3,2}$}
 \end{overpic}
\end{minipage}%
\begin{minipage}[c]{0.5\textwidth}
\centering
\begin{overpic}[scale=1]
{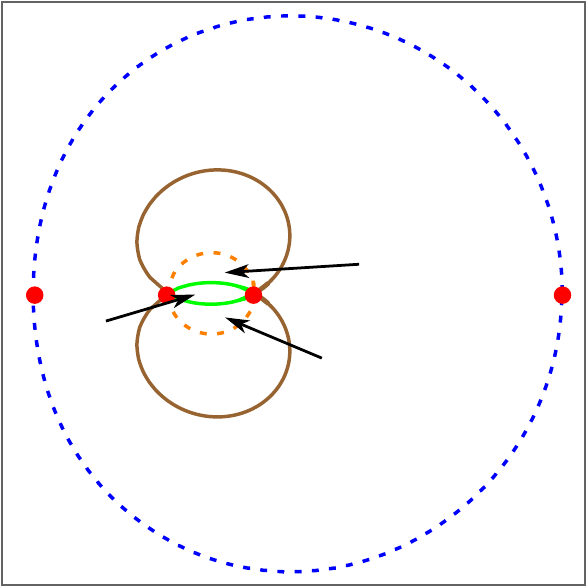}
\put(85,85){$\scriptstyle \mathcal W_1$}
\put(65,70){$\scriptstyle \mathcal W_{2,0}$}
\put(62,54){$\scriptstyle \mathcal W_{2,2}$}
\put(56,37){$\scriptstyle \mathcal W_{2,1}$}
\put(12,42){$\scriptstyle \mathcal W_{3,0}$}
\put(33,62){$\scriptstyle \mathcal W_{3,1}$}
\put(33,34){$\scriptstyle \mathcal W_{3,2}$}
 \end{overpic}
 \end{minipage}
 \caption{The partition of the $w$-plane into the sets $\mathcal W_{j,k}$'s in the three-cut (left panel) and one-cut (right panel) cases. The dashed lines display the inverse 
images of the cuts for the original Riemann surface $\mathcal R$ (compare with Figure \ref{w_partition}), and the solid lines display the new cuts arising after the regluing that 
defines $\mathcal{\widetilde R}$. Numerical output for the choices $r=1/20$ and $a_0=1/10$ (three-cut) and $a_0=1/4$ (one-cut).}\label{figure_mother_body_w_plane_reglued}
 \end{figure}

\begin{remark}
The function $\psi_1$ in \eqref{inverse_rational_parametrization_psi1} is analytic in $\C\setminus \Sigma_*$ and maps $\C\setminus \Omega$ conformally to $\C\setminus \D$. Since 
the point $\infty^{(1)}$ on $\mathcal R_1$ corresponds to the point $\infty$ on the $w$-plane through $\psi$, we automatically get that
$$
\lim_{z\to\infty}\psi_1(z)=\infty.
$$
Furthermore, using the Implicit Function Theorem,
$$
\lim_{z\to\infty} \psi_1'(z)=\lim_{w\to\infty} \frac{1}{h'(w)}=\frac{1}{r}.
$$
This shows that $\psi_1$ in \eqref{inverse_rational_parametrization_psi1} coincides with the function $\psi_1$ appearing in Theorem \ref{thm_wave_factor}.
\end{remark}

\subsection{Construction of the global parametrix in the three-cut case}\label{section_global_parametrix_three_cut}

In \cite{bleher_kuijlaars_normal_matrix_model}, the global parametrix is constructed for $t_1=0$ using meromorphic differentials. In this section we reproduce their arguments 
to construct the parametrix in the general three-cut case.

On the Riemann surface $\widetilde{\mathcal R}$, consider the meromorphic differential $\eta$, defined by the condition that it has simple poles at each of the branch 
points $z_0^{(1)},z_1^{(1)},z_2^{(1)}$ and $\infty^{(2)}$, with residues
\begin{equation}\label{residues_condition}
\res(\eta,z_j^{(1)})=-\frac{1}{2},\quad \re(\eta,\infty^{(2)})=\frac{3}{2},
\end{equation}
and no other poles. Since the sum of residues is zero, such an $\eta$ exists. It is also unique, because the genus of $\widetilde R$ is zero.

The set $\Sigma^{(1)}_{*,+}\cup L^{(2)}_+$ is connected and consists of a finite union of analytic arcs. Its image through the inverse $\psi$ in 
\eqref{inverse_rational_parametrization} is the set $\Xi\cup\Lambda$, which can be geometrically described with standard arguments in conformal mapping, and is displayed in Figure 
\ref{figure_mother_body_w_plane_jumps}. Since $\overline \C\setminus (\Xi\cup \Lambda)$ and  $\widetilde{\mathcal R}\setminus (\Sigma^{(1)}_{*,+}\cup L^{(2)}_+)$ are conformally 
equivalent, it readily follows from Figure \ref{figure_mother_body_w_plane_jumps} that the domain $\widetilde{\mathcal R}\setminus (\Sigma^{(1)}_{*,+}\cup 
L^{(2)}_{+})\subset \widetilde{\mathcal R}$ is simply connected and does not contain poles of $\eta$. In particular, this implies that the  
function
$$
u(p)=\int_{\infty^{(1)}}^p \eta, \quad p\in\widetilde{\mathcal R}\setminus (\Sigma^{(1)}_{*,+}\cup L^{(2)}_{+}),
$$
where the integration goes along any path that does not cross $\Sigma^{(1)}_{*,+}\cup L^{(2)}_{+}$, is well defined and analytic. 

\begin{figure}[t]
\centering
  \begin{overpic}[scale=1]
  {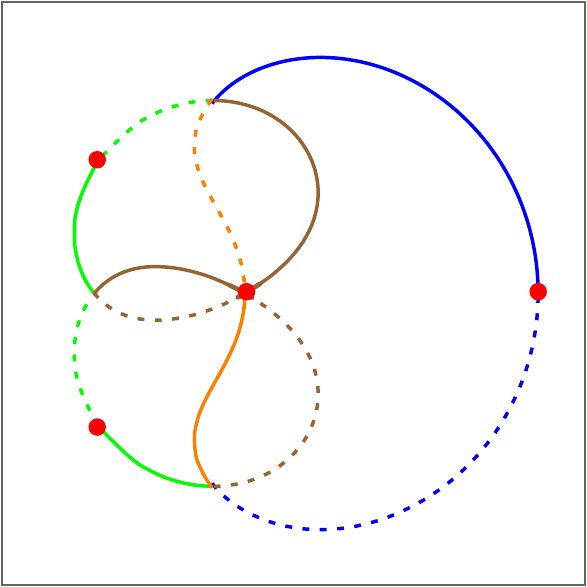}
  \put(80,80){$ \Xi_0$}
  \put(5,63){$\Xi_2$}
  \put(20,15){$\Xi_1$}
  \put(24,56){$\Lambda_1$}
  \put(55,70){$\Lambda_2$}
  \put(29,34){$\Lambda_0$}
 \end{overpic}
 \caption{In solid lines the sets $\Lambda_j$ and $\Xi_j$, $j=0,1,2$, are displayed; these are the image through $\psi$ of the jumps for $m$. In dashed lines, the remaining images 
through $\psi$ of the cuts of $\widetilde{\mathcal R}$ are shown. The solid lines also correspond to the jump contours for $f$ in the three-cut case. Numerical output for $r=1/20$ 
and $a_0=1/10$.}\label{figure_mother_body_w_plane_jumps}
 \end{figure}

From Figure \ref{figure_mother_body_w_plane_jumps} and conformal equivalence, it follows that for a given $p\in \Sigma_{*,+}^{(1)}\cup 
L_+^{(2)}$, we can write the difference $u_+(p)-u_-(p)$ as an integral over a closed contour going around exactly one of the branch points 
$z_j^{(1)}$. Consequently we learn from \eqref{residues_condition} and the Residues Theorem that
\begin{equation}\label{half_period_conditions}
u_+(p)-u_-(p)=\pm\pi i,\quad p\in \Sigma^{(1)}_{*,+}\cup L^{(2)}_{+}.
\end{equation}
We thus define 
$$
m(p)=e^{u(p)},\quad p\in \widetilde{\mathcal R}\setminus \Sigma^{(1)}_{*,+}\cup L^{(2)}_{+}.
$$
Note that \eqref{half_period_conditions} implies that 
\begin{equation}\label{jump_conditions_m_three_cut}
m_+(p)=-m_-(p),\quad p\in \Sigma^{(1)}_{*,+}\cup L^{(2)}_{+}.
\end{equation}
Set
$$
m_k=\restr{m}{\widetilde{\mathcal R}_k},\quad k=1,2,3.
$$
The functions $m_1$, $m_2$ and $m_3$ are analytic in $\C\setminus\Sigma_*$, $\C\setminus(\Sigma_*\cup L)$ and $\C\setminus L$, respectively. Combining with the
condition \eqref{jump_conditions_m_three_cut} we immediately get that
\begin{equation}\label{jumps_first_row}
(m_{1+}(z),m_{2+}(z),m_{3+}(z))=(m_{1-}(z),m_{2-}(z),m_{3-}(z))J_M(z),\quad z\in \Sigma_*\cup L.
\end{equation}
In addition, \eqref{residues_condition} gives
\begin{equation}\label{local_asymptotics_mk_1}
m_k(z)=\Boh\left((z-z_j)^{-1/4}\right),\quad z\to z_j,\quad k=1,2,
\end{equation}
and also
\begin{equation}\label{local_asymptotics_mk_2}
m_k(z)=\Boh(z^{-3/4}),\quad z\in \infty, \quad k=2,3.
\end{equation}
Furthermore, since $u(\infty^{(1)})=0$, we also have 
\begin{equation}\label{local_asymptotics_mk_3}
m_1(z)=1+\Boh(z^{-1}),\quad z\to\infty.
\end{equation}

In summary, \eqref{jumps_first_row}--\eqref{local_asymptotics_mk_3} tell us that the row vector $(m_1,m_2,m_3)$ satisfies the 
conditions for the first row of $M$. 

To construct the remaining rows of $M$, consider a basis $f_1\equiv 1,f_2,f_3 $ of the vector space of functions analytic on $\mathcal{\widetilde R}\setminus\{\infty^{(2)}\}$, 
with at most a double pole at $\infty^{(2)}$. Denote $f_{j,k}=\restr{f_j}{\widetilde{\mathcal R}_k}$ and consider the auxiliary matrix
\begin{equation}\label{definition_prefactor_B}
B=
\begin{pmatrix}
 m_1 		& m_2 		& m_3 \\
 m_1 f_{2,1}	& m_2 f_{2,2}	& m_3 f_{2,3} \\
 m_1 f_{3,1}	& m_2 f_{3,2}	& m_3 f_{3,3}
\end{pmatrix}.
\end{equation}
Using \eqref{jumps_first_row}, we learn
\begin{equation}\label{jumps_matrix_B}
B_+(z)=B_-(z)J_M(z), \quad z\in \Sigma_*\cup L.
\end{equation}
Furthermore, it follows from the analyticity of the $f_{j,k}$'s near finite points, and also the local behavior \eqref{local_asymptotics_mk_1}--\eqref{local_asymptotics_mk_3}, 
that $B$ satisfies the endpoint conditions for $M$.

Because $\det J_M=1$, we also learn from \eqref{jumps_matrix_B} that $\det B$ is entire. Furthermore, a simple analysis of its entries shows that 
$B(z)=\Boh(z^{1/4})$ as $z\to \infty$. Since the functions $f_1,f_2$ and $f_3$ are linearly independent, this is enough to show that $\det B$ is equal to a non-zero constant. In 
particular, $B$ is always invertible.

By inspection one can see that the function $A$ in \eqref{definition_asymptotic_matrix_A} satisfies $A_+=A_-J_M$ on $L$, and using \eqref{jumps_matrix_B} we thus conclude that 
$B A^{-1}$ is analytic on $\C\setminus \Sigma_*$. As $A(z)=\Boh(z^{1/4})$ and $B(z)=\Boh(z^{1/4})$, we see that $BA^{-1}$ is bounded near $\infty$, and thus admits a series 
expansion of the form
$$
(BA^{-1})(z)=C + \Boh(z^{-1}),\quad z\to \infty,
$$
for some constant matrix $C$, which is non-singular because $\det A,\det B \neq 0$. We already observed that $B$ satisfies \eqref{jumps_matrix_B} and also the endpoint 
conditions for $M$. It thus finally follows that
$$
M(z)=C^{-1}B(z),\quad z\in \C\setminus (\Sigma_*\cup L)
$$
is the desired global parametrix.

\subsection{Construction of the global parametrix in the one-cut case}\label{global_parametrix_one_cut}

The Riemann-Hilbert problem for the global parametrix in the one-cut case assumes the following form.

\begin{itemize}
 \item $M:\C\setminus(\Sigma_*\cup L)\to \C^{3\times 3}$ is analytic;
 
 \item $M_+(z)=M_-(z)J_M(z)$, $z\in \Sigma_*\cup L$, where $J_M$ is defined as in \eqref{jump_matrix_global_parametrix}.
 
 \item $M(z)=\Boh((z-z_j)^{-1/4})$ as $z\to z_j$, $j=0,1,2$.
 
 \item $M(z)=(I+\Boh(z^{-1}))A(z)$, as $z\to \infty$.
\end{itemize}

Following the ideas carried out in Section \ref{section_global_parametrix_three_cut}, we start the construction of $M$ from its first row. 

As in Section \ref{section_global_parametrix_three_cut}, there exists a meromorphic differential $\eta$ on $\widetilde{\mathcal R}$ uniquely defined through the conditions that it 
has simple poles at the branch points $z_0^{(1)}$, $z_1^{(1)}$, $z_2^{(2)}$ and $\infty^{(2)}$, with residues
\begin{equation}\label{residue_conditions_one_cut}
\res(\eta, z_0^{(1)})=\res(\eta, z_1^{(1)})=\res(\eta, z_2^{(2)})=-\frac{1}{2},\quad \res(\eta, \infty^{(2)})=\frac{3}{2},
\end{equation}
and no other poles.

We then consider the function
$$
u(p)=\int_{\infty^{(1)}}^p \eta,\quad p\in \widetilde{\mathcal R}\setminus (\Sigma_{*,+}^{(1)}\cup L_{0,+}^{(2)}).
$$
The image of $\Sigma_{*,+}^{(1)}\cup L_{0,+}^{(2)}$ through $\psi$ is the set $\Xi\cup \Lambda_0$, which is shown in the left panel of Figure 
\ref{figure_mother_body_w_plane_jumps_one_cut}. From this figure and conformal equivalence, it easily follows that the set $\widetilde{\mathcal R}\setminus (\Sigma_{*,+}^{(1)}\cup 
L_{0,+}^{(2)})$ is not simply connected, and thus $u(p)$ depends on the path of integration chosen. However, in virtue of \eqref{residue_conditions_one_cut}, it follows after a 
residue calculation that the value $u(p)$ is well defined modulo $2\pi i$. Having this in mind, it also holds true
\begin{equation}\label{jumps_u_one_cut}
u_+(p)-u_-(p)=\pi i \mod 2\pi i,\quad p\in \Sigma_{*,+}^{(1)}\cup L_{0,+}^{(2)}.
\end{equation}
We then define
$$
\widetilde m(p)=e^{u(p)},\quad p\in \widetilde{\mathcal R}\setminus \Sigma_{*,+}^{(1)}\cup L_{0,+}^{(2)},
$$
and also
$$
m(p)=
\begin{cases}
 -\widetilde m(p), & p\in \mathcal W_{2,2}\cup \mathcal W_{3,1} \\
 \widetilde m(p), & \mbox{elsewhere}.
\end{cases}
$$

\begin{figure}[t]
\begin{minipage}[c]{0.5\textwidth}
\centering
  \begin{overpic}[scale=1]
  {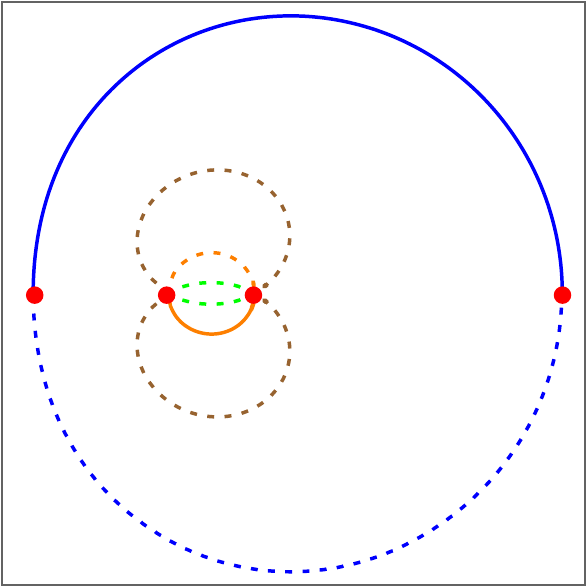}
  \put(83,83){$ \Xi$}
  \put(29,39){$\Lambda_0$}
 \end{overpic}
\end{minipage}%
\begin{minipage}[c]{0.5\textwidth}
\centering
\begin{overpic}[scale=1]
{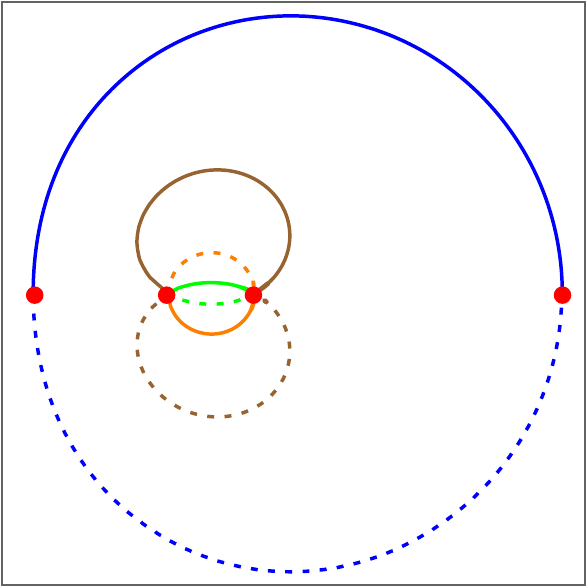}
\put(83,83){$ \Xi$}
\put(29,39){$\Lambda_0$}
\put(35,53){$\Lambda_1$}
\put(42,71){$\Lambda_2$}
 \end{overpic}
 \end{minipage}
 \caption{The solid lines are the images through $\psi$ of the jumps for $\tilde m$ and $m$ in the left and right panels, respectively. The dashed lines are the images of the 
remaining cuts. The solid lines in the right panel also correspond to the jump contours for $f$ in the one-cut case. Numerical output for $r=1/20$ and 
$a_0=1/4$.}\label{figure_mother_body_w_plane_jumps_one_cut}
 \end{figure}

The functions $\tilde m$ and $m$ are analytic on $\widetilde{\mathcal R}\setminus (\Sigma_{*,+}^{(1)}\cup L_{0,+}^{(2)})$ and $\widetilde{\mathcal R}\setminus 
(\Sigma_{*,+}^{(1)}\cup 
L_+^{(2)})$, respectively. Furthermore, from \eqref{jumps_u_one_cut}
$$
\widetilde m_+(p)=-\widetilde m_-(p),\quad p\in \Sigma_{*,+}^{(1)}\cup L_{0,+}^{(2)},
$$
and consequently, from the definition of $m$ we get 
\begin{equation}\label{jump_conditions_m_one_cut}
m_+(p)=-m_-(p),\quad p\in \Sigma_{*,+}^{(1)}\cup L_{+}^{(2)}.
\end{equation}
We refer the reader to Figure \ref{figure_mother_body_w_plane_jumps_one_cut} for a display of the jump contours for $\widetilde m$ and $m$ in the $w$-plane.

Further setting
$$
m_j=\restr{m}{\widetilde{\mathcal R}_j},\quad j=1,2,3,
$$
it follows in the same way as we did in Section \ref{section_global_parametrix_three_cut}, equations \eqref{residues_condition}--\eqref{local_asymptotics_mk_3}, that the first row 
of the global parametrix $M$ is $(m_1,m_2,m_3)$. The remaining rows of $M$ can be obtained in exactly the same way as we did in Section \ref{section_global_parametrix}, equation 
\eqref{definition_prefactor_B} {\it et seq}. We skip the details.

\subsection{Explicit construction of the first row}

In Sections \ref{section_global_parametrix_three_cut} and \ref{global_parametrix_one_cut}, we constructed the first row of the global parametrix in terms of the function $m$, 
which, in virtue of \eqref{jump_conditions_m_three_cut}--\eqref{local_asymptotics_mk_3} (see also \eqref{jump_conditions_m_one_cut}) is unique solution to the following 
Riemann-Hilbert problem.
\begin{itemize}
 \item $m:\widetilde{\mathcal R}\setminus(\Sigma_{*,+}^{(1)}\cup L_+^{(2)} )\to \C$ is analytic;
 
 \item $m_+(z)=-m_-(z)$, $\; z\in \Sigma_{*,+}^{(1)}\cup L_+^{(2)}$;
 
 \item If $z_j$ is a branch point for the sheet $\widetilde{\mathcal R}_k$, then
 \begin{equation}\label{asymptotics_local_coordinate_1}
 m_k(z)=\Boh((z-z_k)^{-1/4}), \quad z \to z_j,
 \end{equation}
 and as $z\to \infty$,
 \begin{equation}\label{asymptotics_local_coordinate_2}
 m_1(z)=1+\Boh(z^{-1}),\quad m_k(z)=\Boh(z^{-3/4}), \; k=2,3.
 \end{equation}
 
 \item In the three-cut case, $m_k(z)$ remains bounded as $z\to z_*$.
 \end{itemize}

It turns out that the Riemann-Hilbert problem above can be solved explicitly with the help of the rational parametrization $h$ and its inverse $\psi$ in 
\eqref{inverse_rational_parametrization}. If we seek for $m$ of the form
\begin{equation}\label{explicit_expression_first_row}
f(w)=m(h(w)),
\end{equation}
then it follows that $f$ should satisfy the following scalar Riemann-Hilbert problem.
\begin{itemize}
 \item $f:\C\setminus (\Xi\cup \Lambda)\to \C$ is analytic;
 
 \item $f_+(w)=-f_-(w)$, $w\in \Xi\cup \Lambda$;
 
 \item $f(w)\to 1$ as $w\to \infty$ and $f(w)\to 0$ as $w\to 0$.
 
 \item $f(w)=\Boh((w-w_j)^{-1/2})$ as $w\to w_j$, $j=0,1,2$.
\end{itemize}

We remark that the $1/4$-blow-ups for $m$ become $1/2$-blow-ups because the points $z_0, z_1$ and $z_2$ are branch points.

The jumps for $f$ are shown in Figure \ref{figure_mother_body_w_plane_jumps} and in the right panel of Figure \ref{figure_mother_body_w_plane_jumps_one_cut} for the three-cut and 
one-cut cases, respectively. 

Thus the natural choice for $f$ is
$$
f(w)=\left(\frac{w^3}{(w-w_0)(w-w_1)(w-w_2)}\right)^{1/2},\quad w\in \C\setminus(\Xi\cup \Lambda),
$$
where the branch of the square root is uniquely determined by the condition that $f(w)\to 1$ as $w\to \infty$ and with branch cuts on $\Xi\cup \Lambda$. 

We know that $w_0,w_1$ and $w_2$ are the zeros of $h'$ (see Lemma \ref{lemma_location_zeros_derivative_h}), and consequently of the monic polynomial $\tilde h(w)=w^3 r^{-1}h'(w)$ 
as in \eqref{equation_zeros_derivative_h}. This means that
$$
\frac{w^3}{r}h'(w)=(w-w_0)(w-w_1)(w-w_2),
$$
which expresses that
\begin{equation}\label{explicit_expression_f}
f(w)=\left( \frac{r}{h'(w)} \right)^{1/2}.
\end{equation}
From the Inverse Function Theorem, we know that $h'(w)=1/\psi'(z)$, where $\psi$ is given in \eqref{inverse_rational_parametrization} and $w=\psi(z)$. Returning back to 
\eqref{explicit_expression_f} and using \eqref{explicit_expression_first_row}, we thus get that $m$ is given by
$$
m(z)=\sqrt{r\psi'(z)},\quad z\in \widetilde{\mathcal R}\setminus (\Sigma_{*,+}^{(1)}\cup L_+^{(2)}).
$$
Recalling \eqref{inverse_rational_parametrization}, we finally arrive at the expressions for the first line $(m_1,m_2,m_3)$ of $M$, namely
\begin{equation}\label{explicit_first_row}
m_j(z)=\sqrt{r \psi'_j(z)}, \quad j=1,2,3,
\end{equation}
where the branch cuts for the square root of $\psi_1$, $\psi_2$ and $\psi_3$ are determined from the ones in \eqref{explicit_expression_f}. In particular, the 
branch cut for $m_1$ is taken on $\Sigma_*$.

\section{Proof of Theorems~\ref{theorem_limiting_counting_measure} and \ref{thm_wave_factor}}\label{section_proof_theorem_limiting_counting_measure}

We now prove Theorems \ref{theorem_limiting_counting_measure} and \ref{thm_wave_factor}. The arguments are valid both in the three-cut and one-cut cases.

\begin{proof}[Proof of Theorem \ref{thm_wave_factor}]

Unfolding the transformations in the Riemann-Hilbert analysis, we get in particular
\begin{align*}
 P_{n,n}(z) & = Y_{1,1}(z) \\
	    & = X_{1,1}(z) \\
	    & = T_{1,1}(z)e^{\frac{n}{t_0}(g_1(z)-V(z)-l_1)},\quad z\in \C\setminus \Sigma,
\end{align*}
we refer to \eqref{P_nn_Y}, \eqref{transformation_X} and \eqref{transformation_T} for the three-cut case, and remind that these transformations are the same in the one-cut case 
(with the appropriate definition of the function $g_1$). For any fixed compact $K\subset \C\setminus \Sigma_*$, we can reduce the lens $\mathscr S$ and 
the set $D_\delta$ in such a way that
$$
K\cap (\overline{\mathscr S} \cup \overline{D_\delta})=\emptyset,
$$
and in this case it follows further that $T_{1,1}=S_{1,1}$ on $K$ (see for instance \eqref{transformation_S_1} and \eqref{transformation_S_2}), so
\begin{equation}\label{pre_asymptotics_P_nn}
P_{n,n}(z)=S_{1,1}(z)e^{-\frac{n}{t_0}(g_1(z)-V(z)-l_1)},\quad z\in K.
\end{equation}

From \eqref{transformation_R} and the estimate \eqref{estimate_R}, we know that as $n\to\infty$
\begin{align*}
S_{1,1}(z) & = R_{1,1}(z)M_{1,1}(z) + R_{1,2}(z)M_{2,1}(z) + R_{1,3}(z)M_{3,1}(z) \\
	   & = (1+\Boh(n^{-1}))M_{1,1}(z), \quad z\in K.
\end{align*}
where for the last equality we also used that the first column of $M$ remains bounded away from $\Sigma_*$, which is a direct consequence of the RHP satisfied by $M$.
Also note that the implicit term above is uniform on $K$. Returning this last equation into 
\eqref{pre_asymptotics_P_nn}, we conclude
\begin{equation}\label{asymptotics_P_nn}
P_{n,n}(z)=(1+\Boh(n^{-1}))M_{1,1}(z)e^{\frac{n}{t_0}(g_1(z)-V(z)-l_1)}
\end{equation}
uniformly on the compact $K\subset \C\setminus \Sigma_*$. 

From \eqref{explicit_first_row} and from the definition of $g_1$ in \eqref{g_functions} and \eqref{g_functions_postcritical}, it 
immediately follows that
$$
M_{1,1}(z)=\sqrt{r\psi_1'(z)},\quad g_1(z)=G(z)+c_1,
$$
where $G$ is as in \eqref{definition_factor_G}. Returning this information back to \eqref{asymptotics_P_nn}, we get \eqref{uniform_asymptotics_polynomial} for the constant 
$c=c_1-l_1$.
\end{proof}

\begin{proof}[Proof of Theorem \ref{theorem_limiting_counting_measure}]
Since $\psi_1$ is the inverse of $h$ on $\widetilde{\mathcal R}_1$, the derivative $\psi_1'$ does not vanish on $\C\setminus\Sigma_*$, so from 
\eqref{uniform_asymptotics_polynomial} we conclude that the zeros of $P_{n,n}$ accumulate on the star $\Sigma_*$ in the large $n$ limit.

Suppose now that $\mu_{n_k}\stackrel{*}{\to}\nu$, where $(\mu_{n_k})$ is a subsequence of the sequence of zero counting measures $(\mu_n)$ defined in \eqref{counting_measures}. 
The zeros of $P_{n,n}$ accumulate on $\Sigma_{*}$, so we must have
\begin{equation*}\label{support_arbitrary_limiting_measure}
 \supp \nu\subset \Sigma_*.
\end{equation*}

For any $z\in \C\setminus \Sigma_*$, it follows from \eqref{asymptotics_P_nn} that
\begin{multline*}
U^{\nu}(z)=-\int \log|s-z|d\nu(s)\\ =-\lim_{k\to\infty} \frac{1}{n_k}\log|P_{n_k,n_k}(z)|=-\frac{1}{t_0}\re(g_1(z)-V(z)-l_1).
\end{multline*}
Having in mind \eqref{identity_cauchy_transform_potential} and \eqref{cauchy_transform_mother_body}, this last identity implies that
\begin{equation*}
C^{\nu}(z)=-\frac{1}{t_0}(g_1'(z)-V'(z))=-\frac{1}{t_0}(\xi_1(z)-V'(z))=C^{\mu_*}(z),\quad z\in \C\setminus \Sigma_*.
\end{equation*}
Using the same arguments as in \eqref{equality_cauchy_transforms} {\it et seq.}, we thus conclude
$$
U^{\nu}(z)=U^{\mu_*}(z),\quad z\in \C\setminus \Sigma_*.
$$
Since $\Sigma_*$ has planar Lebesgue measure zero, the above equation says that the potential of the measures $\nu$ and $\mu_*$ coincide a.e. in $\C$. From the Unicity 
Theorem~\cite[Theorem~II.2.1]{Saff_book} we get $\nu=\mu_*$, concluding the proof.
\end{proof}


\appendix


\section{Analysis of the width parameters}\label{appendix_widths}

In the appendix, we analyze the parameters $\tau_j$'s that were used in Section~\ref{section_quadratic_differential}.
To do so, we need some preliminary lemmas.

\begin{lem}
 For $(t_0,t_1)\in \mathcal F$, it is valid
 \begin{equation}\label{lower_bound_w_1_1}
 w_0>r^{1/3}.
 \end{equation}
 
 Additionally, for $(t_0,t_1)\in \mathcal F_2$,
 \begin{equation}\label{lower_bound_w_1_r}
 |w_1|<r^{1/3},
 \end{equation}
 and 
 \begin{equation}\label{lower_bound_w_1_r_2}
 w_1>-\frac{2a_0}{3r},
 \end{equation}
 and consequently,
 \begin{equation}\label{lower_bound_w_1_2}
 |w_1|<w_0.
 \end{equation}
\end{lem}
\begin{proof}
Recall that $w_0$ is the unique positive solution to
$$
h'(w)=0,
$$
so $h'(w)<0$ for positive $w$ only if $w<w_0$. Simple calculations then show
 $$
 h'(r^{1/3})=r-2a_0r^{1/3}-2r^{2/3}\leq r-2r^{2/3}<0,\quad 0<r<\frac{1}{2}.
 $$
giving us \eqref{lower_bound_w_1_1}.

For the second inequality, we recall that $w_1$ is the smallest (negative) root of $h'$ and, furthermore, $h'(w)\to r>0$ as $w\to-\infty$, so that $h(w)>0$ on the interval 
$(-\infty,w_1)$. Simple computations show that
$$
h'(-r^{1/3})=r\left(1-\frac{2a_0}{r^{2/3}}\right).
$$
For fixed $a_0$, the function $r\mapsto 1-2a_0/r^{2/3}$ is increasing, so it attains its maximum when $r$ is chosen so that the corresponding pair $(t_0,t_1)$ belongs to the 
critical curve $\gamma_c$. Using \eqref{critical_curve_limiting_zero_distribution_a_r_plane}, we see that this maximum is 
$$
1-\frac{3s^2}{s^2}=-2<0,
$$
thus we get that $h'(-r^{1/3})<0$. Since we already observed that $h'$ is positive on $(-\infty,w)$, this is enough to conclude that $w_1<-r^{1/3}<0$, which is equivalent to 
\eqref{lower_bound_w_1_r}. 

To get \eqref{lower_bound_w_1_r_2}, we note that the function 
$$
a_0\mapsto -\frac{2a_0}{3r}
$$
is decreasing, so it attains its maximum value along $\gamma_c$. Recalling that $w_1>-1$ (see Lemma~\ref{lemma_zeros_derivative_h}) and using 
\eqref{critical_curve_zeros_r_a_plane_negative_t}, we get 
$$
-\frac{2a_0}{3r}<-\frac{2 (3s^2/2)}{s^3}=-\frac{1}{s}<-1<w_1.
$$

Finally, the inequality \eqref{lower_bound_w_1_2} trivially follows from \eqref{lower_bound_w_1_1}--\eqref{lower_bound_w_1_r}.
\end{proof}

\begin{lem}
Suppose $(t_0,t_1)\in\mathcal F_2$. Then
\begin{equation}\label{lower_bound_sum_inverse_roots}
\frac{1}{w_1}+\frac{1}{w_2}>-1
\end{equation}
where $w_0$ and $w_1$ are the zeros of $h'$ as in \eqref{lemma_zeros_derivative_h}.
\end{lem}
\begin{proof}
From the explicit expression of $h$ in \eqref{rational_parametrization}, we trivially have
\begin{equation}\label{aux_equation_31}
h''(w)=\frac{4a_0r}{w^3}+\frac{6r^2}{w^4}=\frac{2r}{w^4}(3rw+2a_0)
\end{equation}
Since $h'(w_j)=0$, the chain rule gives us
\begin{equation*}
 \frac{\partial w_j}{\partial w}=-\frac{\frac{\partial h'}{\partial a_0}(w_j)}{h''(w_j)}=\frac{2r}{w_j^2 h''(w_j)}.
\end{equation*}
Using \eqref{aux_equation_31}, we thus get 
\begin{align}
\frac{\partial}{\partial a_0}\left(\frac{1}{w_0}+\frac{1}{w_1}\right) & = -\left(\frac{1}{w_0^2}\frac{\partial w_0}{\partial a_0}+ \frac{1}{w_1^2}\frac{\partial w_1}{\partial 
a_0}\right) \nonumber \\
		  & = -\left( \frac{1}{3r w_0 +2a_0} + \frac{1}{3rw_1+2a_0} \right). \label{aux_equation_32}
\end{align}
We know that $w_0>0$ and also $3rw_1+2a_0>0$, as it follows from Lemma~\ref{lemma_location_zeros_derivative_h} and \eqref{lower_bound_w_1_r_2}, respectively. From 
\eqref{aux_equation_32} we thus conclude that the function
$$
a_0\mapsto \frac{1}{w_0}+\frac{1}{w_1}
$$
is decreasing, so it attains its minimum along the critical curve $\Gamma_c$. On $\Gamma_c$, it follows from \eqref{cusp_critical_curve_positive_t1} that
$$
h'(w)=\frac{s (w-1) \left(2 s+w^2+w\right)}{w^3},
$$
where $s\in (0,1/8)$, so that in this case
$$
w_0=1,\quad w_1=\frac{1}{2} \left(-1-\sqrt{1-8 s}\right),
$$
and consequently for every choice of parameters in $\mathcal F_2$, it holds true
$$
\frac{1}{w_0}+\frac{1}{w_1}> 1-\frac{2}{1+\sqrt{1-8 s}}>-1,
$$
as we want.
\end{proof}

\subsection{Width parameters in the three-cut case}\label{appendix_widths_precritical}

Recall that the non vanishing of the parameters $\tau_j$, $j=1,2,3,4,5$, introduced in \eqref{width_tau_1}--\eqref{width_tau_5}, were used in Section 
\ref{deformation_critical_graph_precritical} to prove that the critical graph of the quadratic differential $\varpi$ remains unchanged in $\mathcal F_1$. We now verify that these 
quantities do not vanish. 

\begin{prop}
For $(t_0,t_1)\in \mathcal F_1$, we have $\tau_5<0$.
\end{prop}
\begin{proof}
Follows directly from \eqref{equalities_inequalities_real_line_xi_functions_precritical_2}.
\end{proof}

\begin{prop}\label{proposition_tau_1}
 For $(t_0,t_1)\in \mathcal F_1$, we have $\tau_1>0$.
\end{prop}
\begin{proof}
The rational parametrization $(\xi,z)=(h(w^-1),h(w))$ given by Theorem \ref{theorem_rational_parametrization_polynomial_curve} induces the change of variables $s=h(w)$, 
$\xi_j=h(w^{-1})$, from which it follows that
$$
\int_{z_2}^{z_0}\xi_1\; ds=\int^{w_0}_{w_2}h\left(\frac{1}{w}\right)h'(w)dw,\quad \int_{z_2}^{z_0}\xi_2\; ds=\int^{\tilde w_0}_{w_2}h\left(\frac{1}{w}\right)h'(w)dw,
$$
where $w_2, w_0, \tilde w_0$ satisfy
\begin{align}
& z_2=h(w_2), && \xi_1(z_2)=h(w_2^{-1})=\xi_2(z_2),  \\
& z_0=h(w_0)=h(\tilde w_0), && \xi_1(z_0)=h(w_0^{-1}), \qquad \xi_2(z_0)=h(\tilde w_0^{-1}). \label{aux_equation_8}
\end{align}

Hence,
\begin{equation}\label{last_conclusion}
\tau_1=\re \int_{z_2}^{z_0}(\xi_1(s)-\xi_2(s))ds=\re \int_{\tilde w_0}^{w_0}h\left(\frac{1}{w}\right)h'(w)dw.
\end{equation}

We can further simplify the integral above in the following way,
\begin{align*}
 \int_{\tilde w_0}^{w_0}h\left(\frac{1}{w}\right)h'(w)dw & =  \int_{\tilde w_0}^{w_0}\left(h\left(\frac{1}{w}\right)-a_0\right)h'(w)dw + a_0 \int_{\tilde w_0}^{w_0}h'(w)dw \\
	      & =  \int_{\tilde w_0}^{w_0}\left(h\left(\frac{1}{w}\right)-a_0\right)h'(w)dw,
\end{align*}
where in the last step we used the first equation in \eqref{aux_equation_8}. 

We use the definition of $h$ in \eqref{rational_parametrization} to compute explicitly the integral above, arriving at
\begin{equation}\label{aux_equation_9}
\int_{\tilde w_0}^{w_0}h\left(\frac{1}{w}\right)h'(w)dw = F(w_0)-F(\tilde w_0)+r^2(1-4 a_0^2 -2 r^2)(\log w_0-\log \tilde w_0),
\end{equation}
where
$$
F(w)=\frac{r^3}{3}w^3 + r^2 a_0 w^2 - 2 r^3 a_0 w + \frac{4r^3 a_0 }{w} +\frac{r^2 a_0}{w^2} + \frac{2 r^3}{3 w^3}.
$$

The next step is to express $\tilde w_0$ in terms of $w_0$. The equation
$$
h(w)-z_0=\frac{r}{w^2}(w^3+\frac{a_0-z_0}{r}w^2+2a_0 w +r)=0
$$
has $w_0$ as a solution with double multiplicity and $\tilde w_0$ as a simple solution, that is
$$
w^3+\frac{a_0-z_0}{r}w^2+2a_0 w +r=(w-w_0)^2(w-\tilde w_0).
$$

This gives us the relation
\begin{equation}\label{aux_equation_10}
\tilde w_0=-\frac{r}{w_0^2}.
\end{equation}

After some calculations, we are thus reduced to
\begin{multline*}
F(w_0)-F(\tilde w_0) = \frac{r+w_0^3}{3w_0^6}\left(2w_0^9-3a_0w_0^7+r(r^2-2)w_0^6 +15a_0r^2w_0^5 \right. \\ \left. +3a_0r(1-2r^2)w_0^4 + r^2(2-r^2)w_0^3 
-3a_0r^3w_0^2 +r^5\right)
\end{multline*}

The expression $h'(w_0)=0$ gives us additionally $w_0^3=2a_0w_0+2r$. Replacing every multiple power of $3$ in the expression under brackets above, we get
\begin{multline}\label{aux_equation_21}
F(w_0)-F(\tilde w_0)=\frac{r+w_0^3}{3w_0^6}\left( (30a_0^2r^2 +4a_0^3)w_0^3+ra_0(r^2(12-8a_0)+10a_0)w_0^2\right. \\ \left. +6a_0r^2(5-r^2)w_0 +12r^3+3r^5\right).
\end{multline}

Recalling Theorem \ref{theorem_monotonicity_r}, we know that $0<a_0,r<1$, so the expression between parentheses above is a polynomial in $w_0$ with positive coefficients. Because 
$w_0>0$ (Lemma \ref{lemma_location_zeros_derivative_h}), we finally conclude 
\begin{equation}\label{aux_equation_11}
F(w_0)-F(\tilde w_0)>0.
\end{equation}

We now take care of the log terms in \eqref{aux_equation_9}. From the definition of $a_0$ in \eqref{definition_a_0}, 
$$
0\leq a_0 \leq \frac{1-4r^2}{2}<\frac{1}{2},
$$
and this gives us $4a_0^2 \leq 2a_0$. Having also in mind $r<1$,
\begin{equation}\label{aux_equation_22}
1-4a_0^2-2r^2>1-2r-2a_0>0,
\end{equation}
where in the last step we used \eqref{inequality_r_a_0_cusp}. Using also \eqref{aux_equation_10}, we get
$$
\re \big( (1-4 a_0^2 -2 r^2)(\log w_0-\log \tilde w_0)\big) = (1-4 a_0^2 -2 r^2) \log\frac{w_0^3}{r}>0,
$$
because $w_0^3>r$, see \eqref{lower_bound_w_1_1}. Plugging this last equation and \eqref{aux_equation_11} into \eqref{aux_equation_9}, and having in mind \eqref{last_conclusion}, 
we arrive at the desired result.
\end{proof}

\begin{remark}\label{remark_relation_simple_double_solution}
 Note that \eqref{aux_equation_10} and \eqref{aux_equation_21} also hold if we replace $w_0$ and $\tilde w_0$ by $w_j$ and $\tilde w_j$, respectively, where $w_j$ is a zero of 
$h'(w)=0$ and $\tilde w_j$ is the simple zero of $h(w)-h(w_j)=0$.
\end{remark}

\begin{remark}\label{remark_critical_curve_negative_t1}
The keen reader might notice that a combination of \eqref{aux_equation_9}--\eqref{aux_equation_21} establishes the equivalence between 
\eqref{implicit_equation_critical_curve_negative_t} and \eqref{critical_curve_zeros_r_a_plane_negative_t}. In fact, the mother body phase transition 
determined by $\gamma_c^-$ corresponds to the vanishing of $\tau_1$. Unlike for $t_1>0$, the transition across $\gamma_c^-$ does not 
correspond to the coalescence of critical points of the quadratic differential $\varpi$ (or, equivalently, of the points $z_j$ and $\hat z_j$ given by 
Theorem~\ref{theorem_singular_points_spectral_curve}). Instead, it corresponds to the shrinking of the domains $\mathcal S_1$ and $\mathcal 
S_6$ in Figure~\ref{figure_planar_graph}.
\end{remark}

\begin{prop}\label{proposition_tau_2_tau_3}
 For $(t_0,t_1)\in \mathcal F_1$, we have
 \begin{align}
\tau_2 & >0,\label{inequality_tau_2}\\
\tau_3 & <0. \label{inequality_tau_3}
 \end{align}
\end{prop}

We have not been able to verify Proposition~\ref{proposition_tau_2_tau_3} analytically, so we verified it numerically as explained next.

We start with \eqref{inequality_tau_2}. As in the proof of Proposition~\ref{proposition_tau_1}, we perform the change of variables $z=h(w)$, $\xi_j=h(w^{-1})$, and arrive at
\begin{equation}\label{aux_equation_12}
\int_{z_0}^{z_2}\xi_1(s)ds=H(w_2)-H(w_0), \quad \int_{z_0}^{z_2}\xi_3(s)ds=H(\tilde w_2)-H(w_0),
\end{equation}
where $w_0,w_2$ are as in Lemma~\ref{lemma_location_zeros_derivative_h}, $\tilde w_2$ is the simple zero of $h(w)-z_2$, so alternatively given by
\begin{equation}\label{aux_equation_13}
\tilde w_2=-\frac{r}{w_2^2},
\end{equation}
see Remark~\ref{remark_relation_simple_double_solution}, and $H(w)$ is the primitive of $h'(w)h(w^{-1})$, explicitly given by
\begin{multline}\label{primitive_change_variables}
H(w)=\frac{r^3}{3}w^3+a_0 r^2 w^2+ a_0r\left(1-2 r^2\right) w  \\ 
+\frac{ 2 a_0^2r+4 a_0 r^3}{w}+\frac{2 a_0 r^2}{w^2}+ \frac{2 r^3}{3 w^3}-r^2 \left(4a_0^2+2 r^2-1\right) \log w
\end{multline}

In the expression above, we choose the main branch of the logarithm - actually the branch chosen is not important, because at the end we will be only interest in the real 
part of $H$.
Taking the difference between the two expressions in \eqref{aux_equation_12}, the integral in Equation~\eqref{inequality_tau_2} gets the form
\begin{align}
\tau_2=\re\int_{z_0}^{z_2}(\xi_1(s)-\xi_3(s))ds &  = \re H(w_2)-\re H(\tilde w_2) \nonumber \\
					 & =\re H(w_2)-\re H\left(-\frac{r}{w_2^2}\right). \label{expression_tau_2_H}
\end{align}

Note that the right hand side of \eqref{expression_tau_2_H} is given only in terms of $a_0,r$. We then use \eqref{expression_tau_2_H} for numerical computation of the integral as 
follows.

For given $r,a_0$, we first solve
$$
h'(w)=0,
$$
pick $w_2$ as the only solution with positive imaginary part (see Lemma~\ref{lemma_location_zeros_derivative_h}), compute $\tilde w_2$ through \eqref{aux_equation_13} and 
finally get the difference $\re H(w_2)-\re H(\tilde w_2)$. By varying $r\in (0,1/2)$ and $a_0\in (0,\alpha)$, where
$$
\alpha=\alpha(r)=\min\{ 3/2r^{2/3}, (1-2r)/2 \}=
\begin{cases}
 \frac{3}{2}r^{2/3},& r\leq \frac{1}{8},\\
 \frac{1-2r}{2}, & \frac{1}{8}<r<\frac{1}{2},
\end{cases}
$$
we are sure to be covering every possible choice $(t_0,t_1)\in \mathcal F_1$ (see Proposition \ref{proposition_change_of_coordinates}).

With this idea in mind, we evaluated \eqref{inequality_tau_2} numerically with Mathematica in $300$-digit precision for the range
\begin{equation}\label{range_numerical_computation}
r=\frac{1}{2}\frac{j}{5000},\quad a_0=\frac{\alpha(r)}{5000}k,\quad j,k=1,\hdots,5000,
\end{equation}
verifying that in this case $\tau_2>0$.

For \eqref{inequality_tau_3} we proceed similarly as before to get
\begin{equation*}
\tau_3=\re\int_{z_2}^{\hat z_2} (\xi_1(s)-\xi_2(s))ds=\re H(\hat w_2)-\re H(\hat w_2^{-1})
\end{equation*}
where $\hat w_2$ is the parameter on the $w$-plane for which $\xi_1(\hat z_2)=h(\hat w_2^{-1})$, $\xi_2(\hat z_2)=h(\hat w_2)$. Recalling Corollary~\ref{corollary_zeros_f}, $\hat 
w_2$ is alternatively characterized as the only zero of the function $f$ appearing in \eqref{equation_polynomial_f_original}--\eqref{equation_polynomial_f} that belongs to 
$\{w\in\C \mid \im w>0, |w|>1\}$. Note that the coefficient $t_0-2t_1$ of $f$ in \eqref{equation_polynomial_f_original}--\eqref{equation_polynomial_f} can be written only in terms 
of $r,a_0$ with the help of the system \eqref{system_change_coordinates_a}--\eqref{system_change_coordinates_b}. 

So the numerical procedure here is to find all the zeros of $f$, select $\hat w_2$ and then compute the left hand side of \eqref{inequality_tau_3} through \eqref{aux_equation_10}. 
\eqref{inequality_tau_3} was again evaluated in the range \eqref{range_numerical_computation} and $300$-digit precision, and we verified that in this case $\tau_3<0$.

The outcome of the numerical evaluation of $\tau_2$ and $\tau_3$ for several choices of $r$ can be seen in 
Figures~\ref{figure_tau_2_three_cut_1}--\ref{figure_tau_2_three_cut_3} and Figures~\ref{figure_tau_3_three_cut_1}--\ref{figure_tau_3_three_cut_3}, respectively.

\begin{prop}
 For $(t_0,t_1)\in \mathcal F_1$ it is valid
 $$
\tau_4>0.
 $$
\end{prop}
\begin{proof}
From \eqref{asymptotics_xi} we see that the residue at infinity of $\xi_1$ is purely real imaginary, thus 
 \begin{equation}\label{aux_equation_14}
 \re \int_{z_2}^{z_*}\xi_1(s)ds+\re\int_{z_*}^{z_1}\xi_1(s)ds=\re \int_{z_2}^{z_1}\xi_1(s)ds,
 \end{equation}
 where, as in \eqref{width_tau_4}, $x_*<z_*$ and on both sides of \eqref{aux_equation_14} the paths of integration are taken in $\C\setminus((-\infty,z_*]\cup \Sigma_*)$.
 
 For the remaining integral in \eqref{width_tau_4}, we deform the path of integration across $\Sigma_*$ to get 
 \begin{equation}\label{aux_equation_15}
 \int_{z_2}^{z_*}\xi_2(s)ds+\int_{z_*}^{z_1}\xi_3(s)ds=\int_{z_2}^{z_0}\xi_1(s)ds+\int_{z_0}^{z_1}\xi_3(s)ds.
 \end{equation}
 
 Combining \eqref{aux_equation_14} and \eqref{aux_equation_15}, we obtain
\begin{align*}
\tau_4=\re\int_{z_2}^{z_*}(\xi_1(s)-\xi_2(s))ds+\re\int_{z_*}^{z_1}(\xi_1(s)-\xi_3(s))ds
& = \re \int_{z_0}^{z_1}(\xi_1(s)-\xi_3(s))ds \\
& = \re \int_{z_0}^{z_2}(\xi_1(s)-\xi_3(s))ds \\
& = \tau_2
\end{align*}
where for the third equality we used the symmetry under conjugation. From \eqref{inequality_tau_2} we get the desired result.
\end{proof}

\subsection{Width parameters in the one-cut case}\label{appendix_widths_supercritical}

We now proceed to the analysis of the $\tau_j$'s in \eqref{width_tau_1_post}--\eqref{width_tau_6_post}.

\begin{prop}
 For $(t_0,t_1)\in\mathcal F_2$, the quantities $\tau_1$, $\tau_2$, $\tau_4$ and $\tau_5$, given respectively by \eqref{width_tau_1_post}, \eqref{width_tau_2_post}, 
\eqref{width_tau_4_post} and 
\eqref{width_tau_5_post}, are never zero.
\end{prop}
\begin{proof}
 Each of the integrals can be deformed to either one of the intervals $[z_2,z_1]$ or $[z_0,\hat z_0]$, where the respective integrand $\xi_j-\xi_k$ is real, continuous and never 
zero (see \eqref{equalities_inequalities_real_line_xi_functions_supercritical}), and hence does not change sign. 
\end{proof}

\begin{prop}
 For $(t_0,t_1)\in\mathcal F_2$, the quantity $\tau_6$ given in \eqref{width_tau_6_post} is strictly negative.
\end{prop}
\begin{proof}
Recall that $w_0,w_1$ and $w_2$ are the zeros of $h'$ (see Lemma~\ref{lemma_location_zeros_derivative_h}) and $\tilde w_j$ is the simple solution to $h(w)-h(w_j)=0$ (see Remark 
\ref{remark_relation_simple_double_solution}). Proceeding in a similar manner as for Proposition~\ref{proposition_tau_1} (see in particular 
\eqref{last_conclusion}--\eqref{aux_equation_21}), we get
\begin{align}
\tau_6 & = \re \int_{w_0}^{w_1}h\left(\frac{1}{w}\right)h'(w)dw-\re\int_{\tilde w_0}^{\tilde w_1}h\left(\frac{1}{w}\right)h'(w)dw \nonumber\\
    & = \frac{r+w_1^3}{3w_1^3}q(w_1)-\frac{r+w_0^3}{3w_0^3}q(w_0) + 3r^2(1-4a_0r^2-2r^2)\log\left|\frac{w_1}{w_0}\right|,\label{width_tau_6_post_eq_1}
\end{align}
where here $q$ is given by
$$
q(w)=\frac{12r^3+3r^5}{w^3}+\frac{6a_0r^2(5-r^2)}{w^2}+\frac{a_0r((12-8a_0)r^2+10a_0)}{w}+30a_0^2r^2+4a_0^3.
$$

From \eqref{lower_bound_w_1_2} and \eqref{aux_equation_22},
\begin{equation}\label{aux_equation_29}
3r^2(1-4a_0r^2-2r^2)\log\left|\frac{w_1}{w_0}\right|<0.
\end{equation}

To deal with the first two terms on the right hand side of \eqref{width_tau_6_post_eq_1}, rewrite
\begin{multline}\label{aux_equation_28}
\frac{r+w_0^3}{3w_0^3}q(w_0)-\frac{r+w_1^3}{3w_1^3}q(w_1) = \\ \left(\frac{r+w_0^3}{3w_0^3}-\frac{r+w_1^3}{3w_1^3}\right)q(w_0)+\frac{r+w_1^3}{3w_1^3}(q(w_0)-q(w_1))
\end{multline}

From the rough estimate $0<a_0,r<1$ (see Theorem \ref{theorem_monotonicity_r}) it follows that the coefficients of $q$ are positive, thus
\begin{equation}\label{aux_equation_26}
q(w_0)>0,
\end{equation}
because $w_0>0$, see Lemma~\ref{lemma_location_zeros_derivative_h}. Furthermore, from \eqref{lower_bound_w_1_r},
\begin{equation}\label{aux_equation_30}
 \frac{r+w_1^3}{3w_1^3}>0.
\end{equation}
Clearly,
\begin{equation}\label{aux_equation_27}
\frac{r+w_0^3}{3w_0^3}-\frac{r+w_1^3}{3w_1^3}=\frac{r(w_0^3-w_1^3)}{w_0^3(-w_1)^3}>0,
\end{equation}
where for the last conclusion we used $w_1<0<w_0$, see Lemma~\ref{lemma_location_zeros_derivative_h}.
Summarizing, a combination of \eqref{aux_equation_29}--\eqref{aux_equation_27} shows that the right-hand side of \eqref{width_tau_6_post_eq_1} is negative if we can 
prove that
\begin{equation}\label{inequality_tau_6_q_w_1_w_0}
q(w_0)-q(w_1)>0.
\end{equation}

To see that \eqref{inequality_tau_6_q_w_1_w_0} holds true, write
\begin{multline}\label{aux_equation_25}
q(w_0)-q(w_1)=(12r^3+3r^5)\left(\frac{1}{w_0^3}-\frac{1}{w_1^3}\right) +\left(\frac{1}{w_0}-\frac{1}{w_1}\right)
\\ \times \left( 6a_0r^2(5-r^2)\left(\frac{1}{w_0}+\frac{1}{w_1}\right) + a_0r((12-8a_0)r^2+10a_0)\right)
\end{multline}

Using again $w_1<0<w_0$,
\begin{equation}\label{aux_equation_24}
\frac{1}{w_0^3}-\frac{1}{w_1^3},\frac{1}{w_0}-\frac{1}{w_1}>0.
\end{equation}
In addition, using \eqref{lower_bound_sum_inverse_roots}
\begin{multline}\label{aux_equation_23}
6a_0r^2(5-r^2)\left(\frac{1}{w_0}+\frac{1}{w_1}\right) + a_0r((12-8a_0)r^2+10a_0) 
\\  
\begin{aligned}
 & \geq  -6a_0r^2(5-r^2)+ a_0r((12-8a_0)r^2+10a_0) \\
 & = a_0r[a_0 \left(10-8 r^2\right)+6 r \left(r^2+2 r-5\right)]
\end{aligned}
\end{multline}

The term between brackets on the right-hand side above is increasing with $a_0$, so it attains its minimum along the critical curve $\gamma_c$. Using 
\eqref{critical_curve_limiting_zero_distribution_a_r_plane}, we get 
$$
a_0 \left(10-8 r^2\right)+6 r \left(r^2+2 r-5\right)=3 s^2 \left(2 s^7-4 s^6+4 s^4-10 s+5\right)>0,
$$
thus the left-hand side of \eqref{aux_equation_23} is positive as well. Combining this with \eqref{aux_equation_25}--\eqref{aux_equation_24}, we conclude 
\eqref{inequality_tau_6_q_w_1_w_0}, and the proof is complete.
\end{proof}

\begin{prop}\label{proposition_width_postcritical_3}
 The width $\tau_3$ in \eqref{width_tau_3_post} is strictly negative.
\end{prop}
As for the Proposition~\ref{proposition_width_postcritical_3}, we verified that $\tau_6<0$ numerically as explained next.

Proceeding as in \eqref{aux_equation_12}--\eqref{expression_tau_2_H} we get 
$$
\int_{z_2}^{\hat z_2}(\xi_1(s)-\xi_2(s))ds=H(\hat w_2)-H\left(\frac{1}{\hat w_2}\right)+H(w_2)-H\left( -\frac{r}{w_2} \right),
$$
where $w_2$ and $\hat w_2$ are given by Lemma~\ref{lemma_location_zeros_derivative_h} and Corollary~\ref{corollary_zeros_f}, respectively (see also 
Remark~\ref{remark_relation_simple_double_solution}), and the function $H$ is given in \eqref{primitive_change_variables}. Thus
\begin{equation}\label{tau_3_post_numerics}
\tau_3=\re H(\hat w_2)-\re H\left(\frac{1}{\hat w_2}\right)+H(w_2)-H\left( -\frac{r}{w_2} \right).
\end{equation}
We use this last expression to verify that $\tau_3<0$ for 
$$ 0<r<\frac{1}{8}, \quad \frac{3}{2}r^{2/3}<a_0<\frac{1-2r}{2},
$$ 
which corresponds to $(t_0,t_1)\in\mathcal F_2$ (see Proposition~\ref{proposition_change_of_coordinates}). 
We evaluated \eqref{tau_3_post_numerics} for
$$
r=\frac{1}{8}\frac{j}{5000},\quad a_0=\frac{3r^{2/3}}{2}\frac{5000-k}{5000}+\frac{1-2r}{2}\frac{k}{5000},\quad j,k=1,\hdots,5000,
$$
using Mathematica with $300$-digit precision and verified that $\tau_3<0$. The outcome for several values of $r$ and the whole corresponding range of $a_0$ can be seen in 
Figures~\ref{figure_tau_3_one_cut_1}--\ref{figure_tau_3_one_cut_2}.

\begin{figure}[t]
  \begin{overpic}[scale=1]
  {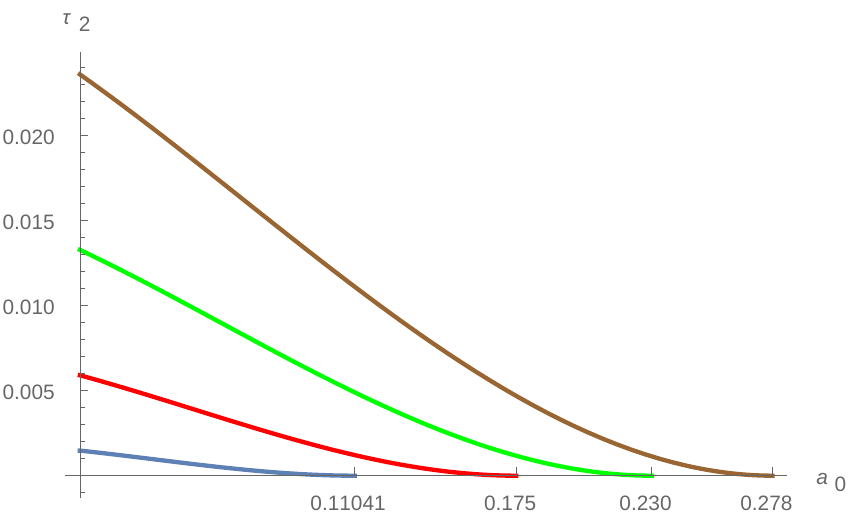}
 \end{overpic}
 \caption{Plot of $\tau_2$ as a function of $a_0$ in the three-cut case for $r=1/50,2/50,3/50,4/50$ (from bottom to top). For these choices of $r$, the extremal values of 
$a_0$ are attained in the critical line $\gamma_c$, so that $a_0$ ranges from $0$ to the correspond critical value $3r^{2/3}/2$, in the present case given by 
$0.1105...,0.1754...,0.2298...$ and $0.2784...$, respectively.}\label{figure_tau_2_three_cut_1}
\end{figure}

\begin{figure}[t]
  \begin{overpic}[scale=1]
  {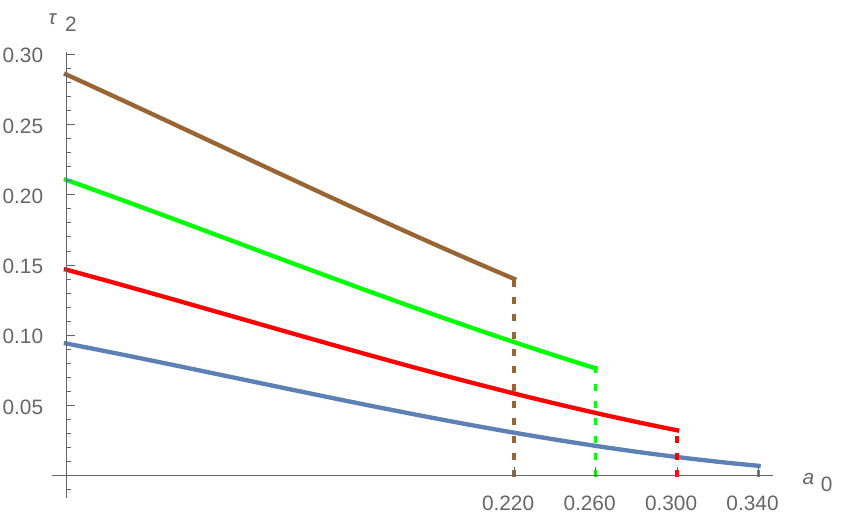}
 \end{overpic}
 \caption{Plot of $\tau_2$ as a function of $a_0$ in the three-cut case for $r=8/50,10/50,12/50,14/50$ (from bottom to top). For these choices of $r$, the extremal values of 
$a_0$ are attained in the critical line $\Gamma_c$, so that $a_0$ ranges from $0$ to the correspond critical value $(1-2r)/2$, in the present case given by 
$0.34,0.3,0.26$ and $0.22$, respectively.}\label{figure_tau_2_three_cut_2}
\end{figure}

\begin{figure}[t]
  \begin{overpic}[scale=1]
  {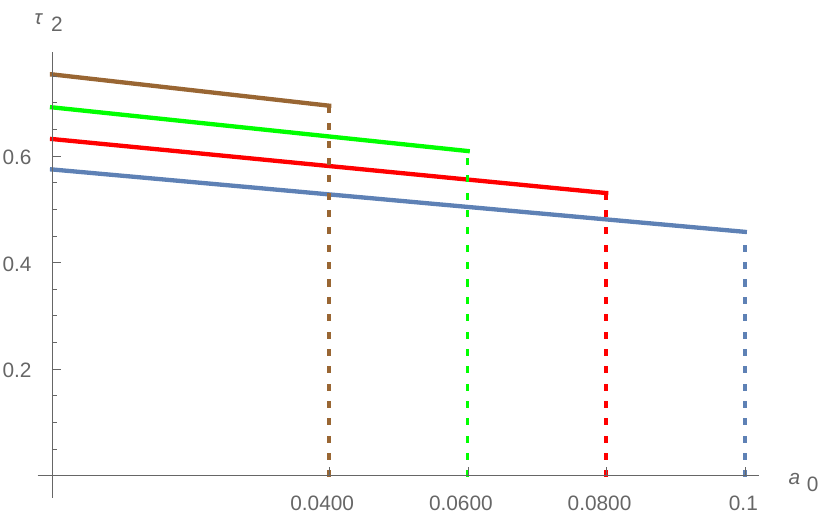}
 \end{overpic}
 \caption{Plot of $\tau_2$ as a function of $a_0$ in the three-cut case for $r=20/50,21/50,22/50,23/50$ (from bottom to top). For these choices of $r$, the extremal values of 
$a_0$ are attained in the critical line $\Gamma_c$, so that $a_0$ ranges from $0$ to the correspond critical value $(1-2r)/2$, in the present case given by 
$0.1,0.08,0.06$ and $0.04$, respectively.}\label{figure_tau_2_three_cut_3}
\end{figure}

\begin{figure}[t]
  \begin{overpic}[scale=1]
  {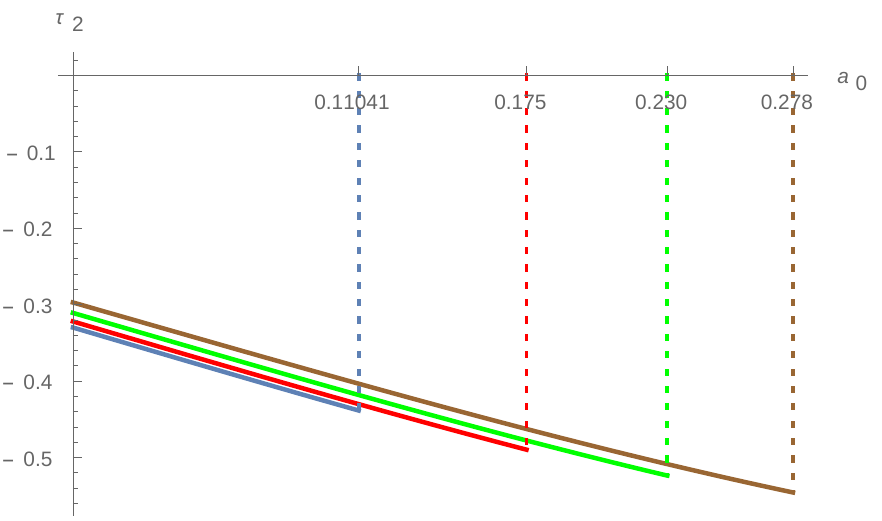}
 \end{overpic}
 \caption{Plot of $\tau_3$ as a function of $a_0$ in the three-cut case for $r=1/50,2/50,3/50,4/50$ (from bottom to top). For these choices of $r$, the extremal values of 
$a_0$ are attained in the critical line $\gamma_c$, so that $a_0$ ranges from $0$ to the correspond critical value $3r^{2/3}/2$, in the present case given by 
$0.1105...,0.1754...,0.2298...$ and $0.2784...$, respectively.}\label{figure_tau_3_three_cut_1}
\end{figure}

\begin{figure}[t]
  \begin{overpic}[scale=1]
  {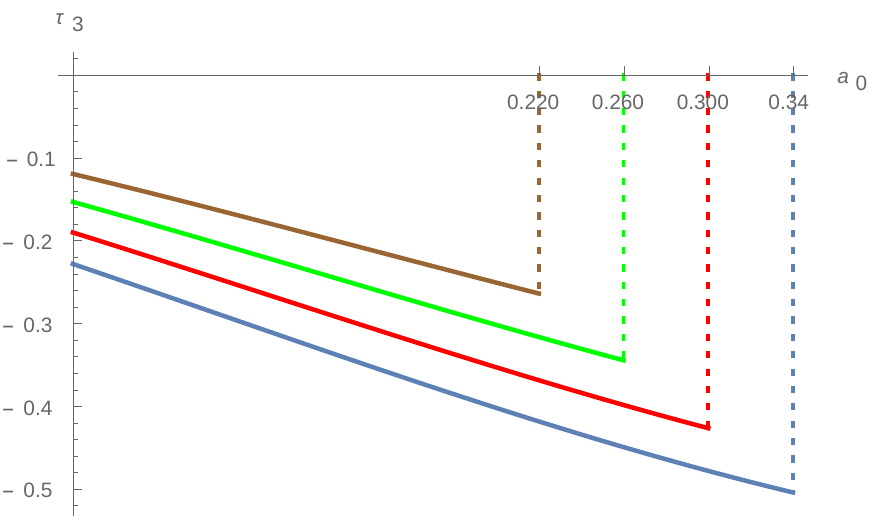}
 \end{overpic}
 \caption{Plot of $\tau_3$ as a function of $a_0$ in the three-cut case for $r=8/50,10/50,12/50,14/50$ (from bottom to top). For these choices of $r$, the extremal values of 
$a_0$ are attained in the critical line $\Gamma_c$, so that $a_0$ ranges from $0$ to the correspond critical value $(1-2r)/2$, in the present case given by 
$0.34,0.3,0.26$ and $0.22$, respectively.}\label{figure_tau_3_three_cut_2}
\end{figure}

\begin{figure}[t]
  \begin{overpic}[scale=1]
  {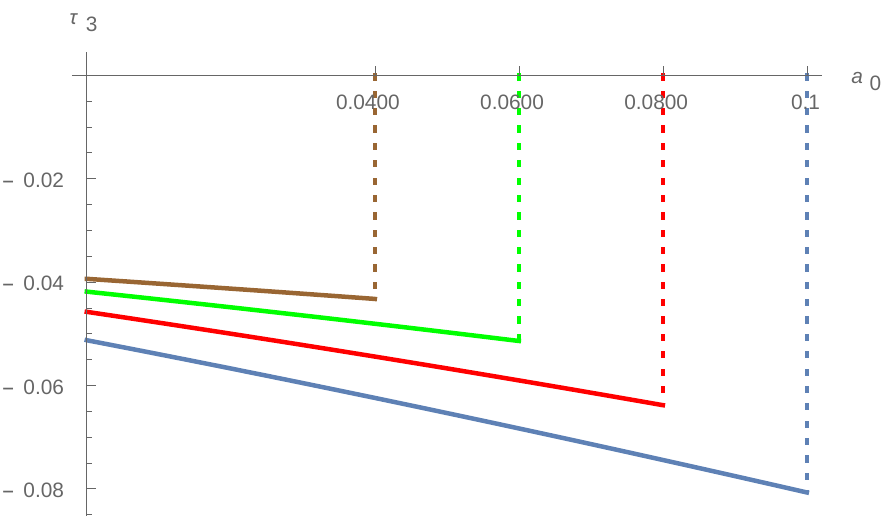}
 \end{overpic}
 \caption{Plot of $\tau_3$ as a function of $a_0$ in the three-cut case for $r=20/50,21/50,22/50,23/50$ (from bottom to top). For these choices of $r$, the extremal values of 
$a_0$ are attained in the critical line $\Gamma_c$, so that $a_0$ ranges from $0$ to the correspond critical value $(1-2r)/2$, in the present case given by 
$0.1,0.08,0.06$ and $0.04$, respectively.}\label{figure_tau_3_three_cut_3}
\end{figure}

\begin{figure}[t]
  \begin{overpic}[scale=1]
  {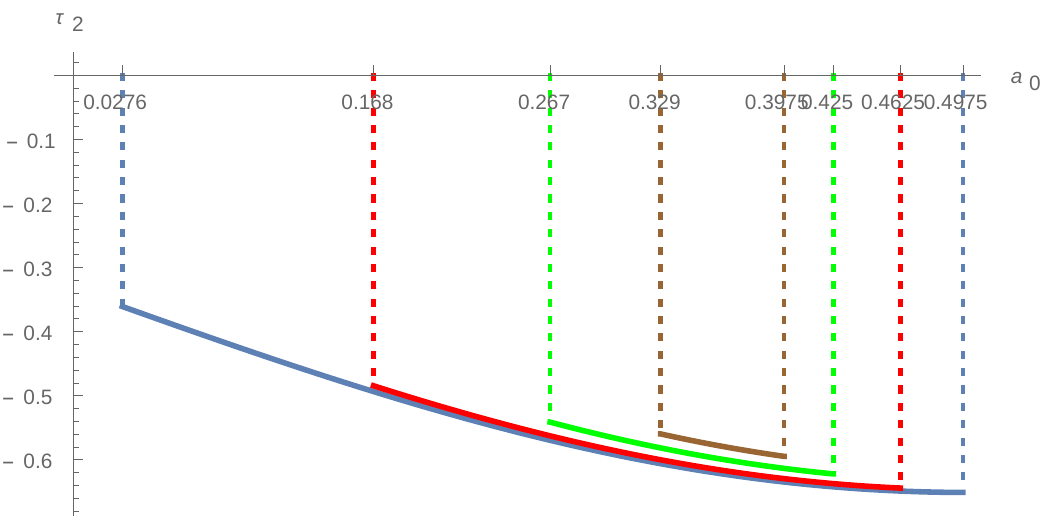}
 \end{overpic}
 \caption{Plot of $\tau_3$ as a function of $a_0$ in the one-cut case for $r=1/400,15/400,30/400,41/400$ (from bottom to top). For these choices of $r$, the minimal 
value for $a_0$ is along $\gamma_c$, thus given by $3r^{2/3}/3$, whereas the maximal value for $a_0$ is along $\Gamma_c$, hence corresponding to $(1-2r)/2$. In the present case, 
the minimal values are $0.0276..., 0.1680...,0.2667...$ and $0.32852...$, respectively, whereas the maximal values are $0.4975,0.4625,0.425$ and $0.3975$, 
respectively.}\label{figure_tau_3_one_cut_1}
\end{figure}

\begin{figure}[t]
  \begin{overpic}[scale=1]
  {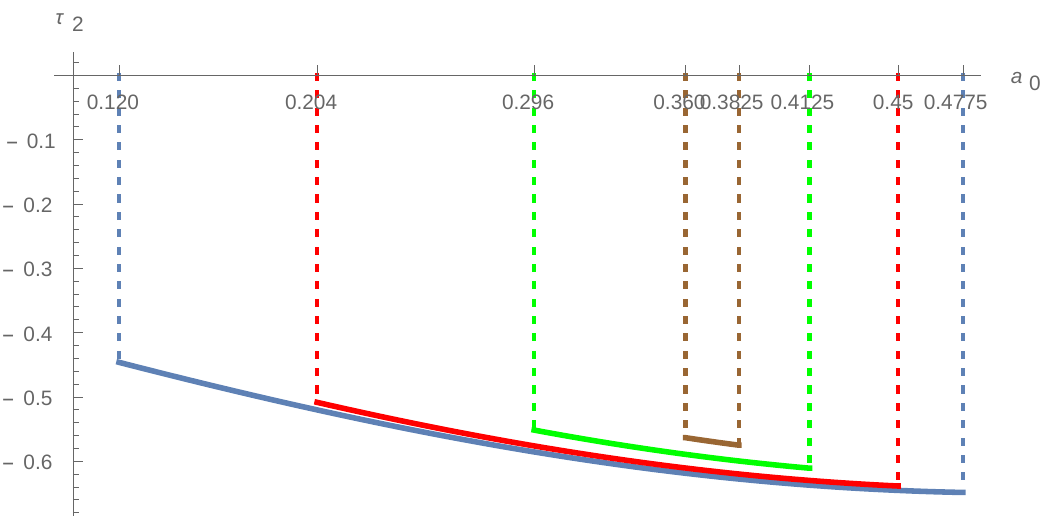}
 \end{overpic}
 \caption{Plot of $\tau_3$ as a function of $a_0$ in the one-cut case for $r=9/400,20/400,35/400,47/400$ (from bottom to top). For these choices of $r$, the minimal 
value for $a_0$ is along $\gamma_c$, thus given by $3r^{2/3}/3$, whereas the maximal value for $a_0$ is along $\Gamma_c$, hence corresponding to $(1-2r)/2$. In the present case, 
the minimal values are $0.1195..., 0.2035...,0.2956...$ and $0.3598...$, respectively, whereas the maximal values are $0.4775,0.45,0.4125$ and $0.3825$, 
respectively.}\label{figure_tau_3_one_cut_2}
\end{figure}

\clearpage


\section*{Acknowledgements}
We thank F.~Balogh, B.~Gustafsson, A.~Kuijlaars, S.-Y.~Lee, R.~Riser and E.~Saff for useful discussions. 

The first author was partially supported by the National Science Foundation project DMS-1265172. 

The second author was supported by FWO Flanders project G.0934.13 and the FP7 IRSES grant RIMMP Random and Integrable Models in Mathematical Physics. He 
gratefully acknowledges the hospitality of the Department of Mathematical Sciences, Indiana University Purdue University Indianapolis, where a substantial part of this work was 
carried out during his visits in April-2014 and October-2014.


\bibliographystyle{amsart}

\end{document}